\documentclass[%
 reprint,
 amsmath,amssymb,
 aps,
]{revtex4-2}
\pdfoutput=1

\usepackage[colorlinks=true, linkcolor=blue, citecolor=blue, urlcolor=blue]{hyperref}

\usepackage{amsmath}
\usepackage{xr}
\usepackage{graphicx}
\usepackage{refcount} 
\usepackage{chngcntr}
\usepackage{amsmath} 
\usepackage{amsfonts} 
\usepackage{amssymb} 
\usepackage{amsthm} 

\usepackage[T1]{fontenc}
\usepackage{graphicx}
\usepackage{amssymb,amsmath}
\usepackage{color}
\usepackage{units}
\usepackage{comment}
\usepackage{multirow}
\usepackage{makecell}
\usepackage{caption}  
\usepackage[table,xcdraw]{xcolor}

\newcommand{\refRoman}[1]{\uppercase\expandafter{\romannumeral\getrefnumber{#1}\relax}}
\usepackage{cleveref} 
\usepackage[dvipsnames]{xcolor}



\usepackage[most]{tcolorbox}

\newcommand{\BOX}[2]{%
\begin{tcolorbox}[
  float*=tp, floatplacement=tp,   
  breakable,                      
  width=\textwidth,               
  colback=gray!6,
  colframe=gray!60,
  boxrule=0.6pt,
  left=6pt,right=6pt,top=6pt,bottom=6pt,boxsep=0pt,
  title={#1}, fonttitle=\bfseries
]#2\end{tcolorbox}%
}

\newcommand{\NUMBOX}[2]{%
\begin{tcolorbox}[
  float*=tp, floatplacement=tp,
  breakable,
  width=\textwidth,
  colback=teal!4,                
  colframe=teal!60,
  boxrule=0.6pt,
  left=6pt,right=6pt,top=6pt,bottom=6pt,boxsep=0pt,
  title={#1}, fonttitle=\bfseries
]#2\end{tcolorbox}%
}
\RequirePackage[normalem]{ulem} 
\RequirePackage{color}\definecolor{RED}{rgb}{1,0,0}\definecolor{BLUE}{rgb}{0,0,1} 
\providecommand{\DIFdeltex}[1]{} 
\providecommand{\DIFaddbegin}{} 
\providecommand{\DIFaddend}{} 
\providecommand{\DIFdelbegin}{} 
\providecommand{\DIFdelend}{} 
\providecommand{\DIFaddFL}[1]{\DIFadd{#1}} 
\providecommand{\DIFdelFL}[1]{\DIFdel{#1}} 
\providecommand{\DIFaddbeginFL}{} 
\providecommand{\DIFaddendFL}{} 
\providecommand{\DIFdelbeginFL}{} 
\providecommand{\DIFdelendFL}{} 
\providecommand{\DIFadd}[1]{\texorpdfstring{\DIFaddtex{#1}}{#1}} 
\providecommand{\DIFdel}[1]{} 

\RequirePackage{settobox} 
\RequirePackage{letltxmacro} 
\newsavebox{\DIFdelgraphicsbox} 
\LetLtxMacro{\DIFOincludegraphics}{\includegraphics} 
\providecommand{\DIFaddincludegraphics}[2][]{
  \DIFOincludegraphics[#1]{#2}
}
\providecommand{\DIFdelincludegraphics}[2][]{
} 
\LetLtxMacro{\DIFOaddbegin}{\DIFaddbegin} 
\LetLtxMacro{\DIFOaddend}{\DIFaddend} 
\LetLtxMacro{\DIFOdelbegin}{\DIFdelbegin} 
\LetLtxMacro{\DIFOdelend}{\DIFdelend} 
\DeclareRobustCommand{\DIFaddbegin}{\DIFOaddbegin \let\includegraphics\DIFaddincludegraphics} 
\DeclareRobustCommand{\DIFaddend}{\DIFOaddend \let\includegraphics\DIFOincludegraphics} 
\DeclareRobustCommand{\DIFdelbegin}{\DIFOdelbegin \let\includegraphics\DIFOincludegraphics} 
\DeclareRobustCommand{\DIFdelend}{\DIFOaddend \let\includegraphics\DIFOincludegraphics} 
\LetLtxMacro{\DIFOaddbeginFL}{\DIFaddbeginFL} 
\LetLtxMacro{\DIFOaddendFL}{\DIFaddendFL} 
\LetLtxMacro{\DIFOdelbeginFL}{\DIFdelbeginFL} 
\LetLtxMacro{\DIFOdelendFL}{\DIFdelendFL} 
\DeclareRobustCommand{\DIFaddbeginFL}{\DIFOaddbeginFL \let\includegraphics\DIFaddincludegraphics} 
\DeclareRobustCommand{\DIFaddendFL}{\DIFOaddendFL \let\includegraphics\DIFOincludegraphics} 
\DeclareRobustCommand{\DIFdelbeginFL}{\DIFOdelbeginFL \let\includegraphics\DIFOincludegraphics} 
\DeclareRobustCommand{\DIFdelendFL}{\DIFOaddendFL \let\includegraphics\DIFOincludegraphics} 
\makeatletter 
\let\sout@orig\sout 
\renewcommand{\sout}[1]{\ifmmode\text{\sout@orig{\ensuremath{#1}}}\else\sout@orig{#1}\fi} 
\makeatother 
\RequirePackage{listings} 
\RequirePackage{color} 
\lstdefinelanguage{DIFcode}{ 
  moredelim=[il][\color{red}\sout]{\%DIF\ <\ }, 
  moredelim=[il][\color{blue}]{\%DIF\ >\ } 
} 
\lstdefinestyle{DIFverbatimstyle}{ 
	language=DIFcode, 
	basicstyle=\ttfamily, 
	columns=fullflexible, 
	keepspaces=true 
} 
\lstnewenvironment{DIFverbatim}[1][]{\lstset{style=DIFverbatimstyle}}{} 
\lstnewenvironment{DIFverbatim*}[1][]{\lstset{style=DIFverbatimstyle,showspaces=true}}{} 
\renewcommand{\DIFdeltex}[1]{}
\renewcommand{\DIFadd}[1]{#1}
\renewcommand{\DIFdel}[1]{}
\renewcommand{\DIFaddFL}[1]{#1}
\renewcommand{\DIFdelFL}[1]{}

\usepackage[most]{tcolorbox}

\newtcolorbox{blueframe}{
  enhanced,
  frame hidden,
  interior hidden,
  boxrule=0pt,
  boxsep=0pt,
  left=0pt,
  right=0pt,
  top=0pt,
  bottom=0pt,
  before skip=0pt,
  after skip=0pt,
  overlay={
    \draw[blue,line width=0.8pt]
    ([xshift=-4pt,yshift=3pt]frame.north west)
    rectangle
    ([xshift=4pt,yshift=-3pt]frame.south east);
  }
}

\begin{document}


\title{Transient Signal Propagation in Heterogeneous Networks under Various Inputs}

\author{Xiaoge Bao$^{1,2}$}
\author{Wei P. Dai$^{1,3,\dagger}$}
\author{Jan Nagler$^{4,5,\dagger}$}
\author{Wei Lin$^{1,2,3,6}$}

\affiliation{$^1$Research Institute of Intelligent Complex Systems, Fudan University, Shanghai, China}
\affiliation{$^2$Institute of Science and Technology for Brain-Inspired Intelligence, Fudan University, Shanghai, China}
\affiliation{$^3$Shanghai Artificial Intelligence Laboratory, Shanghai, China}
\affiliation{$^4$Centre for Human and Machine Intelligence, Frankfurt School of Finance \& Management, Frankfurt, Germany.}
\affiliation{$^5$Deep Dynamics, Frankfurt School of Finance \& Management, Frankfurt, Germany.}
\affiliation{$^6$State Key Laboratory of Medical Neurobiology and MOE Frontiers Center for Brain Science, Fudan University, Shanghai, China}

\email{$^\dagger$corresponding authors: \\
weidai@fudan.edu.cn;  jan.nagler@gmail.com}




\begin{abstract}
Transient signals in neural, biological and technological networks rarely follow a single route. Their strength and timing depend on where the input enters, on the intrinsic dynamics of each node visited, and on the feedforward and recurrent pathways that carry it. Previous theory has clarified important limits of signal propagation, but often for a single input type, averaged network structure, or asymptotic regimes. It has remained unclear how finite heterogeneous networks transform sustained, impulsive, finite-duration and stochastic inputs into node-resolved response timescales. Here, we develop a theory for network dynamics \textcolor{black}{near stable operating states} that answers this question. The framework identifies a canonical set of gains and timescales shared across deterministic and stochastic forcing and expresses them through matrix resolvents, spectra and directed walks. \textcolor{black}{In positive, sufficiently damped networks,} the walk representation gives a simple physical picture: response strength accumulates multiplicatively along a single path, whereas response time accumulates additively, through local intrinsic time constants. Competing pathways, feedback and inhibition further reshape propagation by amplifying, suppressing, reversing or delaying signals. \textcolor{black}{We derive perturbative criteria that specify when these signatures, and the resulting node rankings, persist beyond the positive regime.} Applied to a Wilson--Cowan model of the primate cortex, the theory links structural connectivity, excitatory--inhibitory balance and local biophysical timescales to hierarchical response timescales. \textcolor{black}{Our results provide a mechanistic framework for input-dependent timescales in finite heterogeneous networks.}
\end{abstract}

\maketitle


\section{{INTRODUCTION}}
A localized perturbation spreads through a network along pathways shaped by its connections and local dynamics. The strength and timing of the response vary across nodes and depend on where and how the input is applied~\cite{barzel2013universality,ji2023signal,timme2019propagation,hens2019spatiotemporal}. These dependencies underlie neural communication~\cite{wang2019hierarchical,chaudhuri2015large}, epidemic intervention~\cite{pastor2015epidemic,iannelli2018reaction}, and the response of infrastructure networks to disturbances~\cite{auer2016impact,wang2022measurement}. Predicting these responses at individual nodes requires understanding how successive interactions combine between a source and a target.

Several analytical studies relate these responses to network structure. Scaling theories describe how nonlinear dynamics and degree heterogeneity shape propagation under sustained inputs~\cite{barzel2013universality,hens2019spatiotemporal,harush2017dynamic,bao2022impact}. For stable linear systems, moments of normalized impulse responses provide measures of response strength and time~\cite{wolter2018quantifying,schroder2019dynamic}. Covariance and linear-response methods describe fluctuations and their dependence on network paths and motifs~\cite{pernice2011structure,trousdale2012impact,grytskyy2013unified}. These approaches describe complementary aspects of propagation. \textcolor{black}{However, it remains unclear how detailed coupling weights and local dynamics combine to shape responses at the level of individual source--target pairs in a heterogeneous network. Deterministic and stochastic responses are also usually studied separately, leaving their relationship and shared structural basis unclear.}

{Brain networks provide a concrete setting for this problem, as measured interareal connectivity is weighted, directed, and regionally heterogeneous, with external inputs varying in both location and type}~\cite{nozari2024macroscopic,wang2019hierarchical,chaudhuri2015large,demirtacs2019hierarchical}. These structures support communication between regions whose local dynamics involve distinct biophysical time constants and recurrent excitatory and inhibitory interactions~\cite{avena2018communication,li2022hierarchical,joglekar2018inter}. The resulting neural timescales {are demonstrated by} the duration of evoked responses and the temporal correlations of ongoing activity. {Relating these observations to circuit mechanisms requires distinguishing local time constants from response timescales shaped by network interactions and external inputs, and determining how these factors jointly govern the response of a given region}~\cite{zeraati2026neural}.

Here, we develop an analytical framework for transient propagation in finite heterogeneous networks with stable linear dynamics, including local linearizations of nonlinear systems near stable operating states. We characterize responses to constant, pulse, square, and white-noise inputs through a common set of node-level strength and temporal metrics. Our reference class consists of \textcolor{black}{Negative Strictly Diagonally Dominant Metzler (NSDDM) networks}, in which local decay dominates the total incoming nonnegative coupling. Positive inputs then produce nonnegative responses, and step responses are monotonic. This setting gives normalized response moments a direct interpretation and provides a controlled basis for extensions to signed interactions. We obtain gains and temporal moments exactly and use them to estimate peak times and decay times.

The structural interpretation follows from a Neumann-series expansion~\cite{kato2013perturbation,berman1994nonnegative,Chau2025ExactLinear,WuMichael2013Amif}. Each contribution corresponds to a directed walk from the input node to a target, weighted by the couplings and intrinsic time constants encountered along that walk. Along a single feedforward chain, local gains multiply and local time constants add. Multiple feedforward paths combine these contributions, and feedback introduces repeated visits to nodes. This representation makes the effects of pathway competition and recurrent amplification explicit. The different weighting of walks in strength and temporal metrics also explains their different sensitivity to longer propagation routes.

The input types are linked by the same source-to-target gain and moment-based time. The gain sets the steady response to constant forcing and the integrated response to a pulse {and square}. Under localized white-noise forcing, integrated covariance factorizes into products of these gains, and the moment-based cross-covariance lag is the difference between the corresponding source-to-target times. The stationary covariance additionally measures the temporal overlap of propagated responses. These relations connect deterministic and stochastic propagation while preserving their different temporal profiles.

\textcolor{black}{We extend the analysis to signed networks through motif calculations and perturbation bounds.} Inhibitory pathways can cancel or reverse the net response, and inhibitory feedback can produce overshoots or damped oscillations. The motif analysis identifies the contributions responsible for these effects, including parameter regimes where ratio-based temporal metrics become ill-defined. \textcolor{black}{For signed perturbations of an NSDDM system, we derive separate sufficient conditions for stability and bounds on changes in the response metrics. These bounds specify when the ordering of node responses is preserved.}

We {then} apply the framework to a Wilson--Cowan model of the primate cortex with heterogeneous interareal connectivity and local excitatory and inhibitory populations~\cite{chaudhuri2015large,li2022hierarchical,joglekar2018inter,li2025hierarchy}. The model operates beyond strict diagonal dominance while remaining stable. The response metrics capture the hierarchical temporal organization, and the predicted node rankings are strongly correlated with simulations across the tested inputs. Wilson--Cowan motifs identify how local biophysical time constants and excitatory and inhibitory pathways shape interregional propagation. Together, these analyses give a circuit-level account of input-dependent neural timescales.

{The paper is strcutured as follows: }Section~\ref{sec:modelintro} introduces the model. Section~\ref{sec:quantify} develops the response metrics and tests them on network ensembles and the cortical model. Section~\ref{sec:IV} derives the walk expansion and analyzes heterogeneous chains, motifs, and signed interactions. Section~\ref{sec:discussion} discusses the implications and scope of the results.

\section{Model Introduction}
\label{sec:modelintro}

\subsection{General formalism}
\label{subsec:general}

\par Consider a general dynamical system comprising $N$ interacting components $\mathbf{y}(t) = (y_1(t), \dots, y_N(t))^\top$ governed by  
\begin{equation*}  
\frac{{d}\mathbf{y}}{{d}t} = \mathbf{F}(\mathbf{y}),  
\end{equation*}  
where $\mathbf{F}(\mathbf{y})$ incorporates self-dynamics and pairwise interactions. Assuming the existence of the equilibrium state $\mathbf{y}^*$ that satisfies $\mathbf{F}(\mathbf{y}^*) = \mathbf{0}$, we can characterize local dynamics around the equilibrium using small perturbations $\mathbf{x}(t) \equiv \mathbf{y}(t) - \mathbf{y}^*$. Linearizing $\mathbf{F}$ around $\mathbf{y}^*$ gives the perturbation dynamics as \DIFaddbegin ${{d}\mathbf{x}}/{{d}t} = \mathbf{H}\mathbf{x}$
where $\mathbf{H} \equiv \nabla\mathbf{F}|_{\mathbf{y}^*}$ is the Jacobian matrix\DIFdelbegin \DIFdel{. }\DIFdelend \DIFaddbegin , and higher-order terms in the Taylor expansion of \mbox{
$\mathbf{F}(\mathbf{y})$
}
are neglected. Accordingly, the validity regime of this linearization, and the extent to which the linearized system captures the exact nonlinear dynamics, should be assessed on a case-by-case basis depending on the specific form of the nonlinear terms. The spectral properties of $\mathbf{H}$ provide a foundational framework for analyzing local stability \cite{meena2023emergent, strogatz2024nonlinear, hirsch2013differential, khalil2002nonlinear}, transient responses \cite{yang2023reactivity, neubert1997alternatives, wolter2018quantifying, schroder2019dynamic, hespanha2018linear}, mode decomposition \cite{wang2019hierarchical, brunton2016extracting, klus2020data}, and other critical features of nonlinear systems near equilibrium states \cite{sontag2013mathematical, franklin2002feedback, kuznetsov1998elements}.  

\par Building on the linearized system, we extend our analysis to systems exhibiting heterogeneity and variety: non-uniformity in $\mathbf{F}$, reflected in the detailed organization of the Jacobian $\mathbf{H}$, and spatiotemporally various \DIFaddbegin external \DIFaddend inputs $\mathbf{I}(t)$\DIFaddbegin \DIFadd{~\mbox{
\cite{tyloo2018robustness, ji2023signal,khalil1996robust, chaudhuri2015large}}\hskip0pt
}\DIFaddend . These motivate the generalized driven linear system  
\begin{equation}\label{main_eq}  
\frac{{d}\mathbf{x}}{{d}t} = \mathbf{H} \mathbf{x} + \mathbf{I}(t)\DIFdelbegin \DIFdel{,  
}\DIFdelend \DIFaddbegin \DIFadd{. 
}\DIFaddend \end{equation}  
\DIFdelbegin \DIFdel{where \mbox{
$\mathbf{I}(t)$
}
denotes external inputs that may vary across time and components \mbox{
\cite{tyloo2018robustness, ji2023signal,khalil1996robust, chaudhuri2015large}}\hskip0pt
. }\DIFdelend We focus on systems for which all eigenvalues of $\mathbf{H}$ have negative real parts, \DIFaddbegin {the }\textit{Hurwitz stability}\DIFadd{, }\DIFaddend ensuring that $\mathbf{H}$ is invertible and the dynamics are asymptotically stable. The interaction structure is represented by a weakly connected directed weighted graph, characterized by the irreducibility of
$|\mathbf{H}_{\mathrm{off}}|+|\mathbf{H}_{\mathrm{off}}|^\top$,
where $\mathbf{H}_{\mathrm{off}}$ contains the off-diagonal entries of $\mathbf{H}$. Systems that fail this criterion split into disjoint, connected subgraphs, which are then subjected to respective individual analysis. Four frequently encountered input types are examined in our framework: constant, pulse, square, and white noise, predominantly applied to individual nodes. Extensions to multi-node inputs are also possible.

\par This generalized model is applicable to both linear systems (e.g., linear compartmental models \cite{anderson2013compartmental}, leaky continuous-time Markov chains~\cite{norris1998markov}) and nonlinear systems that are linearized near equilibrium states. Under white noise input, it reduces to a Gaussian linear process \cite{seeger2004gaussian}. Despite the linearity of the model, typical transient-response metrics
remain analytically intractable in structurally heterogeneous networks.
Heterogeneity in $\mathbf H$ generates weighted combinations of multiple
dynamical modes, while different inputs $\mathbf I(t)$ produce distinct
temporal response profiles. Determining quantities such as response
peaks and times therefore generally requires solving transcendental
equations involving sums of exponentials, precluding closed-form
expressions. To address this, we develop a framework that systematically accounts for structural heterogeneity across inputs, focusing on spatial granularity (mesoscale finite-size networks with asymmetric and weighted $\mathbf{H}$, where mean-field approximations break down \cite{barzel2013universality, bao2022impact, hens2019spatiotemporal,harush2017dynamic}) and temporal granularity (finite-time dynamics in addition to $t \to 0$ or $t \to \infty$ asymptotics \cite{ren2005consensus, olfati2004consensus, neubert1997alternatives, yang2023reactivity}). These considerations motivate two central questions answered in our framework:
(Q1) Can transient responses $\mathbf{x}(t)$ be quantitatively characterized in finite-size, structurally heterogeneous networks governed by $\mathbf{H}$ under different input types $\mathbf{I}(t)$?
(Q2) Can we understand how pairwise interactions, encoded in the off-diagonal elements of $\mathbf{H}$, and self-dynamics, encoded in its on-diagonal elements, jointly shape transient responses $\mathbf{x}(t)$ under different input types $\mathbf{I}(t)$, and can common interpretable relationships be identified across these inputs?
\subsection{Special formalism}
\label{subsec:special}
\DIFdelbegin 
A widely studied realization within the general framework is the \emph{Negative Strictly Diagonally Dominant Metzler} (NSDDM) structure \cite{gao1992criteria,kushel2021generalization,liu2010dominant,doroslovavcki2023matrices,sootla2017block,li2021subdirect,horn2012matrix}, defined by $\mathbf{H}\equiv \mathbf{A}-\mathbf{D}-\operatorname{diag}(\beta_i)$. Here $\mathbf{A}$ is the weighted adjacency matrix of a directed network: for $i\neq j$, $A_{ij}\ge 0$ denotes the non-negative connection strength from node $j$ to node $i$, while $A_{ii}=0$. The diagonal matrix $\mathbf{D}=\operatorname{diag}\!\left(\sum_{j=1}^N A_{ij}\right)\equiv \operatorname{diag}(D_i)$ collects the nodal in-degrees, and $\beta_i>0$ denotes the self-decay rate.

Under this construction, $\mathbf{H}$ is a Metzler matrix in the sense that all off-diagonal entries satisfy $H_{ij}\ge 0$ for $i\neq j$, corresponding to non-negative network couplings. In addition, the diagonal entries satisfy $H_{ii}=-(\beta_i+D_i)<0$, providing self-stability when interactions are absent. Moreover, $|H_{ii}|=\beta_i+D_i>\sum_{j\neq i}H_{ij}=D_i$, so $\mathbf{H}$ is strictly diagonally dominant. These properties together ensure that the Gershgorin discs of $\mathbf{H}$ are contained strictly in the left half-plane, which in turn implies Hurwitz stability \cite{horn2012matrix,trefethen2022numerical,strang2000linear,robinson1975stability,guglielmi2018closest,khalil2002nonlinear,hespanha2018linear}.

Substituting $\mathbf{H}$ into Eq.~(\ref{main_eq}) and considering an input applied to node $m$ yields
\begin{equation}
\begin{aligned}
\frac{d x_i}{d t}
&= -\beta_i x_i+\sum_{j=1}^N A_{ij}(x_j-x_i)+\delta_{im} I_m(t) \\
&= \underbrace{-(\beta_i+D_i)x_i}_{\text{self-dynamics}}
+\underbrace{\sum_{j=1}^N A_{ij}x_j}_{\text{interactions}}
+\underbrace{\delta_{im} I_m(t)}_{\text{inputs}}.
\label{eq:NSDD}
\end{aligned}
\end{equation}
In this form, the NSDDM condition states that, at each node, the intrinsic dissipative self-dynamics dominate the total positive coupling received from the network. That is, the intrinsic time constant associated with the diagonal dynamics, $\gamma_i \equiv (\beta_i + D_i)^{-1}$, satisfies $\gamma_i^{-1} > \sum_j A_{ij}$.

\par This structure serves as an analytically tractable foundation for the following reasons. (i) It guarantees stability and well-behaved transient dynamics, such as positivity and monotonicity, enabling rigorous \DIFaddbegin {theoretical }\DIFaddend analysis (Appendix~\ref{sec:A}). (ii) This model architecture and its extensions apply broadly across disciplines, from neuroscience to physiology. In many such systems, positive couplings have clear physical interpretations, representing excitation, activation, diffusion, consensus, or other cooperative interactions~\cite{wang2019hierarchical,olfati2007consensus,newman2018networks,haddad2005stability,pecora1998master,skardal2015control}. (iii) The term $\sum_j A_{ij}(x_j - x_i)$ represents a diffusion process on the network governed by the graph Laplacian, analogous to spatial diffusion in continuous media \cite{newman2018networks}. (iv) The explicit separation of self-dynamics\DIFdelbegin \DIFdel{\mbox{
$(-(\beta_i+D_i) x_i)$
}
), pairwise interactions from the network structure \mbox{
$(A_{ij})$
}
, }\DIFdelend \DIFaddbegin {, \mbox{
$H_{ii}=-\gamma_i^{-1}$
}
, pairwise interactions, \mbox{
$H_{ij}=A_{ij}$
}
, }\DIFaddend and external inputs\DIFdelbegin \DIFdel{\mbox{
$(I_i(t))$
}
}\DIFdelend \DIFaddbegin {, \mbox{
$I_m(t)$
}
, }\DIFaddend provides a unified template for \DIFaddbegin {addressing (Q2) and }\DIFaddend comparing linearized nonlinear systems.
It also facilitates extending the theory to the more general formalism in Eq.~\eqref{main_eq}, especially including inhibitory couplings, and clarifies the limitations of the theory, when the sign constraints on $A_{ij}$ and $\beta_i$ are relaxed~\cite{wolter2018quantifying,schroder2019dynamic}.

\DIFdelbegin 
\DIFdelend \DIFaddbegin \BOX{Box 1.
Intuition from a one-dimensional system under Hurwitz stability}{
A more general assumption than NSDDM is Hurwitz stability~\cite{khalil2002nonlinear, hespanha2018linear}, namely, that dominant eigenvalues of \mbox{
\(\mathbf{H}\)
}
have strictly negative real parts, \mbox{
\(\max_j \operatorname{Re}(\lambda_j)<0\)
}
. Under this assumption, the system \mbox{
\(\dot{ \mathbf{x}}=\mathbf{H}\mathbf{x}\)
}
forgets initial conditions exponentially fast \mbox{
\((e^{\mathbf{H}t}\!\to 0)\)
}
; \mbox{
\(\mathbf{H}\)
}
 is invertible; and steady-state quantities, such as the unique solution of the continuous-time Lyapunov equation, are well defined.

This assumption also clarifies why the steady-state (step) gain is the resolvent at zero frequency, \mbox{
\(\!-\mathbf{H}^{-1}\)
}
. For a constant input \mbox{
\(\mathbf{I}\)
}
, the equilibrium satisfies \mbox{
\(0=\mathbf{H}\mathbf{x}_{\mathrm{ss}}+\mathbf{I}\)
}
, hence \mbox{
\(\textbf{x}_{\mathrm{ss}}=-\mathbf{H}^{-1}\mathbf{I}\)
}
. Equivalently,
\[
\int_{0}^{\infty} e^{\mathbf{H}s}\,{d}s \;=\; -\,\mathbf{H}^{-1},
\]
so the long-time effect of a sustained drive is mediated by \mbox{
\(\!-\mathbf{H}^{-1}\)
}
, the fundamental strength gain shared across input types.

\emph{Intuition in 1d.} For \mbox{
\(\dot x=-\beta x + I(t)\)
}
with \mbox{
\(\beta>0\)
}
(so \mbox{
\(\mathbf{H}=-\beta\)
}
is Hurwitz), a step \mbox{
\(I(t)=I\,\mathbf{1}_{t\ge 0}\)
}
yields
\[
x(t)=\Bigl(x_0-\frac{I}{\beta}\Bigr)e^{-\beta t}+\frac{I}{\beta},
\qquad x_{\mathrm{ss}}=\frac{I}{\beta}=\!-\mathbf{H}^{-1}I.
\]
An impulse-like probe \mbox{
\(I(t)=I\,\delta(t)\)
}
has impulse response \mbox{
\(h(t)=I\,e^{-\beta t}\)
}
for \mbox{
\(t\ge 0\)
}
, whose {area} equals the step gain:
\[
\int_{0}^{\infty} h(t)\,{d}t \;=\; \frac{I}{\beta} \;=\; \!-\mathbf{H}^{-1}I.
\]
A square of duration \mbox{
\(t_s\)
}
, \mbox{
\(I(t)=I\,\mathbf{1}_{0\le t\le t_s}\)
}
, interpolates between impulse- and step-like behavior: as \mbox{
\(t_s\to 1\)
}
the response is impulse-like with amplification \mbox{
\(I/
\beta\)
}
, while as \mbox{
\(t_s\to\infty\)
}
the peak saturates to \mbox{
\(I/\beta\)
}
. Under white-noise input $u(t)=I\xi(t)$, the autocovariance is $C(\tau)=\frac{I^2}{2\beta}e^{-\beta|\tau|}$, yielding a peak $I^2/(2\beta)$, an integrated area $I^2/\beta^2$, with the same decay time constant $1/\beta$.

}
\DIFaddend

\section{GENERIC QUANTIFICATION ACROSS INPUTS}
\label{sec:quantify}

\DIFaddbegin 

\DIFaddend \par Building on the central \DIFdelbegin \DIFdel{questions of how networks transform inputs into responses, we }\DIFdelend \DIFaddbegin {question of how to quantify transient propagation dynamics systematically (Q1), we first }\DIFaddend establish a general mathematical framework \DIFdelbegin \DIFdel{connecting }\DIFdelend \DIFaddbegin {that connects }\DIFaddend input properties to node-resolved response measures across network structures\DIFdelbegin \DIFdel{. This section achieves three interconnected advances: }\DIFdelend \DIFaddbegin {, based on the spectral or matrix-inverse methods, with a glossary of terms and metrics provided in Appendix~\ref{sec:glossary}. The following subsections deliver several interconnected advances.
}
\par First, we derive exact analytical expressions for canonical transient response metrics, including amplification, peak response, time constant, and response time, under four input types (Fig.~\ref{fig:one}; Appendix~\ref{sec:B}), and illustrate the key intuition in a one-dimensional Hurwitz-stable system (Box~1). Second, we validate these expressions across a broad range of network structures within the NSDDM regime, demonstrating robust quantitative accuracy across heterogeneities in topology, coupling weights, self-dynamics, and spectral properties (Fig.~\ref{fig:two}; SM Sec.~I). We further clarify the theoretical limitations for each input type (Appendix~\ref{sec:B}) and systematically assess applicability beyond the NSDDM regime numerically (SM Sec.~II) and theoretically (Appendix~\ref{sec:A}). Third, we uncover relations among these metrics across input types, revealing shared organizing principles of transient responses and identifying common basic units of strength and timing; we also specify the applicable regime of the theory (see Box~2). Finally, we demonstrate the applicability of the framework in a biologically motivated and strongly heterogeneous Wilson--Cowan-type model of the primate cortex. Notably, the theoretical predictions remain in strong agreement with simulations even outside the NSDDM regime (Fig.~\ref{fig:newadd}; Appendix~\ref{sec:F}).
\begin{figure*}
\includegraphics[width=1.0\linewidth]{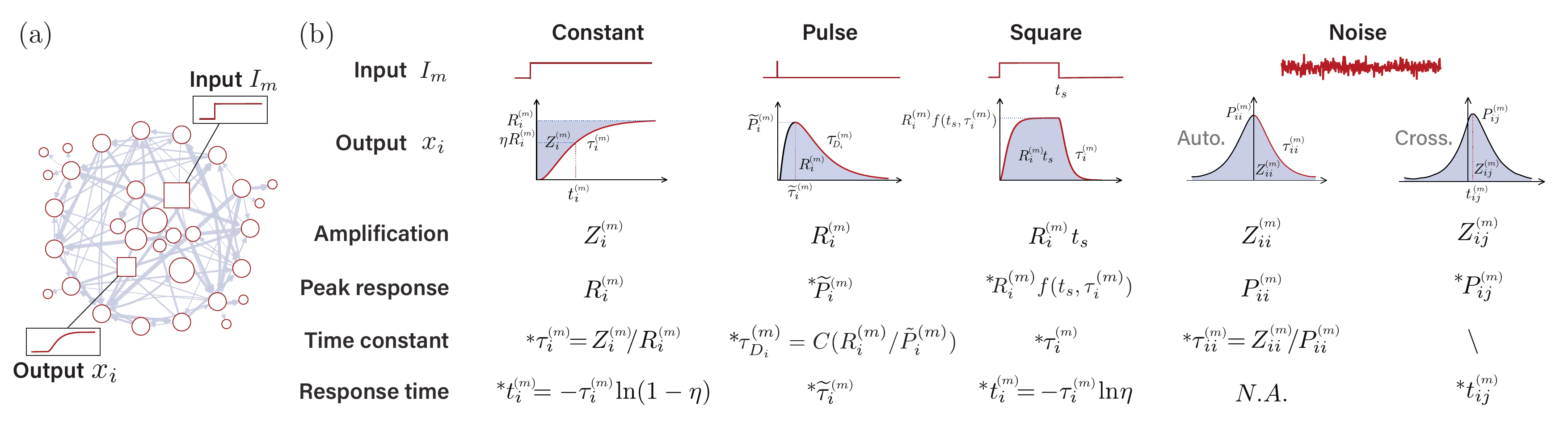}
\vspace{-0.8em}
\caption{\label{fig:one}
\textbf{Response metrics across localized input types.}
(a) Schematic of node-to-node propagation under constant input. A single node $m$ receives a constant input, and the resulting transient responses are quantified for any single node $i$ in a heterogeneous network. (b) Summary of metrics quantifying node $i$ response across four input types. Response strength is quantified by amplification and peak response, while response time is characterized by time constants and response times. Metrics marked with asterisks ($*$'s) indicate approximations derived from transcendental equations, as detailed in Appendix~\DIFdelbeginFL \DIFdelFL{\ref{sec:A}}\DIFdelendFL \DIFaddbeginFL \DIFaddFL{\ref{sec:B}}\DIFaddendFL . Constant input (first column): Amplification ($Z_{i}^{(m)}$, blue area above curve), peak response ($R_{i}^{(m)}$, maximum amplitude), time constant ($\tau_{i}^{(m)}$, rise to $(1 - 1/e)R_{i}^{(m)}$), and relative propagation time ($t_{i}^{(m)}$, time to $\eta R_{i}^{(m)}$). Pulse input (second column):  Amplification ($R_{i}^{(m)}$, blue area under curve), peak response ($\widetilde{P}_{i}^{(m)}$, maximum amplitude), decay time constant ($\tau_{D_{i}^{(m)}}$, drop to $\widetilde{P}_i^{(m)}/e$), and time to peak ($\tilde{\tau}_{i}^{(m)}$). Square input (third column):  Amplification ($R_{i}^{(m)} t_s$, blue area under curve), peak response ($R_{i}^{(m)} f(t_s, \tau_{i}^{(m)})$), decay time constant ($\tau_{i}^{(m)}$, drop to $1/e$ of peak), and response time ($t_i^{(m)}$). White noise input (last two columns): Autocovariance ($i = j$): Amplification ($Z_{ii}^{(m)}$), peak response ($P_{ii}^{(m)}$, zero-lag), time constant ($\tau_{ii}^{(m)}$, decay to $P_{ii}^{(m)}/e$). Crosscovariance ($i \ne j$): Amplification ($Z_{ij}^{(m)}$), peak response ($P_{ij}^{(m)}$), and peak response time ($t_{ij}^{(m)}$, time to peak).}
\end{figure*}

\subsection{Constant input}
\label{subsec:constant}

\par We start by analyzing the system's response to constant inputs modeled as the Heaviside step functions $\mathbf{I}_0^{\text{const}}(t)$, a widely-used approach for studying signal propagation in large-scale systems ~\cite{hens2019spatiotemporal, bao2022impact, barzel2013universality}. Full-time course of node $i$ receiving the constant input based on the steady-state $x^{\text{const}}_i(0)$ is $\Delta x^{\text{const}}_i(t) \equiv x^{\text{const}}_i(t) - x^{\text{const}}_i(0) = \left[\mathbf{H}^{-1} \left( e^{\mathbf{H}t} - \mathbf{I}_N \right) \mathbf{I}_0^{\text{const}}(t) \right]_i$
where $\mathbf{I}_N$ is the identity matrix. Under stability, the trajectory $\Delta x^{\text{const}}_i(t)$ converges to its final steady state denoted as $\Delta x^{\text{const}}_i(\infty)$. We characterize the response of \DIFaddbegin {target }\DIFaddend node $i$ to a single-node constant input $I_0^{\text{const}}$ at \DIFaddbegin {source }\DIFaddend node $m$ using four dynamical metrics (Fig.~\ref{fig:one}): (i) \textit{amplification} $Z_{i}^{(m)}$, defined as the area between the response curve and its final steady state; (ii) \textit{peak response} $R_{i}^{(m)}$, the maximum amplitude (with only positive couplings) or the signed steady-state response (with negative couplings); (iii) \textit{time constant} $\tau_{i}^{(m)}$, the time to reach $(1 - 1/e)$ of increasing $\Delta x^{\text{const}}_i(t)$; and (iv) \textit{relative propagation time} $t_{i}^{(m)}$, the time when $\Delta x^{\text{const}}_i(t_{i}^{(m)})/\Delta x^{\text{const}}_i(\infty) = \eta$~\cite{hens2019spatiotemporal, bao2022impact, kittel2017timing, chen2018estimating,iannelli2018reaction, gautreau2007arrival, gautreau2008global,newman2002spread}, where fraction $\eta$ is the given threshold. Analytical expressions of these metrics are computed and estimated, respectively, as:
\begin{align}
Z_{i}^{(m)} &\triangleq [\mathbf{H}^{-2}]_{im} I_0^{\text{const}} \sim O(1/\lambda_1^2), \label{eq:const:Zim} \\[5pt]
R_{i}^{(m)} &\triangleq -[\mathbf{H}^{-1}]_{im} I_0^{\text{const}} \sim O(1/\lambda_1), \label{eq:const:Rim} \\[5pt]
* \tau_{i}^{(m)} &\triangleq -\frac{[\mathbf{H}^{-2}]_{im}}{[\mathbf{H}^{-1}]_{im}} \sim O(1/\lambda_1), \label{eq:const:tauim} \\[5pt]
* t_{i}^{(m)} &\triangleq -\tau_{i}^{(m)} \ln(1 - \eta) \sim O(1/\lambda_1). \label{eq:const:tim}
\end{align}
Here, $[\cdot]_{im}$ denotes the $(i,m)$ matrix element and $\lambda_1 \equiv \max_j \operatorname{Re}(\lambda_j)$ is the dominant eigenvalue of $\mathbf{H}$. Metrics marked with asterisks ($*$'s) indicate approximations derived from transcendental equations.  \DIFdelbegin \DIFdel{These scaling relations with \mbox{
$\lambda_1$
}
govern system-level response strength and time. }\DIFdelend 
Corresponding spectral decomposition is
\begin{equation*}
\left[\mathbf{H}^{-k}\right]_{im} = \sum_{j=1}^N \frac{u_{im}^{j}}{\lambda_j^k}, \quad
u_{im}^{j} \equiv [\mathbf{U}]_{ij}[\mathbf{U}^{-1}]_{jm},
\end{equation*}
where $k$ is an integer; \( \mathbf{U} \) is the eigenmatrix diagonalizing \( \mathbf{H} \) \((\mathbf{H} = \mathbf{U} \mathbf{\Lambda} \mathbf{U}^{-1},\ \mathbf{\Lambda} = \mathrm{diag}(\lambda_j))\), establishing the metrics via eigenmode projections.

\par Among these metrics, amplification ($Z_{i}^{(m)}$) and peak response ($R_{i}^{(m)}$) admit analytical expressions, whereas the temporal metrics ($\tau_{i}^{(m)}$ and $t_{i}^{(m)}$, marked by asterisks in Fig.~\ref{fig:one}) are estimated under the assumption that the residual response, $\Delta x_i^{\text{const}}(\infty)-\Delta x_i^{\text{const}}(t)$, can be well approximated by an exponential- or sigmoidal-like decay governed by the effective time constant $\tau_{im}$. Full derivations are provided in Appendix~\ref{sec:B}. These metrics are well-defined in NSDDM systems, where positive constant inputs yield strictly positive, monotonic responses, ensuring temporal solution uniqueness. The metric sign conventions are also rigorously maintained: $\left[\mathbf{H}^{-1}\right]_{im} < 0$ and $\left[\mathbf{H}^{-2}\right]_{im} > 0$ for all reachable node pairs $(i,m)$, while $\left[\mathbf{H}^{-k}\right]_{im} = 0$ (for all integers $k \geq 1$) when no path exists from $m$ to $i$ (Appendix~\ref{sec:A}).
Numerical validation across diverse network topologies in NSDDM systems demonstrates high accuracy in metric estimation, with strong Spearman's rank correlations between analytical and simulated values (Fig.~\ref{fig:two}(a); See Supplementary Material (SM) Sec. I). 

Extensions to more general forms reveal robustness: the percentage of nodes with moderate accuracy ($<100\%$) remains above $80\%$ even when $20\%$ of the connections are inhibitory ($A_{ij} < 0$) in large, sparse networks operating near the stability boundary (see SM Fig. S36). Although introducing inhibitory coupling no longer guarantees positive, monotonic responses or well-defined temporal metrics, our analysis still captures the diverse roles of inhibition in shaping transient responses (Figs.~\ref{fig:eight} and~\ref{fig:nine}), as further discussed in Discussion Sec.~\ref{sec:discussion}. Theoretical limitations are also addressed in Appendix~\ref{sec:B:constant}. Compared with the other input types considered below, constant inputs generate more well-behaved monotonic transient responses, which makes the associated metrics more robust to estimate and more amenable to theoretical analysis. Constant drives therefore serve as the most practical probe of the underlying system structure.

\par The derived metrics exhibit two distinct scaling relations with the input amplitude: strength metrics ($Z_{i}^{(m)}$, $R_{i}^{(m)}$) scale linearly with input amplitude ($Z_{i}^{(m)}, R_{i}^{(m)} \propto I_0^{\mathrm{const}}$), while temporal metrics ($\tau_{i}^{(m)} = Z_{i}^{(m)}/R_{i}^{(m)}$) remain invariant. This fundamental distinction enables separate structural interpretations: strength metrics quantify absolute intensities and input amplitude, whereas temporal metrics characterize relative efficiency. 

\subsection{Pulse input}
\label{subsec:pulse}
\par Understanding the impulse response of a linear time-invariant (LTI) system is fundamental for characterizing its transient dynamics, as the response to any input can be derived through its convolution with the system impulse response \cite{oppenheim1997signals}. Response to the Dirac delta input $\delta \mathbf{I}(t)$ with the total impulse strength $\mathbf{I}_0^{\mathrm{pulse}} \equiv \int_0^\infty \delta\mathbf{I}(t)dt$ is governed by $\Delta x_i^{\mathrm{pulse}}(t) \equiv  x_i^{\mathrm{pulse}}(t) - x_i^{\mathrm{pulse}}(0) = \left[e^{\mathbf{H}t}\mathbf{I}_0^{\mathrm{pulse}}\right]_i$. The response to a single-node pulse input $I_0^{\text{pulse}}$ applied at node $m$ has already been systematically characterized in prior works~\cite{wolter2018quantifying, schroder2019dynamic}. Building on these studies, we adopt the same set of metrics to quantify the temporal and strength properties of the response, as illustrated in the second column of Fig.~\ref{fig:one}(b). Additionally, we derive an estimation for the decay time constant. The four metrics are given, respectively, by: (i) \textit{amplification} $R_{i}^{(m)}$, defined as the signed area under the response curve; (ii) \textit{peak response} $\widetilde{P}_{i}^{(m)}$, the maximum amplitude; (iii) \textit{decay time constant} $\tau_{D_{i}^{(m)}}$, the time to reach $1/e$ of the peak during the decay phase; and (iv) \textit{peak response time} $\widetilde{\tau}_{i}^{(m)}$, the time at which the peak occurs, where, more concretely,
\begin{align}
R_{i}^{(m)} &\triangleq -[\mathbf{H}^{-1}]_{im} I_0^{\text{pulse}} \sim O(1/\lambda_1),\label{eq:pulse:Rim} \\[5pt]
*\widetilde{P}_{i}^{(m)} &\triangleq C(d) P_{i}^{(m)}, \\
&= C(d) \frac{([\mathbf{H}^{-1}]_{im})^2 I_0^{\text{pulse}}}{\sqrt{2[\mathbf{H}^{-3}]_{im}[\mathbf{H}^{-1}]_{im} - ([\mathbf{H}^{-2}]_{im})^2}} \sim O(1), \label{eq:pulse:Pim} \\[5pt]
*\tau_{D_{i}^{(m)}} &\triangleq \left(1 - \frac{1}{e}\right) \frac{R_{i}^{(m)}}{\widetilde{P}_{i}^{(m)}} \sim O(1/\lambda_1), \label{eq:pulse:Dim}\\[5pt]
*\widetilde{\tau}_{i}^{(m)} &\triangleq \tau_{i}^{(m)}+ \frac{1}{\lambda_1}= -\frac{[\mathbf{H}^{-2}]_{im}}{[\mathbf{H}^{-1}]_{im}} + \frac{1}{\lambda_1} \sim O(1/\lambda_1). \label{eq:pulse:tauim}
\end{align}
Here, the bias term $C(d) = {(\sqrt{d+1} d^d)} /{(e^d d!)}$ ($ \approx1/\sqrt{2\pi}$ when $d$ is large) depends on the shortest path length $d$ between nodes $m$ and $i$. For a given network, $C(d)$ typically provides lower bound estimations on the simulated peak response, and setting $C(d)=1$ gives an upper bound numerically. Similarly, for $\widetilde{\tau}_{im}$, the bias term $(1/\lambda_1) < 0$ results in an lower-bound approximation~\cite{schroder2019dynamic}, and omitting this term gives practical upper bounds.

\begin{figure}
\includegraphics[width=1.02\linewidth]{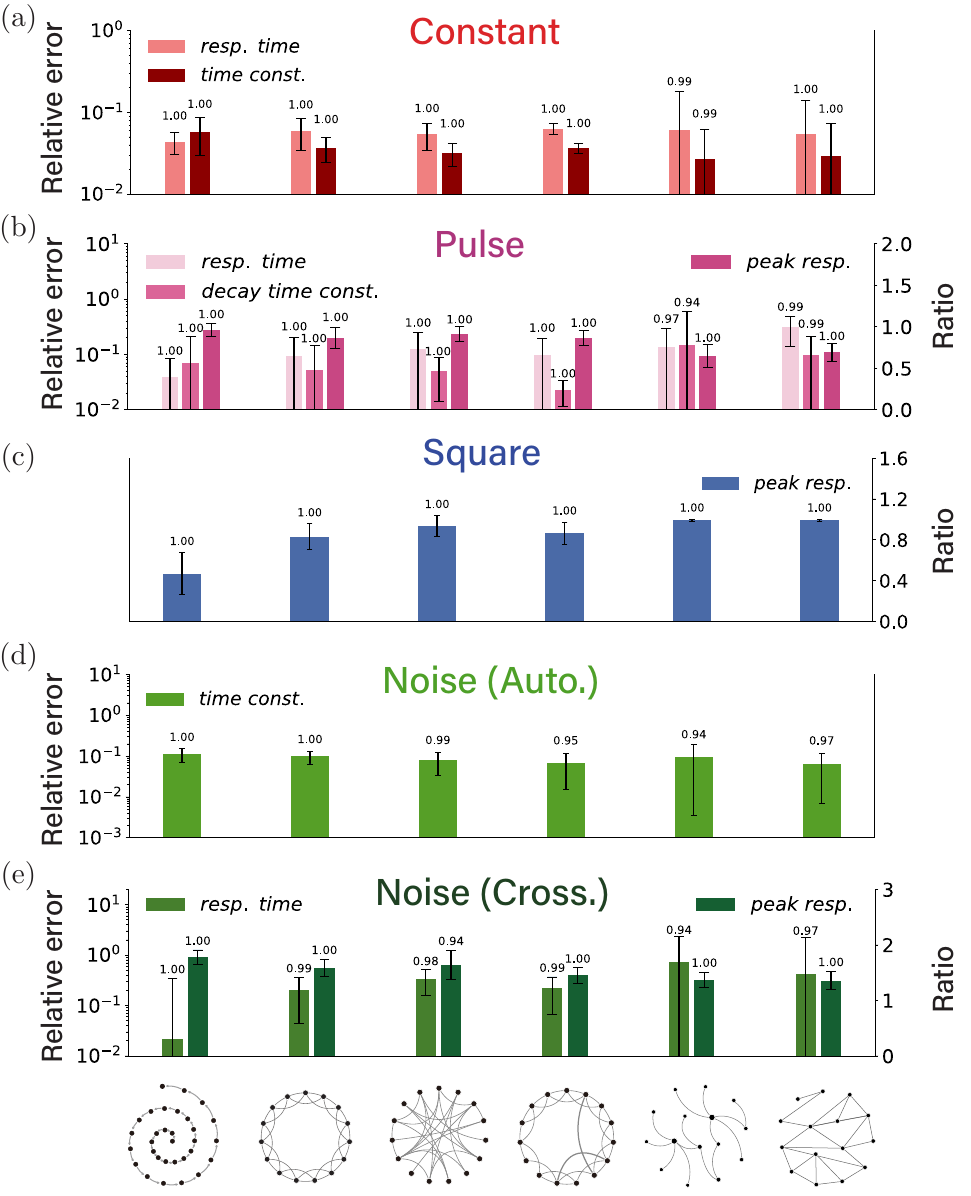}
\DIFaddendFL \caption{\label{fig:two}
\textbf{Validation of estimated response metrics across network topologies.}  
Network types from left to right:  
(i) Chain (directed for deterministic inputs; undirected for noise inputs),  
(ii) Regular lattice (average degree $\sim 4$),  
(iii) Erdős–Rényi (ER) random network (edge probability $\sim 0.02$\DIFaddbeginFL {; higher probabilities are investigated in SM Sec. I}\DIFaddendFL ),  
(iv) Small-world network (average degree $\sim 2$, rewiring probability $\sim 0.5$),  
(v) Scale-free network (preferential attachment parameter $\sim 1$), and  
(vi) Geometric network (connection radius $\sim 0.2$).  
Input types: (a) Constant; (b) Pulse; (c) Square; (d) Noise (for autocovariance); (e) Noise (for crosscovariance). The \DIFaddbeginFL {input white noise is \mbox{
$0$
}
mean and unit standard deviation. The }\DIFaddendFL abscissa (x-axis) is shared across all panels.
All networks contain $100$ nodes with uniform parameters: self-decay rate $\beta = 1$ and connection strengths set to $1$.  
Inputs are applied as follows: to the first node in the chain, randomly assigned in the lattice, and randomly assigned across $100$ independent instances for randomly generated networks (ER, small-world, scale-free, geometric). Time-related metrics (response time, time constant) use relative error $|t_{\text{sim}} - t_{\text{thr}}| / t_{\text{sim}}$ (left bars; $t_{\text{sim}}$: simulated, $t_{\text{thr}}$: theoretical). Strength-related metrics (peak response) use ratio $P_{\text{thr}} / P_{\text{sim}}$ (right bars). Error bars show mean $\pm$ variance across instances. Most time-related errors remain below $10^0$ ($100\%$), while strength ratios cluster near $\sim 1$ (within one order of magnitude), indicating consistent quantitative agreement.  
\DIFdelbeginFL \DIFdelFL{Gray labels }\DIFdelendFL \DIFaddbeginFL {Digits on the bar }\DIFaddendFL indicate Spearman rank correlation for node-wise ordering preservation, with values close to $1$ reflecting strong rank consistency. In chain-like or sparse networks, nodes whose shortest path from the input node is $\geq 15$ are excluded to avoid numerical artifacts caused by rapid response decay.  
}
\end{figure}

\par Pulse-response metrics build on established methods~\cite{wolter2018quantifying, schroder2019dynamic}, where normalized responses to the single-node pulse are interpreted as probability distributions. This framework ensures non-negative dynamics under positive pulse inputs in NSDDM systems, providing a well-grounded basis for interpretation (Appendix~\ref{sec:A}, \cite{wolter2018quantifying})\DIFdelbegin \DIFdel{. We extend prior works by defining the decay time constant \mbox{
$\tau_{D_{im}}$
}
through exponential assumptions (Appendix~\ref{sec:B}). 
}\DIFdelend \DIFaddbegin {, while introducing negative links breaks positivity for certain nodal responses and reduces predictive accuracy (see SM Figs.~S44 and S45; Appendix~\ref{sec:B:pulse}). 
}\DIFaddend Numerical validation across positive networks demonstrates great performance (Fig.~\ref{fig:two}(b); See SM Sec. I), with low relative error and strong rank correlations, particularly in sparse networks with weak network interactions ($\alpha/\lambda_1\ \rightarrow 0$ regime where $\alpha$ refers to identical connection strength \cite{schroder2019dynamic}). 

\par Pulse-input dynamics inherit properties from constant-input responses through their derivative relationship ($d\Delta x_i^{\text{const}}(t) /dt = \Delta x_i^{\text{pulse}}(t)$) under the same input location and amplitude: pulse amplification equals the peak response under constant input (denoted as $R_{i}^{(m)}$ in Eqs. (\ref{eq:const:Rim}) and (\ref{eq:pulse:Rim})). Temporal metrics share complementary interpretations: the constant-input time constant $\tau_{i}^{(m)}$ (Eq. (\ref{eq:const:tauim})) aligns with the pulse-input peak response time $\widetilde{\tau}_{i}^{(m)}$ (Eq. (\ref{eq:pulse:tauim})). In addition, for LTI systems, this equivalence extends to covariance \cite{oppenheim2017signals}: single-node pulse responses mirror the crosscovariance function with time-lag $s$: $\langle \mathbf{x}, \mathbf{I} \rangle = e^{\mathbf{H}s}\mathbf{Q}$ when spectral density matrix $\mathbf{Q}\equiv S_{\mathbf{I}}(w)$ (Fourier transform of the autocovariance function $\mathbb{E}\left[\mathbf{I}(t) \mathbf{I}(t+\tau)^{\top}\right]$) contains only a single non-zero diagonal element $I_0^{\text{pulse}}$ at node $m$ (Appendix~\ref{sec:B}). This existing mathematical equivalence enables direct comparison among inputs while preserving consistent interpretation.
\subsection{Square input}
\label{subsec:square}
\par Square inputs combine analytical simplicity with biological relevance, offering precise temporal control for modeling finite-duration stimuli in physiological experiments \cite{spall2000adaptive, billings2013nonlinear}. The full time course comprises two phases: (i) a driven rising phase corresponding to a constant-input response truncated at the input duration $t_s$, followed by (ii) a post-stimulus phase during which delayed propagation may cause the response to continue rising before it reaches its peak and subsequently decays. We characterize the decay phase by the time constant $\tau_{i}^{(m)}$, which measures the time it takes for the response to drop to $1/e$ of its initial value in response to a single-node square input $I_0^{\text{square}}$ at node $m$. We then derive two strength metrics: the \textit{amplification} $R_{{i}}^{(m)} t_{{s}}$ (representing total integrated response) and the \textit{peak response} $R_{{i}}^{(m)} f(t_{{s}}, \tau_{i}^{(m)})$ (quantifying maximum amplitude), where $R_{i}^{(m)} = -[\mathbf{H}^{-1}]_{{im}}I_0^{\text{square}}$ (see Eqs.~\eqref{eq:const:Rim} and \eqref{eq:pulse:Rim}).

Precisely,
\begin{align}
R_{{i}}^{(m)}t_s &\triangleq -t_s [\mathbf{H}^{-1}]_{{im}}I_0^{\text{square}}, \label{eq:square:amp}\\[5pt]
*R_{{i}}^{(m)}f(t_s,\tau_{{i}}^{(m)}) &\triangleq -(1 - e^{-t_s/\tau_{{i}}^{(m)}})[\mathbf{H}^{-1}]_{{im}}I_0^{\text{square}}.\label{eq:square:peak}
\end{align}
As such, systematic validation across NSDDM systems confirms metric robustness (Fig.~\ref{fig:one}(c)): amplification exhibits negligible error ($<0.1\%$, omitted for clarity), while peak responses achieve near-unity agreement ratios ($>0.8$) under typical topologies (Fig.~\ref{fig:two}(c); Appendix~\ref{sec:B}). For unit input duration ($t_s = 1$), Eq.~\eqref{eq:square:amp} establishes the equivalence linking impulse-integrated amplification (Eq.~\eqref{eq:pulse:Rim}) to constant-input peak response (Eq.~\eqref{eq:const:Rim}). Correspondence for peak response of unit duration (Eq.~\eqref{eq:square:peak}) extends to impulse-response peaks (Eq.~\eqref{eq:pulse:Pim}; See SM Fig.~S18). The asymptotic scaling $f(t_s, \tau_{i}^{(m)}) \sim t_s/\tau_{i}^{(m)}$ (for $t_s \to 0$) and $f(t_s, \tau_{i}^{(m)}) \to 1$ (for $t_s \to \infty$) emerges naturally from $\tau_{i}^{(m)}$-dominated dynamics, confirming time constant (Eq.~\eqref{eq:const:tauim}) as \DIFdelbegin \DIFdel{universal }\DIFdelend \DIFaddbegin {basic temporal }\DIFaddend regulators of transient dynamics.

\subsection{White noise input}
\label{subsec:noise}

\par White noise input, characterized by a flat power spectral density, serves as a fundamental tool to probe broadband system dynamics (e.g., in neural processing \cite{chaudhuri2015large, joglekar2018inter}). \DIFdelbegin \DIFdel{To rigorously model its discontinuous and unbounded nature, we reformulate the system dynamics from Eq.}\DIFdelend {We reformulate Eq.~}\DIFaddend (\ref{main_eq}) as \DIFaddbegin {\mbox{
$d\mathbf{x} = \mathbf{H}\mathbf{x}\,dt + d\boldsymbol{\beta}(t)$
}
, }\DIFaddend where \( d\boldsymbol{\beta}(t) = \mathbf{I}(t)dt \), and \( \boldsymbol{\beta}(t) \) is the Brownian motion process with a zero mean and covariance structure as: $\mathbb{E}[\mathbf{I}(\tau)\mathbf{I}(s)^\top] = \mathbf{Q}\delta(\tau-s).$ Here, \( \mathbf{Q} \) defines the input spectral density matrix, and also the Fourier transform of the autocovariance function $\mathbb{E}\left[\mathbf{I}(\tau) \mathbf{I}(s)^{\top}\right]$. The time-dependent solution, derived \DIFdelbegin \DIFdel{via Itô calculus to accommodate the unbounded and discontinuous variation of \mbox{
\( \boldsymbol{\beta}(t) \)
}
, is
\mbox{
\cite{sarkka2019applied}}\hskip0pt
:
}\DIFdelend \DIFaddbegin {from \mbox{
\cite{sarkka2019applied}}\hskip0pt
, is
}\DIFaddend \begin{equation*}
    \mathbf{x}^{\text{noise}}(t) = e^{\mathbf{H}(t-t_0)}\mathbf{x}^{\text{noise}}(t_0) + \int_{t_0}^t e^{\mathbf{H}(t-\tau)} \, d\boldsymbol{\beta}(\tau).
\end{equation*}
In steady state (\( t \to \infty \)), the stationary covariance function becomes:
\begin{equation*}
    \mathbf{C}(\tau) \triangleq \mathbb{E}[\mathbf{x}(t)\mathbf{x}(t-\tau)^\top] = 
    \begin{cases}
        \mathbf{P}_{\infty} e^{-\mathbf{H}^\top \tau}, & \tau \leq 0, \\
        e^{\mathbf{H} \tau} \mathbf{P}_{\infty}, & \tau > 0,
    \end{cases}
\end{equation*}
satisfying \( \mathbf{C}(\tau) = \mathbf{C}(-\tau)^\top \) with time-lag $\tau$. The steady-state covariance \( \mathbf{P}_{\infty} \) corresponds to the Lyapunov equation: 
\DIFaddbegin {\mbox{
$\mathbf{H}\mathbf{P}_{\infty} + \mathbf{P}_{\infty}\mathbf{H}^{\top} + \mathbf{Q} = 0,$
}
}\DIFaddend and admits equivalent representations:
\DIFaddbegin {\mbox{
$  \mathbf{P}_{\infty} = \int_{0}^\infty e^{\mathbf{H}\tau} \mathbf{Q} e^{\mathbf{H}^\top \tau} \, d\tau$.
}
}
In the one-dimensional case with a stable eigenvalue $\lambda<0$, the covariance reduces to
$
C(\tau)=\frac{Q}{2|\lambda|}e^{\lambda|\tau|}.
$
The steady covariance matrix $\mathbf{P}_{\infty}$ quantifies stationary variability and pairwise covariance under stochastic forcing (Appendix~\ref{sec:B}; \cite{sarkka2019applied}).

 \par The covariance function contains two distinct components: diagonal elements, the autocovariances $C_{ii}(\tau)$, which quantify self-evolutions and are even symmetric with maxima at zero lag, and off-diagonal elements, the cross-covariances $C_{ij}(\tau)$, which capture pairwise relations.
 \DIFdelbegin \DIFdel{NSDD systems }\DIFdelend \DIFaddbegin {NSDDM systems further }\DIFaddend exhibit monotonic decay with a large self-decay rate and strictly positive values (Appendix~\ref{sec:A}). When noise input with spectral density $I_0^{\text{noise}}$ is applied only to node $m$, the autocovariance dynamics at node $i$ are characterized by three metrics (fourth column in Fig. \ref{fig:one}(b)):
(i) \textit{amplification} \(Z_{ii}^{(m)}\), total integrated covariance \cite{pernice2011structure, vanalbada2015scalability}, (ii) \textit{peak response} \(P_{ii}^{(m)}\), maximum instantaneous covariance at zero lag \cite{grytskyy2013unified, vanalbada2015scalability}, and (iii) \textit{decay time constant} \(\tau_{ii}^{(m)}\), \(1/e\) relaxation time \cite{van2021microscopic,chaudhuri2014diversity,chaudhuri2015large}, where

\begin{align}  
Z_{ii}^{(m)} &\triangleq -2\left[\mathbf{H}^{-1}\mathbf{P}_{\infty}\right]_{ii}^m  \sim {O}(1/\lambda_1^2),  \label{eq:auto:Zim}\\[5pt]  
P_{ii}^{(m)} &\triangleq \left[\mathbf{P}_{\infty}\right]_{ii}^m  \sim {O}(1/\lambda_1),  \label{eq:auto:Pim}\\[5pt]  
*\tau_{ii}^{(m)} &\triangleq -\frac{\left[\mathbf{H}^{-1}\mathbf{P}_{\infty}\right]_{ii}^m}{\left[\mathbf{P}_{\infty}\right]_{ii}^m} \sim {O}(1/\lambda_1).\label{eq:auto:tauim}   
\end{align}  
 The eigenmode decomposition for steady-covariance is
\begin{equation}
\left[\mathbf{P}_{\infty}\right]_{ij}^m = -\sum_{p,q} \frac{u_{im}^p u_{jm}^q}{\lambda_p + \lambda_q}I_0^{\text{noise}}.
\label{eq:lyapunov_eigen}
\end{equation}
\par The metrics of crosscovariance dynamics between input node $m$ and node pair $(i,j)$ (last column in Fig.~\ref{fig:one}(b)) are: (i) \textit{amplification} $Z_{ij}^{(m)}$, signed area under the crosscovariance curve \cite{pernice2011structure, vanalbada2015scalability}; (ii) \textit{peak response} $P_{ij}^{(m)}$, maximal value; and (iii) \textit{peak response time} $t_{ij}^{(m)}$, time to maximum \cite{adhikari2010crosscorrelation}, where 
\begin{align}
Z_{ij}^{(m)} &\triangleq -M_{ij}^{m(1)}\sim O(1/\lambda_1^2), \label{eq:cross:Zij} \\[5pt]
*P_{ij}^{(m)} &\triangleq \frac{(M_{ij}^{m(1)})^2 }{\sqrt{4M_{ij}^{m(1)}M_{ij}^{m(3)}- 2\left(M_{ij}^{m(2)}\right)^2}} \sim O(1/\lambda_1), \label{eq:cross:Pij} \\[5pt]
*t_{ij}^{(m)} &\triangleq -\frac{M_{ij}^{m(2)}}{M_{ij}^{m(1)}} \sim O(1/\lambda_1),\label{eq:cross:tij}
\end{align}
with \DIFdelbegin \DIFdel{\mbox{
$M_{ij}^{m(n)} \equiv \left[\mathbf{H}^{-n} \mathbf{P}_{\infty}\right]_{ij}^m +(-1)^{n+1}  \left[\mathbf{H}^{-n} \mathbf{P}_{\infty}\right]_{ji}^m$
}
and \mbox{
$n=1,2,3$
}
}\DIFdelend \DIFaddbegin {\mbox{
$M_{ij}^{m(k)} \equiv \left[\mathbf{H}^{-k} \mathbf{P}_{\infty}\right]_{ij}^m +(-1)^{k+1}  \left[\mathbf{H}^{-k} \mathbf{P}_{\infty}\right]_{ji}^m$
}
and \mbox{
$k=1,2,3$
}
}\DIFaddend. {In particular, the stochastic amplification metrics reduce exactly to
\begin{equation}
Z_{ii}^{(m)}
=
I_0^{\text{noise}}\left(R_i^{(m)}\right)^2,
\qquad
Z_{ij}^{(m)}
=
I_0^{\text{noise}}R_i^{(m)}R_j^{(m)},
\label{stochastic:amplification}
\end{equation}
while the estimated cross-covariance timing satisfies
\begin{equation}
t_{ij}^{(m)}
=
\tau_i^{(m)}-\tau_j^{(m)}.
\label{stochastic:time}
\end{equation}
Here,
\[
R_i^{(m)}=-[\mathbf H^{-1}]_{im},
\qquad
\tau_i^{(m)}
=
-\frac{[\mathbf H^{-2}]_{im}}
{[\mathbf H^{-1}]_{im}},
\]
denote the input-normalized canonical strength gain and temporal metric,
respectively. These
reductions follow directly from the Lyapunov equation
$
\mathbf H^{-1}\mathbf P_\infty
+
\mathbf P_\infty\mathbf H^{-T}
=
-I_0\mathbf H^{-1}\mathbf e_m\mathbf e_m^T\mathbf H^{-T}.
$
Temporal integration over the lag isolates the zero-frequency covariance contribution, retaining the source-to-node propagation gains, while the cross-covariance timing reflects the relative temporal delay between the corresponding node responses. \textcolor{black}{As a result, $R_i^{(m)}$ and $\tau_i^{(m)}$ emerge as common node-to-node quantities for both deterministic and stochastic propagation.} A detailed spectral interpretation and comparison with $\mathbf P_\infty$ are provided in Appendix~\ref{sec:B:noise}, while the relation to previous stochastic-network approaches and the specific contribution of the present framework are discussed in Discussion Sec.~\ref{discuss:stochastic}.}

\par Estimation methods parallel those for pulse inputs, with complete derivations in Appendix~\ref{sec:B}. In NSDDM systems, crosscovariance also preserves strict positivity (Appendix~\ref{sec:A}). We validate all response metrics in NSDDM systems, showing close agreement with simulations and high rank correlations across network settings (Fig.~\ref{fig:two}(d,e)); their applicable range and limitations are further discussed in Appendix~\ref{sec:B}. Compared with deterministic inputs, stochastic responses additionally depend on the stationary covariance $\mathbf P_\infty$, which captures temporal structure specific to stochastic fluctuations\DIFdelbegin \DIFdel{- }\DIFdelend \DIFaddbegin {. It is }\DIFaddend computable through the Lyapunov equation \DIFdelbegin \DIFdel{(Eqs.~(\ref{eq:lyapunov}) and (\ref{eq:lyapunov_int})) }\DIFdelend or its eigenmode \DIFdelbegin \DIFdel{(Eq.~(\ref{eq:lyapunov_eigen})), though }\DIFdelend \DIFaddbegin {for large-scale systems, and }\DIFaddend both approaches lack intuitive interpretation of their dependence with $\mathbf{H}$. To make this dependence on $\mathbf{H}$ more explicit, we invoke the
Wiener--Khinchin theorem~\cite{sarkka2019applied}. The stationary covariance
$\mathbf{P}_{\infty}=\mathbf{C}(0)$ can be written as

\begin{equation*}
\mathbf{P}_{\infty}
=
\frac{1}{2\pi}
\int_{-\infty}^{\infty}
\left(\mathbf{H}-\mathrm{i}\omega\mathbf{I}_N\right)^{-1}
\mathbf{Q}
\left[
\left(\mathbf{H}+\mathrm{i}\omega\mathbf{I}_N\right)^{-1}
\right]^{\mathsf{T}}
\, d\omega .
\label{wiener}
\end{equation*}

This representation explicitly relates $\mathbf{P}_{\infty}$ to the
system matrix $\mathbf{H}$, but its complex-integral form still obscures
how specific network structures contribute to the stochastic response.
This motivates the matrix expansion and complex analysis developed in
Sec.~\hyperref[sec:IV]{IV}, where we further disentangle the contributions
of network structure and local dynamics.

\DIFdelend \DIFaddbegin \NUMBOX{Box 2.
How to probe your network and use our framework for Hurwitz systems}{
Given a network whose linearized dynamics, $d\mathbf{x}/dt =\mathbf{H}\mathbf{x} + \mathbf{I}(t)$, is hurwitz stable,\\
A \emph{constant (step) probe},
$\mathbf{I}(t)=I_0 e_m\,\mathbf{1}_{t\ge0}$,
provides a natural baseline, revealing both the static input--output gain and the characteristic timescale. At steady state, the response from source $m$ to target $i$ is given by the canonical strength metric
$
R_i^{(m)}
=
-I_0\left[\mathbf{H}^{-1}\right]_{im}
$
, and the timescale metric is
$
\tau_i^{(m)}
=
-{\left[\mathbf{H}^{-2}\right]_{im}}/
{\left[\mathbf{H}^{-1}\right]_{im}}
$.
At the system level, this timescale is typically governed by the dominant spectral mode,
$\tau_i^{(m)}\sim 1/|\lambda_1(\mathbf{H})|$. A step input continuously drives the source, while each contributing feedforward walk approaches its own asymptotic response monotonically (Appendix~\ref{sec:A}); together, these contributions determine the steady-state response of the target node. This produces a comparatively simple and well-behaved transient that can be compactly characterized by robust measures of steady-state gain and response timescale, making the constant (step) input a convenient probe of how strongly and how rapidly one node influences another. \\

A \emph{pulse (impulse-like) probe} is directly related to the step probe -- an impulse response is the time derivative of a step response. This relation establishes a simple correspondence between the two inputs: the pulse amplification, defined as the area under the target transient, is exactly equal to the steady-state gain of the step response, $R_i^{(m)}$. Their temporal metrics are also closely related: the step-response timescale provides the contribution to the pulse peak-time estimate, $\widetilde{\tau}_i^{(m)}\approx\tau_i^{(m)}+1/\lambda_1$. Compared with a step, a pulse emphasizes transient timing and peak; practically, when the stimulus duration is much shorter than the system's characteristic timescale, the response approaches the impulse response. This makes pulses useful for resolving propagation delays and separating fast feedforward routes from slower loop-mediated pathways visually. However, pulse-evoked transients are typically more irregular, and inhibitory couplings can produce dramatic nonmonotonic or sign-changing responses with multiple peaks, making typical timing harder to define and estimate robustly. \\

For a \emph{square probe} with duration $t_s=1$, the response peak coincides with that of the impulse response, while the amplification approaches the canonical strength metric $R_i^{(m)}$. In contrast, as $t_s\to\infty$, the final response during stimulation approaches that of a step probe, with its plateau converging to $R_i^{(m)}$. Varying $t_s$ therefore provides a direct way to probe the transition from impulse-like to step-like responses. Since square inputs naturally represent finite-duration perturbations in many real systems, the duration $t_s$ also serves as a practical control parameter for balancing temporal localization against response magnitude. \\

\emph{Zero-mean white-noise forcing} extends these deterministic probes to ongoing fluctuations. At steady state, the integrated covariance amplification satisfies the exact relation
$
Z_{ij}^{(m)}=I_0R_i^{(m)}R_j^{(m)},
$
showing that the same canonical strength metrics also characterize stochastic amplification. The cross-covariance timing is determined by the relative response timescales of the two nodes,
$
\tau_i^{(m)}-\tau_j^{(m)}.
$
Beyond these quantities shared with deterministic responses, the stationary covariance $\mathbf P_\infty$ provides additional stochastic-specific information by retaining the full temporal structure of the fluctuations, which can be interpreted in the time domain through the temporal overlap between pairwise propagated responses. \\

{\textcolor{black}{\emph{Applicable regime}}}

{
Hurwitz stability is the baseline condition considered here. It guarantees that unforced perturbations decay and that white-noise forcing with finite intensity produces a unique stationary covariance. The NSDDM regime provides a more restrictive reference setting in which nonnegative inputs produce nonnegative responses and monotonic step responses. In this regime, the strength and timing metrics have direct physical interpretations, and their analytical estimates are accurate across the tested parameter settings, except for the limiting cases discussed in Appendix~\ref{sec:B}. Beyond NSDDM, Appendix~\ref{sec:A} derives separate sufficient perturbative conditions for preserving Hurwitz stability and node-wise response rankings after introducing signed interactions or violations of strict diagonal dominance. The walk expansion also applies to heterogeneous signed networks whenever its convergence condition is satisfied. Caution is required when positive and negative walk contributions strongly cancel or recurrent inhibition induces pronounced oscillations. In such cases, $R_i^{(m)}$ remains interpretable as a signed steady-state gain under constant input, or equivalently as a signed integrated response under pulse input. By contrast, $\tau_i^{(m)}$ may be negative, undefined when the net gain vanishes, or unrepresentative of the dominant transient timescale. Similarly, when auto- or cross-covariance functions contain multiple extrema, a single decay time or peak lag cannot fully characterize their temporal structure.}

\centering
\includegraphics[width=0.4\linewidth]{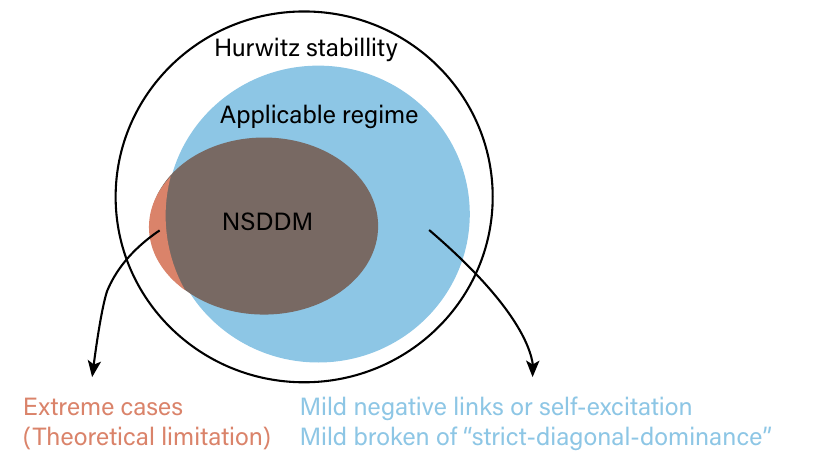}
\captionof{figure}{\label{fig:app}
\textbf{Schematic of the framework applicability.}}
}

\subsection{{Summary and applicability}}
\label{sec:E:1}

From the set of metrics developed above, these responses depend on both global spectral structure and localized, input-specific features, which can be summarized as follows: (i) from the inverse-matrix perspective, heterogeneous connectivity and self-dynamics are encoded in the element-wise inverse terms $[\mathbf{H}^{-k}]_{im}$, $k=1,2,3$, and in the steady-state covariance $[\mathbf{P}_{\infty}]_{ij}^{m}$; (ii) from the spectral perspective, spectral dispersion of $\lambda_j$ and nonuniform eigenmode participation $u_{im}^{j}$ can localize temporal and strength features~\cite{chaudhuri2014diversity, fyodorov2025nonorthogonal}; and (iii) input-specific alignment, captured by the corresponding $[\cdot]_{im}$ terms, governs localized response profiles.

 \begin{figure*}
\includegraphics[width=1.0\linewidth]{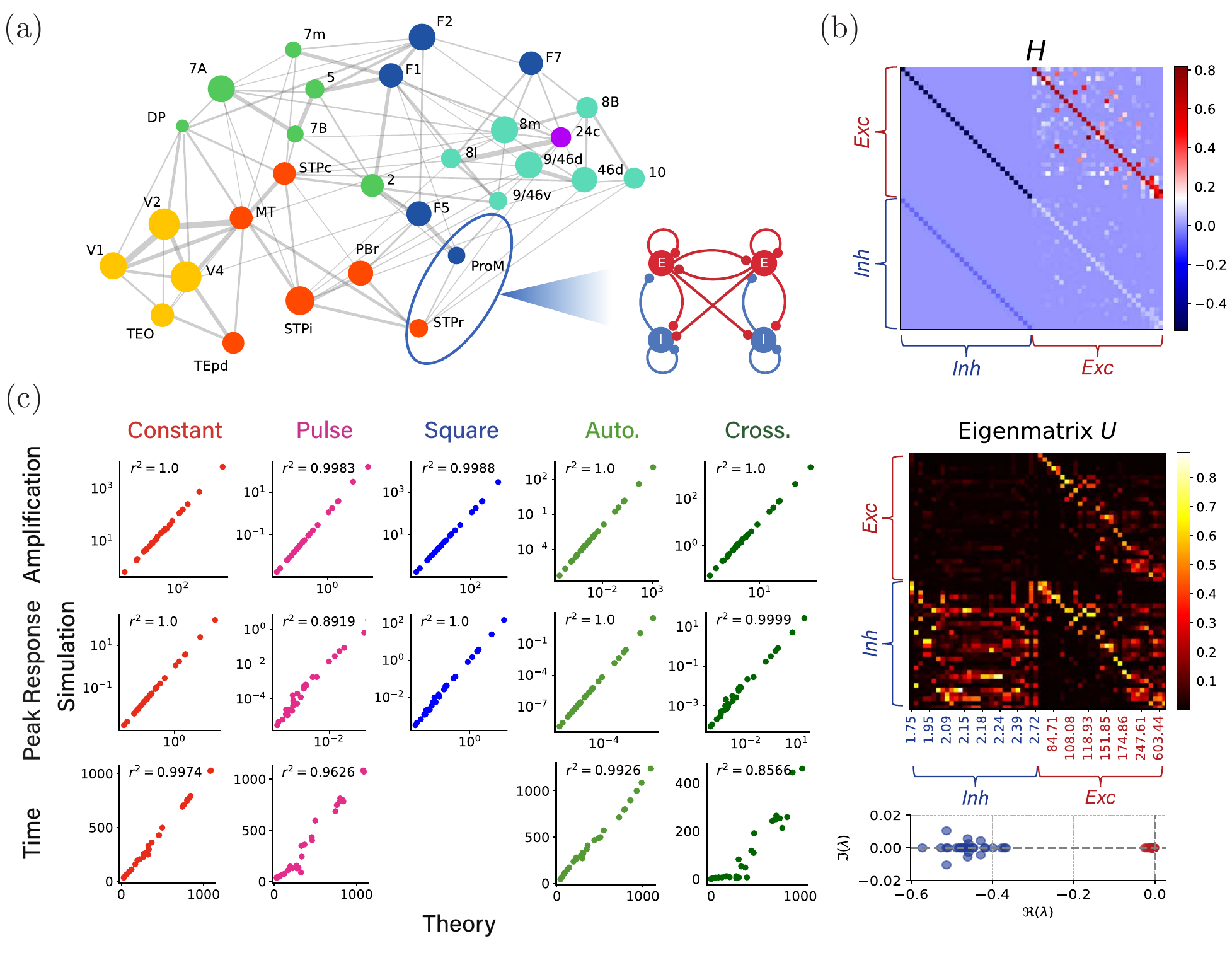}
\vspace{-0.8em}
\caption{\label{fig:newadd}
\textbf{{Quantification of response strength and time in a multi-regional primate cortical network.}}
{(a) Graph representation of cortical areas, where linewidth indicates connection strength (only the strongest \mbox{
$20\%$
}
of connections are shown). Each node (brain region) represents a Wilson--Cowan type model comprising excitatory (E) and inhibitory (I) populations. Brain regions from different lobes are shown in different colors.
(b) Matrix properties of \mbox{
$H$
}
, its eigenvector matrix \mbox{
$U$
}
, and the corresponding eigenvalues. \mbox{
$H$
}
has a non-NSDDM structure: it contains negative off-diagonal entries and positive diagonal elements, and strict diagonal dominance is broken; nonetheless, \mbox{
$H$
}
has all negative eigenvalues, which guarantees Hurwitz stability. The eigenvector matrix \mbox{
$U$
}
exhibits weak localization and weak orthogonality (upper right block, for excitatory populations), corresponding to a slow eigenmode (a cluster of eigenvalues near zero shown in red, while another cluster lies far from zero shown in blue which refers to a fast eigenmode). This specific structure of eigenvector matrix has been mathematically verified to correspond to the hierarchical timescales described in Ref.~\mbox{
\cite{li2022hierarchical}}\hskip0pt
.
(c) Comparison of simulated metrics and theoretical estimations for different input types. A high coefficient of determination \mbox{
$r^2$
}
indicates strong agreement between theory and simulation. Both strength and temporal metrics exhibit strong alignment in response ordering across input types (see SM Fig. S59). A full description of the model is given in Appendix~\ref{sec:F}. }}
\end{figure*} 

{
\textcolor{black}{We notice that two canonical quantities, $R_i^{(m)}$ and $\tau_i^{(m)}$, repeatedly emerge as the basic quantities of response strength and time across input types.} For response strength, $R_i^{(m)}$ gives the steady-state response to a constant input, equals the integrated response to a pulse of matched strength, determines the integrated square-input response, and serves as a basic gain factor of stochastic amplification. For response time, $\tau_i^{(m)}$ characterizes the constant-input response, approximates the post-stimulus decay timescale of the square-input response, provides the leading first-moment contribution to the estimated pulse peak time, and contributes to estimates of cross-covariance peak lags. \textcolor{black}{The stationary covariance $\mathbf{P}_{\infty}$ provides complementary information for stochastic forcing:} each element $\left[\mathbf{P}_{\infty}\right]_{ij}$ accumulates the temporal overlap between transient responses at nodes $i$ and $j$, thereby retaining details of temporal shape, width, and overlap that are lost under temporal integration. Together, $R_i^{(m)}$, $\tau_i^{(m)}$, and $[\mathbf{P}_{\infty}]_{ij}^m$ provide simple and interpretable building blocks for characterizing transient responses across input types.} 
 
Moreover, many integrated-strength and temporal metrics inherit an inverse dependence on the dominant eigenvalue $\lambda_1<0$ in the NSDDM regime. As $\lambda_1\to0^{-}$, the system approaches the stability boundary, and metrics with nonzero participation in the dominant mode generally become increasingly large or slow (SM Fig.~S8), whereas peak responses to pulses or fixed-duration square inputs may remain bounded because they also depend on transient shape and input duration. Practical guidance and the applicable regime of the framework are summarized in Box~2.

\DIFdelbegin 
\addtocounter{section}{-1}
To demonstrate the applicability of our framework, we apply the theory to a Wilson--Cowan-type primate cortical model (see Appendix~\ref{sec:F} for model details; \cite{chaudhuri2015large, li2022hierarchical, joglekar2018inter, li2025hierarchy}). This model reproduces a hierarchy of temporal responses: sensory regions at lower hierarchical levels, characterized by weaker recurrent excitation, respond rapidly, whereas association regions at higher hierarchical levels, with stronger recurrent excitation, exhibit slower and more persistent responses. In the model, each node represents a brain region with interacting excitatory (E) and inhibitory (I) populations, and interregional coupling is given by the fraction of labeled neurons (FLNs), i.e., structural connectivity (Fig.~\ref{fig:newadd}(a)). The structural network is dense ($\sim 65\%$) and highly heterogeneous: weights span from $10^{-6}$ to $10^{-1}$ and follow a log-normal distribution \cite{chaudhuri2015large, science.1238406}. The linearized matrix $\mathbf{H}$ has a four-block structure and lies outside the NSDDM regime because it includes negative couplings, positive self-dynamics, and violations of strict diagonal dominance (Fig.~\ref{fig:newadd}(b)). Its eigenvectors in the upper-right block of the eigenvector matrix for E-E populations are weakly localized and weakly orthogonal \cite{li2022hierarchical}. This spectral organization gives rise to hierarchical timescales that track the anatomical hierarchy, which is associated with spine-count--related excitatory gradients (Fig.~2B in Ref.~\cite{chaudhuri2015large}). In the spectral view, this hierarchy mainly corresponds to a slow cluster of eigenvalues approaching $0$, which dominates the transient dynamics.

Despite this strong heterogeneity, the analytical estimates agree quantitatively with simulations (Fig.~\ref{fig:newadd}(c)) and preserve the qualitative node-wise response organization (SM Fig.~S59), demonstrating that the framework remains informative beyond the NSDDM regime and identifies robust hierarchies of response strength and time across input types.

\section{{STRUCTURAL MECHANISMS OF TRANSIENT PROPAGATION}}
\label{sec:IV}

We have introduced a general set of metrics and derived analytical estimates for quantifying transient responses across different input types (Q1). However, when expressed in spectral or matrix-inverse form, or through Schur decomposition as in Ref.~\cite{goldman2009memory}, these metrics remain opaque in revealing how structural elements $A_{ij}$, local self-dynamics $\gamma_i$, and the input profile $I_m(t)$ jointly shape the quantified transient responses (Q2). This interpretive difficulty becomes particularly pronounced in large-scale systems, where many structural pathways and dynamical modes contribute simultaneously. To overcome it, we develop a directed-walk expansion based on the Neumann series~\cite{kato2013perturbation,berman1994nonnegative,Chau2025ExactLinear,WuMichael2013Amif}, which separates the contributions of network connectivity and local self-dynamics. We carry out this analysis in the NSDDM setting initially, with heterogeneous intrinsic time constants, $\dot{x}_i = - \gamma_i^{-1} x_i + \sum_{j=1}^N A_{ij}x_j + \delta_{im} I_m(t)$, where $\gamma_i = (\beta_i + D_i)^{-1}$, and subsequently extend the analysis beyond the NSDDM regime that allows signed couplings, including negative $A_{ij}$. We show below how the expansion makes explicit the role of walk-based propagation in governing both the strength and timing of transient responses under different input types.

\subsection{Directed chain to sparse random networks}
\label{subsec:chain}

We begin with the simplest case, propagation along a single directed chain with identical connection strengths $A_{d-1\to d}=\alpha$ and decay rate $\beta$ in the NSDDM setting (Appendix~\ref{sec:C}, Fig.~\ref{fig:three}(a)). All downstream nodes share the same local time constant $\gamma=(\beta+\alpha)^{-1}$, while the source node has $\gamma_m=\beta^{-1}$, which we set here to unity for simplicity; thus, a path of length $d$ contains $d$ nontrivial factors of $\gamma$. Here, $\gamma$ depends on both self-decay and connectivity, but beyond the NSDDM setting it can be treated more generally as an independent local dynamical parameter. The strength metrics ($Z,R,P$) therefore vary geometrically with $d$ and scale with the input strength $I_0$:
\begin{equation}
\text{Strength}(d)
\sim
\left(1+\frac{\beta}{\alpha}\right)^{-d}I_0
=
(\gamma\alpha)^d I_0,
\label{chain:strength}
\end{equation}
whereas the temporal metrics ($\tau,t$), measured relative to the source response, increase linearly with $d$:
\begin{equation}
\text{Time}(d)
\sim
\frac{d}{\alpha+\beta}
=
\gamma d.
\label{chain:time}
\end{equation}

These relations reveal two complementary effects of local dynamics and connectivity: response strength varies geometrically with path length $d$ according to the relative value between $\beta$ and $\alpha$, whereas response time increases linearly with $d$ at a rate set by their combined scale $(\alpha+\beta)^{-1}$. The same distance-dependent scaling extends approximately to sparse ER random networks when propagation is dominated by shortest source-to-target paths (shaded areas in Fig.~\ref{fig:three}(b)). More generally, the chain illustrates that local intrinsic time constants enter the two response properties in distinct ways: they combine multiplicatively in response strength but additively in response time along a directed path. Within the NSDDM regime, these relationships are also observed numerically in cases where substantial complex spectral components arise, and their qualitative behavior persists in stiff systems (SM Sec.~III).

\begin{figure}
\DIFdelbeginFL 
\includegraphics[width=1.02\linewidth]{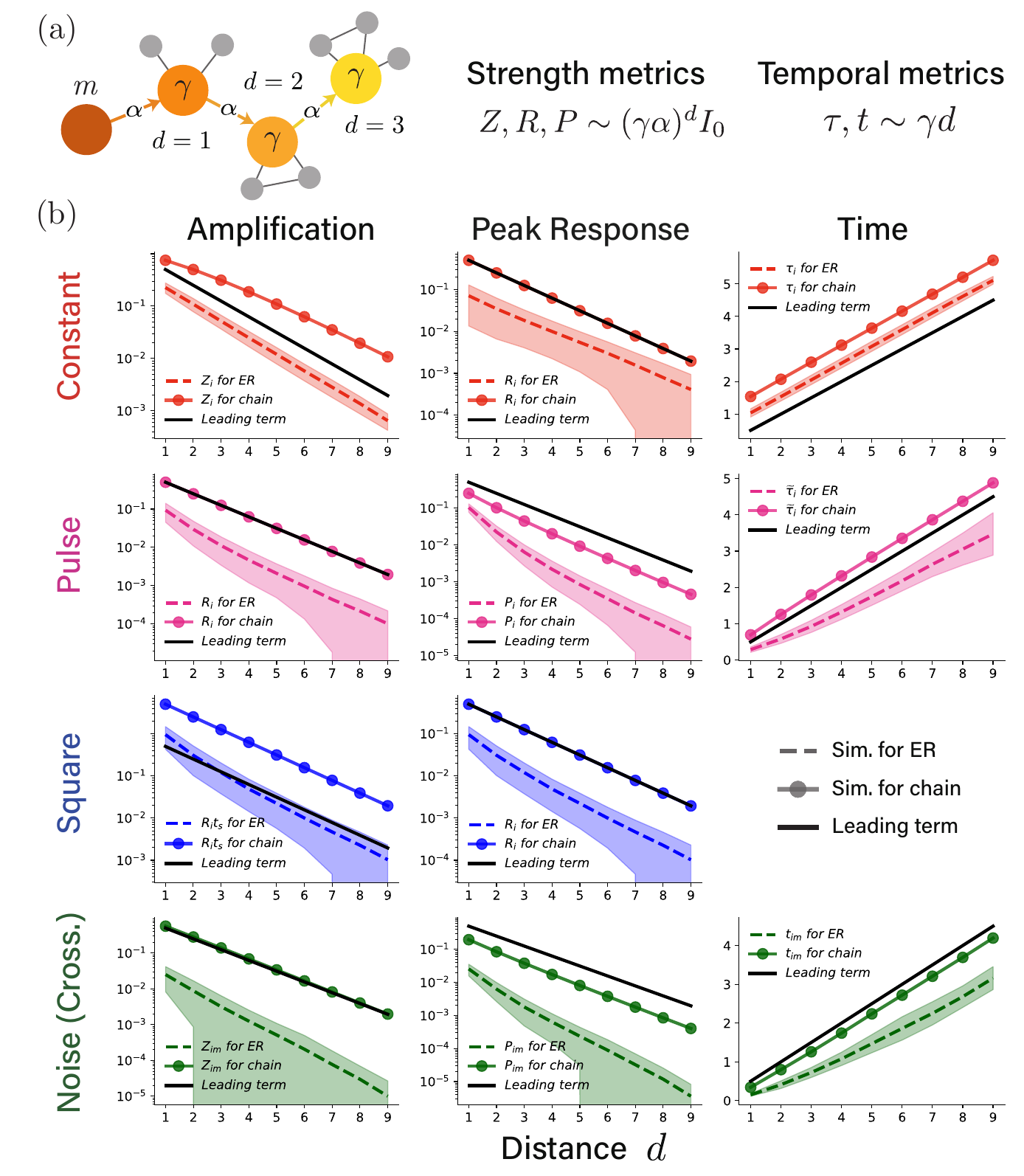}
\DIFaddendFL \caption{\label{fig:three}
\textbf{Propagation from homogeneous directed chains to sparse random networks.}  
(a) Schematic relationships: Directed chain topology with uniform intrinsic time constant $\gamma$, uniform connection strength $\alpha$, and distinct source-target shortest path length $d$. For the leading terms in metrics, strength metrics ($Z$, $R$, $P$) exhibit geometric decay with $d$, while temporal metrics ($\tau$, $t$) show linear path dependence.  
(b) Combined validation: Theoretical predictions (solid black: analytic leading terms for the chain; dotted: numerical chain simulations) and sparse ER random networks (dashed: ensemble mean of $100$ realizations; shading: $\pm 1$ SD; connection probability of random networks: $\sim 0.02$). Alignment enables direct structural comparison. Parameters: $\beta = \alpha = 1.0, \gamma = 0.5$, stimulus duration \DIFdelbeginFL \DIFdelFL{\mbox{
$t_s = 10$
}
}\DIFdelendFL \DIFaddbeginFL {\mbox{
$t_s = 1$
}
}\DIFaddendFL for square input, input strength $I_0$ normalized to unity in simulation.
}
\end{figure}

\subsection{Homogeneous case}
\label{subsec:homo}

\begin{figure*}
\includegraphics[width=1.0\linewidth]{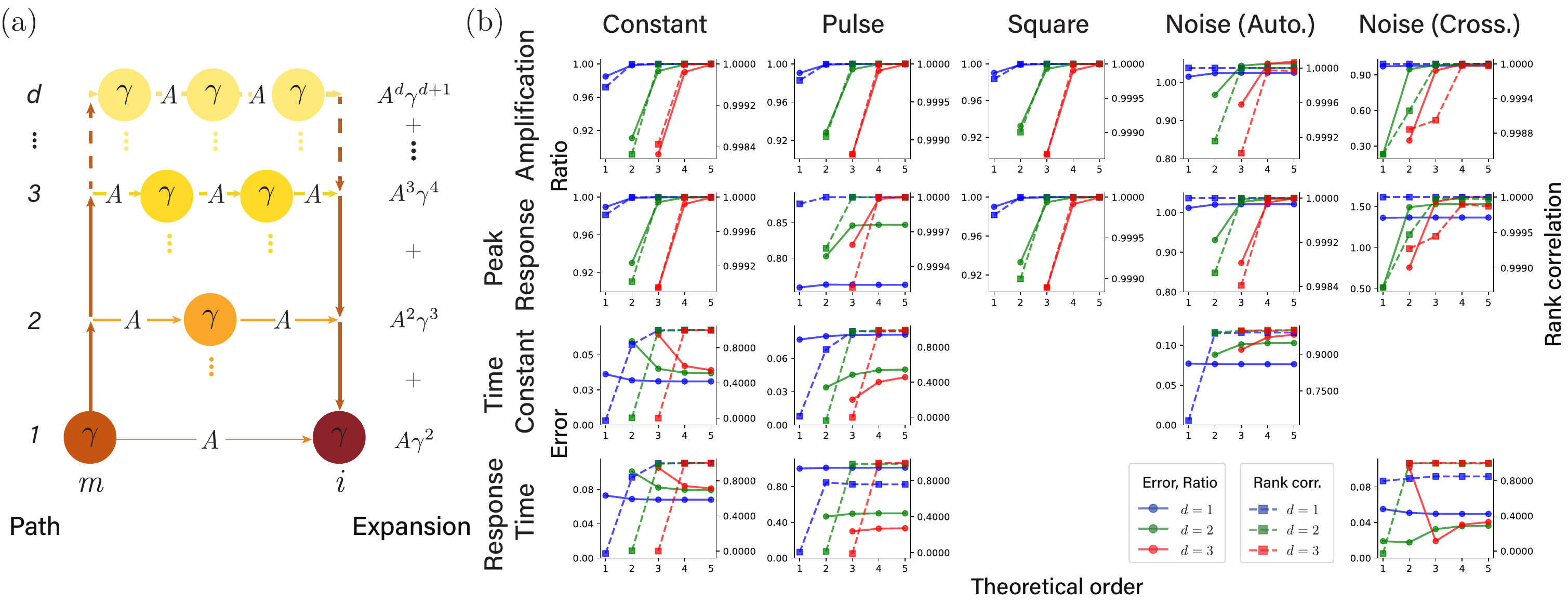}
\caption{
\textbf{Path-length dependence and truncation accuracy in homogeneous networks.}  
(a) Relations for dynamical metric dependence on path length $d$: Peak response $R_{im}$ (constant input) as weighted sums of $\gamma^{(d+1)} \times A^d$ terms. The $A^d$ factor accounts for path multiplicity (all length-$d$ paths between nodes), while temporal metrics are derived from ratio relationships.  
(b) Truncation order effects: Strength metric ratios (circles, left axis; $\text{Ratio} = {P_{\text{trunc}}}/{P_{\text{sim}}}$) and temporal metric relative errors (circles, left axis; $\text{Error} = {|t_{\text{trunc}} - t_{\text{sim}}|} /{t_{\text{sim}}}$) vs. rank correlations (squares, right axis). Colors denote shortest path lengths (blue: $d=1$, green: $d=2$, red: $d=3$) and values are averaged from $100$ network realizations. Strength metrics require $\geq d$-order truncation (ratio $>0.9$, error $<10\%$, rank correlation $>0.9$); temporal metrics need $\geq(d+1)$-order (error $<10\%$, rank correlation $>0.8$). Crosscovariance $C_{im}^m(\tau)$ truncated separately in $\mathbf{H}$ and $\mathbf{P}_{\infty}$. Parameters: $\gamma = 0.1$ (uniform intrinsic time constant); $\alpha=0.1$ (identical connection strength); $N=100$ (network size); $p=0.08$ (connection probability). 
}
\label{fig:four}
\end{figure*}

\par We next generalize directed chains and sparse random networks to homogeneous in-degree NSDDM systems, in which every node has the same in-degree, $D_i \equiv \sum_{j=1}^N A_{ij} = D$ and also we set $\beta_i$ as uniform. Consequently, all nodes share identical intrinsic time constant $\gamma = (\beta + D)^{-1}$. Structural heterogeneity in this setting, therefore, arises solely from differences in the directed and weighted walks through which signals propagate across the network. Such a configuration is characteristic of many artificial neural networks and synthetic biological circuits \cite{van1996chaos, renart2010asynchronous, brunel2000dynamics, litwin2012slow}. For clarity, we use \textit{chains} to denote acyclic walks, and \textit{walks} for the general case, which may include repeated nodes. Without loss of generality, we set the input strength to unity.

\par To make the explicit roles of propagation walks and self-dynamics, we employ the Neumann series expansion \cite{kato2013perturbation, berman1994nonnegative, Chau2025ExactLinear, WuMichael2013Amif}, a standard tool in perturbation theory, which expresses the inverse operator as $(\mathbf{I}_N - \mathbf{M})^{-1} = \sum_{p=0}^{\infty} \mathbf{M}^p$, valid when the spectral radius $\rho(\mathbf{M}) < 1$. This expansion decomposes the system response into contributions from propagation walks of all orders $p$. For instance, the canonical strength metric,
$R_i^{(m)} = -[\mathbf{H}^{-1}]_{im}$, can be written as
\begin{equation}
\begin{aligned}
    R_i^{(m)} &= -\left[(\mathbf{A} - (\beta + D) \mathbf{I}_N)^{-1}\right]_{im} \\
    &= \sum_{p=0}^{\infty} \frac{[\mathbf{A}^p]_{im}}{(\beta + D)^{p+1}} =\sum_{p=0}^{\infty} \gamma^{p+1} [\mathbf{A}^p]_{im},
\label{eq:peak}
\end{aligned}
\end{equation}
where $[\mathbf{A}^p]_{im}
\equiv
\sum_{j_1,\ldots,j_{p-1}}
A_{i j_1}A_{j_1 j_2}\cdots A_{j_{p-1}m}$ denotes the \emph{$p$-step cumulative walk term} from source $m$ to target $i$, i.e., the total weighted contribution of all directed walks of length $p$ from $m$ to $i$. It is obtained by summing, over all possible intermediate nodes $\{j_1,\dots,j_{p-1}\}$, the product of edge weights along each walk. Accordingly, $[\mathbf{A}^p]_{im}$ is large when there are many such walks, and/or when their edge weights are strong (Fig.~\ref{fig:four}(a)). Each $p$-step walk carries the same intrinsic-timescale factor $\gamma^{p+1}$. This also quantifies the intuition that longer paths can contribute more strongly to the peak response when the intrinsic time constant along the path is large, and/or when the corresponding $p$-step cumulative walk term $[\mathbf{A}^p]_{im}$ is large. The expansion in Eq.~\eqref{eq:peak} connects the strength metric to several established all-walk quantities, including matrix-exponential communicability \cite{estrada2008communicability,estrada2012physics}, Katz similarity \cite{katz1953new}, and the static Green-function \cite{economou2006green}. All these quantify the effective connection between a source and a target by aggregating contributions from walks of all possible lengths. Systematic metric analysis is provided in Appendix~\ref{sec:D}.


\par The precise convergence condition for this expansion in the general case is
$\rho(\mathbf{A}-\mathbf{D}) < \beta$. A convenient, sufficient, but not
necessary, condition follows from Gershgorin's theorem: strong decay dominance,
$\beta > 2D_{\max}$, guarantees convergence of the expansion; in this homogeneous
in-degree setting, this reduces to $\beta > 2D$. Alternatively, by absorbing the full diagonal term $\mathbf{D}+\beta\mathbf{I}$ into the local dynamics, the inverse can be expanded in powers of $\mathbf{\Gamma}\mathbf{A}$, where $\mathbf{\Gamma}=(\mathbf{D}+\beta\mathbf{I})^{-1}$ is the matrix of intrinsic time constants, under the convergence condition $\rho(\mathbf{\Gamma}\mathbf{A})<1$. For the NSDDM setting,
this expansion converges for any $\beta>0$. Thus, the
convergence condition depends on how the diagonal dynamics are
separated from the interaction term. Practical implementations often
tolerate weaker decay rates in the former expansion, which may increase
systematic estimation bias but retain the ranking order
(see SM Figs.~S46 and S47; Appendix~\ref{sec:D}).

\par The series expansions naturally motivate appropriate truncations that identify dominant contributions and simplify metrics while preserving accuracy. \DIFdelbegin \DIFdel{We }\DIFdelend \DIFaddbegin {In Fig.~\ref{fig:four}(b), we }\DIFaddend analyze simulated results for node pairs (source $m$, target $i$) with \DIFdelbegin \DIFdel{walk }\DIFdelend \DIFaddbegin {shortest-path }\DIFaddend length $d$\DIFdelbegin \DIFdel{and truncate expansions }\DIFdelend \DIFaddbegin {, using truncation }\DIFaddend at order $p$ in \DIFdelbegin \DIFdel{random networks (Fig. ~\ref{fig:four}(b)). This }\DIFdelend \DIFaddbegin {sparse random networks with significant self-decay. In this regime, the metrics are determined mainly by the dominant shortest-path contribution, whereas the aggregate effect of longer walks remains subleading. The }\DIFaddend order $p$ determines the maximal walk length ($[\mathbf{A}^p]_{im}$) incorporated in the metrics. Although all expansions converge asymptotically, required truncation depths differ between metric classes: strength metrics (e.g., $R_{i}^{(m)}$) need at least $p=d$, while temporal metrics (e.g., $\tau_{i}^{(m)}$) require at least $p=d+1$. This distinction originates from their mathematical forms \DIFdelbegin \DIFdel{. Notice that \mbox{
$R_i^{(m)} = \sum_{p=1}^\infty \left({[\mathbf{A}^p]_{im}}/{(\beta+D)^{p+1}}\right)$
}
}\DIFdelend \DIFaddbegin {when the intrinsic time constant is sufficiently small that longer-walk contributions are progressively suppressed. Notice that the linear weighted summation \mbox{
$R_i^{(m)} = \sum_{p=1}^\infty {\gamma^{p+1}[\mathbf{A}^p]_{im}}$
}
}\DIFaddend is dominated by minimal-length walks ($p=d$), whereas 
\begin{equation*}
\tau_i^{(m)} = \left( \sum_{p=1}^\infty {(1+p)\gamma^{p+2}[\mathbf{A}^p]_{im}} \right) \!\Bigg/\! \left( \sum_{p=1}^\infty \gamma^{p+1}{[\mathbf{A}^p]_{im}} \right)
\end{equation*}
demands $p=d+1$ terms to resolve the $(d+1)$-vs-$d$ balance between numerator and denominator. \DIFdelbegin \DIFdel{Temporal }\DIFdelend \DIFaddbegin {Hence, temporal }\DIFaddend metrics are more sensitive to longer walk lengths than strength metrics, which may limit the generalizability of localized approximations for these metrics~\cite{hens2019spatiotemporal, bao2022impact} (see Discussion Sec.~\hyperref[sec:discussion]{V}). \DIFaddbegin {For details on the truncation error analysis, see Appendix~\ref{sec:D}.
}\DIFaddend 

\par \DIFdelbegin \DIFdel{When }\DIFdelend \DIFaddbegin {When }\DIFaddend dominant walks of length $p$ exist between source-target pairs (i.e., $\left[\mathbf{A}^p\right]_{im} \gg \left[\mathbf{A}^q\right]_{im}$ for $q \neq p$), strength metric scaling reveals \DIFdelbegin \DIFdel{universal }\DIFdelend alignment patterns across input \DIFdelbegin \DIFdel{classes }\DIFdelend \DIFaddbegin {types }\DIFaddend when we omit effects of other walks. By isolating these terms\DIFdelbegin \DIFdel{including dominant walksand omitting \mbox{
$(\beta+D)$
}
-independent coefficients }\DIFdelend \DIFaddbegin {, namely the dominant walks, and omitting coefficients independent of the intrinsic time constants}\DIFaddend , we obtain:

\begin{align}
& Z_{i}^{(m)} \sim Z_{im}^{(m)} \sim \DIFdelbegin \DIFdel{\frac{1}{(\beta+D)^{p+2}}}\DIFdelend \DIFaddbegin {{\gamma^{p+2}}}\DIFaddend \left[\mathbf{A}^p\right]_{im}, \label{eq:homo:Zim} \\
& R_{i}^{(m)} \sim P_{im}^{(m)} \sim \DIFdelbegin \DIFdel{\frac{1}{(\beta+D)^{p+1}}}\DIFdelend \DIFaddbegin {{\gamma^{p+1}}}\DIFaddend \left[\mathbf{A}^p\right]_{im},\label{eq:homo:Rim}\\
& P_{i}^{(m)} \sim \DIFdelbegin \DIFdel{\frac{1}{(\beta+D)^p}}\DIFdelend \DIFaddbegin {\gamma^{p}}\DIFaddend \left[\mathbf{A}^p\right]_{im} ,\label{eq:homo:Pim}\\
& P_{i i}^{(m)} \sim \DIFdelbegin \DIFdel{\frac{1}{(\beta+D)^{2 p+1}}}\DIFdelend \DIFaddbegin {\gamma^{2p+1}}\DIFaddend \left[\mathbf{A}^p\right]_{im}^2, \label{eq:homo:Pii}\\
& Z_{i i}^{(m)} \sim \DIFdelbegin \DIFdel{\frac{1}{(\beta+D)^ {2 p+2}}}\DIFdelend \DIFaddbegin {\gamma^{2p+2}}\DIFaddend \left[\mathbf{A}^p\right]_{im}^2.\label{eq:homo:Zii}
\end{align}
\DIFdelbegin 

\DIFdelend Fundamental relationships emerge under the dominant walks case: (i) Constant-input amplification ($Z_{i}^{(m)}$) scales with noise-driven crosscovariance amplification ($Z_{im}^{(m)}$), as Eq.~(\ref{eq:homo:Zim}) shows; (ii) Peak responses under constant ($R_{i}^{(m)}$), amplification under pulse and unit square inputs all map to noise-driven crosscovariance peaks ($P_{im}^{(m)}$), as Eq.~(\ref{eq:homo:Rim}) shows; (iii) Impulse-response peaks ($P_{i}^{(m)}$) serve as natural reference units ($\sim O(1)$), with metric differences arising solely through \DIFdelbegin \DIFdel{\mbox{
$1/(\beta+D)$
}
}\DIFdelend \DIFaddbegin {\mbox{
$\gamma$
}
}\DIFaddend scaling (Eq.~(\ref{eq:homo:Pim})). (iv) The autocovariance peak $P_{ii}^{(m)}$ and amplification $Z_{ii}^{(m)}$ are quadratic in the cumulative walk terms and $Z_{ii}^{(m)}$ scales as $(R_i^{(m)})^2$ (Eq.~(\ref{eq:homo:Zii})). Additionally, the time constants now share the same scaling, $\tau_i^{(m)} \sim \tau_{ii}^{(m)} \sim \gamma$. These relationships reveal consistent walk-length-dependent scaling underlying transient dynamics across deterministic and stochastic inputs in the dominant-walk setting, as verified in Appendix~\ref{sec:D}.

\subsection{{Heterogeneous case}}
\label{subsec:heter}

\par Heterogeneous in-degree distributions are ubiquitous in real-world networks, spanning biological, transportation, and social systems \cite{roxin2011role, larremore2011predicting, wang2011synchronous, barrat2008dynamical}. In NSDDM-like systems, including diffusion- or spreading-type dynamics, such structural heterogeneity can also induce heterogeneous local self-dynamics through $\gamma_i=(\beta_i+D_i)^{-1}$. Previous approaches have incorporated degree heterogeneity through coarse-grained or mean-field descriptions, yielding degree-dependent scaling of local timescales such as $\tau_i\sim D_i^{\theta}$, with the exponent $\theta$ determined by the interaction dynamics and also self-dynamics ~\cite{barzel2013universality,harush2017dynamic,hens2019spatiotemporal,bao2022impact}. Such formulations capture qualitative effects of degree heterogeneity in sufficiently large networks but not detailed node- and link-level variations. Here, we instead retain network topology, coupling strengths, and local self-dynamics explicitly in finite-size heterogeneous networks. To systematically interpret these structural heterogeneities, we extend the walk-based expansion for arbitrary in-degree distributions. By introducing $\gamma_i = (\beta_i + D_i)^{-1}$, the resulted expressions still consist of sums of contributing walks as in the homogeneous case, but now weighted by their heterogeneous time constants.

\textcolor{black}{{\textbf{Generalized Walk Expansion for Global Network Propagation.}} Canonical metrics decompose naturally into weighted sums over all contributing walks as follows:
\begin{align}
R_{i}^{(m)} &= \sum_{w \in \mathcal{W}(m \to i)} \left( \prod_{v \in w} \gamma_v
\right) {\mathcal{A}_w},\label{heter:strength} \\
\tau_{i}^{(m)} &= \frac{\sum_{w \in \mathcal{W}(m \to i)} \left( \prod_{v \in w} \gamma_v
\right) \left( \sum_{v \in w} \gamma_v
\right) {\mathcal{A}_w} }{\sum_{w \in \mathcal{W}(m \to i)} \left( \prod_{v \in w} \gamma_v
\right) {\mathcal{A}_w}, }\label{heter:time}
\end{align}
where $\mathcal{W}(m \to i)$ denotes the set of all walks from node $m$ to node $i$; for $i=m$, this set includes the zero-length walk, for which $\mathcal{A}_w=1$. For each walk of length $p$, represented by the node sequence $w=(w_0,w_1,\dots,w_p)$ with source $w_0=m$ and target $w_p=i$, the corresponding {walk term} is $\mathcal{A}_w \equiv \prod_{t=0}^{p-1} A_{w_{t+1}w_t}$, and weighted by the intrinsic time constants of the nodes appearing along the walk. In the expansion of strength metric $R_i^{(m)}$, this weight is the product of $\gamma_v$ over all node occurrences along the walk (Eq.~\eqref{heter:strength}), so it contains exactly $p+1$ nodal intrinsic time constants, counted with multiplicity if a node is revisited. For the temporal metric $\tau_i^{(m)}$, the corresponding weight additionally involves a sum of the nodal $\gamma_v$ in the numerator (Eq.~\eqref{heter:time}). Although these weights are invariant under permutations of the visited nodes, a node with high walk centrality between source $m$ and target $i$ (i.e., one participating in many connecting walks) can substantially change the relative weight of walks~\cite{fink2011maximum}.}

\textcolor{black}{{For stochastic inputs, the main quantity is the steady-state covariance matrix $\mathbf{P}_{\infty}$. Using Cauchy's residue theorem~\cite{brown2009complex}, the covariance takes the form}
\begin{equation}
\left[\mathbf P_\infty\right]_{ij}^{m}
=
\sum_{\substack{ w\in\mathcal{W}(m\to i)\\ v\in\mathcal{W}(m\to j) }} 
\left( \prod_{p \in w} \gamma_p
\right)\left( \prod_{q \in v} \gamma_q
\right)\, \Phi(w,v)\, \mathcal{A}_w\mathcal{A}_v. 
\label{P_inf_compact}
\end{equation} 
Here, $\Phi(w,v)$ quantifies the temporal overlap between the responses propagated along walks $w$ and $v$, i.e., it is large when the two walk-induced responses evolve over similar time ranges and small when their temporal profiles are well separated (see Eq.~\eqref{Phi_wv}); this temporal-overlap factor distinguishes the stochastic response from the deterministic case; useful simplifications of the source--target cross-covariance $\left[\mathbf P_\infty\right]_{im}^{m}$ and target autocovariance $\left[\mathbf P_\infty\right]_{ii}^{m}$ based on this matrix are provided in Appendix~\ref{sec:E}.}

{We also examine self-responses ($i=m$; Appendix~\ref{sec:E}). The leading zero-length walk is governed by the local timescale $\gamma_i$, while recurrent-return walks provide higher-order effects. These effects vanish in the absence of recurrent loops and remain small for weak couplings or local time constants. For stochastic responses, they are further weighted by the temporal-overlap factor $\Phi(w,v)$, which can be suppressed when heterogeneous timescales produce well-separated path-response profiles. Thus, self-responses are dominated by $\gamma_i$ in weakly recurrent networks, whereas strong recurrence produces larger deviations, particularly in temporal metrics (Fig.~\ref{fig:mm}). Detailed results are provided in Appendix~\ref{sec:E}.}

{Overall, the walk-based expansion answers (Q2) by explicitly separating the roles of pairwise interactions and local self-dynamics in shaping transient propagation. Pairwise interactions, $A_{ij}$, determine which source-to-target walks are available and their combined strength through $\mathcal{A}_w$, whereas the local dynamics, $\gamma_i$, set the timescale at each visited node, controlling how rapidly activity evolves there. These factors combine along each walk, so the transient response reflects the interplay between where activity can propagate and how quickly it evolves locally. This separation persists across input types, as reflected by the canonical response strength $R_i^{(m)}$ and response time $\tau_i^{(m)}$ shared by deterministic and stochastic propagation, and is complemented by the stationary covariance $\mathbf{P}_\infty$ providing additional information specific to stochastic forcing.}

\begin{figure*}
\includegraphics[width=1.0\linewidth]{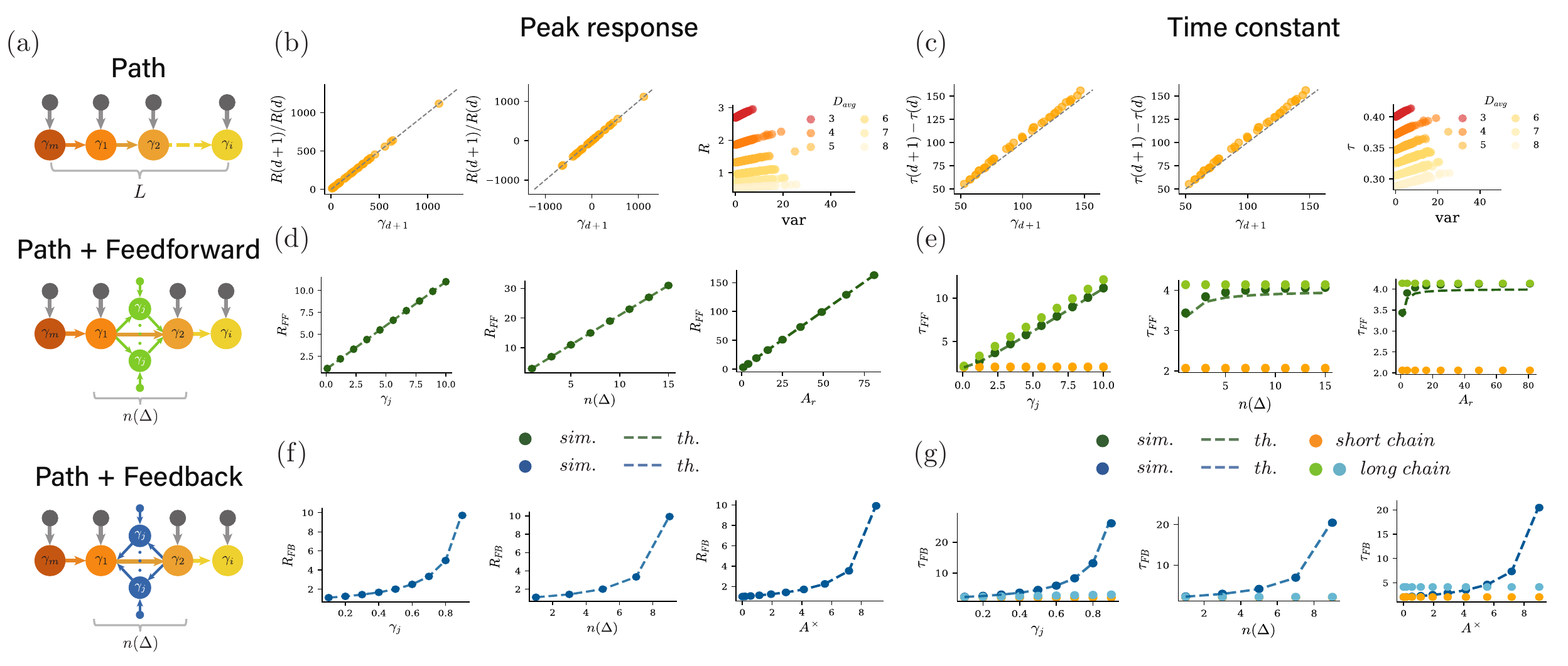}
\caption{\label{fig:five}
\textbf{Effects of nodal heterogeneity and local motifs on constant-input propagation.} (a) Schematic of a single path of length $L$ with heterogeneous intrinsic time constants $\{\gamma_0=\gamma_m,\gamma_1,\dots,\gamma_L=\gamma_i\}$ along the chain, with additional feedforward (second row) and feedback triangles (third row). (b, c) Propagation laws for constant input: $R(d+1)/R(d)=A_{d\to d+1}\gamma_{d+1}$ and $\tau(d+1)-\tau(d)=\gamma_{d+1}$, where $d$ is the shortest path length. Negative couplings $A_{d\to d+1}$ simply reverse the response sign, as shown in Fig. S40. Path degree reduces $R$ and $\tau$ by reducing $\gamma$, while variance enhances them. (d, e) Feedforward triangle effects: intrinsic time constant of the off-chain nodes, $\gamma_j=\gamma_{\Delta}$ (single triangle, $n(\Delta)=1$; for multiple triangles, all off-chain nodes are assigned the same value $\gamma_{\Delta}$), triangle count $n(\Delta)$, and weight ratio $A_r=(A_{1\to j}A_{j\to 2})/A_{1\to 2}$. Theory (lines) and simulations (dots) show linear peak response dependence (d). The time constant $\tau$ interpolates between pure length-$L$ and pure length-$(L+1)$ chain limits, saturating with triangle number and weight ratio (e). (f, g) Feedback triangle effects: intrinsic time constant $\gamma_j=\gamma_{\Delta}$ (single triangle, $n(\Delta)=1$; for multiple triangles, all off-chain nodes are assigned the same value $\gamma_{\Delta}$), triangle count $n(\Delta)$, and weight multiplication $A^{\times}=A_{2\to j}A_{j\to 1}A_{1\to 2}$. Both peak response and time constant show superlinear increases until the stability boundary is reached. Parameters: $\gamma$ uniformly random in $[50,150]$; coupling weights drawn from a Gaussian distribution with zero mean and variance $\sigma^2=9$ (negative couplings included) or absolute values (positive-only, b,c); path length $L=5$; $\gamma_m=\gamma_i=1.0$; $\gamma_{\Delta}=2.0$ (FF) and $0.1$ (FB). All path weights are unity except in weight-variation panels.
}
\end{figure*}
\textcolor{black}{\textbf{{Propagation Laws Along a Single Feedforward Chain.}}}
 {Specializing the all-walk decomposition described above to a single weighted directed feedforward chain, a simple iterative principle emerges across input types:
\begin{align}
\frac{\text{Strength}(d+1)}{\text{Strength}(d)}
&\to A_{d\to d+1}\gamma_{d+1},
\label{strength:base}\\
\text{Time}(d+1)-\text{Time}(d)
&\to \gamma_{d+1}.
\label{time:base}
\end{align}
Thus, the response strength is propagated multiplicatively through the coupling $A_{d\to d+1}$ and the intrinsic time constant $\gamma_{d+1}$ of the next node, whereas the temporal contribution accumulates additively, independent of the coupling weights (first column of Fig.~\ref{fig:five}(b,c)). These relations follow directly from the expansions of $R_i^{(m)}$ (Eq.~\eqref{heter:strength}) and $\tau_i^{(m)}$ (Eq.~\eqref{heter:time}), become particularly apparent at large propagation distances across input types (Fig.~\ref{fig:diff}), and reduce to the homogeneous-chain scaling laws derived above. For autocovariance-based strength metrics, the corresponding local factor enters quadratically as $\left(A_{d\to d+1}\gamma_{d+1}\right)^2$, whereas the analogous scaling for temporal metrics holds when the upstream timescale dominates (see Appendix~\ref{sec:E} for a detailed analysis). \textcolor{black}{The principle requires $\gamma_i>0$ for all nodes along the chain but any negative coupling can be included:} when $A_{d\to d+1}<0$, the sign of the propagated response at the next node is reversed, while the multiplicative scaling of its magnitude is preserved (Fig.~S52; second column of Fig.~\ref{fig:five}(b,c)). Consequently, the sign of a walk contribution is determined by the parity of its negative edges: odd parity gives a negative contribution and even parity a positive one. Moreover, along a single feedforward directed walk, the response to a constant input remains monotonic as verified in Appendix~\ref{sec:A}. The feedforward chain therefore provides a natural baseline for more complex architectures involving parallel feedforward pathways or recurrent feedback, analyzed below.}

{For this single walk of length $L$ from source $w_0=m$ to target $w_L=i$, with degree sequence $D=\{D_0,\dots,D_L\}$ in the NSDDM setting, two key statistical effects arise: (i) Increasing the mean in-degree $\langle D\rangle$ suppresses propagation by reducing the local time constants $\gamma_v=1/(\beta+D_v)$, thereby decreasing both response strength, approximately as $(\beta+\langle D\rangle)^{-L}$, and response time, approximately as $L(\beta+\langle D\rangle)^{-1}$. (ii) For fixed mean in-degree $\langle D\rangle$, increasing degree heterogeneity, quantified by the variance $\sigma_D$, enhances both response strength and time, with both metrics minimized at homogeneity ($\sigma_D=0$), as shown in the third column of Fig.~\ref{fig:five}(b,c) and Fig.~\ref{fig:heter}, and derived analytically in Appendix~\ref{sec:E}.}

\textcolor{black}{{\textbf{Network Motifs as Building Blocks of Propagation Dynamics.}} The general expansion describes propagation through all possible source-to-target walks, while the single-walk case represents propagation along one selected pathway. These limits correspond naturally to diffusion-like and routing-like propagation, respectively, with real systems, particularly brain networks across species, often lying between the two~\cite{avena2018communication}. Between these two limits, local network motifs provide simple structural building blocks for understanding how a small number of interacting paths and loops reshape propagation~\cite{mori2023information,luo2021architectures,prill2005dynamic}.}

{ At the initial stage of this motif-based analysis, we focus on two representative and mechanistically distinct motifs: feedforward (FF) and feedback (FB) triangles (Fig.~\ref{fig:five}(a)). An FF triangle is the simplest motif containing two source-to-target paths of different lengths: a direct path and a longer indirect path through an intermediate node. Its effect can therefore be understood as adding an extra longer-path contribution to the baseline chain. By contrast, an FB triangle is a three-node loop through which activity can circulate repeatedly. Its effect is naturally interpreted as repeated traversal of the loop, equivalent to an infinite summation of finite single-chain contributions, which mathematically gives rise to a geometric resummation. For identical motif copies attached to a fixed baseline chain, these two structures isolate how motif count $n(\Delta)$, local intrinsic time constant $\gamma_v$, and coupling weights $A_{ij}$ shape the constant-input response:}

{
\begin{equation*}
{\begin{aligned}
R_{\mathrm{FF}}(L) &= R(L)\left(1 + g_{FF}\right),\\
\tau_{\mathrm{FF}}(L) &= \tau(L) + \gamma_{\Delta}\left(\frac{g_{FF}}{1 + g_{FF}}\right), \\
R_{\mathrm{FB}}(L) &= \frac{R(L)}{1 - g_{FB}}, \\
\tau_{\mathrm{FB}}(L) &= \frac{\tau(L)}{1-g_{FB}} + \gamma_{\Delta}\left(\frac{g_{FB}\,}{1 - g_{FB}}\right),
\end{aligned}
\label{triangle}
}\end{equation*}}
{where $R(L)$ and $\tau(L)$ denote the baseline response strength and response time of a single directed chain with fixed length $L$. Here $\gamma_{\Delta}$ denotes the intrinsic time constant associated with the motif node(s) not belonging to the baseline chain. The quantities $g_{\mathrm{FF}} \equiv n(\Delta)\,\gamma_{\Delta} A_r$ and $g_{\mathrm{FB}} \equiv n(\Delta)\,\gamma^{\times} A^{\times}$ define the effective gains of the FF and FB motifs, respectively: $g_{\mathrm{FF}}$ captures the additional contribution of the indirect FF path, with $A_r$ denoting the ratio of indirect to direct coupling strength. In contrast, $g_{\mathrm{FB}}$ captures the recurrent amplification generated by the whole loop, where $\gamma^{\times}$ is the product of the intrinsic time constants of the nodes in the FB motif and $A^{\times}$ is the product of all the coupling weights along the loop.}

{For FF motifs, the three structural elements enter through the product $g_{\mathrm{FF}}$, so their effects combine multiplicatively. Since the peak response is linear in $g_{\mathrm{FF}}$, it grows linearly with each of $\gamma_{\Delta}$, $n(\Delta)$, and $A_r$ when the other two are held fixed (Fig.~\ref{fig:five}(d)). The time constant instead interpolates between the short-path temporal limit $\tau(L)$ and the long-path temporal limit $\tau(L)+\gamma_{\Delta}$. When the indirect long-path contribution is weak, the time constant remains approximately linear in these three elements; as $n(\Delta)$ or $A_r$ increases and the indirect pathway becomes dominant, it gradually saturates toward $\tau(L)+\gamma_{\Delta}$ (Fig.~\ref{fig:five}(e)).}

 {In contrast, for FB motifs, both the peak response and the time constant exhibit superlinear growth, reflecting the geometric summation inherent in the feedback loop (Fig.~\ref{fig:five}(f,g)). In the strong self-decay regime, where $g_{\mathrm{FB}}$ is small, both quantities grow approximately linearly with each element in $g_{\mathrm{FB}}$. As self-decay weakens and loop reinforcement accumulates, however, both the peak response and the time constant grow superlinearly and ultimately diverge as $g_{\mathrm{FB}}\to 1$, signaling the onset of instability.}

 {The distinction between FF and FB motifs follows directly from their structure. In FF motifs, indirect walks act in parallel with the baseline chain, increasing response strength linearly, while response time interpolates between direct- and indirect-path limits, saturating toward the latter as the indirect pathway dominates; such acyclic additions preserve the triangular form of the system matrix under topological ordering, leaving eigenvalues and stability unchanged when diagonal self-dynamics are fixed. FB motifs, by contrast, introduce recurrent loops that modify the eigenvalue spectrum, enabling repeated circulation, superlinear growth in both response strength and time, and eventual instability as feedback approaches the stability boundary. For white-noise input, the results remain qualitatively consistent with the constant-input case and naturally extend the same motif-based picture: motif-induced enhancement persists, with strength metrics increasing with both $g_{\mathrm{FF}}$ and $g_{\mathrm{FB}}$, although with distinct scaling behaviors. The temporal metrics exhibit different trends: they tend to saturate with increasing $g_{\mathrm{FF}}$, whereas under feedback the cross-covariance peak time remains approximately unchanged while the autocovariance time constant continues to increase (Appendix~\ref{sec:E}; Fig.~\ref{fig:six}; SM Sec.~III).}
\begingroup
\subsection{Application}
\label{sec:app:whole}
Biophysical systems contain more diverse motifs, including inhibitory couplings that can produce qualitatively different transient responses. \textcolor{black}{We thus consider representative motif classes following Refs.~\cite{avena2018communication,luo2021architectures}, including convergent and divergent motifs, feedforward and feedback inhibition, lateral inhibition composed of multiple feedforward inhibitory pathways, mutual inhibition, and disinhibition (Fig.~\ref{fig:eight}(a)).} These basic motifs provide distinct mechanisms for modulating and routing signal propagation.

\par \textcolor{black}{These dynamics are commonly characterized by an explicit biophysical time constant~\cite{zeraati2026neural},} which for a heterogeneous network satisfies
\begin{equation}
\tau_{\mathrm{param}}^{i}\frac{dx_i}{dt} = \sum_j H_{ij}x_j,
\end{equation}
where $\tau_{\mathrm{param}}^{i}$ may vary across nodes (see Box~2 of Ref.~\cite{zeraati2026neural}). This parameter should be distinguished from the local intrinsic time constant $\gamma_i$ introduced above: $\gamma_i$ arises from the effective local decay associated with self dynamics, whereas $\tau_{\mathrm{param}}^{i}$ rescales the entire dynamical equation of node $i$. For simplicity, we fix the source-input strength to unity throughout this analysis. 

After rescaling,
\[
\widetilde A_{ij} = \frac{A_{ij}}{\tau_{\mathrm{param}}^{i}}, \qquad
\widetilde\gamma_i = \tau_{\mathrm{param}}^{i}\gamma_i.
\]
The chain propagation laws then become
\begin{align*}
\frac{R(d+1)}{R(d)} &= \widetilde A_{d\to d+1} {\widetilde\gamma}_{d+1} = A_{d\to d+1}\gamma_{d+1}, \\
\tau(d+1)-\tau(d) &= {\widetilde\gamma}_{d+1}.
\end{align*}
The two propagation laws reveal a simple distinction: $\tau_{\mathrm{param}}^{i}$ rescales coupling and local decay inversely, leaving the stepwise strength gain $A_{d\to d+1}\gamma_{d+1}$ unchanged, but it directly sets the time constant $\widetilde{\gamma}_i=\tau_{\mathrm{param}}^{i}\gamma_i$. Thus, the biophysical timescale changes response time but not strength.

\begin{figure*}
\includegraphics[width=1.0\linewidth]{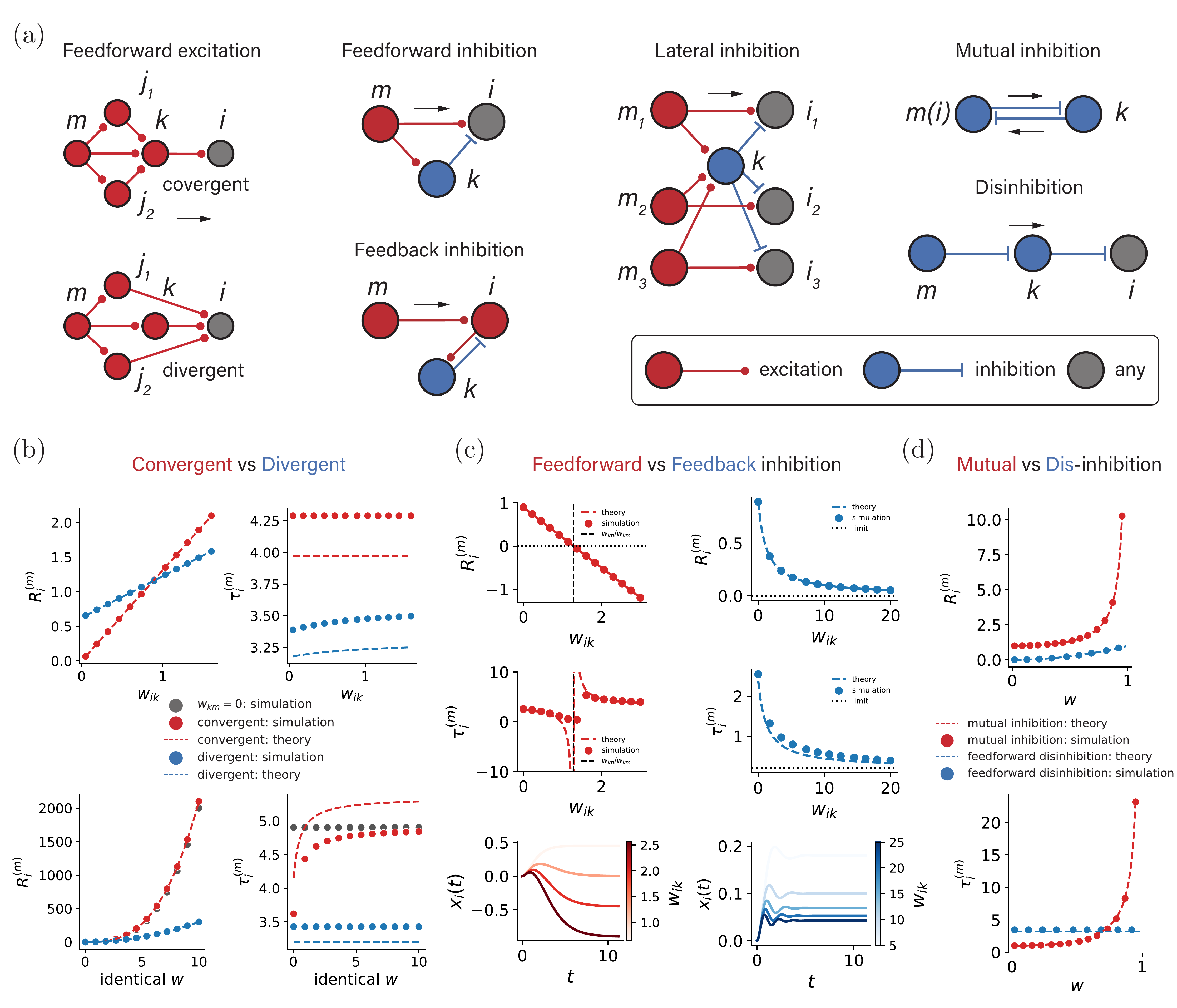}
\caption{\label{fig:eight}
\textbf{Transient propagation through canonical excitatory and inhibitory motifs.} 
(a) Motif schematic panels. An input is applied to source node $m$ and propagates to target node $i$ through different local motifs. Red and blue nodes denote excitatory and inhibitory populations, respectively. These motifs commonly occur in biological and neural systems (adapted from Refs.~\cite{avena2018communication, luo2021architectures}). 
(b) Effects of different motif structures on transient propagation. For convergent and divergent motifs, the coupling $w_{ik}$ plays distinct roles. In the convergent case, where a single pathway governs bunch of propagation pathways to the target, increasing the weight $w_{ik}$ linearly increases the response strength without altering the response time constant. For feedforward (FF) and feedback (FB) inhibition, increasing the inhibitory coupling $w_{ik}$ produces distinct transient effects: the final steady-state response decreases in the FF case, whereas an initial overshoot emerges in the FB case. The moment-estimated time constant becomes ill-defined when the final steady responses approach zero. Lateral inhibition can be viewed as a combination of feedforward inhibitory pathways sharing the same inhibitory node $k$. Mutual inhibition and disinhibition also produce distinct effects as the coupling involving node $k$ is increased. Mutual inhibition introduces feedback, producing a superlinear response dependence, whereas in the disinhibition motif, the response time constant is independent of the coupling strength. 
Parameters: intrinsic time constants: $\gamma_m=1.0$, $\gamma_k=1.4$, $\gamma_i=0.9$, $\gamma_{j_1}=1.8$, $\gamma_{j_2}=0.7$ for all cases; weights: $w_{km}=0.6$, $w_{ik}=0.8$, $w_{j_1m}=0.5$, $w_{j_2m}=0.4$, $w_{kj_1}=0.7$, $w_{kj_2}=0.9$, $w_{ij_1}=0.65$, $w_{ij_2}=0.75$ (convergent vs divergent cases); $w_{im}=0.9$, $w_{km}=0.7$, $w_{ik}=0.8$, $w_{ki}=0.6$ (FF and FB inhibition); identical weight variation for mutual inhibition and disinhibition.
}
\end{figure*}

\textbf{Representative Biological Motifs with Signed Couplings.}
We first consider convergent and divergent feedforward excitation~\cite{avena2018communication}, which provide simple extensions of a dominant propagation pathway (Fig.~\ref{fig:eight}(a)). For a baseline pathway
$m\to k\to i$, with additional branches
$m\to j_n\to k\to i$, the convergent motif gives
\begin{equation}
R_{im}^{\mathrm{conv}}
=
\gamma_m\gamma_k\gamma_i A_{ik}
\left(
A_{km}
+
\sum_n
\gamma_{j_n}A_{kj_n}A_{j_nm}
\right),
\end{equation}
and
\begin{equation}
\tau_{im}^{\mathrm{conv}}
=
\widetilde{\gamma}_m
+
\widetilde{\gamma}_k
+
\widetilde{\gamma}_i
+
\frac{
\sum_n
\gamma_{j_n}\widetilde{\gamma}_{j_n}
A_{kj_n}A_{j_nm}
}{
A_{km}
+
\sum_n
\gamma_{j_n}A_{kj_n}A_{j_nm}
}.
\end{equation}
Thus, the common downstream edge $A_{ik}$ scales the response amplitude
linearly but its weight does not affect the response time
(Fig.~\ref{fig:eight}(b)). As the branch weights increase, the direct path becomes progressively less important and the response time approaches the limit $A_{km}\to0$:
\begin{equation}
\tau_{im}^{\mathrm{conv}}
\rightarrow
\widetilde{\gamma}_m
+
\widetilde{\gamma}_k
+
\widetilde{\gamma}_i
+
\frac{
\sum_n
\gamma_{j_n}\widetilde{\gamma}_{j_n}
}{
\sum_n
\gamma_{j_n}
}.
\end{equation}
for all identical branch weights.

If all intermediate nodes share the same 
$\gamma_j$, while their 
$\widetilde{\gamma}_{j_n}$ remain heterogeneous because of biophyscial time constants, the limit further
reduces to
\begin{equation}
\tau_{im}^{\mathrm{conv}}
\rightarrow
\widetilde{\gamma}_m
+
\widetilde{\gamma}_k
+
\widetilde{\gamma}_i
+
\frac{1}{n}\sum_n \widetilde{\gamma}_{j_n}.
\end{equation}
Thus, increasing the feedforward branch strength continuously shifts
the response time from that of the direct pathway toward the timescale of the longer parallel pathways.

For the divergent motif, where activity propagates through the parallel
paths $m\to k\to i$ and $m\to j_n\to i$, the response strength is
\begin{equation}
R_{im}^{\mathrm{div}}
=
\gamma_m\gamma_i
\left(
\gamma_k A_{ik}A_{km}
+
\sum_n
\gamma_{j_n}A_{ij_n}A_{j_nm}
\right),
\end{equation}
and the corresponding response time is
\begin{equation}
\tau_{im}^{\mathrm{div}}
=
\widetilde{\gamma}_m+\widetilde{\gamma}_i
+
\frac{
\gamma_k\widetilde{\gamma}_k A_{ik}A_{km}
+
\sum_n
\gamma_{j_n}\widetilde{\gamma}_{j_n}
A_{ij_n}A_{j_nm}
}{
\gamma_k A_{ik}A_{km}
+
\sum_n
\gamma_{j_n}A_{ij_n}A_{j_nm}
}.
\end{equation}
Thus, divergent excitation adds the strengths of parallel paths, while the response time is their strength-weighted average timescale.

Multiplying all parallel-path weights by a common factor increases response strength without changing response time, because the factor cancels in the weighted average. If, in addition, all intermediate
nodes share the same local intrinsic time constant,
$\gamma_k=\gamma_{j_n}\equiv\gamma$, and all parallel paths have equal
weight products, the response time reduces to
\begin{equation}
\tau_{im}^{\mathrm{div}}
=
\widetilde{\gamma}_m+\widetilde{\gamma}_i
+
\frac{
\widetilde{\gamma}_k
+
\sum_n\widetilde{\gamma}_{j_n}
}{
n+1
}.
\end{equation}
Thus, under equal path strengths, changing their common coupling
strength affects only the response strength, while the propagation
time is determined by the average effective timescale of the parallel
pathways (Fig.~\ref{fig:eight}(b)).

\par Next, we examine inhibition. Feedforward and feedback inhibition
suppress responses through distinct mechanisms
(Fig.~\ref{fig:eight}(a)). For feedforward inhibition (FFI), with a
direct excitatory pathway $m\to i$ and an indirect inhibitory pathway
$m\to k\dashv i$, the two pathways compete. Taking $A_{ik}>0$ to denote
the magnitude of the inhibitory coupling, the signed response is
\begin{equation}
R_{im}^{\mathrm{FFI}}
=
\gamma_m\gamma_i
\left(
A_{im}
-
\gamma_k A_{ik}A_{km}
\right),
\end{equation}
and the corresponding time constant is
\begin{equation}
\tau_{im}^{\mathrm{FFI}}
=
\widetilde{\gamma}_m
+
\widetilde{\gamma}_i
-
\frac{
\gamma_k\widetilde{\gamma}_k
A_{ik}A_{km}
}{
A_{im}
-
\gamma_k A_{ik}A_{km}
}.
\end{equation}

\textcolor{black}{Thus, feedforward inhibition acts through cancellation between the
direct excitatory and indirect inhibitory paths.} Increasing inhibitory-path strength reduces the signed response, can reverse its sign, and may produce a non-monotonic transient. Consequently, once monotonicity is lost,
$R_i^{(m)}$ is interpreted as the signed steady-state response
rather than the peak response in the NSDDM setting. At the balance condition
$
A_{im}
=
\gamma_k A_{ik}A_{km},
$
the net signed response vanishes, and the moment-based temporal metric
becomes ill-defined because its denominator approaches zero
(second row of Fig.~\ref{fig:eight}(c)). When inhibition dominates,
$
A_{im}
<
\gamma_k A_{ik}A_{km},
$
the response time can instead be defined as the first time at which
$x_i(t)$ reaches a fraction $\eta=1-1/e$ of the total initial-to-final
change,
$
x_i(t)
=
x_i(0)
+
\eta\left[x_i(\infty)-x_i(0)\right].
$
Using this alternative definition, the theoretical estimation remains accurate.

\textcolor{black}{More generally, the signed-walk interpretation of the FFI motif extends naturally to arbitrary acyclic feedforward networks with signed couplings.} Importantly, after a topological ordering of the nodes, the system matrix becomes triangular, so, for fixed diagonal self-dynamics, the signs and magnitudes of off-diagonal couplings do not affect its eigenvalues or linear stability~\cite{guo2021exploring,bhattacharya2022discovering}. Acyclicity also guarantees finitely multiple source-to-target walks, giving an exact finite decomposition of the steady-state response into signed pathway contributions. Since the intrinsic factors along each walk are positive, the sign of each contribution is determined solely by the parity of its negative couplings: even parity gives a positive contribution, odd parity a negative one. The total response is therefore a signed sum over all source-to-target walks. The decomposition remains exact for response strength \(R\), whereas \(\tau\) is less robust because cancellation between opposite-sign paths can produce non-monotonic responses or make its normalization denominator vanish near zero net response (Fig.~S42).

Lateral inhibition can be viewed as a set of FFI motifs sharing a common inhibitory node \(k\), so \(\gamma_k\) and \(\widetilde{\gamma}_k\) are fixed across source--target pairs \((m_n,i_n)\), whereas the direct and indirect couplings can vary (Fig.~\ref{fig:eight}(a)). This shared pathway enables selective routing: strengthening direct excitation \(A_{i_n m_n}\) or weakening lateral inhibition \(A_{i_n k}\) opens route \(m_n\to i_n\), while the opposite closes it. Thus, lateral inhibition extends FFI competition to multiple parallel routes, allowing a common inhibitory node to selectively redistribute propagation.

In contrast, feedback inhibition (FBI), with
$m\to i\to k\dashv i$, repeatedly suppresses activity through the
negative feedback loop. Taking $A_{ik}>0$ to denote the magnitude
of the inhibitory coupling $k\dashv i$, we define the feedback gain as
$
g_{\mathrm{FBI}}
\equiv
\gamma_i\gamma_k A_{ik}A_{ki}.
$
The response strength forms a geometric series,
\begin{equation}
\begin{aligned}
R_{im}^{\mathrm{FBI}}
&=
\gamma_m\gamma_i A_{im}
\left[
1-g_{\mathrm{FBI}}
+g_{\mathrm{FBI}}^2-\cdots
\right]
\\
&=
\frac{
\gamma_m\gamma_i A_{im}
}{
1+g_{\mathrm{FBI}}
}.
\end{aligned}
\end{equation}
The corresponding signed response time is
\begin{equation}
\tau_{im}^{\mathrm{FBI}}
=
\widetilde{\gamma}_m
+
\frac{
\widetilde{\gamma}_i
-
g_{\mathrm{FBI}}\widetilde{\gamma}_k
}{
1+g_{\mathrm{FBI}}
}.
\end{equation}

\textcolor{black}{Thus, feedforward inhibition suppresses the steady response through pathway cancellation, whereas feedback inhibition acts through recurrent attenuation by the factor $1/(1+g_{\mathrm{FBI}})$.} Increasing $g_{\mathrm{FBI}}$ continuously
reduces the signed response strength, and the temporal metric
also decreases (first and second rows of Fig.~\ref{fig:eight}(c)). For strong inhibitory feedback, the recurrent interaction can generate
an initial overshoot (third row of Fig.~\ref{fig:eight}(c)), while the temporal metric
can still provide a meaningful estimate of the response timing according to the total initial-to-final change.
With sufficiently strong feedback, recurrent inhibition can generate damped oscillatory transients.
In this regime, the moment-based temporal metric should no longer be interpreted
directly as a conventional response time, nor as the time required to
reach a fixed fraction of the initial-to-final response, because
cancellation between positive and negative phases can strongly distort
its value. The estimation may become negative, and becomes ill-defined when
the integrated signed response approaches zero.

Mutual inhibition (MI) and disinhibition both contain two
inhibitory interactions (Fig.~\ref{fig:eight}(a)), so that the product of their signs is positive,
but recurrence fundamentally distinguishes their dynamics. For mutual inhibition, $m\dashv i\dashv m$, define the effective
recurrent gain as
$
g_{\mathrm{MI}}
\equiv
\gamma_m\gamma_i A_{im}A_{mi},
$
where $A_{im}$ and $A_{mi}$ denote the magnitudes of the inhibitory
couplings. For a constant input to node $m$, its self-response strength is
\begin{equation}
R_{mm}^{\mathrm{MI}}
=
\frac{\gamma_m}{1-g_{\mathrm{MI}}},
\qquad
g_{\mathrm{MI}}<1,
\end{equation}
with the corresponding response time
\begin{equation}
\tau_{mm}^{\mathrm{MI}}
=
\frac{
\widetilde{\gamma}_m
+
g_{\mathrm{MI}}\widetilde{\gamma}_i
}{
1-g_{\mathrm{MI}}
}.
\end{equation}

With only one node driven, this mutual-inhibition motif acts as recurrent amplification rather than the competitive mechanism typical of comparably driven populations~\cite{blake1989neural,seely2011mutual}. Its effect on the source response is analogous to mutual excitation, producing geometric amplification and temporal prolongation as $g_{\mathrm{MI}}$ increases; the other node $k$ simply responds with the opposite sign (Fig.~\ref{fig:eight}(a)). This behavior relies on the linearized dynamics, where activity can become negative. In nonlinear firing-rate models with rectified or bounded gains~\cite{seely2011mutual}, the inhibited population saturates at its lower bound, limiting further disinhibition and exhibiting superlinear amplification rather than saturation.

By contrast, disinhibition 
$m\dashv k\dashv i$ contains only a single disinhibitory pathway,
yielding
\begin{equation}
R_{im}^{\mathrm{FD}}
=
\gamma_m\gamma_k\gamma_i
A_{ik}A_{km},
\end{equation}
and
\begin{equation}
\tau_{im}^{\mathrm{FD}}
=
\widetilde{\gamma}_m
+
\widetilde{\gamma}_k
+
\widetilde{\gamma}_i.
\end{equation}

\textcolor{black}{Thus, disinhibition produces a positive response through
the product of two inhibitory couplings but introduces no recurrent
amplification: its response strength scales directly with the path
weights and local time constants, while its response time is simply
the sum of the effective local time constant along the path. Mutual
inhibition, in contrast, repeatedly recycles the same positive
disinhibitory effect, producing the characteristic recurrent factor
$1/(1-g_{\mathrm{MI}})$ and a strong increase in both response strength
and response time. More generally, for a directed loop with arbitrary source--target pairs $(m,i)$, the signed loop gain $g_{\mathrm{loop}}=\prod_{j=1}^{N}\gamma_j a_j$ and loop time $T_{\mathrm{loop}}=\sum_{j=1}^{N}\widetilde{\gamma}_j$ yield
\[
R_i^{(m)}=\frac{R_0}{1-g_{\mathrm{loop}}},\qquad
\tau_i^{(m)}=T_0+\frac{g_{\mathrm{loop}}}{1-g_{\mathrm{loop}}}T_{\mathrm{loop}},
\]
where $R_0$ and $T_0$ are the response strength and accumulated local time along the shortest directed source--target path. The response sign follows from the parity of negative couplings along this path, whereas the sign of $g_{\mathrm{loop}}$ follows from their parity around the loop: an even number gives $g_{\mathrm{loop}}>0$, increasing both $R_i^{(m)}$ and $\tau_i^{(m)}$, while an odd number gives $g_{\mathrm{loop}}<0$, reducing both (Fig.~S43). MI and FBI are special cases of this result.}

\textcolor{black}{\textbf{Wilson–Cowan Motif Mechanisms Shaping Neural Response Timescales.}}
In the previous section \ref{sec:E:1}, we captured the robust hierarchy of response timescales in the Wilson--Cowan model across different inputs, demonstrating our framework's predictive power. Here, we examine the same model from a motif-based perspective to understand how specific circuit elements shape transient responses, addressing a central question in theoretical neuroscience: how neural response timescales emerge from the interplay between intrinsic dynamics and network structure. This analysis explicitly connects local biophysical time constants, excitatory--inhibitory interactions, propagation pathways, and network motifs to the resulting response timescales~\cite{chaudhuri2015large,li2022hierarchical,zeraati2026neural}. Furthermore, our framework provides a simple basis for dissecting the relationship between amplification and temporal slowing, including detailed modulation where strong amplification occurs without a corresponding increase in response timescale \cite{murphy2009balanced,joglekar2018inter}.

\begin{figure*}
\includegraphics[width=1.0\linewidth]{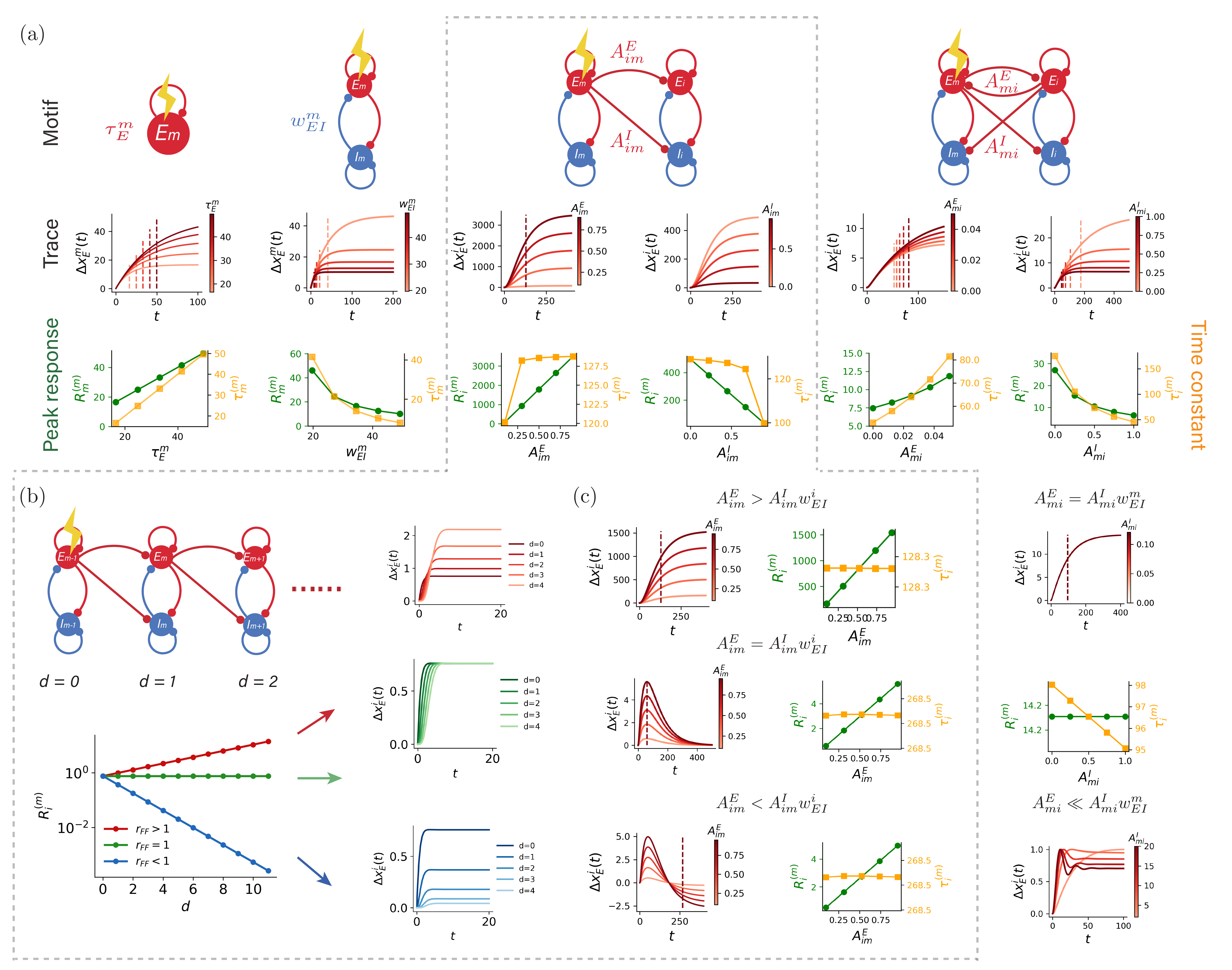}
\caption{\label{fig:nine}
\textbf{Response strength and time in Wilson–Cowan network motifs.}
{(a)} Effects of increasingly complex network motifs. For a single excitatory population, both peak response and time vary linearly with $\tau_E^m$. Adding local inhibition reduces both, while strong recurrent inhibition can generate damped oscillations (Fig. S38). Coupling two E--I motifs through feedforward connections preserves multiplicative response strengths and additive timescales. Increasing excitatory feedforward coupling $A_{im}^E$ linearly increases the peak response while response time approaches saturation. Inhibitory coupling $A_{im}^I$ has the opposite effect. Feedback connections produce superlinear dependence on $A_{mi}^E$ and can destabilize the system. These examples illustrate how local motifs act as building blocks for transient propagation.
{(b)} Amplification and attenuation along chains. The region-to-region response ratio determines propagation: ratio $>1$ amplifies, $=1$ preserves, and $<1$ attenuates the signal with propagation distance.
{(c)} Modulation through excitation--inhibition balance. In the feedforward motif (first two columns of panel~c), proportionally scaling excitatory and inhibitory contributions at fixed ratio leaves response time unchanged and rescales peak response linearly. Strong inhibition makes the transient nonmonotonic and the timescale metric ill-defined. At unity E--I ratio, the response becomes pulse-like; peak time and maximal amplitude cam remain the modulation. When inhibition dominates, the peak response and the time for first crossing of relative threshold $\eta$ satisfying $\Delta x^{\text{const}}_i(t_{im})/\Delta x^{\text{const}}_i(\infty)=\eta$ can still be modulated. In the feedback motif (final column), balanced excitatory and inhibitory feedback can cancel their effects on peak response; increasing both while preserving balance leaves the peak approximately unchanged while reducing the timescale. Stronger inhibitory feedback can induce an initial overshoot or damped oscillations.
All parameter settings are the same as in Appendix F, except for the varied parameters.
}
\end{figure*}

The linearized Wilson--Cowan network for $N$ interacting brain regions, each comprising excitatory and inhibitory population, is
given by
\begin{equation}
\begin{aligned}
\tau_E^i\frac{d {x}_E^i}{d t}
&=
-x_E^i
+
\sum_{j=1}^{N} w_{EE}^{ij}x_E^j
-
\sum_{j=1}^{N} w_{EI}^{ij}x_I^j
+\tau_E^iI_E^i(t),
\\
\tau_I^i\frac{d {x}_I^i}{d t}
&=
-x_I^i
+
\sum_{j=1}^{N} w_{IE}^{ij}x_E^j
-
\sum_{j=1}^{N} w_{II}^{ij}x_I^j
+ \tau_I^iI_I^i(t),
\end{aligned}
\label{eq:WC_linear_general}
\end{equation}
where $\tau_E^i$ and $\tau_I^i$ are the biophysical time constants of
the excitatory and inhibitory populations, respectively. The
parameters
$w_{EE}^{ij},w_{EI}^{ij},w_{IE}^{ij},w_{II}^{ij}\geq0$
denote coupling magnitudes, with the signs of excitatory
and inhibitory projections written explicitly in
Eq.~(\ref{eq:WC_linear_general}). 

We first consider an isolated excitatory population at node $m$. Under a unit perturbation, both the response strength and the response time are $\tau_E^m$ (Fig.~\ref{fig:nine}(a), first column).

We next consider a local E--I loop,
$E_m\rightarrow I_m\rightarrow E_m$, and define its recurrent gain as
$
g_m
\equiv
w_{EI}^m w_{IE}^m .
$
The effective local quantities then become
\begin{equation}
R_m
\equiv
R_m^{(m)}
=
\frac{\tau_E^m}{1+g_m},
\qquad
\tau_m
\equiv
\tau_m^{(m)}
=
\frac{\tau_E^m-g_m\tau_I^m}{1+g_m}.
\label{eq:WC_gamma}
\end{equation}
Thus, although the biophysical time constants $\tau_E^m$ and
$\tau_I^m$ remain fixed, recurrent E--I interactions renormalize the
effective response gain $R_m$ and temporal contribution
$\tau_m$ through $g_m$. Increasing $g_m$ suppresses the steady-state
response and, when $\tau_E^m>g_m\tau_I^m$, also shortens the response time (Fig.~\ref{fig:nine}(a), second column).

The local E--I circuit is Hurwitz stable when
$
1+g_m>0,
$
which is automatically satisfied for
$w_{EI}^m,w_{IE}^m>0$. Its transient is nonoscillatory when
$
(\tau_E^m-\tau_I^m)^2
\geq
4g_m\tau_E^m\tau_I^m,
$
and becomes damped oscillatory otherwise (Fig.~S38), in which case the moment-based temporal metric may lose its interpretability and can become ill-defined. Hence, increasing the negative coupling in the recurrent E--I circuit can drive a monotonic-to-oscillatory transition while remaining stable.

We next consider two local E--I motifs, with node $m$ as the source
and node $i$ as the target. Retaining only the excitatory feedforward
connection $E_m\to E_i$, with strength $w_{EE}^{im}$, gives
\begin{equation}
R_{im}^{(\mathrm{FF})}
=
\frac{w_{EE}^{im}}{\tau_E^i}R_iR_m,
\qquad
\tau_{im}^{(\mathrm{FF})}
=
\tau_m+\tau_i.
\end{equation}
Thus, response strengths combine multiplicatively through feedforward coupling, whereas effective response times add, as verified in a chain embedded in the real network (Fig.~S58).

We next include an additional feedforward projection
$E_m\to I_i$, with strength $w_{IE}^{im}$, which generates a
feedforward inhibitory pathway through the local connection
$I_i\to E_i$ with strength $w_{EI}^{i}$ (Fig.~\ref{fig:nine}(a), third column).
The response strength then becomes
\begin{equation}
R_{im}^{(\mathrm{FFI})}
=
\frac{
w_{EE}^{im}-w_{EI}^{i}w_{IE}^{im}
}{
\tau_E^i
}
R_iR_m,
\end{equation}
while the corresponding response time is
\begin{equation}
\tau_{im}^{(\mathrm{FFI})}
=
\tau_m+\tau_i
-
\frac{
w_{EI}^{i}w_{IE}^{im}\tau_I^i
}{
w_{EE}^{im}-w_{EI}^{i}w_{IE}^{im}
}.
\end{equation}
For a fixed feedforward inhibitory pathway, increasing the direct
excitatory feedforward strength increases the response strength
linearly and progressively reduces the inhibitory temporal correction,
such that
$\tau_{im}^{(\mathrm{FFI})}\to\tau_m+\tau_i$ in the strong-excitation
limit. Conversely, increasing feedforward inhibition suppresses both
the response strength and the moment-based temporal metric, provided
that
$w_{EE}^{im}>w_{EI}^{i}w_{IE}^{im}$.
At the balance condition
$
w_{EE}^{im}=w_{EI}^{i}w_{IE}^{im},
$
the signed steady-state response vanishes and the moment-based temporal
metric becomes ill-defined. Beyond this point, the inhibitory pathway
dominates and steady response becomes negative.

\textcolor{black}{This balance provides a method to modulate the transient response.} When
$A_{im}^{E}>w_{EI}^iA_{im}^{I}$, scaling $A_{im}^{E}$ and $A_{im}^{I}$ while keeping their ratio fixed increases the peak response without changing the time constant (first row of Fig.~\ref{fig:nine}(c)). At the balanced point,
$A_{im}^{E}=w_{EI}^iA_{im}^{I}$, the two feedforward contributions cancel, so
$R_{im}^{(\mathrm{FFI})}=0$ and the ratio-based temporal metric becomes ill-defined
(second row of Fig.~\ref{fig:nine}(c)). Nevertheless, increasing the balanced feedforward strengths while maintaining
$A_{im}^{E}/A_{im}^{I}=w_{EI}^i$ still modulates the transient, which now takes a pulse-like form and can be characterized by its maximal amplitude and the corresponding peak time. When
$A_{im}^{E}<w_{EI}^iA_{im}^{I}$, feedforward inhibition dominates and reverses the sign of the final response; we then use the initial-to-final response timing (third row of Fig.~\ref{fig:nine}(c)) to achieve the modulation.

\textcolor{black}{We next concatenate a series of feedforward motifs into a chain (Fig.~\ref{fig:nine}(b)).} The multiplicative propagation law remains valid, with each step depending only on the downstream node's local properties and the inter-node coupling. For a connection from node $m$ to node $i$, the response strength ratio is
\begin{equation}
r_{d \to d+1} \equiv \frac{R_{0 \to d+1}^{(FFI)}}{R_{0 \to d}^{(FFI)}} =  \frac{w_{EE}^{d+1,d} - w_{EI}^{d+1}w_{IE}^{d+1,d}}{1 + g_{d+1}},
\label{eq:WC_local_ratio}
\end{equation}
where $d$ denotes the path length. This ratio determines the change in response strength along the chain: $>1$ amplifies, $=1$ preserves, and $<1$ attenuates. We verify these regimes in a homogeneous chain (Fig.~\ref{fig:nine}(b)).

We next add an excitatory feedback connection $E_i\to E_m$ with
strength $w_{EE}^{mi}$. Now, the effective gain of the closed excitatory loop is
\begin{equation}
\ell_E
=
r_{m \to i}
\frac{w_{EE}^{mi}}{1+g_m}.
\end{equation}
For $r_{m \to i}>0$, $\ell_E>0$ represents a positive feedback loop.
Summing the repeated propagation around this loop gives the
closed-loop response
\begin{equation}
R_{im}^{(\mathrm{FB})}
=
\frac{
R_{im}^{(\mathrm{FF})}
}{
1-\ell_E
},
\qquad
\tau_{im}^{(\mathrm{FB})}
=
\frac{
\tau_{im}^{(\mathrm{FF})}
}{
1-\ell_E
},
\label{eq:WC_feedback}
\end{equation}
within the stable regime $\ell_E<1$. Thus, excitatory feedback
simultaneously amplifies the propagated response and prolongs its
moment-based response time, with both effects increasing as the loop
gain approaches unity.

Including an additional feedback connection $E_i\!\to I_m$ with strength
$A_{mi}^{I}$ introduces the inhibitory loop gain
\begin{equation}
\ell_I
=
r_{m \to i}\frac{w_{EI}^{m}A_{mi}^{I}}{1+g_m}.
\end{equation}
The response becomes
\begin{equation}
R_{im}^{(\mathrm{FB})}
=
\frac{R_{im}^{(\mathrm{FF})}}
{1-\ell_E+\ell_I},
\qquad
\tau_{im}^{(\mathrm{FB})}
=
\frac{\tau_{im}^{(\mathrm{FF})}-\ell_I\tau_I^m}
{1-\ell_E+\ell_I}.
\end{equation}
The excitatory feedback pathway increases both the peak response and the time constant in a super-linear manner, whereas the inhibitory feedback pathway reduces both quantities (Fig.~\ref{fig:nine}(a), final column). \textcolor{black}{A balance between the two feedback pathways occurs when $\ell_E=\ell_I$, equivalently, $A_{mi}^{E}=w_{EI}^{m}A_{mi}^{I}$.} At this balance point, the feedback contributions to the final response cancel, so that $R_{im}^{(\mathrm{FB})}=R_{im}^{(\mathrm{FF})}$, while the effective response time remains reduced by the inhibitory pathway when keeping the balance. Further increasing the inhibitory feedback strength can produce an initial overshoot in the transient response (Fig.~\ref{fig:nine}(c), final column).
\endgroup
\DIFaddend

\section{Discussion}
\label{sec:discussion}
\subsection*{Summary}
\label{subsec:summary}

\textcolor{black}{We have related the strength and timing of localized transient responses to the structure and intrinsic dynamics of finite heterogeneous networks. Two canonical source-to-target quantities, the response gain and the moment-based response time, connect the measures used for constant, pulse, square, and white-noise inputs. The walk expansion resolves these quantities into contributions from individual propagation routes. It thereby relates a {local} response to the couplings, local time constants, and recurrent interactions that generate it.}

The feedforward chain gives the simplest expression of this relationship. Successive couplings and local time constants multiply to determine response gain, while local time constants add to determine the canonical response time. These exact chain relations provide a reference for the approximate scaling of other metrics {in} sparse networks. They also connect the present analysis to earlier work on path-based propagation times~\cite{timme2019propagation,hens2019spatiotemporal}. Multiple feedforward routes contribute according to their weights, and positive feedback repeatedly propagates activity around loops. This recurrence increases gain and prolongs the response as the stability boundary is approached. Feedforward propagation and recurrence thus provide distinct mechanisms for regulating signal strength and duration.

\subsection*{Roles of degree and intrinsic-timescale heterogeneity}
\label{subsec:degree_heter}

Degree heterogeneity affects propagation through both the available pathways and the local dynamics of their constituent nodes~\cite{roxin2011role,hens2019spatiotemporal,barzel2013universality}. In the NSDDM model, the local time constant is $\gamma_i=(\beta_i+D_i)^{-1}$, where $D_i$ is the weighted in-degree. Increasing $D_i$ shortens this time constant and reduces the local gain associated with {the} given incoming connections. Along a feedforward path, these effects strengthen attenuation and reduce each node's contribution to the accumulated response time. The degree-ensemble calculations further show how increasing degree variance at fixed mean can enhance response strength and time through the distribution of local factors along a walk.

The general walk expansion also accommodates local time constants specified independently of degree. This allows comparison between diffusion-like systems, where coupling and local decay are linked, and neural models with separate biophysical parameters for each population. In both cases, the expansion retains the arrangement of heterogeneous nodes along propagation routes. The truncation analysis shows that temporal metrics {are} especially sensitive to longer walks {in sparse networks}. Multi-step connectivity and {distinct} motifs provide {additional} information about response {time and strength} alongside degree statistics~\cite{hens2019spatiotemporal,bao2022impact,hu2013motif}.

\subsection*{\textcolor{black}{Role of inhibitory couplings}}

The effect of inhibition depends on its {exact} position within a circuit. On a single feedforward chain with positive local time constants, the sign of the response follows the product of the coupling signs, while the multiplicative gain and additive time relations retain their form. When several paths connect a source and a target, their signed contributions compete. Feedforward inhibition can progressively attenuate the net gain, cancel it at balance, or reverse its sign. At exact cancellation, a finite transient can coexist with a zero steady response, as illustrated by the Wilson--Cowan motifs in Fig.~\ref{fig:nine}.

Recurrent inhibition produces a different set of effects. Mutual inhibition supplies effective positive feedback {that} can amplify responses and destabilize the circuit{, even without excitatory connections}. A feedback loop with one excitatory and one inhibitory coupling supplies negative feedback and can {attenuate} the steady response while generating {transient} overshoots or damped oscillations (Figs.~\ref{fig:eight} and \ref{fig:nine}). The Wilson--Cowan feedback motif gives a particularly direct example of temporal control. Balanced excitatory and inhibitory loop gains preserve the steady response while shortening its moment-based response time. Thus, the net response gain and its temporal organization can be tuned through circuit structure~\cite{murphy2009balanced, joglekar2018inter}.

The perturbative analysis quantifies the robustness of these response measures near the NSDDM reference class. For a diagonalizable baseline, a spectral criterion relates stability to the perturbation size and eigenvector conditioning. Resolvent bounds then control changes in the canonical gain and temporal metric. Under preserved stability, a Lyapunov-operator bound controls changes in stationary covariance (Appendix~\ref{sec:A}). Rankings of distinct responses are preserved when the corresponding error bounds are smaller than half their pairwise separations. The criteria are sufficient, and the spectral bound can be restrictive for strongly nonnormal baselines. Motif {analysis} and numerical tests explore larger departures from the reference class.

\subsection*{\textcolor{black}{Treatment of stochastic inputs}}
\label{discuss:stochastic}

The common source-to-target quantities give a precise relation between deterministic responses and noise-driven correlations. For white noise applied at node $m$, the covariance integrated over all lags satisfies $Z_{ij}^{(m)}=I_0^{\mathrm{noise}}R_i^{(m)}R_j^{(m)}$. Here $R_i^{(m)}=-[\mathbf H^{-1}]_{im}$ is the input-normalized deterministic gain. The corresponding moment-based lag satisfies $t_{ij}^{(m)}=\tau_i^{(m)}-\tau_j^{(m)}$. These identities hold for Hurwitz-stable systems whenever the quantities are well defined, including networks with recurrent feedback. The lag identity is exact for the temporal moment used here. Its use to estimate the position of a cross-covariance peak is an additional approximation.

The stationary covariance matrix $\mathbf P_\infty$, evaluated at zero lag, contains complementary information. Its entries integrate the temporal overlap of the impulse responses from the driven source to the two target nodes. They therefore depend on the shapes of these responses as well as their integrated strengths. In the walk representation, this dependence involves pairs of propagation walks. Distinguishing stationary covariance from covariance integrated over all lags specifies which stochastic observables follow directly from the canonical deterministic quantities and which require additional temporal information.

These results build on spectral and covariance methods for stochastic networks~\cite{hawkes1971spectra,pernice2011structure,trousdale2012impact,grytskyy2013unified,achab2018uncovering,hu2022spectrum}. Our analysis resolves the contribution of a specified input source to correlations between individual targets. Expressing these correlations through the same gains and temporal moments used for deterministic inputs makes their dependence on heterogeneous connectivity and local dynamics explicit.

\subsection*{Applications and implications}
\label{subsec:application}

The primate cortical model relates these results to the generation of neural timescales~\cite{chaudhuri2015large,li2022hierarchical,zeraati2026neural}. Local biophysical time constants specify {intrinsic} dynamics, interareal connectivity determines propagation routes, and excitatory and inhibitory interactions modify the gain and timing of the response. Across the tested inputs, the estimates capture the model's hierarchical organization, with predicted node rankings strongly correlated with simulations even outside strict diagonal dominance (Fig.~\ref{fig:newadd}). The motif analysis identifies circuit contributions that account for this dependence on local dynamics and connectivity.

The distinction between biophysical time constants, evoked response timescales, and activity timescales is important for interpreting these predictions~\cite{zeraati2026neural}. In our formulation, these correspond to parameters of local dynamics, {temporal} measures of the driven response, and the temporal structure of noise-driven covariance. The timescale {chosen to describe} a target region therefore depends on the input conditions. Tracking the source and its propagation pathways relates these different measures within the same circuit model. It also helps connect anatomical structure to spontaneous and evoked activity~\cite{avena2018communication,gollo2014mechanisms,luo2025mapping}.

The strength and timing relations also suggest ways to regulate propagation. Feedforward chains accumulate response time through successive local filters, with the product of local gains determining whether the signal is amplified or attenuated. This provides a structural account of long-lived activity generated by sequential feedforward propagation~\cite{goldman2009memory}. Recurrent excitation can further amplify and prolong responses, while inhibitory pathways provide additional control over their gain and shape. These mechanisms are relevant to balanced amplification in cortical circuits~\cite{murphy2009balanced,joglekar2018inter} and to temporal processing in recurrent network models~\cite{mastrogiuseppe2018linking,sussillo2009generating,maass2002real,rajan2016recurrent}. More generally, the treatment of finite-duration and fluctuating inputs supports comparisons between different local probes of neural and synthetic biological networks~\cite{nag2023dynamic,friston2011functional,cornelius2013realistic}, as well as the analysis of disturbances in infrastructure and ecological systems~\cite{gu2020performance,hastings2018transient}.

\subsection*{Limitations and future directions}
\label{subsec:extension}

The analysis applies to stable linear dynamics and, locally, to nonlinear systems near a stable operating state. Hurwitz stability ensures finite integrated responses and a stationary covariance under white-noise forcing. The walk expansion requires its own convergence condition, $\rho(\mathbf\Gamma\mathbf A)<1$, where $\mathbf\Gamma$ contains the positive local time constants. NSDDM networks provide nonnegative impulse responses and monotonic step responses, making normalized temporal moments particularly useful. Peak times and threshold or decay times are estimated from these moments and retain an approximation error in this reference class (Appendix~\ref{sec:B}). {In signed systems with moderate inhibitory interactions, the temporal metrics can remain informative, while response strength remains exact but should be interpreted as a signed steady-state response. However, exact cancellation can render temporal ratios undefined, and oscillations can introduce multiple relevant timescales.} These cases call for measures of transient shape and sensitivity that complement the present metrics and system-level measures of reactivity and transient growth~\cite{neubert1997alternatives,yang2023reactivity,krakovska2024resilience,hellmann2016survivability,fyodorov2025nonorthogonal}.

Temporally correlated inputs provide a natural extension~\cite{hanggi1995colored}. Initial tests in the primate model show that response rankings can remain robust as the noise correlation time is varied (Fig.~S60). A systematic treatment would relate the input spectrum to the temporal overlap of propagated signals. Inputs delivered to several nodes would further allow the analysis of signal integration in convergent circuits~\cite{luo2021architectures,hu2013motif,battaglia2012dynamic,shao2025impact}. Higher temporal moments offer another route to refining peak-time estimates. The chain calculations in Appendix~\ref{sec:C} show improvements in specific regimes, with residual bias in others. The first-order measures used here balance accuracy across the tested networks with a direct {and a simple} structural interpretation.

Extending the structural description to larger nonlinear responses will require accounting for changes in the local dynamics during propagation. The sensitivity of temporal metrics to longer walks suggests a useful direction for such extensions. Incorporating multi-step pathways and motifs into degree-based descriptions could improve predictions of transient timing while retaining the analytical advantages of coarse-grained models~\cite{hens2019spatiotemporal,bao2022impact,hu2013motif}.


\clearpage

\newtheorem{theorem}{Theorem}
\newtheorem{lemma}{Lemma}
\newtheorem{definition}{Definition}
\newcommand{\reflabeq}[2]{
  \ifcsname r@#2\endcsname
    \label{#2}
  \fi
  \begin{equation}
    #1
    \reflabeq{#2}
  \end{equation}
}
\appendix
\section{GLOSSARY OF TERMS}
\label{sec:glossary}

\subsection*{General terms}

\paragraph{Number of nodes $N$.}
The total number of nodes, or components, in the networked system.

\paragraph{Linearized matrix $\mathbf{H}$.}
The $N\times N$ Jacobian matrix describing the linearized dynamics of the system.

\paragraph{External input $\mathbf{I}(t)$.}
The external input is applied to a single source node $m$ and is written as
$I(t)\mathbf{e}_m$. Its amplitude is denoted by $I$. We mainly consider four
canonical input types: constant, pulse, square, and white-noise inputs.

\paragraph{Adjacency matrix $A_{ij}$.}
The weighted directed adjacency matrix of the network, with size $N\times N$ and
zero diagonal entries. In the NSDDM setting, its off-diagonal entries are
nonnegative, corresponding to a Metzler structure. In the general setting,
$A_{ij}$ may also take negative values. In some examples, $\alpha$ denotes a
uniform coupling weight.

\paragraph{In-degree $D_i$.}
The in-degree of node $i$ is defined as the total incoming coupling to that node,
given by the row sum
\[
D_i=\sum_{j=1}^{N}A_{ij}.
\]

\paragraph{Shortest path length or distance $d$.}
The graph distance between a source node and a target node. When multiple walks
of different lengths exist, $d$ denotes the length of the shortest one.

\paragraph{Expansion order $p$.}
The order in the expansion, corresponding to the contribution from walks of
length $p$.

\paragraph{Total length of a directed chain $L$.}
The total number of edges in a directed chain. For a chain of length $L$, the
corresponding degree sequence is denoted by
\[
D=\{D_0,\ldots,D_L\},
\]
where $D_0=D_m$ and $D_L=D_i$ are the degrees of the source and target nodes,
respectively.

\paragraph{Self-decay rate $\beta_i$.}
The local decay parameter governing the self-dynamics of node $i$. In the NSDDM
setting, it is assumed to be positive; in the expansion scheme, it is taken to be
uniform across nodes.

\paragraph{Intrinsic time constant $\gamma_i$.}
The time constant of node $i$ when network interactions are omitted, defined as $\gamma_i \equiv (\beta_i+D_i)^{-1}$ in NSDDM systems.

\paragraph{Number of motifs $n(\Delta)$.}
The number of motifs of a given type, i.e., the motif count.

\paragraph{Motif gain parameter $g$.}
An effective parameter that combines motif count, intrinsic time constants, and
coupling weights to quantify the strength of motif-induced effects.

\paragraph{Eigen-decomposition $\lambda_i$, $\mathbf{U}$.}
The eigen-decomposition of the linearized matrix $\mathbf{H}$, where
$\lambda_i$ are the eigenvalues and $\mathbf{U}$ is the eigenvector matrix. The
dominant eigenvalue, i.e., the eigenvalue with the largest real part, is denoted
by
\[
\lambda_1=\max_j \operatorname{Re}(\lambda_j).
\]

\subsection*{Metric terms}

\paragraph{Metric convention.}
All metrics are defined with respect to a source node $m$ receiving the input and
a target node $i$ receiving the propagated response. For covariance-based
metrics, the target is a node pair $(i,j)$.

\paragraph{Amplification.}
The amplification metrics are defined as follows:
\[
\begin{array}{ll}
\text{constant input:} & Z_{i}^{(m)},\\
\text{pulse input:} & R_{i}^{(m)},\\
\text{square input:} & R_{i}^{(m)}t_s,\\
\text{autocovariance:} & Z_{ii}^{(m)},\\
\text{crosscovariance:} & Z_{ij}^{(m)}.
\end{array}
\]

\paragraph{Peak response.}
The peak-response metrics are defined as follows:
\[
\begin{array}{ll}
\text{constant input:} & R_{i}^{(m)},\\
\text{pulse input:} & P_{i}^{(m)}\ \text{or}\ \widetilde{P}_{i}^{(m)},\\
\text{square input:} & R_{i}^{(m)}f(t_s,\tau_{i}^{(m)}),\\
\text{autocovariance:} & P_{ii}^{(m)},\\
\text{crosscovariance:} & P_{ij}^{(m)}.
\end{array}
\]

\paragraph{Time constant.}
The time-constant metrics are defined as follows:
\[
\begin{array}{ll}
\text{constant input:} & \tau_{i}^{(m)},\\
\text{pulse input:} & \tau_{D_{i}}^{(m)},\\
\text{square input:} & \tau_{i}^{(m)},\\
\text{autocovariance:} & \tau_{ii}^{(m)}.
\end{array}
\]

\paragraph{Response time.}
The response-time metrics are defined as follows:
\[
\begin{array}{ll}
\text{constant input:} & t_{i}^{(m)},\\
\text{pulse input:} & \tau_{i}^{(m)}\ \text{or}\ \widetilde{\tau}_{i}^{(m)},\\
\text{square input:} & t_{i}^{(m)},\\
\text{crosscovariance:} & t_{ij}^{(m)}.
\end{array}
\]

\DIFaddend \section{PROPERTIES OF NSDDM MATRICES AND PERTURBATIVE ROBUSTNESS BEYOND THE NSDDM REGIME} \label{sec:A}

\renewcommand\theequation{A\arabic{equation}}
\setcounter{equation}{0}

\par In the NSDDM system, the linear matrix $\mathbf{H} \equiv \mathbf{A} - \mathbf{D} - \operatorname{diag}(\beta_i)$ is strictly diagonally dominant with negative diagonal entries, ensuring all eigenvalues have negative real parts. Moreover, $-\mathbf{H}$ is an $M$-matrix, a structure with many desirable properties, as shown below.

\begin{lemma}
[All nodal dynamics are positive under positive pulse and constant inputs; under constant inputs, they increase monotonically]
The NSDDM system exhibits positive activity across all nodes after positive inputs $I_0 >0$: the time course $\Delta x^{\text{const}}_i(t)$ satisfies $\Delta x^{\text{const}}_i(t) > 0$ and ${d\Delta x^{\text{const}}_i(t)}/{dt} = \Delta x^{\text{pulse}}_i(t) > 0$ for all $i \in \{1, \dots, N\}$.
\label{lemma1}
\end{lemma}

\begin{proof}
Based on the definition, it is equivalently to prove that $\overline{\Delta x^\text{const}_i(t)} \equiv \Delta x^\text{const}_i(t)-\Delta x^\text{const}_i(\infty)=[e^{Ht}\,\overline{\Delta \mathbf{x}^\text{const}(0)}]_i< 0$ and that $\overline{\Delta x^\text{const}_i(t)}$ increases monotonically.

First, we prove that the initial value $\overline{\Delta x^\text{const}_i(0)} = \left[\mathbf{H}^{-1} \mathbf{I}_0^\text{const}\right]_i < 0$, i.e., $\left[\mathbf{H}^{-1}\right]_{im} \leqslant 0$ for all $i, m$ in the case of a single-node input to $m$ with positive scalar ${I}_0^\text{const}$.  

In the NSDDM system, the matrix $-\mathbf{H}$ is a non-singular $M$-matrix, which can be expressed as $-\mathbf{H} = \mathbf{S} - \mathbf{B}$, where $\mathbf{B} = (b_{ij})$ with $b_{ij} \geqslant 0$. The diagonal matrix $\mathbf{S}$ satisfies $S_{ii} \geq \max |\lambda_i(\mathbf{B})|$, meaning each diagonal entry exceeds the largest absolute eigenvalue of $\mathbf{B}$. Since the inverse of a non-singular $M$-matrix is non-negative, we have $\left[\mathbf{H}^{-1}\right]_{im} \leqslant 0$ for all $i, m$.

Next, we prove that $\overline{\Delta x^\text{const}_i(t)} < 0$. The expression can be expanded as
$$
\overline{\Delta x^\text{const}_i(t)} =\sum_{j=1}^N\left[e^{\mathbf{H}t}\right]_{ij} \overline{\Delta x^\text{const}_j(0)}.
$$
Define $\mathbf{C} = \mathbf{H} + b\mathbf{I}_N$, where $b > \max \{|H_{ii}|\}$; then $\left[\mathbf{C}\right]_{ii} > 0$ for all $i$, and $\left[\mathbf{C}\right]_{ij} \geq 0$ for all $i \ne j$. Hence, all entries of $\mathbf{C}^n$ are non-negative for any $n \geq 0$, and there exists $n^*$ such that for all $n > n^*$, $\left[\mathbf{C}^n\right]_{ij} > 0$. Thus, 
$$
\left[e^{\mathbf{C} t}\right]_{ij} = \sum_n \frac{t^n \left[\mathbf{C}^n\right]_{ij}}{n!} > 0.
$$

Now,
\begin{align}
\overline{\Delta x^{\text{const}}_i(t)} &= \sum_{j=1}^N \left[e^{\mathbf{H}t}\right]_{ij} \overline{\Delta x^{\text{const}}_j(0)}, \\
&= \sum_{j=1}^N \left[e^{ \left(\mathbf{C} - b\mathbf{I}_N\right)t}\right]_{ij} \overline{\Delta x^{\text{const}}_j(0)}, \\
&= \sum_{j=1}^N e^{-bt} \left[e^{\mathbf{C}t}\right]_{ij} \overline{x^{\text{const}}_j(0)} > 0.
\end{align}

Finally, note that ${d\Delta x^{\text{const}}_i(t)}/{dt} = \Delta x^{\text{pulse}}_i(t) > 0$ under the same input nodes and amplitudes, with the positivity of $\Delta x^{\text{pulse}}_i(t)$ established in \cite{wolter2018quantifying}. This monotonic behavior ensures the uniqueness of the solution for the temporal metrics we defined.
\end{proof}

\begin{lemma}[Signatures of $\mathbf{H}$ at negative integer powers]  In the NSDDM system, the matrix powers of $\mathbf{H}$ satisfy the following:  
$\left[ \mathbf{H}^{-p} \right]_{km} \geqslant 0$ for all positive even integers $p$, and  
$\left[ \mathbf{H}^{-p} \right]_{km} \leqslant 0$ for all positive odd integers $p$.  
Moreover, $\left[ \mathbf{H}^{-p} \right]_{km} = 0$ indicates that an input at node $m$ cannot reach node $k$.
\end{lemma}

\begin{proof}
From Lemma~\ref{lemma1}, we know that $\left[\mathbf{H}^{-1}\right]_{km} \leqslant 0$ for all $k, m$. Then, we have:
\begin{equation}
\begin{aligned}
\left[\mathbf{H}^{-2}\right]_{km} &= \sum_j \left[\mathbf{H}^{-1}\right]_{kj} \left[\mathbf{H}^{-1}\right]_{jm} \geqslant 0, \\
\left[\mathbf{H}^{-3}\right]_{km} &= \sum_j \left[\mathbf{H}^{-2}\right]_{kj} \left[\mathbf{H}^{-1}\right]_{jm} \leqslant 0, \\
&\vdots
\end{aligned}
\end{equation}
This pattern holds iteratively, establishing the sign structure of $\mathbf{H}^{-p}$ for all positive integers $p$.

\par Consider the Neumann series expansion
\begin{equation}
(\mathbf{I}_N - \mathbf{A})^{-1} = \sum_{q=0}^{\infty} \mathbf{A}^q,
\end{equation}
which implies that if node $k$ cannot be reached from node $m$, then $\left[\mathbf{A}^q\right]_{km} = 0$ for all $q$, since each power $q$ represents walks of length $q$ between nodes. 

Given that $\mathbf{H} = \mathbf{A} - \mathbf{D} - \operatorname{diag}_{i \in \{1, \ldots, N\}}(\beta_i)$, the additional diagonal terms do not affect the connectivity between node pairs. Therefore, $\left[\mathbf{H}^{-1}\right]_{km} = 0$ if and only if node $m$ cannot influence node $k$ through the network. This property is preserved for any power of $\mathbf{H}$, and thus $\left[\mathbf{H}^{-p}\right]_{km} = 0$ for all $p$ when no path exists from $m$ to $k$.
\end{proof}

\begin{lemma}[Initial decrease of input node under pulse input]  
In the NSDDM system, the response of the input node initially decreases: the solution $x^{pulse}_m(t)$ decreases at the onset of a pulse input.  
\label{lemma2}
\end{lemma}

\begin{proof}
We aim to show that $d\Delta x^{\text{pulse}}_m(t)/dt = \left[\mathbf{H} e^{\mathbf{H}t} \mathbf{I}_0^{\text{pulse}}\right]_m < 0$.

\begin{align}
&\frac{d}{dt} \Delta x_m^{\text{pulse}}(t) = \sum_j \left[\mathbf{H} e^{\mathbf{H} t}\right]_{mj} \left[\mathbf{I}_0\right]_j \\
&= \left[\mathbf{H} e^{\mathbf{H} t}\right]_{mm} I_0^{\text{pulse}} = \sum_j \left[\mathbf{H}\right]_{mj} \left[e^{\mathbf{H} t}\right]_{jm} I_0^{\text{pulse}}.
\end{align}

Define $\mathbf{C} = \mathbf{H} + b \mathbf{I}_N$, where $b > \max \{ |H_{ii}| \}$. Then $\left[\mathbf{C}\right]_{ii} > 0$ for all $i$ and $\left[\mathbf{C}\right]_{ij} \geq 0$ for all $i \neq j$. Substituting $\mathbf{H} = \mathbf{C} - b \mathbf{I}_N$, we get:
\begin{equation}
\begin{aligned}
&\sum_j \left[\mathbf{H}\right]_{mj} \left[e^{\mathbf{C} t - b \mathbf{I}_N t}\right]_{jm} I_0^{\text{pulse}} \\
&= \sum_j e^{-bt} \left[\mathbf{H}\right]_{mj} \left[e^{\mathbf{C} t}\right]_{jm} I_0^{\text{pulse}} \\
&= I_0^{\text{pulse}} e^{-bt} \left( \left[\mathbf{H}\right]_{mm} \left[e^{\mathbf{C} t}\right]_{mm} + \sum_{j \neq m} \left[\mathbf{H}\right]_{mj} \left[e^{\mathbf{C} t}\right]_{jm} \right).
\end{aligned}
\end{equation}

Since the first two terms are positive, we focus on the summation terms. Using the diagonal dominance of $\mathbf{H}$, we obtain:

\begin{equation}
\begin{aligned}
\left[\mathbf{H}\right]_{mm} \left[e^{\mathbf{C} t}\right]_{mm} &+ \sum_{j \neq m} \left[\mathbf{H}\right]_{mj} \left[e^{\mathbf{C} t}\right]_{jm}  \\
&< \sum_{j \ne m} \left[\mathbf{H}\right]_{mj} \left( \left[e^{\mathbf{C} t}\right]_{jm} - \left[e^{\mathbf{C} t}\right]_{mm} \right).
\end{aligned}
\end{equation}

Because $\left[\mathbf{H}\right]_{mj} > 0$, we examine the sign of the difference. Expanding $\left[e^{\mathbf{C} t}\right]_{jm}$ and $\left[e^{\mathbf{C} t}\right]_{mm}$ in Taylor series, we find that for small $t$, the leading order term of the difference is dominated by:

\[
\frac{t^d [\mathbf{C}^d]_{jm}}{d!} - 1,
\]
where $d$ is the shortest path length from node $m$ to node $j$. Therefore, for sufficiently small $t$, we have $t^d < \frac{d!}{[\mathbf{C}^d]_{jm}}$, implying the expression is negative.

\end{proof}

\begin{lemma}[Autocovariance and crosscovariance are positive]  
If $\mathbf{H}$ is diagonalizable and both nodes $i$ and $j$ are reachable from the input source $m$, then in the NSDDM system, both the autocovariance and crosscovariance satisfy $C_{ij}(\tau) = \left[e^{\mathbf{H} \tau}\, \mathbf{P}_{\infty}\right]_{ij} > 0$, for all $i, j$.
\end{lemma}

\begin{proof}
To show that $\left[e^{\mathbf{H} \tau} \mathbf{P}_{\infty}\right]_{ij} > 0$, it suffices to prove that $\sum_k [e^{\mathbf{H} \tau}]_{ik} [\mathbf{P}_{\infty}]_{kj} > 0$.  
In the NSDDM system, $[e^{\mathbf{H} \tau}]_{ik} \geq 0$, so we only need to prove that $[\mathbf{P}_{\infty}]_{kj} > 0$.

\par Since $\mathbf{H}$ is diagonalizable, we have:
\begin{equation}
\begin{aligned}
[\mathbf{P}_{\infty}]_{kj} 
&= \int_{-\infty}^{t} 
\!\left[ e^{\mathbf{H}(t-\tau)}\, \mathbf{Q}\, e^{\mathbf{H}^{\top}(t-\tau)} \right]_{kj} d\tau \\[4pt]
&= \int_{-\infty}^{t} 
\sum_{m} [e^{\mathbf{H}(t-\tau)}]_{k m}\, Q_{m m}\, [e^{\mathbf{H}^{\top}(t-\tau)}]_{m j}\, d\tau \\[4pt]
&= Q_{m m}\int_{-\infty}^{t} 
\sum_{m} [e^{\mathbf{H}(t-\tau)}]_{k m}\, [e^{\mathbf{H}(t-\tau)}]_{j m}\, d\tau,
\end{aligned}
\end{equation}

\par If $k = j$, this becomes $[\mathbf{P}_{\infty}]_{kk} = \int_{-\infty}^t [e^{\mathbf{H}(t - \tau)}]_{km}^2\, d\tau > 0$.  
If $k \ne j$, then $[e^{\mathbf{H}(t - \tau)}]_{km} > 0$ for $(t - \tau) > 0$, ensuring positivity of the integral.

\par Note that the derivative of autocovariance is not always positive, implying non-monotonic decay in certain cases, particularly with strong interactions or feedback loops. The derivative is given by:
\begin{equation}
\begin{aligned}
C_{ii}'(\tau) &= \left[\mathbf{H} e^{\mathbf{H} \tau} \mathbf{P}_{\infty}\right]_{ii}, \\
&= \sum_j [\mathbf{H}]_{ij} C_{ji}(\tau), \\
&= [\mathbf{H}]_{ii} C_{ii}(\tau) + \sum_{j \ne i} [\mathbf{H}]_{ij} C_{ji}(\tau), \\
&= \left(-\sum_{j \ne i} A_{ij} - \beta\right) C_{ii}(\tau) + \sum_{j \ne i} A_{ij} C_{ji}(\tau), \\
&= \sum_{j \ne i} A_{ij} \left(C_{ji}(\tau) - C_{ii}(\tau)\right) - \beta C_{ii}(\tau).
\end{aligned}
\end{equation}

\par This expression indicates that to ensure monotonic decay of $C_{ii}(\tau)$, one can either increase the self-decay rate $\beta$, or ensure $C_{ji}(\tau) < C_{ii}(\tau)$, meaning self-dynamics dominate.
\end{proof}

\begin{lemma}[\textcolor{black}{Monotonicity along a single signed feedforward walk}]
Consider a single feedforward walk
$$
m=i_0\to i_1\to\cdots\to i_d=i,
$$
with positive intrinsic timescales $\gamma_{i_\ell}>0$ and arbitrary nonzero signed couplings
$A_{i_{\ell+1}i_\ell}$. Under a constant input applied at node $m$, the response at every downstream node is monotonic. In particular, $x_i(t)$ increases monotonically if
$$
I_0\prod_{\ell=0}^{d-1}A_{i_{\ell+1}i_\ell}>0,
$$
and decreases monotonically if this product is negative.
\end{lemma}

\begin{proof}
It suffices to show that the derivative of the step response has a fixed sign for all $t>0$.

For a unit impulse input, the local impulse response of each node $i_\ell$ is
$$
g_{i_\ell}(t)=e^{-t/\gamma_{i_\ell}}>0,\qquad t\ge0.
$$
The source-to-target impulse response is therefore
$$
h_{im}(t)
=
\left(\prod_{\ell=0}^{d-1}A_{i_{\ell+1}i_\ell}\right)
\left(g_{i_0}*g_{i_1}*\cdots*g_{i_d}\right)(t),
$$
where $*$ denotes convolution. Since every local kernel is positive, the convolution is positive for $t>0$, and hence
$$
\operatorname{sgn} h_{im}(t)
=
\operatorname{sgn}
\left(\prod_{\ell=0}^{d-1}A_{i_{\ell+1}i_\ell}\right),
\qquad t>0.
$$

For a constant input $I_0\mathbf{1}_{t\ge0}$, the step response is
$$
x_i(t)=I_0\int_0^t h_{im}(s)\,ds,
$$
so
$$
\dot{x}_i(t)=I_0 h_{im}(t).
$$
Thus
$$
\operatorname{sgn}\dot{x}_i(t)
=
\operatorname{sgn}
\left(
I_0\prod_{\ell=0}^{d-1}A_{i_{\ell+1}i_\ell}
\right),
\qquad t>0.
$$
The derivative therefore has a fixed sign, so the trajectory is monotonic and approaches its final value without overshoot or sign reversal.

Non-monotonic transients therefore require interference between multiple walks with different signs and/or temporal profiles, or recurrent interactions.
\end{proof}

\begingroup
 
\begin{lemma}[\textcolor{black}{Sufficient spectral stability criterion under perturbation}]
Let $\mathbf{H}_0$ be Hurwitz and diagonalizable, so that $\mathbf{H}_0 = \mathbf{V}\boldsymbol{\Lambda}\mathbf{V}^{-1}$ with $\boldsymbol{\Lambda} = \operatorname{diag}(\lambda_1,\dots,\lambda_N)$. Define the stability margin
\[
\gamma \equiv -\max_i \operatorname{Re}(\lambda_i) > 0,
\]
and let the condition number $\kappa_2(\mathbf{V}) \equiv \lVert \mathbf{V} \rVert_2 \lVert \mathbf{V}^{-1} \rVert_2$. If a perturbation matrix $\mathbf{E}$ satisfies
\[
\kappa_2(\mathbf{V}) \lVert \mathbf{E} \rVert_2 < \gamma,
\tag{1}
\]
then the perturbed matrix $\mathbf{H} = \mathbf{H}_0 + \mathbf{E}$ remains Hurwitz; that is, $\operatorname{Re}(\lambda) < 0$ for every eigenvalue $\lambda$ of $\mathbf{H}$.
\end{lemma}

\begin{proof}
By the Bauer--Fike theorem, for every eigenvalue $\lambda$ of $\mathbf{H}$, there exists an eigenvalue $\lambda_i$ of $\mathbf{H}_0$ such that
\[
|\lambda - \lambda_i| \le \kappa_2(\mathbf{V}) \lVert \mathbf{E} \rVert_2.
\]
Taking real parts and applying the triangle inequality gives
\[
\operatorname{Re}(\lambda)
\le
\operatorname{Re}(\lambda_i) + |\lambda - \lambda_i|
\le
-\gamma + \kappa_2(\mathbf{V}) \lVert \mathbf{E} \rVert_2<0
\]
Hence every eigenvalue of $\mathbf{H}$ lies in the open left half-plane, so $\mathbf{H}$ is Hurwitz.

For normal matrices, $\kappa_2(\mathbf{V}) = 1$, and the bound reduces to $\lVert \mathbf{E} \rVert_2 < \gamma$. For highly nonnormal systems, $\kappa_2(\mathbf{V})$ can be large, making the criterion restrictive. 
\end{proof}

\begin{lemma}[\textcolor{black}{Preservation of deterministic response rankings under perturbation}]
Let $\mathbf{H}_0$ be a nonsingular NSDDM matrix, and let $\mathbf{H} = \mathbf{H}_0 + \mathbf{E}$ satisfy 
\[
\eta \equiv \lVert \mathbf{H}_0^{-1}\mathbf{E} \rVert_\infty < 1.
\]
Then $\mathbf{H}$ is nonsingular. For a fixed input node $m$, define the response and temporal-numerator vectors
\[
\mathbf{R}^{(m)}(\mathbf{H}) = -I_0\mathbf{H}^{-1}\mathbf{e}_m, \qquad
\mathbf{Z}^{(m)}(\mathbf{H}) = I_0\mathbf{H}^{-2}\mathbf{e}_m,
\]
with components $R_i^{(m)}(\mathbf{H}) = -I_0\mathbf{e}_i^\top \mathbf{H}^{-1}\mathbf{e}_m$ and $Z_i^{(m)}(\mathbf{H}) = I_0\mathbf{e}_i^\top \mathbf{H}^{-2}\mathbf{e}_m$.

Let $\boldsymbol{\mathcal R}(\mathbf{H}) = -I_0\mathbf{H}^{-1}$ and suppose $\lVert \boldsymbol{\mathcal R}(\mathbf{H}_0) \rVert_\infty \le R_{\max}$. Define
\[
\epsilon_R \equiv \frac{R_{\max} \eta}{1-\eta}.
\]
Then the componentwise perturbation bounds hold:
\begin{align}
\left\lVert \mathbf{R}^{(m)}(\mathbf{H}) - \mathbf{R}^{(m)}(\mathbf{H}_0) \right\rVert_\infty &\le \epsilon_R, \\
\left\lVert \mathbf{Z}^{(m)}(\mathbf{H}) - \mathbf{Z}^{(m)}(\mathbf{H}_0) \right\rVert_\infty 
&\le \frac{\epsilon_R}{I_0}\bigl(2R_{\max}+\epsilon_R\bigr). 
\end{align}

Define the minimum nonzero pairwise response gap
\[
\Delta_R^{(m)}
=
\min_{\substack{i\neq j \\ R_i^{(m)}(\mathbf{H}_0) \neq R_j^{(m)}(\mathbf{H}_0)}}
\left|
R_i^{(m)}(\mathbf{H}_0) - R_j^{(m)}(\mathbf{H}_0)
\right|.
\]
If $2\epsilon_R < \Delta_R^{(m)}$, then the peak response ordering is preserved.

See the temporal metric $\tau_i^{(m)}(\mathbf{H}) = Z_i^{(m)}(\mathbf{H}) / R_i^{(m)}(\mathbf{H})$, and let 
\[
r_*^{(m)} \equiv \min_i \left|R_i^{(m)}(\mathbf{H}_0)\right| > 0.
\]
If $\epsilon_R < r_*^{(m)}$, then every $\tau_i^{(m)}(\mathbf{H})$ is well-defined and satisfies
\begin{align}
&\left|
\tau_i^{(m)}(\mathbf{H}) - \tau_i^{(m)}(\mathbf{H}_0)
\right|
\le
\epsilon_\tau,\\
&\epsilon_\tau \equiv
\frac{\epsilon_R}{I_0 (r_*^{(m)} - \epsilon_R)}
\left(
2R_{\max} + \epsilon_R + \frac{R_{\max}^2}{r_*^{(m)}}
\right).
\end{align}
Define the minimum nonzero pairwise temporal gap
\[
\Delta_\tau^{(m)}
=
\min_{\substack{i\neq j \\ \tau_i^{(m)}(\mathbf{H}_0) \neq \tau_j^{(m)}(\mathbf{H}_0)}}
\left|
\tau_i^{(m)}(\mathbf{H}_0) - \tau_j^{(m)}(\mathbf{H}_0)
\right|.
\]
If $2\epsilon_\tau < \Delta_\tau^{(m)}$, then the temporal ordering is preserved.
\label{proof:perturb}
\end{lemma}
\begin{proof}
We bound the inverse matrix first.
Since
\[
\mathbf{H}
=
\mathbf{H}_0+\mathbf{E}
=
\mathbf{H}_0
\left(
\mathbf{I}+\mathbf{H}_0^{-1}\mathbf{E}
\right),
\]
and
\[
\lVert
\mathbf{H}_0^{-1}\mathbf{E}
\rVert_\infty
=
\eta
<
1,
\]
the matrix
\[
\mathbf{I}+\mathbf{H}_0^{-1}\mathbf{E}
\]
is nonsingular by the Neumann-series argument. Hence $\mathbf{H}$ is
nonsingular and
\[
\mathbf{H}^{-1}
=
\left(
\mathbf{I}+\mathbf{H}_0^{-1}\mathbf{E}
\right)^{-1}
\mathbf{H}_0^{-1}.
\]

The Neumann series gives
\[
\left(\mathbf{I}+\mathbf{H}_0^{-1}\mathbf{E}\right)^{-1}
=
\sum_{\ell=0}^{\infty}
\left(-\mathbf{H}_0^{-1}\mathbf{E}\right)^\ell.
\]
Therefore,
\[
\left\lVert
\left(\mathbf{I}+\mathbf{H}_0^{-1}\mathbf{E}\right)^{-1}
-
\mathbf{I}
\right\rVert_\infty
\le
\sum_{\ell=1}^{\infty}
\left\lVert\mathbf{H}_0^{-1}\mathbf{E}\right\rVert_\infty^\ell
=
\frac{\eta}{1-\eta}.
\]
It follows, by submultiplicativity of the induced infinity norm, that
\[
\begin{aligned}
\left\lVert
\mathbf{H}^{-1}
-
\mathbf{H}_0^{-1}
\right\rVert_\infty
&=
\left\lVert
\left[
\left(\mathbf{I}+\mathbf{H}_0^{-1}\mathbf{E}\right)^{-1}
-
\mathbf{I}
\right]
\mathbf{H}_0^{-1}
\right\rVert_\infty \\
&\le
\left\lVert
\left(\mathbf{I}+\mathbf{H}_0^{-1}\mathbf{E}\right)^{-1}
-
\mathbf{I}
\right\rVert_\infty
\left\lVert
\mathbf{H}_0^{-1}
\right\rVert_\infty \\
&\le
\frac{
\left\lVert\mathbf{H}_0^{-1}\right\rVert_\infty \eta
}{
1-\eta
}
\end{aligned}
\]

Since
\[
\boldsymbol{\mathcal R}(\mathbf{H})
=
-I_0\mathbf{H}^{-1},
\]
we have
\[
\begin{aligned}
\lVert
\boldsymbol{\mathcal R}(\mathbf{H})
-
\boldsymbol{\mathcal R}(\mathbf{H}_0)
\rVert_\infty
&=
I_0
\lVert
\mathbf{H}^{-1}
-
\mathbf{H}_0^{-1}
\rVert_\infty \le
\epsilon_R.
\end{aligned}
\]
For every fixed input node $m$,
\[
\begin{aligned}
\lVert
\mathbf{R}^{(m)}(\mathbf{H})
-
\mathbf{R}^{(m)}(\mathbf{H}_0)
\rVert_\infty
&=
\left\lVert
\left[
\boldsymbol{\mathcal R}(\mathbf{H})
-
\boldsymbol{\mathcal R}(\mathbf{H}_0)
\right]\mathbf{e}_m
\right\rVert_\infty \\
&\le
\lVert
\boldsymbol{\mathcal R}(\mathbf{H})
-
\boldsymbol{\mathcal R}(\mathbf{H}_0)
\rVert_\infty
\lVert \mathbf{e}_m \rVert_\infty
\le
\epsilon_R.
\end{aligned}
\]

Observe that
\[
\boldsymbol{\mathcal R}(\mathbf{H})^2
=
I_0^2\mathbf{H}^{-2},
\]
and hence
\[
\boldsymbol{\mathcal Z}(\mathbf{H})
=
I_0\mathbf{H}^{-2}
=
\frac{1}{I_0}
\boldsymbol{\mathcal R}(\mathbf{H})^2.
\]
Therefore,
\[
\begin{aligned}
\boldsymbol{\mathcal Z}(\mathbf{H})
-
\boldsymbol{\mathcal Z}(\mathbf{H}_0)
&=
\frac{1}{I_0}
\left[
\boldsymbol{\mathcal R}(\mathbf{H})^2
-
\boldsymbol{\mathcal R}(\mathbf{H}_0)^2
\right] \\
&=
\frac{1}{I_0}
\Big[
\left(
\boldsymbol{\mathcal R}(\mathbf{H})
-
\boldsymbol{\mathcal R}(\mathbf{H}_0)
\right)
\boldsymbol{\mathcal R}(\mathbf{H})
\\
&\qquad\qquad+
\boldsymbol{\mathcal R}(\mathbf{H}_0)
\left(
\boldsymbol{\mathcal R}(\mathbf{H})
-
\boldsymbol{\mathcal R}(\mathbf{H}_0)
\right)
\Big].
\end{aligned}
\]
By submultiplicativity,
\[
\begin{aligned}
\lVert
\boldsymbol{\mathcal Z}(\mathbf{H})
-
\boldsymbol{\mathcal Z}(\mathbf{H}_0)
\rVert_\infty
&\le
\frac{1}{I_0}
\lVert
\boldsymbol{\mathcal R}(\mathbf{H})
-
\boldsymbol{\mathcal R}(\mathbf{H}_0)
\rVert_\infty
\\
&\quad\times
\left(
\lVert
\boldsymbol{\mathcal R}(\mathbf{H})
\rVert_\infty
+
\lVert
\boldsymbol{\mathcal R}(\mathbf{H}_0)
\rVert_\infty
\right).
\end{aligned}
\]
Moreover, the triangle inequality gives
\[
\begin{aligned}
\lVert
\boldsymbol{\mathcal R}(\mathbf{H})
\rVert_\infty
&\le
\lVert
\boldsymbol{\mathcal R}(\mathbf{H}_0)
\rVert_\infty
+
\lVert
\boldsymbol{\mathcal R}(\mathbf{H})
-
\boldsymbol{\mathcal R}(\mathbf{H}_0)
\rVert_\infty
\\
&\le
R_{\max}+\epsilon_R.
\end{aligned}
\]
Consequently,
\[
\begin{aligned}
\lVert
\boldsymbol{\mathcal Z}(\mathbf{H})
-
\boldsymbol{\mathcal Z}(\mathbf{H}_0)
\rVert_\infty
&\le
\frac{\epsilon_R}{I_0}
\left(
R_{\max}+\epsilon_R+R_{\max}
\right)
\\
&=
\frac{\epsilon_R}{I_0}
\left(
2R_{\max}+\epsilon_R
\right)
\equiv
\epsilon_Z.
\end{aligned}
\]
Thus, for every fixed input node $m$,
\[
\lVert
\mathbf{Z}^{(m)}(\mathbf{H})
-
\mathbf{Z}^{(m)}(\mathbf{H}_0)
\rVert_\infty
\le
\epsilon_Z.
\]
Consider any pair $i\neq j$ such that
\[
R_i^{(m)}(\mathbf{H}_0)
>
R_j^{(m)}(\mathbf{H}_0).
\]
By definition of $\Delta_R^{(m)}$,
\[
R_i^{(m)}(\mathbf{H}_0)
-
R_j^{(m)}(\mathbf{H}_0)
\ge
\Delta_R^{(m)}.
\]
Using the componentwise perturbation bound,
\[
\left|
R_i^{(m)}(\mathbf{H})
-
R_i^{(m)}(\mathbf{H}_0)
\right|
\le
\epsilon_R
\]
and
\[
\left|
R_j^{(m)}(\mathbf{H})
-
R_j^{(m)}(\mathbf{H}_0)
\right|
\le
\epsilon_R.
\]
Therefore,
\[
\begin{aligned}
R_i^{(m)}(\mathbf{H})
-
R_j^{(m)}(\mathbf{H})
&=
R_i^{(m)}(\mathbf{H}_0)
-
R_j^{(m)}(\mathbf{H}_0)
\\
&\quad+
\left[
R_i^{(m)}(\mathbf{H})
-
R_i^{(m)}(\mathbf{H}_0)
\right]
\\
&\quad-
\left[
R_j^{(m)}(\mathbf{H})
-
R_j^{(m)}(\mathbf{H}_0)
\right]
\\
&\ge
\Delta_R^{(m)}
-
2\epsilon_R
>
0.
\end{aligned}
\]
Hence
\[
R_i^{(m)}(\mathbf{H})
>
R_j^{(m)}(\mathbf{H})
\]
whenever
\[
2\epsilon_R
<
\Delta_R^{(m)}.
\]
By the reverse triangle
inequality,
\[
\begin{aligned}
\left|
R_i^{(m)}(\mathbf{H})
\right|
&\ge
\left|
R_i^{(m)}(\mathbf{H}_0)
\right|
-
\left|
R_i^{(m)}(\mathbf{H})
-
R_i^{(m)}(\mathbf{H}_0)
\right|
\\
&\ge
r_*^{(m)}-\epsilon_R
>
0.
\end{aligned}
\]
Thus $\tau_i^{(m)}(\mathbf{H})$ is well-defined. Adding and subtracting
\[
\frac{
Z_i^{(m)}(\mathbf{H}_0)
}{
R_i^{(m)}(\mathbf{H})
}
\]
gives
\[
\begin{aligned}
&\left|
\tau_i^{(m)}(\mathbf{H})
-
\tau_i^{(m)}(\mathbf{H}_0)
\right|
=
\left|
\frac{
Z_i^{(m)}(\mathbf{H})
}{
R_i^{(m)}(\mathbf{H})
}
-
\frac{
Z_i^{(m)}(\mathbf{H}_0)
}{
R_i^{(m)}(\mathbf{H}_0)
}
\right|
\\
&\le
\left|
\frac{
Z_i^{(m)}(\mathbf{H})
-
Z_i^{(m)}(\mathbf{H}_0)
}{
R_i^{(m)}(\mathbf{H})
}
\right|
\\
&\quad+
\left|
Z_i^{(m)}(\mathbf{H}_0)
\right|
\left|
\frac{1}{
R_i^{(m)}(\mathbf{H})
}
-
\frac{1}{
R_i^{(m)}(\mathbf{H}_0)
}
\right|.
\end{aligned}
\]
The first term satisfies
\[
\left|
\frac{
Z_i^{(m)}(\mathbf{H})
-
Z_i^{(m)}(\mathbf{H}_0)
}{
R_i^{(m)}(\mathbf{H})
}
\right|
\le
\frac{\epsilon_Z}{r_*^{(m)}-\epsilon_R}.
\]
For the second term,
\[
\begin{aligned}
\left|
\frac{1}{
R_i^{(m)}(\mathbf{H})
}
-
\frac{1}{
R_i^{(m)}(\mathbf{H}_0)
}
\right|
&=
\frac{
\left|
R_i^{(m)}(\mathbf{H})
-
R_i^{(m)}(\mathbf{H}_0)
\right|
}{
\left|
R_i^{(m)}(\mathbf{H})
\right|
\left|
R_i^{(m)}(\mathbf{H}_0)
\right|
}
\\
&\le
\frac{
\epsilon_R
}{
r_*^{(m)}(r_*^{(m)}-\epsilon_R)
}.
\end{aligned}
\]

Furthermore,
\[
\boldsymbol{\mathcal Z}(\mathbf{H}_0)
=
\frac{1}{I_0}
\boldsymbol{\mathcal R}(\mathbf{H}_0)^2,
\]
so
\[
\begin{aligned}
\lVert
\boldsymbol{\mathcal Z}(\mathbf{H}_0)
\rVert_\infty\le
\frac{1}{I_0}
\lVert
\boldsymbol{\mathcal R}(\mathbf{H}_0)
\rVert_\infty^2\le
\frac{R_{\max}^2}{I_0}.
\end{aligned}
\]
Consequently,
\[
\left|
Z_i^{(m)}(\mathbf{H}_0)
\right|
\le
\frac{R_{\max}^2}{I_0}.
\]
Combining the preceding estimates gives
\[
\begin{aligned}
\left|
\tau_i^{(m)}(\mathbf{H})
-
\tau_i^{(m)}(\mathbf{H}_0)
\right|
&\le
\frac{\epsilon_Z}{r_*^{(m)}-\epsilon_R}
+
\frac{
R_{\max}^2\epsilon_R
}{
I_0r_*^{(m)}(r_*^{(m)}-\epsilon_R)
}
\equiv
\epsilon_\tau.
\end{aligned}
\]
Substituting
\[
\epsilon_Z
=
\frac{\epsilon_R}{I_0}
\left(
2R_{\max}+\epsilon_R
\right)
\]
yields the equivalent expression
\[
\epsilon_\tau
=
\frac{\epsilon_R}{
I_0(r_*^{(m)}-\epsilon_R)
}
\left(
2R_{\max}
+
\epsilon_R
+
\frac{R_{\max}^2}{r_*^{(m)}}
\right).
\]
Similarly, whenever
\[
2\epsilon_\tau
<
\Delta_\tau^{(m)},
\]
the order preserves.
\end{proof}

\begin{lemma}[\textcolor{black}{Perturbation of the steady-state covariance}]
Let $\mathbf{H}_0$ be a NSDDM matrix and let $Q$ be a symmetric positive semidefinite matrix. Define the Lyapunov operator and its perturbation by
\[
\mathcal L_{\mathbf{H}}(P) = \mathbf{H}P + P\mathbf{H}^\top, \qquad
\Delta \mathcal L_{\mathbf{E}}(P) = \mathbf{E}P + P\mathbf{E}^\top,
\]
so that $\mathcal L_{\mathbf{H}} = \mathcal L_{\mathbf{H}_0} + \Delta \mathcal L_{\mathbf{E}}$. 
Let $P_\infty(\mathbf{H}_0)$ denote the unique solution of
\[
\mathcal L_{\mathbf{H}_0}(P_\infty(\mathbf{H}_0)) + Q = 0.
\]
If the induced operator norm satisfies
\[
\alpha_E \equiv \left\lVert \mathcal L_{\mathbf{H}_0}^{-1} \Delta \mathcal L_{\mathbf{E}} \right\rVert_{\mathrm{op}} < 1,
\]
then $\mathcal L_{\mathbf{H}}$ is invertible, the unique solution $P_\infty(\mathbf{H})$ of $\mathcal L_{\mathbf{H}}(P_\infty(\mathbf{H})) + Q = 0$ exists, and
\[
P_\infty(\mathbf{H})
=
\left( I + \mathcal L_{\mathbf{H}_0}^{-1} \Delta \mathcal L_{\mathbf{E}} \right)^{-1}
P_\infty(\mathbf{H}_0).
\]
Consequently, the following uniform perturbation bound holds:
\[
\left\lVert P_\infty(\mathbf{H}) - P_\infty(\mathbf{H}_0) \right\rVert_{\mathrm{op}}
\le
\frac{\alpha_E}{1-\alpha_E}
\left\lVert P_\infty(\mathbf{H}_0) \right\rVert_{\mathrm{op}}.
\]
\end{lemma}

\begin{proof}
Since $\mathbf{H}_0$ is a NSDDM matrix, the Lyapunov operator $\mathcal L_{\mathbf{H}_0}$ is invertible and $P_\infty(\mathbf{H}_0) = -\mathcal L_{\mathbf{H}_0}^{-1}(Q)$. For $\mathbf{H} = \mathbf{H}_0 + \mathbf{E}$, the steady covariance satisfies
\[
\mathcal L_{\mathbf{H}_0}(P_\infty(\mathbf{H})) + \Delta \mathcal L_{\mathbf{E}}(P_\infty(\mathbf{H})) + Q = 0.
\]
Applying $\mathcal L_{\mathbf{H}_0}^{-1}$ to both sides and using $P_\infty(\mathbf{H}_0) = -\mathcal L_{\mathbf{H}_0}^{-1}(Q)$ gives
\[
P_\infty(\mathbf{H}) + \mathcal L_{\mathbf{H}_0}^{-1} \Delta \mathcal L_{\mathbf{E}}(P_\infty(\mathbf{H})) = P_\infty(\mathbf{H}_0).
\]
As $\lVert \mathcal L_{\mathbf{H}_0}^{-1} \Delta \mathcal L_{\mathbf{E}} \rVert_{\mathrm{op}} = \alpha_E < 1$, so the Neumann series guarantees that $(I + \mathcal L_{\mathbf{H}_0}^{-1} \Delta \mathcal L_{\mathbf{E}})$ is invertible. Hence
\[
P_\infty(\mathbf{H}) = (I + \mathcal L_{\mathbf{H}_0}^{-1} \Delta \mathcal L_{\mathbf{E}})^{-1} P_\infty(\mathbf{H}_0).
\]
Moreover,
\[
(I + \mathcal L_{\mathbf{H}_0}^{-1} \Delta \mathcal L_{\mathbf{E}})^{-1} - I = \sum_{k=1}^{\infty} (-\mathcal L_{\mathbf{H}_0}^{-1} \Delta \mathcal L_{\mathbf{E}})^k,
\]
so taking operator norms yields
\[
\left\lVert (I + \mathcal L_{\mathbf{H}_0}^{-1} \Delta \mathcal L_{\mathbf{E}})^{-1} - I \right\rVert_{\mathrm{op}}
\le
\sum_{k=1}^{\infty} \lVert \mathcal L_{\mathbf{H}_0}^{-1} \Delta \mathcal L_{\mathbf{E}} \rVert_{\mathrm{op}}^k
=
\frac{\alpha_E}{1-\alpha_E}.
\]
\end{proof}
\endgroup

\section{DERIVATION OF METRICS ACROSS INPUTS} \label{sec:B}

\renewcommand\theequation{B\arabic{equation}}
\setcounter{equation}{0}

\subsection*{Constant input}
\label{sec:B:constant}
\par In this subsection, we analyze the system response to constant input (i.e., Heaviside step function), a typical and analytically tractable case often used in large-scale complex systems.

Assume the system is initially at steady state at $t = 0$, where $\mathbf{I}(0)$ satisfies $0 = \dot{\mathbf{x}}(0) = \mathbf{H} \mathbf{x}(0) + \mathbf{I}(0)$. Define variations from this steady state as $\Delta \mathbf{x}(t) \equiv \mathbf{x}(t) - \mathbf{x}(0)$ and $\Delta \mathbf{I}(t) \equiv \mathbf{I}(t) - \mathbf{I}(0)$, yielding the dynamics  
$\Delta \dot{\mathbf{x}}(t) = \mathbf{H} \Delta \mathbf{x}(t) + \Delta \mathbf{I}(t)$.

\par For constant input $\mathbf{I}_0^\text{const} \equiv \Delta \mathbf{I}(t)$, the system experiences a constant perturbation that eventually stabilizes to a final steady state $\Delta \mathbf{x}(\infty)$. By asymptotic stability,  
\[
0 = \Delta \dot{\mathbf{x}}(\infty) = \mathbf{H} \Delta \mathbf{x}(\infty) + \mathbf{I}_0^\text{const},
\]  
and the \textit{peak response} at node $i$ is  
\begin{equation}
\refstepcounter{equation}\label{peak_response}
\tag{B\arabic{equation}}
    R_i \triangleq \Delta x_i(\infty) = -[\mathbf{H}^{-1} \mathbf{I}_0^\text{const}]_i.
\end{equation}

\par Define $\overline{\Delta \mathbf{x}(t)} \equiv \Delta \mathbf{x}(t) - \Delta \mathbf{x}(\infty)$ as the deviation from final steady state. Then,  
\[
\dot{\overline{\Delta \mathbf{x}(t)}} = \mathbf{H} \overline{\Delta \mathbf{x}(t)},
\]  
with solution  
\begin{equation}
\refstepcounter{equation}\label{general_dynamic}
\tag{B\arabic{equation}}
    \overline{\Delta x^\text{const}_i(t)} \equiv [e^{\mathbf{H}t}\overline{\Delta \mathbf{x}(0)}]_i,
\end{equation}
where $\overline{\Delta \mathbf{x}(0)} = -\Delta \mathbf{x}(\infty) = \mathbf{H}^{-1} \mathbf{I}_0^\text{const}$.

\par Note: the form of $\overline{\Delta \mathbf{x}(t)}$ can be regarded as the evolution in the impulse response case (see Eq. (\ref{pulse:eq})), but with multiple nonzero elements in $\overline{\Delta \mathbf{x}(0)}$, since each column of $\mathbf{H}^{-1}$ typically contains multiple nonzero entries due to the weak connectivity of the graph.

\par The evolution can also be expressed as $\overline{\Delta x_i^\text{const}(t)} = \sum_j {(u_{i m}^j e^{\lambda_j t} I_0)}/{\lambda_j}$ in response to a scalar input $I_0$ at node $m$, provided that $\mathbf{H}$ is diagonalizable. Here, $u_{i m}^j = [\mathbf{U}]_{ij} [\mathbf{U}^{-1}]_{jm}$, where $\mathbf{U}$ and $\mathbf{U}^{-1}$ are the eigenmatrix and its inverse, respectively. If $\mathbf{H}$ is a normal or real symmetric matrix, $\mathbf{U}^{-1}$ can be replaced by $\mathbf{U}^*$ or $\mathbf{U}^\top$, respectively.

\par Substituting $\overline{\Delta \mathbf{x}(0)}$ into Eq.~(\ref{general_dynamic}) and adding $\Delta \mathbf{x}(\infty)$ yields the full time-dependent activity at node $i$:
\begin{equation}
\refstepcounter{equation}\label{constant_course}
\tag{B\arabic{equation}}
\Delta x^\text{const}_i(t) \equiv \left[\mathbf{H}^{-1}(e^{\mathbf{H}t} - \mathbf{I}_N)\, \mathbf{I}_0^\text{const} \right]_i.
\end{equation}

\par In Fig.~\ref{constant_sketch}(a), we show the time course of a single node receiving constant input (red line), followed by the propagation of a constant input to other nodes along a directed chain, which responds more slowly and weakly (gray lines). To characterize temporal properties of $\Delta x^\text{const}_i(t)$ under constant input (Eq. (\ref{constant_course})), we introduce two temporal metrics: the \textit{relative propagation time} ${t}_{i}$, defined by the relative threshold ${\eta} = {\Delta x^\text{const}_i({t}_{i})}/{\Delta x^\text{const}_i(\infty)}$ \cite{hens2019spatiotemporal, bao2022impact}, and the \textit{absolute propagation time} $\widetilde{t}_{i}$, defined by the absolute threshold $\widetilde{\eta} = \Delta x^\text{const}_i(\widetilde{t}_{i})$.

\begin{figure}[!ht]
\centerline{\includegraphics[scale=0.8]{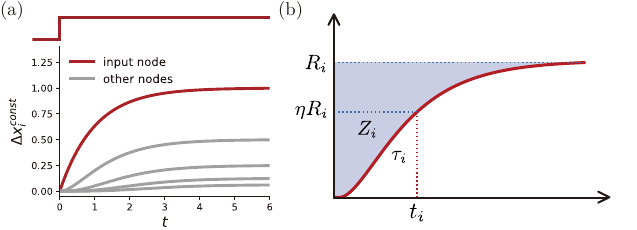}}
    \caption{
    Quantifying the impact of constant input on network dynamics.  
(a) Schematic showing the constant input and the resulting responses from the input node and other network nodes. The input node responds immediately with a stronger reaction, while other nodes respond more slowly and weakly.  
(b) Illustration of response metrics under constant input. Key metrics include peak response ($R_i$), amplification ($Z_i$), time constant ($\tau_i$), and relative propagation time (${t}_i$) at threshold ${\eta} R_i$.
}
\label{constant_sketch}
\end{figure}

\par Thresholds ${\eta}$ $(0 < {\eta} < 1)$ and $\widetilde{\eta}$ $(0 < \widetilde{\eta} < R_i)$ are chosen based on the timescale of interest: smaller values highlight early responses, while larger values capture slower, sustained dynamics. The half-maximum relative threshold ($\eta \sim \frac{1}{2}$) is commonly used in physics and chemistry, corresponding to the system's half-response time or half-period.
The time constant, corresponding to a $(1/e)$-fraction of the relative response, is also fundamental in physics, engineering, and neuroscience for describing convergence speed to steady state \cite{chaudhuri2014diversity}. The absolute propagation time $\widetilde{t}_i$ is particularly relevant when the absolute threshold itself is meaningful, such as in the global workspace theory of consciousness \cite{murphy2009balanced, joglekar2018inter}.

\par For the NSDDM system, we prove that $\overline{\Delta x^\text{const}_{i}(t)}$ monotonically increases and remains negative under positive constant inputs. This confirms the positivity and monotonic growth of $\Delta x^\text{const}_{i}(t)$ (Appendix~\ref{sec:A}). If negative entries $A_{ij}$ are included, the traces may not remain positive. However, the relative response time can still be estimated as ${\eta} = {|\Delta x^\text{const}_i ({t}_{i})|}/{|\Delta x^\text{const}_i(\infty)|}$, see Supplementary Material (SM) Sec. II.

\par Based on this, we approximate the activity variation of node $i$ using a single exponential function with \textit{time constant} $\tau_{i}$:
\begin{equation}
	\overline{\Delta x^\text{const}_{i}(t)} \approx [e^{-t / \tau_{i}}\overline{\Delta \mathbf{x}(0)}]_i.
\end{equation}

\par The time constant is derived by integrating the response over time. The integral, termed the response \textit{amplification}, is
\begin{equation}
   Z_i \triangleq - \int_{0}^{+\infty}\overline{\Delta x^\text{const}_{i}(t)} \, dt = [\mathbf{H}^{-2} \mathbf{I}_0^\text{const}]_i,
\end{equation}
and the time constant $\tau_i$ is approximated by:
\begin{equation}
\refstepcounter{equation}\label{const_time}
\tag{B\arabic{equation}}
\tau_i = -\frac{\left[\mathbf{H}^{-2} \mathbf{I}_0^\text{const}\right]_i}{\left[\mathbf{H}^{-1} \mathbf{I}_0^\text{const}\right]_i}.
\end{equation}

\par The time constant $\tau_i$ characterizes how quickly the response evolves: for decaying $\overline{\Delta x^\text{const}_i(t)}$, it marks the time to reach $1/e \approx 36.8\%$ of the initial value $\overline{\Delta x^\text{const}_i(0)}$; for growing $\Delta x^\text{const}_i(t)$, it corresponds to reaching $(1 - 1/e)$ of the steady-state value $\Delta x^\text{const}_i(\infty)$.

\par The relative propagation time is:
\begin{equation}\label{eq_time}
	{t}_i = -\tau_i \ln(1 - {\eta}),
\end{equation}
and the absolute propagation time is:
\begin{equation}\label{eq_time_abs}
	\widetilde{t}_i = -\tau_i \ln\left(1 - \frac{\widetilde{\eta}}{R_i}\right).
\end{equation}

\par The system's evolution can also be approximated by a sigmoid-shaped curve:
\begin{equation}
	\overline{\Delta x^\text{const}_{i}(t)} \approx \left[\frac{2}{1 + e^{t / \tau_i^s}} \, \overline{\Delta \mathbf{x}(0)}\right]_i.
\end{equation}
The sigmoid time constant is given by $\tau_i^s = {\tau_i}/{(2 \ln 2)}$. The corresponding relative and absolute propagation times are:
\begin{align*}
{t}_i^s &= -\frac{\tau_i}{2 \ln 2} \ln\left(\frac{1 - {\eta}}{1 + {\eta}}\right), \quad \\
\widetilde{t}_i^s &= -\frac{\tau_i}{2 \ln 2} \ln\left(\frac{\Delta x^\text{const}_i(\infty) - \widetilde{\eta}}{\Delta x^\text{const}_i(\infty) + \widetilde{\eta}}\right).
\end{align*}
\par For a fixed ${\eta}$, both exponential and sigmoid models yield a relative propagation time proportional to the time constant $\tau_i$. The absolute propagation time additionally depends on the peak response. The peak response is determined by the product of the input amplitude and $\mathbf{H}^{-1}$, while the amplification results from multiplying the input amplitude by $\mathbf{H}^{-2}$. Their ratio gives the time constant. Fig.~\ref{constant_sketch}(b) summarizes these metrics for quantifying response dynamics under constant input. Numerical simulations across a broad range of network topologies\DIFdelbegin \DIFdel{and connection strengths }\DIFdelend \DIFaddbegin {, connection strengths and settings beyond NSDDM }\DIFaddend are shown in SM Sec. II.  
\par All the metrics receiving  constant input at source $m$ targeted with $i$ are
\begin{equation}
\begin{aligned}
Z_i^{(m)} &\triangleq [\mathbf{H}^{-2}]_{im} I_0^{\text{const}} = \sum_{j=1} \frac{u_{im}^j}{\lambda_j^2} I_0^{\text{const}} \sim O(1/\lambda_1^2), \\
R_i^{(m)} &\triangleq -[\mathbf{H}^{-1}]_{im} I_0^{\text{const}} = -\sum_{j=1} \frac{u_{im}^j}{\lambda_j} I_0^{\text{const}} \sim O(1/\lambda_1), \\  
\tau_i^{(m)} &\triangleq -\frac{[\mathbf{H}^{-2}]_{im}}{[\mathbf{H}^{-1}]_{im}} = -\frac{\sum_{j=1} \frac{u_{im}^j}{\lambda_j^2}}{\sum_{j=1} \frac{u_{im}^j}{\lambda_j}} \sim O(1/\lambda_1), \\ 
t_i^{(m)} &\triangleq -\tau_{im} \ln(1 - \eta) \sim O(1/\lambda_1).
\end{aligned}
\label{constant}
\end{equation}

\subsection*{Pulse input}
\label{sec:B:pulse}  
\par Understanding the impulse response of a linear time-invariant system is essential for analyzing its dynamics. The system's output to any input can be constructed from its impulse response in LTI systems.

\par We first consider the case without external input $\mathbf{I}(t)$. With initial condition $\mathbf{x}_0 \triangleq \Delta \mathbf{x}(0)$ at $t = 0$, we analyze the response defined as $\Delta \mathbf{x}(t) \equiv \mathbf{x}(t) - \mathbf{x}(0)$. The dynamics follow $\Delta \dot{\mathbf{x}}(t) = \mathbf{H} \Delta \mathbf{x}(t)$ with solution $\Delta x_i^\text{pulse}(t) \equiv [e^{\mathbf{H}t}\, \mathbf{x}_0]_i$.

\par Alternatively, introducing an external input $\delta \mathbf{I}(t)$ modeled as a Dirac delta function, and defining $\mathbf{I}_0^\text{pulse} \equiv \int_0^{\infty} \delta \mathbf{I}(t)\, dt$, the response becomes:
\begin{equation}
\refstepcounter{equation}\label{pulse:eq}
\tag{B\arabic{equation}}
\begin{aligned}
\Delta x_i^\text{pulse}(t) &\equiv \left[\int_0^t \exp[\mathbf{H}(t - \tau)]\, \delta \mathbf{I}(\tau)\, d\tau\right]_i, \\
&= [e^{\mathbf{H}t}\, \mathbf{I}_0^\text{pulse}]_i.
\end{aligned}
\end{equation}

\par These two methods produce identical outputs when $\mathbf{x}_0 = \mathbf{I}_0^\text{pulse}$. In practice, since a Dirac delta is not implementable in simulation, we approximate it by applying $\mathbf{I}_0^\text{pulse}$ at the first time step.

\par Additionally, if $\mathbf{H}$ is diagonalizable, the response to an input $I_0$ at node $m$ can also be written as:  
\begin{equation}
\Delta x_i^\text{pulse}(t) = \sum_j u_{i m}^j e^{\lambda_j t} I_0.
\label{eq:eigen_pulse}
\end{equation}  

\begin{figure}[!ht]
\centerline{\includegraphics[scale=0.8]{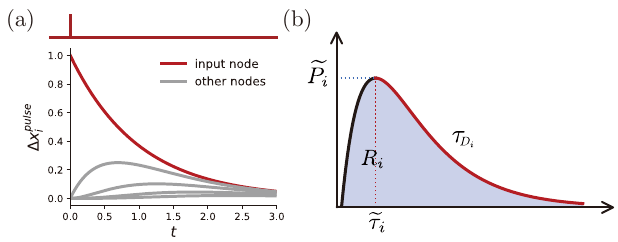}}
    \caption{Quantifying the impact of a pulse input on network dynamics.  
(a) Schematic showing the pulse input and resulting responses from the input node and other network nodes. The input node exhibits immediate decay from the initial input, whereas other nodes rise first and then decay, reaching lower peaks.  
(b) Illustration of response metrics for pulse input. Key metrics include the amended peak response (${\widetilde{P}_i}$), amplification ($R_i$), amended peak time ($\widetilde{\tau}_i$), and decay time constant ($\tau_{D_i}$).}
\label{pulse_sketch}
\end{figure}

\par The complete time courses are shown in Fig.~\ref{pulse_sketch}(a). The input node responds to a single pulse with immediate decay (Appendix~\ref{sec:A}), while responses of other nodes increase first and decay slowly in the NSDDM system. Although the full temporal evolution can be described analytically, calculating specific temporal metrics is challenging due to the presence of transcendental equations. Methods have been developed to characterize system responses using effective probability distributions \cite{wolter2018quantifying, schroder2019dynamic}. This approach normalizes the response into a probability distribution, providing a compact view of transient dynamics, briefly introduced below.

\par In the NSDDM system, the response of each node to a positive input remains positive over time \cite{wolter2018quantifying}. Normalizing the response trajectory by the total response strength, defined as \textit{amplification}
\begin{equation}
\refstepcounter{equation}\label{pulse_amp}
\tag{B\arabic{equation}}
R_i \triangleq \int_0^{\infty} \Delta x^\text{pulse}_i(t)\, dt = -[\mathbf{H}^{-1} \mathbf{I}_0^\text{pulse}]_i, 
\end{equation}
yields the \textit{probability density}  
\[
\rho_i(t) \triangleq \frac{\Delta x_i^\text{pulse}(t)}{R_i}.
\]
\textit{Peak response time} is then defined as the expected value:
\begin{equation}
\refstepcounter{equation}\label{pulse_time}
\tag{B\arabic{equation}}
\tau_i\triangleq \int_0^{\infty} t \rho_i(t)\, dt = -\frac{\left[\mathbf{H}^{-2} \mathbf{I}_0^\text{pulse}\right]_i}{\left[\mathbf{H}^{-1} \mathbf{I}_0^\text{pulse}\right]_i}.
\end{equation}

\par Notably, the expressions for Eqs.~(\ref{const_time}) and (\ref{pulse_time}) are similar under the same input amplitude and location. For constant input, the time constant is estimated when the response reaches $(1 - 1/e)$ of the steady state; for pulse input, it estimates the time to peak response.

\par Additionally, the duration and amplitude of the response are characterized by the \textit{standard deviation} $\sigma_i$ and the \textit{peak response} $P_i$. The peak response is defined as the ratio of amplification to standard deviation. In analogy to the normal distribution, where the total area is fixed at $1$, a higher peak implies a narrower spread (shorter duration), and a lower peak implies a broader response.

\begin{equation}
\begin{aligned}
\sigma_i & \triangleq \sqrt{\int_0^{\infty}\left(t-\tau_i \right)^2 \rho_i(t) d t}  \\
&=\sqrt{\frac{2\left[\mathbf{H}^{-3} \mathbf{I}^\text{pulse}_0\right]_i}{\left[\mathbf{H}^{-1} \mathbf{I}^\text{pulse}_0\right]_i}-\left(\frac{\left[\mathbf{H}^{-2} \mathbf{I}_0^\text{pulse}\right]_i}{\left[\mathbf{H}^{-1} \mathbf{I}_0^\text{pulse}\right]_i}\right)^2}, \\
P_i&\triangleq \frac{R_i}{\sigma_i}\\ &=\frac{\left(\left[\mathbf{H}^{-1} \mathbf{I}_0^\text{pulse}\right]_i\right)^2}{\sqrt{2\left[\mathbf{H}^{-3} \mathbf{I}^\text{pulse}_0\right]_i\left[\mathbf{H}^{-1} \mathbf{I}^\text{pulse}_0\right]_i-\left(\left[\mathbf{H}^{-2} \mathbf{I}^\text{pulse}_0\right]_i\right)^2}}.
\end{aligned}
\end{equation}

\par A refined method \cite{schroder2019dynamic} improves accuracy by correcting estimation bias:

\begin{equation}
\begin{aligned}
\widetilde{\tau}_i &\triangleq \tau_i + \frac{1}{\lambda_1}, \\
\widetilde{P}_i &\triangleq \frac{\sqrt{d+1}\, d^d}{e^d\, d!} P_i = \frac{P_i}{\sqrt{2\pi}} + O(d^{-1}),
\end{aligned}
\label{refined}
\end{equation}
where $\widetilde{\tau}_i$ is the refined peak time, and $\widetilde{P}_i$ is the refined peak response. The correction involves $\lambda_1 \equiv \max_j \operatorname{Re}(\lambda_j)<0$, which is the dominant eigenvalue of $\mathbf{H}$. The first expression in $\widetilde{P}_i$ applies for small shortest path lengths $d$, while the asymptotic form applies for large $d$.
\par These corrections are based on the observation that the response approximates the form $\Delta x^\text{pulse}_i(t) = t^d e^{\lambda_1 t}$, where $d$ is the shortest path length from the input node to node $i$. The rising phase scales with $t^d$, and the decay phase is governed by $\lambda_1$. This form is similar to the alpha function seen in chain-like structures (Appendix~\ref{sec:C}). The refined estimators are especially suited for weakly coupled systems, where the identical coupling strength $\alpha$ satisfies $\alpha / \lambda_1 \rightarrow 0$ \cite{schroder2019dynamic}.

\begin{figure}[!ht]
    \centerline{\includegraphics[scale=0.8]{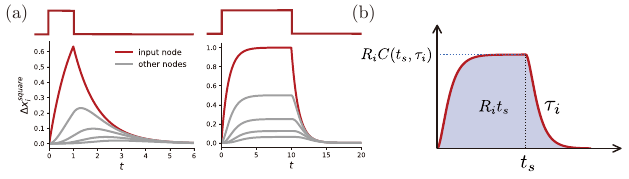}}
    \caption{Quantifying the impact of a square input on network dynamics.  
(a) Schematic showing the square input and corresponding responses from the input node and other nodes for both short and long durations. For short durations, the response shape resembles the impulse response in downstream nodes.  
(b) Illustration of response metrics for a square input with duration time $t_s$. Key metrics include the peak response ($R_{i}f(t_s,\tau_{i})$), amplification ($R_it_s$), and time constant ($\tau_i$).
}
\label{square}
\end{figure}

\par Notably, the impulse response curve is often asymmetric: rising sharply to its peak and then decaying more gradually, typically following an exponential-form trend. To characterize this decay phase, we estimate a decay rate once the peak response and its timing have been identified. We define a \textit{decay time constant} $\tau_{D_i}$, approximating the descent using a single exponential function: $\int_{\widetilde{\tau}_i}^{+\infty} [e^{\mathbf{H}t}\mathbf{I}_0^\text{pulse}]_i\, dt = \widetilde{P}_i \int_0^{+\infty} e^{-t/\tau_{D_i}}\, dt$, which yields: $\tau_{D_i} = -\frac{[e^{\mathbf{H} \widetilde{\tau}_i} \mathbf{H}^{-1} \mathbf{I}_0^\text{pulse}]_i}{\widetilde{P}_i}$.

\par This expression involves a matrix exponential and two estimated metrics ${\widetilde{\tau}_i}$ and ${\widetilde{P}_i}$, which is hard for analysis. Additional insight can be gained by estimating the area under the impulse response from $\widetilde{\tau}_i$ to infinity, approximating $(1-1/e) R_i$. This leads to a simplified estimate for the decay rate:
\begin{equation}
    \tau_{D_i} = \left(1 - \frac{1}{e}\right) \frac{R_i}{\widetilde{P}_i}.
\end{equation}

\par Three limitations for estimations might emerge, especially for the temporal metrics: (i) strong interactions may induce multi-peak responses in recurrent loop structures (See SM Fig. \DIFdelbegin \DIFdel{S13}\DIFdelend \DIFaddbegin \DIFadd{S16}\DIFaddend ), (ii) hub-nonhub node pairs in scale-free networks can exhibit obvious asymmetric time course, which leads to biased estimations for peak response time (See SM Figs. \DIFdelbegin \DIFdel{S11, }\DIFdelend S14\DIFaddbegin \DIFadd{, S17}\DIFaddend ), and (iii) inhibitory connections may violate non-negativity assumptions (See SM Figs. \DIFdelbegin \DIFdel{S15, S16}\DIFdelend \DIFaddbegin \DIFadd{S41, S42}\DIFaddend ). However, these limitations mainly affect numerical accuracy but \DIFaddbegin {mostly }\DIFaddend do not alter the response order.

\par \DIFaddbegin {We emphasize here that, even under the same network structure and parameter setting, different response metrics do not in general preserve the same ordering across input types. In a single chain, it is natural to expect a clear ordering of which nodes respond earlier or later, or more strongly or weakly. However, this intuition does not generally extend to networks with multiple walks and prominent hubs, especially in dense networks. In such cases, the ranking of nodes can differ substantially across input types (see SM Figs. S32 and S33). 
}

\par \DIFaddend All the metrics receiving  pulse input at source $m$ targeted with $i$ are
\begin{equation}
\begin{aligned}
R_i^{(m)} &\triangleq -[\mathbf{H}^{-1}]_{im} I_0^{\text{pulse}} = -\sum_{j=1} \frac{u_{im}^j}{\lambda_j} I_0^{\text{pulse}} \sim O(1/\lambda_1), \\
P_i^{(m)} &\triangleq C(d) \frac{([\mathbf{H}^{-1}]_{im})^2 I_0^{\text{pulse}}}{\sqrt{2[\mathbf{H}^{-3}]_{im}[\mathbf{H}^{-1}]_{im} - ([\mathbf{H}^{-2}]_{im})^2}} \\
&= C(d) \frac{\left(\sum_p \frac{u_{im}^p}{\lambda_p}\right)^2 I_0^{\text{pulse}}}{\sqrt{\sum_{p,q} \frac{u_{im}^p u_{im}^q}{\lambda_p^2 \lambda_q(2\lambda_p - \lambda_q)}}} \sim O(1), \\
\tau_{D_i^{(m)}} &\triangleq \left(1 - \frac{1}{e}\right) \frac{R_i^{(m)}}{\widetilde{P}_i^{(m)}} \sim O(1/\lambda_1), \\
\widetilde{\tau}_i^{(m)} &\triangleq -\frac{[\mathbf{H}^{-2}]_{im}}{[\mathbf{H}^{-1}]_{im}} + \frac{1}{\lambda_1} = -\frac{\sum_j \frac{u_{im}^j}{\lambda_j^2}}{\sum_j \frac{u_{im}^j}{\lambda_j}} + \frac{1}{\lambda_1} \sim O(1/\lambda_1). \label{pulse}
\end{aligned}
\end{equation}

\subsection*{Square input}
\label{sec:B:square}
\par Square inputs are commonly used for their simplicity and analytical tractability, especially in neuroscience and image processing. Conceptually, a square input is a truncated constant input. Assuming exponential evolution with the corresponding time constant (Eq. (\ref{const_time})), the \textit{peak response} is estimated as:
\begin{equation}
R_{i}f(t_s,\tau_{i}) \triangleq -(1 - e^{-t_s / \tau_i}) [\mathbf{H}^{-1} \mathbf{I}_0^\text{square}]_i.
\end{equation}
As $t_s \to \infty$, $f(t_s,\tau_{i}) \to 1$.

\par \textit{Amplification} combines two parts: the truncated constant-input phase and the decay phase, yielding:
\begin{equation}
R_it_s = -t_s [\mathbf{H}^{-1} \mathbf{I}_0^\text{square}]_i.
\end{equation}
All the metrics are shown in Fig. \ref{square}(b).

\par Under unit input amplitude and duration conditions ($I_0 = 1$, $t_s = 1$), amplification equals both the pulse-integrated amplification and the constant-input peak response $R_i$ (Eqs. \eqref{peak_response}, \eqref{pulse_amp}). Amplification increases linearly with $t_s$ at rate $R_i$ and remains robust across parameter regimes (See SM Fig.~\DIFdelbegin \DIFdel{S17}\DIFdelend \DIFaddbegin \DIFadd{S18}\DIFaddend ). As shown in Fig.~\ref{square}(a), for short durations, the response closely resembles the impulse response. For unit-duration square input, peak responses match those of impulse input, excluding the input node (See SM Fig. \DIFdelbegin \DIFdel{S17}\DIFdelend \DIFaddbegin \DIFadd{S18}\DIFaddend ).  

\par Peak response accuracy may degrade under: weak interactions and long paths ($\alpha \ll \beta$; See SM Fig. \DIFdelbegin \DIFdel{S17}\DIFdelend \DIFaddbegin \DIFadd{S18}\DIFaddend (k), \DIFdelbegin \DIFdel{S18}\DIFdelend \DIFaddbegin \DIFadd{S19}\DIFaddend (d)); or short duration input ($t_s \ll \tau_{im}$), where non-exponential transients dominate \cite{schroder2019dynamic}.

\subsection*{Noise input}
\label{sec:B:noise}
\par We consider white noise input in this subsection with zero mean and spectral density $\mathbf{Q}$. As white noise is discontinuous and unbounded, we adopt the Itô interpretation and rewrite the equation as $d\mathbf{x} = \mathbf{H}\mathbf{x}\,dt + d\boldsymbol{\beta}$, where $\mathbf{I} = d\boldsymbol{\beta}/dt$ and $\boldsymbol{\beta}$ is Brownian motion. The complete solution is given by:
\begin{equation}
\mathbf{x}^{\text{noise}}(t) = e^{\mathbf{H}t}\mathbf{x}^{\text{noise}}(0) + \int_0^t e^{\mathbf{H}(t-\tau)}\, d\boldsymbol{\beta}(\tau).
\end{equation}

\par Taking expectations and covariances yields:
\begin{equation}
\begin{aligned}
&\mathbb{E}[\mathbf{x}(t)] 
= e^{\mathbf{H}t}\, \mathbb{E}[\mathbf{x}(0)] 
\triangleq \mathbf{m}(t), \\[6pt]
&\mathbb{E}\!\left[(\mathbf{x}(t) - \mathbf{m}(t))(\mathbf{x}(t) - \mathbf{m}(t))^{\top}\right] \\ 
&= e^{\mathbf{H}t}\, 
\mathbb{E}\!\left[(\mathbf{x}(0) - \mathbf{m}(0))(\mathbf{x}(0) - \mathbf{m}(0))^{\top}\right]
e^{\mathbf{H}^{\top}t}\\ &+ \int_0^t 
e^{\mathbf{H}(t - \tau)}\, \mathbf{Q}\, e^{\mathbf{H}^{\top}(t - \tau)}\, d\tau 
\triangleq \mathbf{P}(t).
\end{aligned}
\end{equation}

\par At steady state, we obtain:
\begin{equation}
\begin{aligned}
\frac{d \mathbf{m}(t)}{dt} &= \mathbf{H} \mathbf{m}(t) = 0, \\
\frac{d \mathbf{P}(t)}{dt} &= \mathbf{H} \mathbf{P}(t) + \mathbf{P}(t) \mathbf{H}^{\top} + \mathbf{Q} = 0,
\end{aligned}
\end{equation}
yielding $\mathbf{m}_\infty = 0$ and the Lyapunov equation $\mathbf{H} {\mathbf{P}}_\infty + {\mathbf{P}}_\infty \mathbf{H}^{\top} + \mathbf{Q} = 0$, where ${\mathbf{P}}_\infty$ can be solved numerically.

\par Taking the limit $t \to \infty$, the stationary covariance becomes:
\begin{equation}
\mathbf{C}(\tau) \triangleq \mathbb{E}[\mathbf{x}(t)\mathbf{x}(t - \tau)] =
\begin{cases}
{\mathbf{P}}_\infty \exp(-\mathbf{H} \tau)^{\top}, & \tau \leq 0 \\
\exp(\mathbf{H} \tau)\, {\mathbf{P}}_\infty, & \tau > 0
\end{cases}
\end{equation}
with $\mathbf{C}(\tau) = \mathbf{C}(-\tau)^{\top}$. The steady-state covariance ${\mathbf{P}}_\infty$ can also be expressed as $\int_{-\infty}^{t} e^{\mathbf{H}(t - \tau)}\, \mathbf{Q}\, e^{\mathbf{H}^{\top}(t - \tau)}\, d\tau$, or $\mathcal{L}^{-1}\left((s\mathbf{I}_N - \mathbf{H})^{-1} \mathbf{Q} (s\mathbf{I}_N - \mathbf{H}^{\top})^{-1}\right)$ via inverse Laplace transform. The complete derivation of the stationary covariance can be found in \cite{sarkka2019applied}.

\par If $\mathbf{Q} = \sigma^2 \mathbf{I}_N$ (i.e., uncorrelated white noise for all nodes) and $\mathbf{H}$ is normal ($\mathbf{H}\mathbf{H}^{\top} = \mathbf{H}^{\top}\mathbf{H}$), then ${\mathbf{P}}_\infty = -\sigma^2(\mathbf{H} + \mathbf{H}^{\top})^{-1}$. For a $1D$ system, the variance simplifies to $\mathbf{C}(\tau) = -{Q} e^{\lambda|\tau|}/(2\lambda)$. 

\par Crosscovariance between input and activity can also be derived.

\begin{equation}
\tilde{\mathbf{C}}(s) \triangleq \mathbb{E}[\mathbf{x}(t)\mathbf{I}^{\top}(t - s)] =
\begin{cases}
e^{\mathbf{H}s}\, \mathbf{Q}, & s \geq 0, \\
0, & s < 0.
\end{cases}
\end{equation}

\par When the input is applied only at node $m$, with strength $I_0^{\text{noise}}$, and $\mathbf{H}$ is diagonalizable, the $(p,q)$-th element of the covariance matrix in the eigendecomposition is given by:
\begin{equation}
\begin{aligned}
C_{pq}^m(\tau) &\triangleq \mathbb{E}[x_p(t)\, x_q(t - \tau)] \\
&=
\begin{cases}
- \sum_j \sum_k \frac{u_{pm}^j u_{qm}^k}{\lambda_j + \lambda_k} e^{-\lambda_k \tau} I_0^{\text{noise}}, & \tau \leq 0, \\
- \sum_j \sum_k \frac{u_{pm}^j u_{qm}^k}{\lambda_j + \lambda_k} e^{\lambda_j \tau} I_0^{\text{noise}}, & \tau > 0.
\end{cases}
\label{cov_diag}
\end{aligned}
\end{equation}

\par The corresponding element of the crosscovariance matrix between input and activity is \cite{oppenheim2017signals}:
\begin{equation}
\tilde{{C}}_{im}(s) \triangleq \mathbb{E}[x_i(t)\, I_m(t - s)] =
\begin{cases}
\sum_j u_{im}^j e^{\lambda_j s} I_0^{\text{noise}}, & s \geq 0, \\
0, & s < 0.
\label{cov_im}
\end{cases}
\end{equation}

\par For $s \geq 0$, the crosscovariance between the input and the activity as a function of lag $s$ (Eq. (\ref{cov_im})) resembles the impulse response (Eq. (\ref{eq:eigen_pulse})) when input strength $I_0$ and location $m$ are identical. This correspondence holds for LTI systems, allowing impulse-response-based metrics to be directly applied to crosscovariance analysis.

\par The diagonal entries of the covariance matrix (Eq. (\ref{cov_diag})) represent autocovariance, which are even functions, attaining their maximum at zero lag by the Cauchy-Schwarz inequality. Assuming an exponential form with \textit{peak response} $P_{ii} \triangleq [\mathbf{P}_{\infty}]_{ii}$ and decay governed by a \textit{time constant} $\tau_{ii}$, we write:
\begin{equation}
C_{ii}(\tau) \approx [\mathbf{P}_{\infty}]_{ii} e^{-\tau / \tau_{ii}}.
\end{equation}

\par This form characterizes the time constant of the autocovariance and, upon normalization, the autocorrelation. \textit{Amplification} is defined as: $Z_{ii} \triangleq 2\int_0^\infty C_{ii}(\tau)\, d\tau = -2[\mathbf{H}^{-1} \mathbf{P}_{\infty}]_{ii}$, and the time constant is expressed as:
\begin{equation}
\tau_{ii} = -\frac{[\mathbf{H}^{-1} \mathbf{P}_{\infty}]_{ii}}{[\mathbf{P}_{\infty}]_{ii}}.
\end{equation}

\par \textit{Relative response time} for autocovariance is then:
\begin{equation}
\bar{t}_{ii} = -\tau_{ii} \ln \bar{\eta},
\end{equation}
where $\bar{\eta} = \frac{C_{ii}(\bar{t}_i)}{C_{ii}(0)}$. While the structure of $\mathbf{P}_{\infty}$ is implicit, its spectral form under input $I_0^{\text{noise}}$ at node $m$ yields more clear form:
\[
Z_{ii}^{(m)} = 2\int_0^\infty C_{ii}^m(\tau)\, d\tau = 2\sum_j \sum_k \frac{u_{im}^j u_{im}^k}{\lambda_j + \lambda_k} \frac{1}{\lambda_j} I_0^{\text{noise}},
\]
and the corresponding time constant:
\[
\tau_{ii}^{(m)} = -\frac{\sum_j \sum_k \frac{u_{im}^j u_{im}^k}{\lambda_j + \lambda_k} \frac{1}{\lambda_j}}{\sum_j \sum_k \frac{u_{im}^j u_{im}^k}{\lambda_j + \lambda_k}}.
\]
\begin{widetext}
\DIFaddend \begin{figure*}[t]
\centering
\includegraphics[scale=0.8]{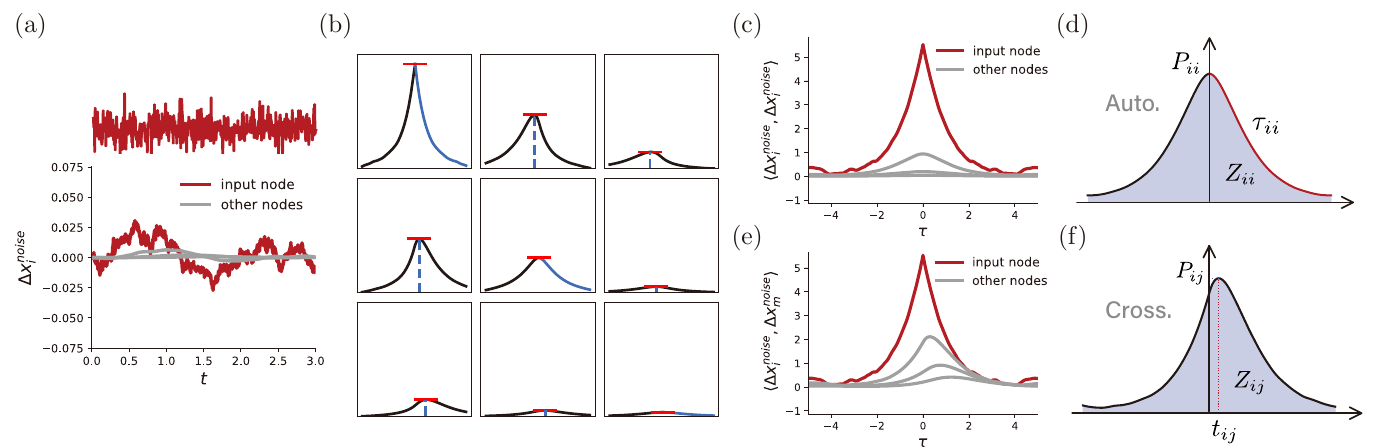}
\caption{Quantifying the impact of white noise input on network dynamics.  
(a) Noise input and resulting responses from the input node and other nodes, showing stochastic fluctuations.  
(b) Covariance matrix ($3 \times 3$) with diagonal elements representing autocovariance and off-diagonal elements representing crosscovariance. 
(c,d) Autocovariance corresponds to diagonal entries in (b). Key metrics include peak response ($P_{ii}$), amplification ($Z_{ii}$), and time constant ($\tau_{ii}$).  
(e,f) Crosscovariance corresponds to off-diagonal entries in (b). Key metrics include peak response ($P_{ij}$), amplification ($Z_{ij}$), and peak response time ($t_{ij}$).}
\label{noise_sketch}
\end{figure*}
\end{widetext}

\par For off-diagonal elements of the covariance matrix (Eq. (\ref{cov_diag})), $C_{ij}(\tau)$ is asymmetric. To characterize the temporal properties of crosscovariance, we apply the effective probability distribution framework by treating the response as a probability distribution. \textit{Amplification} is defined as:
\begin{equation}
Z_{ij}\triangleq \int_{-\infty}^{\infty} C_{ij}(\tau)\, d\tau = -\left[\mathbf{P}_{\infty} \mathbf{H}^{-\top} + \mathbf{H}^{-1} \mathbf{P}_{\infty}\right]_{ij}.
\end{equation}

\par We define the normalized \textit{probability density} $\rho_{ij}(\tau) = C_{ij}(\tau) / Z_{ij}$, and estimate the \textit{peak response time} via the expected value:
\begin{equation}
t_{ij}\triangleq \frac{\left[\mathbf{P}_{\infty} (\mathbf{H}^{-\top})^2\right]_{ij} - \left[\mathbf{H}^{-2} \mathbf{P}_{\infty}\right]_{ij}}{\left[\mathbf{P}_{\infty} \mathbf{H}^{-\top}\right]_{ij} + \left[\mathbf{H}^{-1} \mathbf{P}_{\infty}\right]_{ij}}.
\end{equation}

\par \textit{Peak response} is given by:
\begin{widetext}
\begin{equation}
P_{ij} \triangleq \frac{\left(\left[\mathbf{P}_{\infty} \mathbf{H}^{-\top}\right]_{ij} + \left[\mathbf{H}^{-1} \mathbf{P}_{\infty}\right]_{ij}\right)^2 }{\sqrt{4\left(\left[\mathbf{P}_{\infty} (\mathbf{H}^{-\top})^3\right]_{ij} + \left[\mathbf{H}^{-3} \mathbf{P}_{\infty}\right]_{ij}\right)\left(\left[\mathbf{P}_{\infty} \mathbf{H}^{-\top}\right]_{ij} + \left[\mathbf{H}^{-1} \mathbf{P}_{\infty}\right]_{ij}\right) - 2\left(\left[\mathbf{P}_{\infty} (\mathbf{H}^{-\top})^2\right]_{ij} - \left[\mathbf{H}^{-2} \mathbf{P}_{\infty}\right]_{ij}\right)^2}}.
\end{equation}
\end{widetext}

\par For diagonalizable $\mathbf{H}$ with input applied at node $m$ with strength $I_0^{\text{noise}}$, the metrics simplify as:
\begin{widetext}
\begin{align}
Z_{ij}^{(m)} &= \sum_{p,q} \frac{u_{ij}^m(p,q)}{\lambda_p \lambda_q} I_0^{\text{noise}}, \\
t_{ij}^{(m)} &= -\frac{\sum_{p,q} u_{ij}^m(p,q) \left(\frac{1}{\lambda_p \lambda_q}\left(\frac{1}{\lambda_p} - \frac{1}{\lambda_q}\right)\right)}{\sum_{p,q} u_{ij}^m(p,q) \frac{1}{\lambda_p \lambda_q}}, \\
P_{ij}^{(m)} &= \frac{\left(\sum_{p,q} u_{ij}^m(p,q) \frac{1}{\lambda_p \lambda_q}\right)^2 I_0^{\text{noise}}}{\sqrt{4 \sum_{p,q} u_{ij}^m(p,q)\left(\frac{1}{\lambda_p^3} + \frac{1}{\lambda_q^3}\right)\frac{1}{\lambda_p + \lambda_q} \sum_{p,q} u_{ij}^m(p,q) \frac{1}{\lambda_p \lambda_q} - 2 \left(\sum_{p,q} u_{ij}^m(p,q) \frac{1}{\lambda_p \lambda_q} \left(\frac{1}{\lambda_p} - \frac{1}{\lambda_q}\right)\right)^2}},
\end{align}
\end{widetext}
where $u_{ij}^m(p,q) \equiv u_{im}^p u_{jm}^q = [\mathbf{U}]_{ip}[\mathbf{U}^{-1}]_{pm} [\mathbf{U}]_{jq}[\mathbf{U}^{-1}]_{qm}$. Compared to the explicit variance formula for one-dimensional systems, we include a correction term of $1/\sqrt{2}$ for peak response estimation.

\par All the metrics receiving noise input at source $m$ for describing autocovariance are
\begin{widetext}
\begin{equation} 
\begin{aligned}  
Z_{ii}^{(m)} &\triangleq -2\left[\mathbf{H}^{-1}\mathbf{P}_{\infty}\right]_{ii}^m = 2\sum_{j,k} \frac{u_{im}^j u_{im}^k}{(\lambda_j + \lambda_k)\lambda_j}I_0^{\text{noise}} \sim {O}(1/\lambda_1^2), \\ 
P_{ii}^{(m)} &\triangleq \left[\mathbf{P}_{\infty}\right]_{ii}^m = -\sum_{j,k} \frac{u_{im}^j u_{im}^k}{\lambda_j + \lambda_k}I_0^{\text{noise}} \sim {O}(1/\lambda_1), \\ 
\tau_{ii}^{(m)} &\triangleq -\frac{\left[\mathbf{H}^{-1}\mathbf{P}_{\infty}\right]_{ii}^m}{\left[\mathbf{P}_{\infty}\right]_{ii}^m} = \frac{\sum_{j,k} \frac{u_{im}^j u_{im}^k}{(\lambda_j + \lambda_k)\lambda_j}}{\sum_{j,k} \frac{u_{im}^j u_{im}^k}{\lambda_j + \lambda_k}} \sim {O}(1/\lambda_1).  
\end{aligned}  
\label{eq:auto_metrics} 
\end{equation}
\end{widetext}
\DIFdelbegin 

\DIFdelend \DIFaddbegin
\DIFaddend 

\DIFaddbegin
\begingroup
 
Further, the amplification for covariance admits a particularly simple form. The network
transfer function is
$\mathbf G(\omega)=(\mathrm{i}\omega\mathbf I-\mathbf H)^{-1}$,
and, for a white-noise input applied at node $m$, the cross-spectrum between
nodes $i$ and $j$ is
$S_{ij}^{(m)}(\omega)=I_0G_{im}(\omega)G_{jm}^{*}(\omega)$.
Because $Z_{ij}^{(m)}$ integrates the cross-covariance over all time lags,
the Wiener--Khinchin theorem gives
$
Z_{ij}^{(m)}
=
\int_{-\infty}^{\infty}C_{ij}^{(m)}(\tau)\,d\tau
=
S_{ij}^{(m)}(0).
$
At zero frequency, $\mathbf G(0)=-\mathbf H^{-1}$, and therefore
\begin{equation}
Z_{ij}^{(m)}
=
I_0 R_i^{(m)}R_j^{(m)},
\label{stochastic:amplification}
\end{equation}
where $R_i^{(m)}=-[\mathbf H^{-1}]_{im}$. Thus, the lag-integrated
stochastic response reduces exactly to the product of the static gains
from the common source to nodes $i$ and $j$.

This simplification has a direct spectral interpretation. Integration over
the covariance lag selects its zero-frequency component and therefore
removes the detailed temporal structure associated with nonzero
frequencies. As a result, $Z_{ij}^{(m)}$ retains only the total static gains
from the common source, explaining why the same quantity $R_i^{(m)}$ also
appears in deterministic propagation.

Through the Wiener--Khinchin theorem, the distinction between the integrated covariance matrix $\mathbf{Z}$ composed of on-diagonal elements $Z_{ii}^{(m)}$ and off-diagonal $Z_{ij}^{(m)}$, and the stationary covariance $\mathbf P_\infty$ can be made precise. The stationary covariance is the covariance evaluated at zero time lag,
\begin{equation}
\mathbf P_\infty = \mathbf C(0),
\end{equation}
which, in the frequency domain, corresponds to integrating the spectrum over all frequencies:
\begin{equation}
\mathbf P_\infty
=
\frac{1}{2\pi}
\int_{-\infty}^{\infty}
\mathbf S_{\mathbf x}(\omega) \, d\omega.
\end{equation}
Conversely, the amplification metric integrates the covariance over all time lags,
\begin{equation}
\mathbf Z
=
\int_{-\infty}^{\infty}
\mathbf C(\tau) \, d\tau,
\end{equation}
which corresponds to evaluating the spectrum at the single frequency $\omega = 0$:
\begin{equation}
\mathbf Z = \mathbf S_{\mathbf x}(0).
\end{equation}
\textcolor{black}{The two quantities are therefore complementary projections of the same spectrum: $\mathbf P_\infty$ is local in time but integrates over the full frequency range, whereas $\mathbf Z$ is integrated over time but local in frequency.} Consequently, $\mathbf P_\infty$ retains information about the complete temporal filtering and overlap of stochastic fluctuations across all timescales (Eq.~\eqref{inv_expand_long}), while $\mathbf Z$ isolates the static, zero-frequency network gain and discards all transient temporal information.

Similarly, the cross-covariance timing metric admits a simple
node-resolved form,
\begin{equation}
t_{ij}^{(m)}
=
\tau_i^{(m)}-\tau_j^{(m)}
=
-\left.
\frac{d}{d\omega}
\arg S_{ij}^{(m)}(\omega)
\right|_{\omega=0}.
\label{stochastic:time}
\end{equation}
This relation can be understood from the low-frequency behavior of the
network transfer function. For propagation from source $m$ to node $i$,
\[
G_i^{(m)}(\omega)
=
R_i^{(m)}
\left[
1-\mathrm{i}\omega\tau_i^{(m)}
+\mathcal O(\omega^2)
\right].
\]
Accordingly, the cross-spectrum between nodes $i$ and $j$ satisfies
\[
S_{ij}^{(m)}(\omega)
=
Z_{ij}^{(m)}
\left[
1-\mathrm{i}\omega
\left(\tau_i^{(m)}-\tau_j^{(m)}\right)
+\mathcal O(\omega^2)
\right].
\]
Thus, its phase slope at zero frequency is determined
by the difference in the characteristic response times
$\tau_i^{(m)}-\tau_j^{(m)}$. The quantities $R_i^{(m)}$ and
$\tau_i^{(m)}$ therefore provide complementary low-frequency
descriptions of stochastic propagation, characterizing its strength
and relative timing, respectively. This also explains why the same
canonical quantities recur in the corresponding deterministic
response metrics.

\textcolor{black}{The quantity $t_{ij}^{(m)}$ characterizes the first-moment, or low-frequency, relative delay of the cross-covariance rather than its detailed temporal profile.} For well-behaved, sign-definite responses in NSDDM networks, this metric provides an intuitive estimate of relative timing. Its applicability is restricted under the following respects. First, in the presence of signed pathways, strong cancellation or sign reversals can produce cross-covariances with multiple peaks, and $t_{ij}^{(m)}$ becomes physically ambiguous or ill-defined when the zero-frequency gain approaches zero. Second, because $t_{ij}^{(m)}$ is defined through the zero-frequency phase slope, it neglects higher-frequency contributions and is therefore unsuitable when high-frequency effects dominate the response or when the input noise is colored with substantial spectral weight away from zero frequency; in such cases, additional frequency-dependent corrections are required.
\endgroup

\section{CHAIN STRUCTURE}
\label{sec:C}

\renewcommand\theequation{C\arabic{equation}}
\setcounter{equation}{0}
\makeatletter
\makeatother
\subsection*{Homogeneous directed chain}
\label{sec:C:chain}
\par For the homogeneous directed chain, the matrix $\mathbf{H}$ takes the form:
\begin{equation}
\left(\begin{array}{cccc}
-\beta & & & \\
\alpha & -(\beta+\alpha) & & \\
& \ddots & \ddots & \\
& & \alpha & -(\beta+\alpha)
\end{array}\right),
\end{equation}
where $\alpha$ denotes connection strength and $\beta$ the self-decay rate. The eigenvalues of $\mathbf{H}$ are $-\beta$ (with eigenvector $(1,1,\dots,1)$) and $-(\beta + \alpha)$ (with multiplicity $N-1$ and eigenvector $(0,\dots,0,1)$).
\par The analytical time courses for unit pulse input is given by:
\begin{equation}
\refstepcounter{equation}\label{homo:indegree}
\tag{C\arabic{equation}}
x_i(t)= \begin{cases}
e^{-\beta t}, & d = 1, \\
e^{-\beta t}\left(1 - e^{-\alpha t} \sum_{j=0}^{d-2} \frac{(\alpha t)^j}{j!} \right), & d \geq 2.
\end{cases}
\end{equation}
\par The corresponding metric expressions are:
\begin{equation}
\begin{aligned}
& -\left[\mathbf{H}^{-1}\right]_{dm} = \frac{\alpha^{d}}{(\alpha + \beta)^{d} \beta}, \\
&\left[\mathbf{H}^{-2}\right]_{dm} = \frac{\alpha^{d}(\alpha + (d+1)\beta)}{(\alpha + \beta)^{d+1} \beta^2}, \\
& -\frac{\left[\mathbf{H}^{-2}\right]_{dm}}{\left[\mathbf{H}^{-1}\right]_{dm}} = \frac{\alpha + (d+1)\beta}{\alpha \beta + \beta^2}, \\
& \frac{\left(\left[\mathbf{H}^{-1}\right]_{dm}\right)^2}{\sqrt{2\left[\mathbf{H}^{-3}\right]_{dm} \left[\mathbf{H}^{-1}\right]_{dm} - \left(\left[\mathbf{H}^{-2}\right]_{dm}\right)^2}} \\
&= \frac{(\alpha + \beta)}{\sqrt{\alpha^2 + 2\alpha\beta + (d+1)\beta^2}} \left( \frac{\alpha}{\alpha + \beta} \right)^{d}, \\
&\left[ \mathbf{P}_{\infty}\right]_{{dm }}^m =\frac{1}{2 \beta}\left(\frac{\alpha}{2 \beta+\alpha}\right)^d.
\end{aligned}
\end{equation}
Here, $m = 0$ denotes the first node of the chain, and $d = 1, 2, 3, \dots, N$ refers to the $d$-th node on the chain.

\par The directed chain serves as a minimal structure to reveal how metrics scale with path length $d$. The main metrics are:
\begin{equation}
\begin{aligned}
\ln R_{dm}
&=-d\ln\left(1+\frac{\beta}{\alpha}\right)-\ln\beta,
\\
\ln P_{dm}
&=-d\ln\left(1+\frac{\beta}{\alpha}\right)
-\frac{1}{2}\ln\left[
1+d\left(\frac{\beta}{\alpha+\beta}\right)^2
\right],
\\
\ln Z_{dm}
&=-d\ln\left(1+\frac{\beta}{\alpha}\right)
+\ln\left(1+\frac{d\beta}{\alpha+\beta}\right)
-2\ln\beta,
\\
\ln Z_{dm}^{m}
&=-d\ln\left(1+\frac{\beta}{\alpha}\right)
-2\ln\beta,
\\
\ln\left[\mathbf{P}_{\infty}\right]_{dm}^{m}
&=-d\ln\left(1+\frac{2\beta}{\alpha}\right)
-\ln(2\beta),
\\
\ln P_{dm}^{m}
&=
-d\ln\left(1+\frac{\beta}{\alpha}\right)
-\ln(2\beta)
\\
&\quad
-\frac{1}{2}\ln\left[
1+\frac{d}{2}
\left(\frac{\beta}{\alpha+\beta}\right)^2
\right].
\\
\tau_{dm}
&=\frac{d}{\alpha+\beta}+\frac{1}{\beta},
\\
t_{dm}^{m}
&=\frac{d}{\alpha+\beta}.
\end{aligned}
\end{equation}

\par Through the metrics, we notice that the scaling of response strength (amplification and peak response) with path length follows approximately $(1 + \beta/\alpha)^{-d}$, and temporal metrics (time constant) basically scale as ${d}/{(\alpha + \beta)}$, especially when $\beta$ dominants.

\par We define $S \equiv - \ln(1+\frac{\beta}{\alpha})d $ and $T \equiv \frac{1}{\alpha+\beta}d$. We then perturb the identical weight $\alpha$ by $\Delta \alpha$ to become $\alpha + \Delta \alpha$, and find
\begin{equation}
\begin{aligned}
& |\Delta S|=|S(\alpha+\Delta \alpha)-S(\alpha)|=S\left(\frac{\alpha}{\delta \alpha}\right), \\
& |\Delta T|=|T(\alpha+\Delta \alpha)-T(\alpha)|=T(\alpha) \delta \alpha,
\end{aligned}
\end{equation}
where $0<\delta \alpha \equiv \left(1+\frac{\alpha+\beta}{\Delta \alpha}\right)^{-1}<1$. From these expressions, we conclude that:
\begin{enumerate}
\item As $d$ increases, both $|\Delta S|$ and $|\Delta T|$ increase.
\item For $\alpha > \beta$ and $\beta > 1$, $|\Delta S| > |\Delta T|$ and ${|\Delta S|}/{S} \approx {|\Delta T|}/{T} = \delta \alpha$.
\item For $\beta > \alpha$, $|\Delta S| > |\Delta T|$ and ${|\Delta S|}/{S} > {|\Delta T|}/{T} = \delta \alpha$.
\end{enumerate}

Interestingly, closing a homogeneous directed chain by adding a single feedback edge from its last node to its first node qualitatively reorganizes the spectrum. The defective two-eigenvalue structure of the open chain is replaced by $N$ eigenvalues,
$\lambda_k=-(\alpha+\beta)+\alpha e^{2\pi\mathrm{i}k/N}$, distributed on a circle in the complex plane. Simultaneously, the localized generalized eigenvectors of the chain are replaced by the orthogonal and delocalized discrete Fourier modes of the loop. Despite this substantial spectral reorganization, the response metrics change only through the loop gain $g=[\alpha/(\alpha+\beta)]^N$ when compared with the corresponding open chain having the same local decay. Specifically, response strength is multiplied by $(1-g)^{-1}$, while the characteristic response time receives the additional term $Ng/[(\alpha+\beta)(1-g)]$. Therefore, when $g\ll1$, the metric--distance relations remain close to those of the open chain even though their spectral representations are qualitatively different (SM Fig.~S48). The effects of heterogeneous local dynamics and link weights, including inhibitory couplings, are examined in Fig.~S43.

\par {We further tested robustness under increased coupling strength by setting the connection strength to \mbox{
$\alpha=1000$ while $\beta=1$
}
, which markedly increases the stiffness ratio \mbox{
$(\max_i |\lambda_i|\,/\,\min_i |\lambda_i|) \sim O(10^3)$
}
. In this stiff regime, the leading-term approximation continues to capture the scaling of most response metrics; the main exception is the cross-covariance peak time, for which the loop-induced contribution $g\approx1$ becomes dominant (see SM Fig.~S49).
}

\DIFaddend \subsection*{Alpha function}
\label{sec:C:alpha}
\par We notice that the excitatory postsynaptic potentials (EPSPs) modelling, both the amplitude (efficacy) and the temporal dynamics (time constant) of postsynaptic conductance changes, are linked with the chain structure. Interestingly, the evolution of EPSPs strongly resembles the impulse response, reflecting a fundamental aspect of synaptic transmission. The conductance dynamics governed by the synaptic time constant $\tau$ are \cite{rall1967distinguishing}:
\begin{equation}
\ddot{g} + \frac{2}{\tau} \dot{g} + \frac{1}{\tau^2} g = G_{\text{norm}} u(t),
\end{equation}
which is equivalent to the two-dimensional system:
\begin{equation}
\left\{
\begin{array}{l}
\frac{d z}{dt} = -\frac{z}{\tau} + G_{\text{norm}} u(t), \\
\frac{d g}{dt} = -\frac{g}{\tau} + z(t),
\end{array}
\right.
\end{equation}
with $z \equiv \frac{g}{\tau} + \dot{g}$. For a pulse input $u(t)$, the solution (alpha function) becomes:
\begin{equation}
g(t) = G_{\text{norm}}\, t\, e^{-t/\tau}, \quad 
\dot{g}(t) = G_{\text{norm}} \left( e^{-t/\tau} - \frac{t}{\tau} e^{-t/\tau} \right).
\end{equation}
The peak of $g(t)$ occurs at $t = \tau$.  If we set $G_{\text{norm}} = \frac{g_{\text{peak}}}{\tau / e}$, the peak response is $g(\tau) = G_{\text{norm}}\, \tau\, e^{-1} = g_{\text{peak}}$.

\par This system can be rewritten in compact form:
\begin{equation}
\left(
\begin{array}{c}
\dot{f} \\
\dot{g}
\end{array}
\right)
=
\left(
\begin{array}{cc}
-\frac{1}{\tau} & 0 \\
1 & -\frac{1}{\tau}
\end{array}
\right)
\left(
\begin{array}{c}
f \\
g
\end{array}
\right), \quad
\left\{
\begin{array}{l}
f = \frac{c}{\tau} e^{-t/\tau}, \\
g = \frac{c t}{\tau} e^{-t/\tau},
\end{array}
\right.
\end{equation}
assuming a single pulse. Generalizing to an $N$-dimensional directed chain yields:
\begin{equation}
\dot{\mathbf{x}}
=
\begin{pmatrix}
-(\beta+\alpha) & & & \\
\alpha & \ddots & & \\
& \ddots & \ddots & \\
& & \alpha & -(\beta+\alpha)
\end{pmatrix}
\mathbf{x},
\label{alpha}
\end{equation}
with a solution for node $d$:
\begin{equation}
\refstepcounter{equation}\label{alpha:sol}
\tag{C\arabic{equation}}
x_d(t) = \frac{\alpha^{d-1}\, t^{d-1}\, e^{-(\beta + \alpha)t}}{(d-1)!},
\end{equation}
where $\alpha$ denotes the connection strength. This form closely resembles the response $\Delta x^{\text{pulse}}_i(t) = t^d e^{\lambda_1 t}$ described in \cite{schroder2019dynamic}, which is employed to amend the metrics for impulse responses.

\subsection*{High-order moment estimations}
\label{sec:C:highorder}
\par The primary difference in the solution for the homogeneous directed chain and the alpha function arises from the difference at the first node. Even in a simple two-node model, the resulting dynamics can differ markedly. In the first scenario, inspired by the alpha function, the matrix has repeated eigenvalues $\lambda_1 = \lambda_2 = -(\beta + \alpha)$ and is non-diagonalizable. Applying a pulse input to the first node leads to a peak response at the second node occurring at time $1/(\beta + \alpha)$. While the estimated metrics of time constant at any integer order $k$ can be computed as:
\begin{equation}
\refstepcounter{equation}\label{order:alpha}
\tag{C\arabic{equation}}
-\frac{[\mathbf{H}^{-(k+1)}]_{im}}{[\mathbf{H}^{-k}]_{im}} = \left( \frac{k+1}{k} \right) \frac{1}{\alpha + \beta}.
\end{equation}
This expression indicates that estimation improves with increasing order $k$ (Fig.~\ref{high_order}(a) and (c)). 

\par In the second scenario, considering the homogeneous directed chain, the second node responds most strongly at time $-\ln(\beta / (\alpha + \beta)) / \alpha$, which approximates $1/(\alpha + \beta)$ for large $\beta$. The corresponding estimate is:
\begin{equation}
\refstepcounter{equation}\label{order:homo}
\tag{C\arabic{equation}}
-\frac{[\mathbf{H}^{-(k+1)}]_{im}}{[\mathbf{H}^{-k}]_{im}} = \frac{1}{\alpha + \beta} + \frac{\alpha}{\beta} \cdot \frac{1}{(\alpha + \beta) - \beta \left( \frac{\beta}{\alpha + \beta} \right)^{k-1}}.
\end{equation}
As $k \to \infty$, the expression converges to $(1+\alpha /\beta)/(\alpha + \beta) \approx 1 / (\alpha + \beta)$ for large $\beta$ (Fig.~\ref{high_order}(d)). However, for strong interactions (large $\alpha$), the theoretical estimations exhibit a consistent bias relative to simulation results (Fig.~\ref{high_order}(b)), even when high orders are considered. This suggests that additional bias correction terms are needed to improve metric accuracy. It is worth noting that convergence to high accuracy with increasing order $k$ holds in this simplified case, but does not necessarily extend to general network topologies. This is the main reason why we choose order $p=1$, as we adhere to the principle of focusing on the generalizability of the metrics, and this choice naturally ensures the easy interpretability of the metrics.


\begin{figure}[!ht]
    \centerline{\includegraphics[scale=0.45]{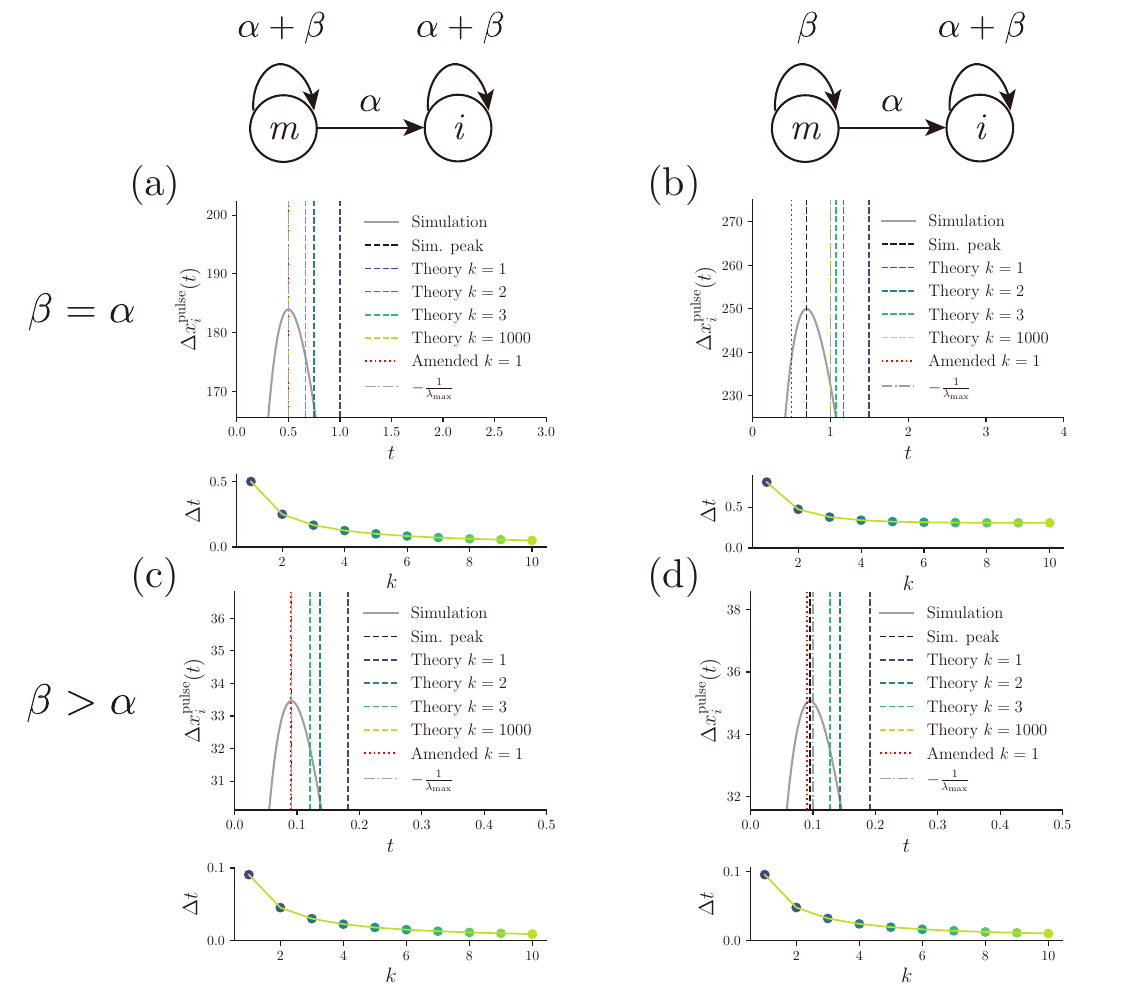}}
    \caption{Order dependence of estimated peak response time in two-node models.  
Panels (a) and (c) correspond to the alpha-function-inspired model (Eq. (\ref{alpha:sol})), while (b) and (d) depict the model based on Eq. (\ref{homo:indegree}). The estimation follows the form $-[\mathbf{H}^{-(k+1)}]_{im} / [\mathbf{H}^{-k}]_{im}$, with the order $k$ varied to assess accuracy. The refined estimator $\widetilde{\tau}_i$ (Eq. 
(\ref{refined})) is also included. In (a) and (b), the self-decay $\beta = 1$ matches the interaction $\alpha = 1$, while in (c) and (d), $\beta = 10$ exceeds $\alpha = 1$. A pulse input is applied to the first node, and the response of the second node is observed. In (a), (c), and (d), higher-order $k$ yields accurate estimates (Eqs. (\ref{order:alpha}) and (\ref{order:homo})), with both the refined estimate and $-1/\max_j \operatorname{Re}(\lambda_j)$ aligning well with simulations. In (b), although increasing $k$ improves accuracy, a notable bias persists relative to the simulated peak time.}
\label{high_order}
\end{figure}

\section{HOMOGENEOUS IN-DEGREE NETWORKS}
\label{sec:D}

\renewcommand\theequation{D\arabic{equation}}
\setcounter{equation}{0}
\makeatletter
\makeatother
\subsection*{Convergence of the expansion}
The Neumann series expansion satisfies
$(\mathbf{I}_N - \mathbf{M})^{-1} = \sum_{p=0}^{\infty} \mathbf{M}^p$,
which converges if and only if the spectral radius satisfies $\rho(\mathbf{M})<1$.

For the NSDDM system, we write
\[
\mathbf{H} = \mathbf{A} - \mathbf{D} - \beta \mathbf{I}_N = -\beta\!\left(\mathbf{I}_N - \frac{\mathbf{A}-\mathbf{D}}{\beta}\right),
\]
so that $\mathbf{H}^{-1} = -\frac{1}{\beta} \left(\mathbf{I}_N - \frac{\mathbf{A}-\mathbf{D}}{\beta}\right)^{-1}$.

Therefore, the Neumann expansion applies with $\mathbf{M} = (\mathbf{A}-\mathbf{D})/\beta$,
and the exact convergence condition is
\[
\rho\!\left(\frac{\mathbf{A}-\mathbf{D}}{\beta}\right) < 1.
\]
Equivalently, all eigenvalues $\lambda$ of $\mathbf{A}-\mathbf{D}$ must satisfy
$|\lambda| < \beta$.

Note that the NSDDM condition guarantees Hurwitz stability of $\mathbf{H}$, i.e., $\max_i \operatorname{Re}\!\bigl(\lambda_i(\mathbf{H})\bigr) < 0$, but neither Hurwitz stability nor NSDDM alone ensure convergence of the Neumann series.

By Gershgorin's theorem, for any adjacency matrix $\mathbf{A}$,
\[
\max_{\lambda \in \sigma(\mathbf{A}-\mathbf{D})} |\lambda| \;\le\; 2D_{\max},
\]
where $D_{\max}$ denotes the maximum node degree. Hence,
\[
\beta > 2D_{\max}
\]
provides a sufficient (though not necessary) condition for convergence.

Alternatively, defining the matrix for intrinsic time constants
$\mathbf{\Gamma}=(\mathbf{D}+\mathbf{B})^{-1}$, where
$\mathbf{B}=\operatorname{diag}(\beta_1,\ldots,\beta_N)$, gives
\[
-\mathbf{H}^{-1}
=
\sum_{p=0}^{\infty}
(\mathbf{\Gamma}\mathbf{A})^p\mathbf{\Gamma},
\]
which converges when
$\rho(\mathbf{\Gamma}\mathbf{A})<1$.
For the NSDDM setting,
\[
\rho(\mathbf{\Gamma}\mathbf{A})
=
\max_i
\frac{D_i}{D_i+\beta_i}
<1,
\]
so convergence is guaranteed whenever $\beta_i>0$ for all nodes.

Generally, small intrinsic time constants and weak path weights favor convergence of the walk expansion, particularly in recurrent networks where activity can repeatedly circulate through feedback loops. By contrast, for a purely feedforward acyclic network, the walk series contains only finitely many terms and therefore terminates exactly.

\DIFaddbegin \subsection*{{Expansion for homogeneous in-degree}}
\label{sec:D:expansion}
\DIFaddend \par \DIFaddbegin {We firstly analyze the \mbox{
$\mathbf{H}^{-1}$
}
expansion under homogeneous in-degree conditions ($\mathbf{D} = DI$) in NSDDM systems. 
}

\par \DIFaddend The expansion is
\begin{equation}
\refstepcounter{equation}\label{homo:first}
\tag{D\arabic{equation}}
\begin{aligned}
\mathbf{H}^{-1}&=
\left(\mathbf{A}-(\beta+D)\mathbf{I}_N\right)^{-1}, \\
&=-\sum_{p=0}^{\infty}  \frac{\mathbf{A}^p}{(\beta + D)^{p + 1}}.  
\end{aligned}
\end{equation}
Cases for $\mathbf{H}^{-2}$ and $\mathbf{H}^{-3}$ are similar. 
\begin{equation}
\refstepcounter{equation}\label{homo:second}
\tag{D\arabic{equation}}
\begin{aligned}
\mathbf{H}^{-2} & =\frac{1}{\beta^2} \sum_{p=0}^{\infty} \frac{p(\mathbf{A}-D\mathbf{I}_N)^{p-1}}{\beta^{p-1}},  \\
& = \sum_{p=0}^{\infty}  \frac{(p+1)\mathbf{A}^p}{(\beta + D)^{p + 2}},
\end{aligned}
\end{equation}
and
\begin{equation}
\refstepcounter{equation}\label{homo:third}
\tag{D\arabic{equation}}
\begin{aligned}
\mathbf{H}^{-3} & =-\frac{1}{2\beta^3} \sum_{p=0}^{\infty} \frac{(p+1)(p+2)(\mathbf{A}-D\mathbf{I}_N)^{p}}{\beta^{p}},  \\
& =- \sum_{p=0}^{\infty}  \frac{(p+1)(p+2)\mathbf{A}^p}{2(\beta + D)^{p + 3}}.
\end{aligned}
\end{equation}

The expansions for $\mathbf{H}^{-1}, \mathbf{H}^{-2}$ and $\mathbf{H}^{-3}$ cover all metrics for deterministic inputs. Here, we do not consider the bias terms, as they are system-wide parameters primarily introduced to compensate for numerical inaccuracies.

\par To investigate metrics under noise inputs, we derive the expansion for the steady-state covariance: 
\begin{equation}
\refstepcounter{equation}\label{p_ij}
\tag{D\arabic{equation}}
\begin{aligned}
\left[\mathbf{P}_{\infty}\right]_{ij}^m &= \frac{I_0}{2\pi} \int_{-\infty}^\infty [(\mathbf{H} - \mathrm{i}\omega \mathbf{I}_N)^{-1}]_{im} [(\mathbf{H} + \mathrm{i}\omega \mathbf{I}_N)^{-1}]_{jm} d\omega \\
&=I_0\sum_{p,q=0}^\infty \frac{(p+q)!}{p!q!} \frac{[\mathbf{A}^p]_{im}[\mathbf{A}^q]_{jm}}{(2(\beta+D))^{p+q+1}} = \left[\mathbf{P}_{\infty}\right]_{ji}^m.
\end{aligned}
\end{equation}
And thus,
\begin{equation}
\begin{aligned}
\left[\mathbf{P}_{\infty}\right]_{im}^m =\left[\mathbf{P}_{\infty}\right]_{mi}^m 
=I_0\sum_{p=0}^\infty \frac{[\mathbf{A}^p]_{im}}{(2(\beta+D))^{p+1}},
\end{aligned}
\end{equation}
which is similar with Eq.~\eqref{homo:first}.

\par For the autocovariance, the three metrics can be written as
\begin{align}
Z_{ii}^{(m)}
&=
I_0\left(R_i^{(m)}\right)^2,
\\[4pt]
P_{ii}^{(m)}
&=
I_0
\sum_{p,q=0}^{\infty}
\binom{p+q}{p}
\left(\frac{\gamma}{2}\right)^{p+q+1}
[\mathbf A^p]_{im}
[\mathbf A^q]_{im},
\\[4pt]
\tau_{ii}^{(m)}
&=
\frac{
I_0\left(R_i^{(m)}\right)^2
}{
2P_{ii}^{(m)}
}.
\end{align}

\par Similarly, for the cross-covariance,
\begin{align}
Z_{ij}^{(m)}
&=
I_0R_i^{(m)}R_j^{(m)},
\\[4pt]
P_{ij}^{(m)}
&=
\frac{
I_0\left|R_i^{(m)}R_j^{(m)}\right|
}{
\sqrt{
2\left[
\left(\sigma_i^{(m)}\right)^2
+
\left(\sigma_j^{(m)}\right)^2
\right]
}
},
\\[4pt]
t_{ij}^{(m)}
&=
\tau_i^{(m)}-\tau_j^{(m)},
\end{align}
where the variance is
\begin{equation}\left(\sigma_i^{(m)}\right)^2
=
\frac{
\displaystyle
\sum_{p=0}^{\infty}
(p+1)(p+2)\gamma^{p+3}
[\mathbf A^p]_{im}
}{
R_i^{(m)}
}
-
\left(\tau_i^{(m)}\right)^2. 
\end{equation}
\begingroup
 
\subsection*{\textcolor{black}{Scaling in the dominant-walk case}}
For a dominant walk length $p$, we retain only the structural contribution $\left[\mathbf{A}^p\right]_{im}$ and its dependence on the homogeneous intrinsic time constant $\gamma=(\beta+D)^{-1}$, while omitting some proportionality coefficients unrelated to $\gamma$ or $\left[\mathbf{A}^p\right]_{im}$. From the expansions of $\mathbf{H}^{-1}$ and $\mathbf{H}^{-2}$, we directly obtain
\begin{equation}
R_i^{(m)}
\sim
\gamma^{p+1}\left[\mathbf{A}^p\right]_{im},
\qquad
Z_i^{(m)}
\sim
\gamma^{p+2}\left[\mathbf{A}^p\right]_{im}.
\end{equation}
For pulse-peak-response metric,
\begin{equation}
P_i^{(m)}
=
\frac{([\mathbf H^{-1}]_{im})^2 I_0^{\mathrm{pulse}}}
{\sqrt{
2[\mathbf H^{-3}]_{im}[\mathbf H^{-1}]_{im}
-
([\mathbf H^{-2}]_{im})^2
}},
\end{equation}
yields
\begin{equation}
P_i^{(m)}
\sim
\frac{I_0^{\mathrm{pulse}}}{\sqrt{p+1}}
\gamma^p
[\mathbf A^p]_{im}.
\end{equation}
The prefactor $C(p)$ is independent of both the intrinsic time constant and the accumulated walk length.
\begin{equation}
\widetilde P_i^{(m)}
=
C(p)P_i^{(m)}
\sim
\gamma^p[\mathbf A^p]_{im}.
\end{equation}

For the autocovariance, both factors propagate from $m$ to $i$ through $p$-step walks, yielding
\begin{equation}
P_{ii}^{(m)}
=
\left[\mathbf{P}_{\infty}\right]_{ii}^{(m)}
\sim
\gamma^{2p+1}
\left[\mathbf{A}^p\right]_{im}^{2}.
\end{equation}
For autocovariance amplification, the exact relation between deterministic and stochastic responses gives
\begin{equation}
Z_{ii}^{(m)}
=
\left(R_i^{(m)}\right)^2.
\end{equation}
Therefore,
\begin{equation}
Z_{ii}^{(m)}
\sim
\gamma^{2p+2}
\left([\mathbf A^p]_{im}\right)^2.
\end{equation}

For the crosscovariance metrics, multiplication by $\mathbf{H}^{-k}$ introduces $k$ additional powers of $\gamma$ while retaining the same $p$-step structural contribution. Hence,
\begin{equation}
M_{im}^{m(k)}
\sim
\gamma^{p+k+1}
\left[\mathbf{A}^p\right]_{im},
\qquad k=1,2,3.
\end{equation}
Therefore,
\begin{equation}
Z_{im}^{(m)}
=
-M_{im}^{m(1)}
\sim
\gamma^{p+2}
\left[\mathbf{A}^p\right]_{im} \sim Z_{i}^{(m)}.
\end{equation}
Substituting
\begin{equation}
P_{im}^{(m)}
=
\frac{\left(M_{im}^{m(1)}\right)^2}
{\sqrt{
4M_{im}^{m(1)}M_{im}^{m(3)}
-
2\left(M_{im}^{m(2)}\right)^2
}},
\end{equation}
we obtain
\begin{equation}
P_{im}^{(m)}
\sim
\frac{
\gamma^{2p+4}\left[\mathbf{A}^p\right]_{im}^{2}
}{
\gamma^{p+3}\left[\mathbf{A}^p\right]_{im}
}
\sim
\gamma^{p+1}
\left[\mathbf{A}^p\right]_{im} \sim R_{i}^{(m)}.
\end{equation}
\endgroup

\DIFaddbegin \subsection*{{Truncation error}}
\label{sec:D:error}

\par {In this subsection, we mainly analyze the asymptotic truncation error for the peak response \mbox{
$R_{im}$
}
and time constant \mbox{
$\tau_{im}$
}
in the simplified case of homogeneous in-degree. This special case provides a baseline for understanding how the approximation accuracy scales with the truncation order \mbox{
\(P\)
}
. Note that we mainly consider sparse random networks with strong self-decay, so that longer-walk contributions are strongly suppressed and the expansion is dominated by the shortest-path term. By contrast, in denser networks, the combinatorial growth in the number of longer walks or recurrent loops may compete with or outweigh this decay. In such cases, the shortest-path truncation may lose quantitative accuracy, and the appropriate truncation order should therefore be assessed and validated on a case-by-case basis.
}

\par {Truncating the Neumann series for the peak response \mbox{
\(R_{i}^{(m)}\)
}
at order \mbox{
\(P\)
}
yields the approximation \mbox{
\(R_{i m}^{(P)}=\sum_{p=1}^P \gamma^{p+1} [A^p]_{im}\)
}
. The corresponding error is \mbox{
\(E_{i m}^{(P)}\equiv R_{i }^{(m)}- R_{i m}^{(P)}= \sum_{p=P+1}^{\infty} \gamma^{p+1}[A^p]_{im}\)
}
. Using the inequality \mbox{
\(|[A^p]_{im}| \leqslant \| A\|^p\)
}
and assuming the series convergence condition \mbox{
\(r = |\gamma| \| A \| < 1\)
}
holds, we obtain the bound
}\begin{equation}
{\left|E_{im}^{(P)}\right| \leqslant |\gamma| \frac{r^{\,P+1}}{1-r}.
}\end{equation}
{When the intrinsic time scale is short (\mbox{
\(|\gamma| \sim O(\varepsilon)\)
}
with \mbox{
\(\varepsilon \ll 1\)
}
) and the network norm is order one (\mbox{
\(\| A \| \sim O(1)\)
}
), the error scales as
}\begin{equation}
{\left|E_{im}^{(P)}\right| \sim O\left(\varepsilon^{P+2}\right).
}\end{equation}
\par {For the time constant \mbox{
\(\tau_{im}\)
}
, the truncation error is naturally defined as a relative measure. Let the approximation using series truncated at order \mbox{
\(P\)
}
be \mbox{
\(\tau_{im}^{(P)}\)
}
. We define the relative error as
}
\begin{equation}
{\begin{aligned}
\tilde{E}_{im}^{(P)} \equiv \frac{\tau_{i}^{(m)}}{\tau_{im}^{(P)}}-1.
\end{aligned}
}\end{equation}
{To find the asymptotic scaling, we assume the sums are dominated by their lowest-order terms for small \mbox{
\(\gamma\)
}
. Let \mbox{
\(p_0\)
}
be the smallest integer such that \mbox{
\([A^{p_0}]_{im} \ne 0\)
}
(typically \mbox{
\(p_0\)
}
is the length of the shortest directed walk from node \mbox{
\(m\)
}
to \mbox{
\(i\)
}
). Retaining terms up to order \mbox{
\(p_0\)
}
, we find
}\begin{equation}
{\tilde{E}_{im}^{(P)}\sim \frac{1+\frac{O\left(\varepsilon^{P +3}\right)}{O\left(\varepsilon^{P_0+2}\right)}}{1+\frac{O\left(\varepsilon^{P+2}\right)}{O\left(\varepsilon^{P_0 +1}\right)}}-1 \sim O\left(\varepsilon^{P-P_0+1}\right).
}\end{equation}
{when the truncation order \mbox{
\(P\)
}
exceeds the leading walk length \mbox{
\(p_0\)
}
. 
}

\DIFaddend \section{HETEROGENEOUS IN-DEGREE NETWORKS}
\label{sec:E}

\renewcommand\theequation{E\arabic{equation}}
\setcounter{equation}{0}
\makeatletter
\makeatother

\subsection*{Expansion for heterogeneous in-degree}
\label{sec:E:heter}

\par Previous analyses considered homogeneous in-degree configurations, but realistic networks exhibit heterogeneous in-degrees that modifies signal propagation. We now derive expansions for NSDDM systems with arbitrary nodal in-degrees while maintaining uniform decay rates ($\beta_i = \beta$). The matrix elements governing deterministic responses admit generalized expansions:
\begin{equation}
\mathbf{H}^{-n} = (-1)^n \sum_{p=1}^{\infty}\left(\frac{C_{p+n-1}^{n-1}}{(\beta+D)^{p+n}}\right)\mathbf{A}^p
\end{equation}

We start from the expansion for $\mathbf{H}^{-1}$. For $i \neq m$, the element-wise expansion reveals degree-dependent expansion:
\begin{widetext}
\begin{equation}
\refstepcounter{equation}\label{H_inv}
\tag{E\arabic{equation}}
\begin{aligned}
&\left[\mathbf{H}^{-1}\right]_{i m} =-\frac{1}{\beta^2} \Biggl[
\left(1-\frac{D_i+D_m}{\beta}+\frac{D_i^2+D_m^2+D_i D_m}{\beta^2}-\cdots\right)[\mathbf{A}]_{i m}  \\
& +\frac{1}{\beta}\left(1-\frac{D_i+D_m}{\beta}+\cdots\right)[\mathbf{A}^2]_{i m} +\frac{1}{\beta^2}\left(1-\frac{D_i+D_m}{\beta}+\cdots\right)[\mathbf{A}^3]_{i m} + \cdots \Biggr], \\
&=-\frac{1}{\beta^2} \sum_{d=0}^{\infty} \sum_{j_1, \cdots, j_d \in\{m \rightarrow i\}} \frac{1}{\beta^d}\sum_{k=0}^{\infty} \frac{(-1)^k}{\beta^k}\left(\sum_{p_1+\cdots+p_{d+2}=k} D_{j_1}^{p_1} \cdots D_m^{p_{d+2}}\right)A_{m \rightarrow j_1 \rightarrow \cdots \rightarrow i}, \\
&= -\frac{1}{\beta^2} \sum_{d=0}^{\infty} \sum_{j_1  \cdots j_d \in\{m \rightarrow i\}} \frac{1}{\beta^d} \frac{A_{m \rightarrow j_1 \rightarrow \cdots \rightarrow i}}{\left(1+\frac{D_{j_1}}{\beta}\right)\cdots\left(1+\frac{D_{j_d} }{\beta}\right)\left(1+\frac{D_i}{\beta}\right)\left(1+\frac{D_m}{\beta}\right)}, \\
&\equiv - \sum_{w \in \mathcal{W}(m \to i)} h_w^1 {\mathcal{A}_w},
\end{aligned}
\end{equation}
\end{widetext}
where \DIFdelbegin \DIFdel{\mbox{
$h_w^1 \equiv \prod_{v \in w}(\beta+D_i)^{-1}$
}
}\DIFdelend \DIFaddbegin {\mbox{
$h_w^1 \equiv \prod_{v \in w}(\beta+D_i)^{-1} = \prod_{v \in w} \gamma_i$
}
}\DIFaddend , and $\mathcal{A}_w = A_{m \rightarrow j_1 \rightarrow \cdots \rightarrow i} \equiv \prod_{t=0}^{k-1} A_{w_{t+1} w_{t}}$ corresponds to the product of edge weights along the walk $w = (w_0, w_1, \dots, w_k)$ with $w_0 = m$ and $w_k = i$. Under homogeneous conditions, this reduces to Eq.~(\ref{homo:first}).

The critical identity for these expansions is established via induction:
\begin{equation*}
\sum_{k=0}^{\infty}\left(-\frac{1}{\beta}\right)^k \sum_{p_1+\cdots=k} D_1^{p_1} \cdots D_n^{p_n}=\frac{1}{\left(1+\frac{D_1}{\beta}\right) \cdots\left(1+\frac{D_n}{\beta}\right)} 
\end{equation*}

{Inductive step:} Assume validity for $n$ nodes. For $n+1$ nodes:
\begin{equation}
\refstepcounter{equation}\label{H_second_inv}
\tag{E\arabic{equation}}
\begin{aligned}
& \sum_{k=0}^{\infty}\left(-\frac{1}{\beta}\right)^k \sum_{p_1+\cdots+p_{n+1}=k} D_1^{p_1} \cdots D_n^{p_n} D_{n+1}^{p_{n+1}}, \\
&=\sum_{k=0}^{\infty}\left(-\frac{1}{\beta}\right)^k \sum_{p_{n+1}=0}^k D_{n+1}^{p_{n+1}} \sum_{p_1+\cdots+p_n=k-p_{n+1}} D_1^{p_1} \cdots D_n^{p_n}, \\
&= \frac{1}{\left(1+\frac{D_1}{\beta}\right) \cdots\left(1+\frac{D_n}{\beta}\right)}\left(1-\frac{D_{n+1}}{\beta}+\cdots\right)=h_w^1. \\
\end{aligned}
\end{equation}

The quadratic inverse operator exhibits diverse path weights:

\begin{widetext}
\begin{equation}
\begin{aligned}
&\left[\mathbf{H}^{-2}\right]_{i m}
= \frac{1}{\beta^3} \Biggl[
\left(2-\frac{3(D_i+D_m)}{\beta}+\cdots\right)[\mathbf{A}]_{im} + \cdots \Biggr],  \\
&=\frac{1}{\beta^3} \sum_{d=0}^{\infty} \sum_{j_1, \cdots, j_d \in\{m \rightarrow i\}} \frac{1}{\beta^d}\sum_{k=0}^{\infty} \frac{(-1)^k}{\beta^k}(k+d+2)\left(\sum_{p_1+\cdots+p_{q+2}=k} D_{j_1}^{p_1} \cdots D_m^{p_{q+2}}\right)A_{m \rightarrow j_1 \rightarrow \cdots \rightarrow i}, \\
&= \frac{1}{\beta^3} \sum_{d=0}^{\infty} \sum_{j_1  \cdots j_d \in\{m \rightarrow i\}} \frac{1}{\beta^d} \frac{(d+2)+(d+1)\left(\frac{D_{j_1}}{\beta}+\cdots+\frac{D_{j_d}}{\beta}+\frac{D_i}{\beta}+\frac{D_m}{\beta}\right)+d\left(\frac{D_{j_1}}{\beta} \frac{D_{j_{2}}}{\beta}+\cdots\right)+\cdots}{\left(1+\frac{D_{j_1}}{\beta}\right)^2\cdots\left(1+\frac{D_{j_d} }{\beta}\right)^2\left(1+\frac{D_i}{\beta}\right)^2\left(1+\frac{D_m}{\beta}\right)^2} A_{m \rightarrow j_1 \rightarrow \cdots \rightarrow i}, \\
&\equiv \sum_{w \in \mathcal{W}(m \to i)} h_w^2 {\mathcal{A}_w},
\end{aligned}
\end{equation}
\end{widetext}
where the self-dynamics factor is
\DIFdelbegin \[\DIFdel{h_w^2 \equiv \frac{1}{\beta^{d+3}} \frac{ \sum_{r=0}^{d+2}(d+2-r) \sum_{\substack{T \subset S \\|T|=r}} \prod_{t \in T} \frac{D_t}{\beta}}{\left(1+\frac{D_{j_1}}{\beta}\right)^2\cdots\left(1+\frac{D_{j_d} }{\beta}\right)^2\left(1+\frac{D_i}{\beta}\right)^2\left(1+\frac{D_m}{\beta}\right)^2}.}\]
\DIFdelend \DIFaddbegin \begin{equation}
{\begin{aligned}
h_w^2 &\equiv \frac{1}{\beta^{d+3}} \frac{ \sum_{r=0}^{d+2}(d+2-r) \sum_{\substack{T \subset S \\|T|=r}} \prod_{t \in T} \frac{D_t}{\beta}}{\left(1+\frac{D_{j_1}}{\beta}\right)^2\cdots\left(1+\frac{D_{j_d} }{\beta}\right)^2\left(1+\frac{D_i}{\beta}\right)^2\left(1+\frac{D_m}{\beta}\right)^2}, \\
&= \prod_{v \in w} \gamma_i ( \sum_{v \in w} \gamma_i).
\end{aligned}
}\end{equation}\DIFaddend 
\DIFdelbegin 

\DIFdelend Homogeneous reduction confirms consistency, matching Eq. \eqref{homo:second}.
The cubic operator is introduced as:

\begin{widetext}
\begin{equation}
\begin{aligned}
&\left[\mathbf{H}^{-3}\right]_{i m} =-\frac{1}{\beta^4} \Biggl[
\left(3-\frac{6(D_i+D_m)}{\beta}+\cdots\right)[\mathbf{A}]_{im} + \cdots \Biggr] , \\
&=-\frac{1}{\beta^4} \sum_{d=0}^{\infty} \sum_{j_1, \cdots, j_d \in\{m \rightarrow i\}} \frac{1}{\beta^d}\sum_{k=0}^{\infty} \frac{(-1)^k}{\beta^k}\frac{(k+d+2)(k+d+3)}{2}\left(\sum_{p_1+\cdots+p_{q+2}=k} D_{j_1}^{p_1} \cdots D_m^{p_{q+2}}\right)A_{m \rightarrow j_1 \rightarrow \cdots \rightarrow i}, \\
&\equiv -\sum_{w \in \mathcal{W}(m \to i)} h_w^3 {\mathcal{A}_w},
\end{aligned}
\end{equation}
\end{widetext}
Accordingly, the self-dynamics factor is
\begin{equation}
h_w^{3}
\equiv
\frac{1}{2}
\left(\prod_{v\in w}\gamma_v\right)
\left[
\left(\sum_{v\in w}\gamma_v\right)^2
+
\sum_{v\in w}\gamma_v^2
\right],
\end{equation}

Homogeneous reduction confirms consistency, matching Eq. \eqref{homo:third}.

For the general expansion for any inverse power of $\mathbf{H}$, define the common terms
\begin{equation}
\Gamma_w
\equiv
\prod_{v\in w}\gamma_v,
\qquad
S_k(w)
\equiv
\sum_{v\in w}\gamma_v^k.
\end{equation}

The self-dynamics factor $h_w^{(n)}$ can be written either through the complete homogeneous symmetric polynomials,
\begin{equation}
{
h_w^{(n)}
=
\Gamma_w\,
h_{n-1}\!\left(\{\gamma_v\}_{v\in w}\right),
}
\label{hw_general_symmetric}
\end{equation}
where
\begin{equation}
h_r\!\left(\{\gamma_v\}_{v\in w}\right)
\equiv
\sum_{\substack{
p_v\geq 0\\
\sum_{v\in w}p_v=r
}}
\prod_{v\in w}\gamma_v^{p_v},
\end{equation}
or, equivalently, through complete exponential Bell polynomials,
\begin{equation}
{
h_w^{(n)}
=
\frac{\Gamma_w}{(n-1)!}
B_{n-1}
\left(
S_1,\,
1!S_2,\,
2!S_3,\,
\ldots,\,
(n-2)!S_{n-1}
\right).
}
\label{hw_general_Bell}
\end{equation}

Thus, the walk expansion of any inverse power of $\mathbf H$ is
\begin{equation}
{
\left[\mathbf H^{-n}\right]_{im}
=
(-1)^n
\sum_{w\in\mathcal W(m\to i)}
h_w^{(n)}
\mathcal A_w.
}
\label{Hn_general}
\end{equation}

The first four coefficients are
\begin{align*}
h_w^{(1)}
&=
\Gamma_w,
\\[1mm]
h_w^{(2)}
&=
\Gamma_w S_1,
\\[1mm]
h_w^{(3)}
&=
\frac{\Gamma_w}{2}
\left(
S_1^2+S_2
\right),
\\[1mm]
h_w^{(4)}
&=
\frac{\Gamma_w}{6}
\left(
S_1^3+3S_1S_2+2S_3
\right),
\end{align*}
so that
\begin{align*}
\left[\mathbf H^{-1}\right]_{im}
&=
-
\sum_{w\in\mathcal W(m\to i)}
\Gamma_w\mathcal A_w,
\\
\left[\mathbf H^{-2}\right]_{im}
&=
\sum_{w\in\mathcal W(m\to i)}
\Gamma_w S_1
\mathcal A_w,
\\
\left[\mathbf H^{-3}\right]_{im}
&=
-\frac{1}{2}
\sum_{w\in\mathcal W(m\to i)}
\Gamma_w
\left(
S_1^2+S_2
\right)
\mathcal A_w,
\\
\left[\mathbf H^{-4}\right]_{im}
&=
\frac{1}{6}
\sum_{w\in\mathcal W(m\to i)}
\Gamma_w
\left(
S_1^3+3S_1S_2+2S_3
\right)
\mathcal A_w.
\end{align*}

For noise input, the main part is to expand the steady-covariance matrix $\mathbf{P}_{\infty}$. 
\begin{equation}
\begin{aligned}
& {\left[\mathbf{P}_{\infty}\right]_{i j}^m }=\left[\mathbf{P}_{\infty}\right]_{ji}^m \\
&=\frac{I_0}{2 \pi} \int_{-\infty}^{\infty}[(\mathbf{H}-(\mathrm{i} \omega) \mathbf{I}_N)^{-1}]_{i m}\left[(\mathbf{H}+(\mathrm{i}\omega) \mathbf{I}_N)^{-1}\right]_{j m} d \omega.
\end{aligned}
\end{equation}

The corresponding expansion in the integral (See Eq.~\eqref{H_inv}) is
\begin{widetext}
\begin{equation}
\refstepcounter{equation}
\label{inv_expand}
\tag{E\arabic{equation}}
\begin{aligned}
 \left[(\mathbf{H}-(\mathrm{i} \omega) \mathbf{I}_N)^{-1}\right]_{i m}&= - \sum_{d_i=0}^{\infty} \sum_{i_1  \cdots i_{d_i} \in\{m \rightarrow i\}}  \frac{A_{m \rightarrow i_1 \rightarrow \cdots \rightarrow i}}{\left({\beta}+{D_{i_1}}+\mathrm{i}w\right)\cdots\left({\beta}+{D_{i_d} }+\mathrm{i}w\right)\left({\beta}+{D_i}+\mathrm{i}w\right)\left({\beta}+{D_m}+\mathrm{i}w\right)}, \\
 \left[(\mathbf{H}+(\mathrm{i}\omega) \mathbf{I}_N)^{-1}\right]_{j m}&= - \sum_{d_j=0}^{\infty} \sum_{j_1  \cdots j_{d_j} \in\{m \rightarrow j\}}  \frac{A_{m \rightarrow j_1 \rightarrow \cdots \rightarrow j}}{\left({\beta}+{D_{j_1}}-\mathrm{i}w\right)\cdots\left({\beta}+{D_{j_d} }-\mathrm{i}w\right)\left({\beta}+{D_j}-\mathrm{i}w\right)\left({\beta}+{D_m}-\mathrm{i}w\right)}.
\end{aligned}
\end{equation}
\end{widetext}
After substitution, the expansion will be 

\begin{widetext}
\begin{equation}
\refstepcounter{equation}
\label{inv_expand_long}
\tag{E\arabic{equation}}
\begin{aligned}
& \left[ \mathbf{P}_{\infty}\right]_{i j}^m
= \frac{I_0}{2\pi} \int_{-\infty}^{\infty}
\left[
\sum_{d_i} \sum_{i_1,\dots,i_{d_i}}
A_{m \rightarrow \cdots \rightarrow i}
\,
\prod_{k \in \{i_1,\dots,i_{d_i}, m\}} \frac{\gamma_{i_1} \cdots \gamma_m}{1 + \mathrm{i}\omega \gamma_k}
\right]
\left[
\sum_{d_j} \sum_{j_1,\dots,j_{d_j}}
A_{m \rightarrow \cdots \rightarrow j}
\,
\prod_{\ell \in \{j_1,\dots,j_{d_j}, m\}} \frac{\gamma_{j_1} \cdots \gamma_m}{1 - \mathrm{i}\omega \gamma_\ell}
\right]
d\omega.
\end{aligned}
\end{equation}
\end{widetext}

We employ Cauchy's residue theorem to evaluate this integral of rational functions, following the approach outlined in Ref.~\cite{brown2009complex}. 
For a walk
$
w=(i_0,i_1,\ldots,i_{d_i})
$
from $m$ to $i$, and a walk
$
v=(j_0,j_1,\ldots,j_{d_j})
$
from $m$ to $j$, we set
$
i_0=j_0=m
$,
$
i_{d_i}=i
$,
and
$
j_{d_j}=j
$.
Thus, the expression is

\begin{widetext}
\begin{equation}
\refstepcounter{equation}\label{P_inf_span}
\begin{aligned}
\left[\mathbf{P}_{\infty}\right]_{ij}^{m}
&=
I_0
\sum_{d_i,d_j}
\sum_{\substack{i_1\cdots i_{d_i-1}\\
j_1\cdots j_{d_j-1}}}
\sum_{p=0}^{d_i}
\left[
\prod_{q=0}^{d_j}
\frac{1}{
2\beta+D_{i_p}+D_{j_q}
}
\prod_{\substack{k=0\\k\neq p}}^{d_i}
\frac{1}{
D_{i_k}-D_{i_p}
}
\right]
\mathcal{A}_w\mathcal{A}_v,
\\
&\equiv
I_0
\sum_{\substack{
w\in\mathcal{W}(m\to i)\\
v\in\mathcal{W}(m\to j)
}}
\Gamma_w\Gamma_v\,
\Phi(w,v)\,
\mathcal{A}_w\mathcal{A}_v,
\end{aligned}
\end{equation}
\end{widetext}
where
$
\mathcal{A}_w
$
and
$
\mathcal{A}_v
$
are the products of edge weights along the two walks. The terms
$
\Gamma_w
=
\prod_{p=0}^{d_i}\gamma_{i_p},
,
\Gamma_v
=
\prod_{q=0}^{d_j}\gamma_{j_q},
$
are the self-dynamics factor for each walk. The temporal-overlap factor is
\begin{equation}
\Phi(w,v)
=
\sum_{p=0}^{d_i}
\frac{
\gamma_{i_p}^{\,n_w+n_v-2}
}{
\displaystyle
\prod_{q=0}^{d_j}
\left(
\gamma_{i_p}+\gamma_{j_q}
\right)
\prod_{\substack{k=0\\k\neq p}}^{d_i}
\left(
\gamma_{i_p}-\gamma_{i_k}
\right)
}.
\label{Phi_wv}
\end{equation}

Equivalently, evaluating the residues using the poles associated with the second walk gives

\begin{equation}
\Phi(w,v)
=
\sum_{q=0}^{d_j}
\frac{
\gamma_{j_q}^{\,n_w+n_v-2}
}{
\displaystyle
\prod_{p=0}^{d_i}
\left(
\gamma_{i_p}+\gamma_{j_q}
\right)
\prod_{\substack{k=0\\k\neq q}}^{d_j}
\left(
\gamma_{j_q}-\gamma_{j_k}
\right)
}.
\end{equation}

Eq.~\eqref{P_inf_span} therefore separates each pairwise contribution to the stationary covariance into a structural and a temporal component. 
The factor
$
\mathcal{A}_w\mathcal{A}_v
$
describes the product of edge weights along the two propagation walks,
$
\Gamma_w\Gamma_v
$
collects the intrinsic time constants of all nodes along them, and
$
\Phi(w,v)
$
quantifies the temporal overlap between the two propagated responses.

Although the residue representation appears singular when two intrinsic time constants coincide, due to factors such as
$
(\gamma_{i_p}-\gamma_{i_k})^{-1}
$,
these singularities are removable. 
When the contributions from all residues are combined, the corresponding factors cancel and the original integral remains finite~\cite{berrut2004barycentric,stoer1980introduction,brown2009complex}. 
Equal intrinsic time constants therefore correspond to repeated poles and can be obtained either by taking the appropriate limit of the expression above or by directly evaluating the corresponding higher-order residues. 
After taking this limit, the homogeneous case reduces to Eq.~\eqref{p_ij}.

For the source--target covariance $\left[\mathbf{P}_{\infty}\right]_{im}^{m}$,
the second walk reduces to the zero-length walk $v=(m)$, with
$\mathcal{A}_v=1$. The general path-pair expansion therefore reduces to
\begin{align}
\left[\mathbf{P}_{\infty}\right]_{im}^{m}
&=
I_0
\sum_{w\in\mathcal{W}(m\to i)}
\left[
\prod_{u\in w}
\frac{1}{2\beta+D_m+D_u}
\right]
\mathcal{A}_w
\nonumber\\
&=
\frac{I_0\gamma_m}{2}
\sum_{w\in\mathcal{W}(m\to i)}
\left[
\prod_{\substack{u\in w\\u\neq m}}
\frac{\gamma_m\gamma_u}
{\gamma_m+\gamma_u}
\right]
\mathcal{A}_w.
\end{align}
The effective temporal-overlap factor between node $u$ and the
input node $m$ is defined as
$$
\widetilde{\gamma}^{(m)}_u
\equiv
\frac{\gamma_m\gamma_u}
{\gamma_m+\gamma_u}
=
\frac{1}{2\beta+D_m+D_u}.
$$
The source--target covariance can then be written in the compact pathwise form
\begin{equation}
\left[\mathbf{P}_{\infty}\right]_{im}^{m}
=
\frac{I_0\gamma_m}{2}
\sum_{w\in\mathcal{W}(m\to i)}
\left(
\prod_{\substack{u\in w\\u\neq m}}
\widetilde{\gamma}^{(m)}_u
\right)
\mathcal{A}_w.
\label{heter_p_inf}
\end{equation}

This expression has the similar pathwise structure as the heterogeneous
deterministic response, except that each intrinsic time constant $\gamma_u$
is replaced by the effective overlap factor
$\widetilde{\gamma}^{(m)}_u$. In the limit
$\gamma_m\gg\gamma_u$ for the nodes along the propagation walks,
$
\widetilde{\gamma}^{(m)}_u
\rightarrow
\gamma_u,
$
and therefore
$
{2}
\left[\mathbf{P}_{\infty}\right]_{im}^{m}
\rightarrow
R_i^{(m)}.
$

For the autocovariance $\left[\mathbf{P}_{\infty}\right]_{ii}^{m}$, the expression can be written as
\begin{equation}
\left[\mathbf{P}_{\infty}\right]_{ii}^{m}
=
I_0
\sum_{\substack{
w\in\mathcal{W}(m\to i)\\
v\in\mathcal{W}(m\to i)
}}
\Phi(w,v)
\mathcal{R}_w^{(m)}
\mathcal{R}_v^{(m)}.
\label{heter_p_ii}
\end{equation}
where
$$
\mathcal{R}_w^{(m)}
\equiv
\Gamma_w\mathcal{A}_w
=
\left(
\prod_{u\in w}\gamma_u
\right)
\mathcal{A}_w.
$$
Using the symmetry $\Phi(w,v)=\Phi(v,w)$, Eq.~\eqref{heter_p_ii}
can equivalently be expressed as

\begin{equation}
\begin{aligned}
\left[\mathbf{P}_{\infty}\right]_{ii}^{m}
=
I_0 \Bigg[
&\sum_w \Phi(w,w)\left(\mathcal{R}_w^{(m)}\right)^2
\\
&\quad + 2\sum_{w<v} \Phi(w,v)\mathcal{R}_w^{(m)}\mathcal{R}_v^{(m)}
\Bigg].
\end{aligned}
\end{equation}
Here, $\Phi(w,w)$ weights the contribution of an individual walk, which links the quadratic term, while $\Phi(w,v)$ measures the temporal overlap between two distinct walks between $m$ and $i$.

The relation with the deterministic response emerges
when the intrinsic time constant of the input node is much larger than
those of the downstream nodes, $\gamma_m\gg\gamma_u$. In this limit,
$
\Phi(w,v)
\rightarrow
\frac{1}{2\gamma_m}
$
for all relevant pairs of walks, and Eq.~\eqref{heter_p_ii} becomes

\begin{equation}
\left[\mathbf{P}_{\infty}\right]_{ii}^{m}
\rightarrow
\frac{I_0}{2\gamma_m}
\left(
R_i^{(m)}
\right)^2.
\end{equation}

A simple example illustrates the different covariance terms. Consider two direct walks $w:m\to i$ and $v:m\to j$, with intrinsic-timescale factors $\Gamma_w=\gamma_m\gamma_i$ and $\Gamma_v=\gamma_m\gamma_j$. Substituting these walks into Eq.~\eqref{Phi_wv} gives
$$
\Phi(w,v)
=
\frac{
2\gamma_i\gamma_j
+
\gamma_m(\gamma_i+\gamma_j)
}{
2(\gamma_i+\gamma_j)
(\gamma_m+\gamma_i)
(\gamma_m+\gamma_j)
}.
$$
Hence, their contribution to the cross-covariance between nodes $i$ and $j$ is
$$
\left[\mathbf{P}_{\infty}\right]_{ij}^{m}
=
I_0
\gamma_m^2\gamma_i\gamma_j
\frac{
2\gamma_i\gamma_j
+
\gamma_m(\gamma_i+\gamma_j)
}{
2(\gamma_i+\gamma_j)
(\gamma_m+\gamma_i)
(\gamma_m+\gamma_j)
}
\mathcal{A}_w\mathcal{A}_v.
$$

For the source--target covariance $\left[\mathbf{P}_{\infty}\right]_{im}^{m}$, the second walk reduces to the zero-length walk $v=(m)$, with $\mathcal{A}_v=1$. Therefore, 
$$
\left[\mathbf{P}_{\infty}\right]_{im}^{m}
=
I_0
\frac{
\gamma_m^2\gamma_i
}{
2(\gamma_m+\gamma_i)
}
\mathcal{A}_w.
$$

For the target autocovariance, $i=j$, the two walks are identical, $w=v:m\to i$, and the temporal-overlap factor simplifies to
$$
\Phi(w,w)
=
\frac{1}{
2(\gamma_m+\gamma_i)
}.
$$
Thus,
$$
\left[\mathbf{P}_{\infty}\right]_{ii}^{m}
=
I_0
\frac{
\gamma_m^2\gamma_i^2
}{
2(\gamma_m+\gamma_i)
}
\mathcal{A}_w^2.
$$

These expressions illustrate how path structure enters covariance: source--target covariance depends linearly on the propagation-path weight from source to the target, target autocovariance quadratically, while covariance between two downstream nodes depends on the product of their respective path weights and their temporal overlap.

We further investigate the steady covariance matrix Eq.~\eqref{P_inf_span} using three three-node motifs: a chain, a feedforward triangle, and a feedback triangle. Consider the covariance between nodes $1$ and $2$ generated by a stochastic input at node $0$, and define
$
x_{ij} \equiv 2\beta+D_i+D_j = \gamma_i^{-1}+\gamma_j^{-1}.
$

\textbf{Chain} $0\to1\to2$: only one pair of walks contributes,
$$
w:0\to1, \qquad v:0\to1\to2.
$$
The factor is
$$
\Theta_{\mathrm{ch}}
\equiv
\Gamma_w\Gamma_v\Phi(w,v)
=
\frac{
x_{01}+x_{02}+x_{12}
}{
x_{00}x_{11}x_{01}x_{02}x_{12}
},
$$
giving
\begin{equation}
{
P_{12}^{\mathrm{ch}}
=
I_0\Theta_{\mathrm{ch}}
[\mathbf A]_{10}^{2}[\mathbf A]_{21}.
}
\label{P_chain}
\end{equation}

\textbf{Feedforward triangle}: the additional edge $0\to2$ introduces a second walk $v':0\to2$. The new walk pair $w:0\to1$ and $v':0\to2$ has factor
$$
\Theta_{\mathrm{dir}}
\equiv
\Gamma_w\Gamma_{v'}\Phi(w,v')
=
\frac{
x_{01}+x_{02}
}{
x_{00}x_{01}x_{02}x_{12}
}.
$$
Thus,
\begin{equation}
{
P_{12}^{\mathrm{FF}}
=
P_{12}^{\mathrm{ch}}
+
I_0\Theta_{\mathrm{dir}}
[\mathbf A]_{10}[\mathbf A]_{20}.
}
\label{P_FF}
\end{equation}

\textbf{Feedback triangle}: the additional edge $2\to0$ closes the loop $0\to1\to2\to0$. Define
$$
A_{\Delta} \equiv [\mathbf A]_{10}[\mathbf A]_{21}[\mathbf A]_{02},
\quad
\rho \equiv A_{\Delta}\gamma_0\gamma_1\gamma_2,
\quad
\chi \equiv \frac{A_{\Delta}}{x_{01}x_{02}x_{12}}.
$$
Summing all recurrent walks gives
\begin{equation}
{
P_{12}^{\mathrm{FB}}
=
\frac{P_{12}^{\mathrm{ch}}}{(1-\rho)(1+\chi)}.
}
\label{P_FB}
\end{equation}

In summary, compared with a single chain case, a feedforward edge adds a new path contribution, while a feedback edge reuses the original pathway through recurrent multiplicative factors. Fig.~\ref{fig:P_inf} compares these theoretical predictions with Lyapunov solutions and with zero-lag covariances obtained from simulations. The agreement is excellent, with negligible relative errors, confirming the accuracy of the path-based expansion.

\begin{figure}
\centering
\includegraphics[scale=0.55]{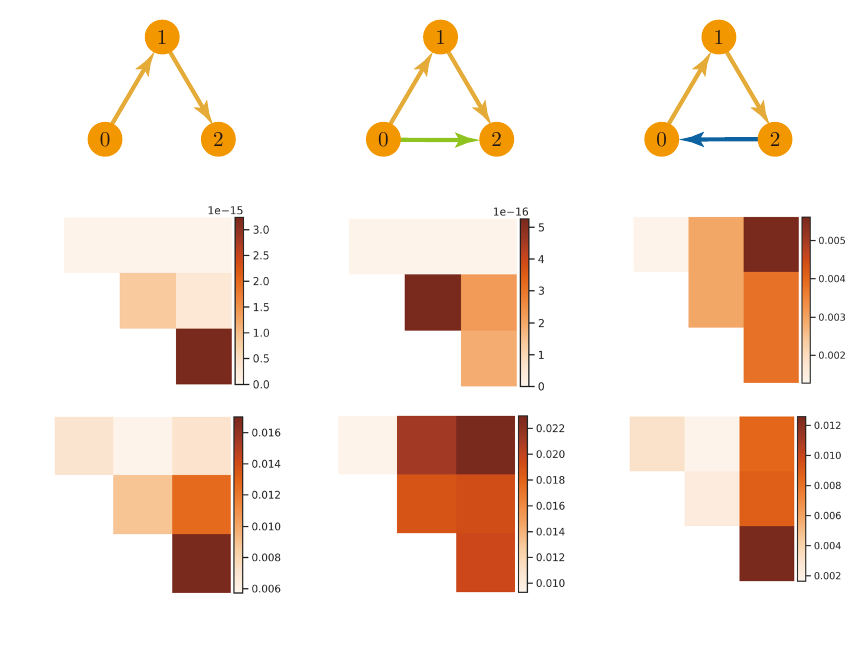}
\caption{\label{fig:P_inf}
Steady-state covariance expansions ($\mathbf{P}_{\infty}$) for three-node motifs (chain, feedforward, and feedback triangles), all with input at node $0$: Top row shows motif schematics; middle row displays relative error between Lyapunov equation solutions and theoretical expansions; bottom row presents relative error between time-lag-zero covariance values averaged for $100$ rounds and theoretical expansions. Parameters: $\beta=10$, edge weights $A_{0\to1}=1$, $A_{1\to2}=2$, $A_{2\to1}=4$.   
}
\end{figure}

\par In the general case, the three autocovariance metrics therefore become
\begin{align}
Z_{ii}^{(m)}
&=
I_0\left(R_i^{(m)}\right)^2,
\\[4pt]
P_{ii}^{(m)}
&=
I_0
\sum_{\substack{
w,v\in\mathcal W(m\to i)
}}
\Phi(w,v)
\mathcal{R}_w^{(m)}
\mathcal{R}_v^{(m)},
\\[4pt]
\tau_{ii}^{(m)}
&=
\frac{
I_0\left(R_i^{(m)}\right)^2
}{
2P_{ii}^{(m)}
}.
\end{align}

\par Similarly, the three cross-covariance metrics reduce to
\begin{align}
Z_{ij}^{(m)}
&=
I_0R_i^{(m)}R_j^{(m)},
\\[4pt]
P_{ij}^{(m)}
&=
\frac{
I_0\left|R_i^{(m)}R_j^{(m)}\right|
}{
\sqrt{
2\left[
\left(\sigma_i^{(m)}\right)^2
+
\left(\sigma_j^{(m)}\right)^2
\right]
}
},
\\[4pt]
t_{ij}^{(m)}
&=
\tau_i^{(m)}-\tau_j^{(m)},
\end{align}
where the temporal variance is
\begin{equation}
\begin{aligned}
\left(\sigma_i^{(m)}\right)^2
&=
\frac{1}{R_i^{(m)}}
\sum_{w\in\mathcal{W}(m\to i)}
\Gamma_w
\left[S_1(w)^2 + S_2(w)\right]
\mathcal{A}_w
\\
&\quad
-\left(\tau_i^{(m)}\right)^2.
\end{aligned}
\end{equation}

\subsection*{Self-responses}

In this subsection, we analyze the self-response case, where the source and target nodes coincide, i.e., \( i = m \).

\begin{figure*}
\centering
\includegraphics[scale=0.45]{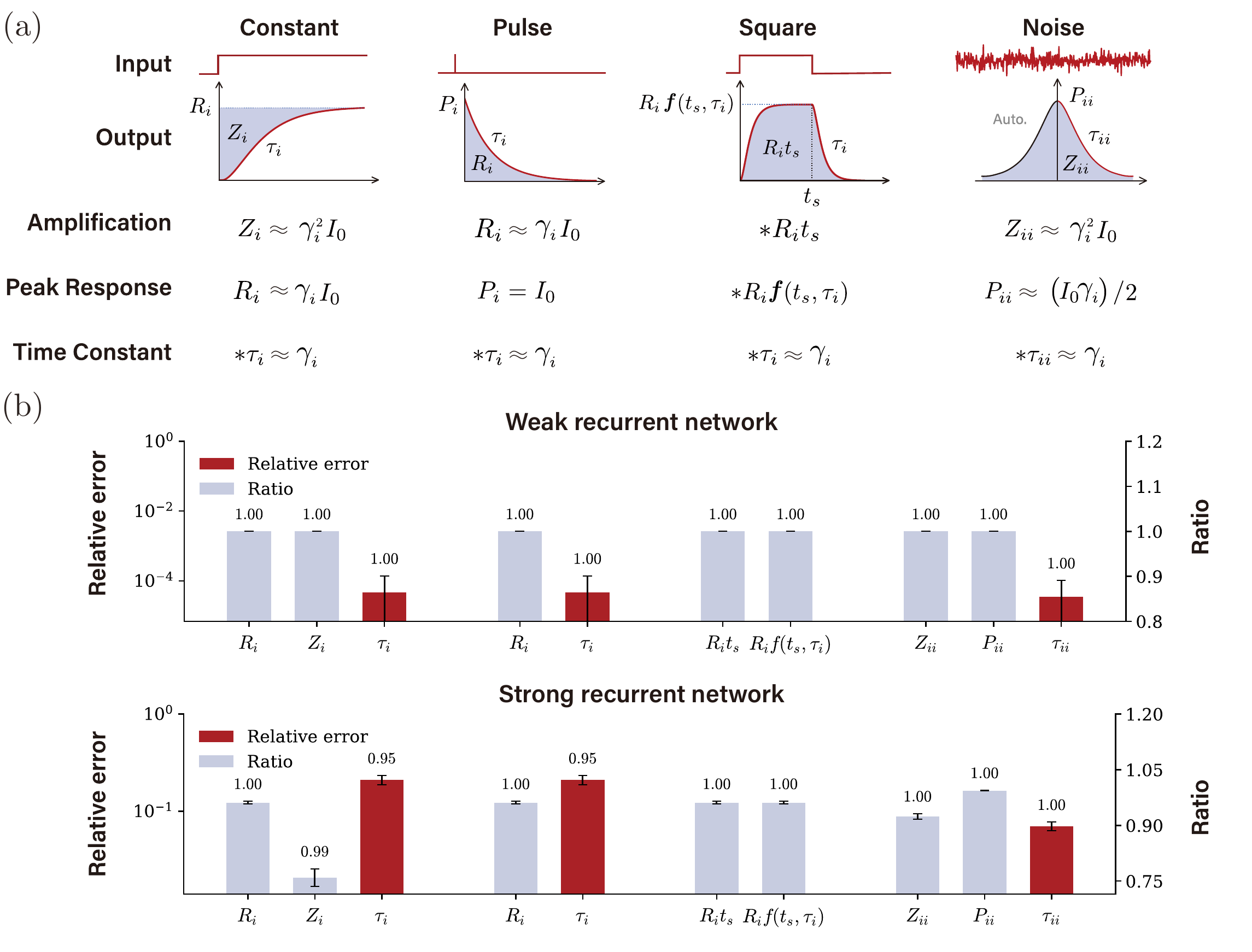}
\caption{
Summary of the self-response metrics of source node \(i\) under constant, pulse, square, and white-noise inputs. Blue bars show the mean theory-to-simulation ratios (right) for response-strength metrics, whereas red bars show the mean relative errors (left) for temporal metrics; error bars denote standard deviations. The response-strength ratios are close to unity, and all temporal relative errors are below \(10\%\). Numbers shown on the bars denote Spearman rank correlations between theoretical predictions and full-network numerical results. The theoretical predictions retain only the local, leading-order term (\(p=0\)), which is governed by
\(\gamma_i\).
The intrinsic time constants were independently sampled as
\(\gamma_i\sim\mathcal{U}(0.08,0.12)\).
Results were averaged over \(100\) realizations of directed Erd\H{o}s--R\'enyi networks with \(N=100\), connection probability \(p=0.05\) for weak recurrent networks, and \(p=0.5\) for strong recurrent networks, and homogeneous positive edge weight \(\alpha=0.1\) for weak while \(\alpha=1.0\) for strong, with \(10\) source nodes randomly selected from each realization. The input amplitude was unity, and the square-pulse duration was \(t_s=10\). The corresponding scatter plots are presented in Sec.~IV of the Supplemental Material. Weakly recurrent networks show excellent quantitative agreement between theory and simulation, whereas strongly recurrent networks exhibit systematic deviations, particularly in the temporal metrics, because higher-order recurrent effects are not captured by the leading-order approximation.
}
\label{fig:mm}
\end{figure*}

We begin with $\mathbf{H}^{-1}$:

\begin{equation}
\begin{aligned}\label{self:deter_inv}
&\left[\mathbf{H}^{-1}\right]_{ii} = \left[(\mathbf{A} - \mathbf{D} - \beta \mathbf{I}_N)^{-1}\right]_{ii}, \\
&= -\frac{1}{\beta}\left[\left(\mathbf{I}_N + \frac{1}{\beta} \mathbf{D}\right)^{-1} + \frac{1}{\beta} \left(\mathbf{A} + \frac{1}{\beta} \mathbf{A}^2 + \cdots\right)\right]_{ii}, \\
&= -\left[\frac{1}{\beta+D_i}+\sum_{d=1}^{\infty} \sum_{j_1  \cdots j_d \in\{i \rightarrow i\}} \frac{A_{i \rightarrow j_1 \rightarrow \cdots \rightarrow i}}{\left({\beta}+D_{j_1}\right)\cdots\left({\beta}+{D_{j_d}}\right)} \right], \\
&\equiv -\sum_{d=0}^{\infty} \sum_{j_1  \cdots j_d \in\{i \rightarrow i\}} \frac{A_{i \rightarrow j_1 \rightarrow \cdots \rightarrow i}}{\left({\beta}+D_{j_1}\right)\cdots\left({\beta}+{D_{j_d}}\right)\left({\beta}+{D_i}\right)}, \\
&\equiv  -\gamma_i - \sum_{w \in \mathcal{W}(i \to i)} \left(\prod_{v \in w} \gamma_i\right)  {\mathcal{A}_w}.
\end{aligned}
\end{equation}

This result generalizes to $\mathbf{H}^{-2}$:

\begin{equation}   
\refstepcounter{equation}
\label{self:deter}
\tag{E\arabic{equation}}
\begin{aligned} 
&\left[\mathbf{H}^{-2}\right]_{ii}  = \gamma_i^2
+
\sum_{w\in\mathcal{W}(i\to i)}
(\prod_{u\in w}\gamma_u) (\sum_{u\in w}\gamma_u)\,\mathcal A_w.
\end{aligned}
\end{equation}

\medskip

Expansion of the steady-state covariance matrix has similar forms with Eq. \eqref{P_inf_span} under unit amplitude ($I_0=1$) for simplicity.
\begin{widetext}
\begin{equation}
\begin{aligned}
 \left[ \mathbf{P}_{\infty}\right]_{i i}^i =
\frac{\gamma_i}{2}
+
2\sum_{w\in\mathcal{W}(i\to i)}
\gamma_i\Gamma_w
\Phi(w,w)\mathcal{A}_w
+
\sum_{\substack{
w\in\mathcal{W}(i\to i)\\
v\in\mathcal{W}(i\to i)
}}
\Gamma_w\Gamma_v
\Phi(w,v)
\mathcal{A}_w\mathcal{A}_v.
\end{aligned}
\end{equation}
\end{widetext}
When node $i$ does not participate in any nontrivial closed walk, the
self-response is determined exactly by the local timescale $\gamma_i$.
More generally, the same approximation remains accurate when recurrent
closed-walk contributions are weak, for example when recurrent couplings
or the products of local time constants along the loops are small. In this
regime,
\[
[\mathbf H^{-1}]_{ii}\approx-\gamma_i,
\qquad
[\mathbf H^{-2}]_{ii}\approx\gamma_i^2,
\qquad
P_{ii}^{(i)}\approx\frac{\gamma_i}{2}.
\]

The integrated stochastic amplification is determined exactly by the
source-to-source static gain,
\begin{equation}
\refstepcounter{equation}
\label{homo:self_amp}
\tag{E\arabic{equation}}
Z_{ii}^{(i)}
=
\left(R_i^{(i)}\right)^2
\approx
\gamma_i^2,
\qquad
R_i^{(i)}=-[\mathbf H^{-1}]_{ii}.
\end{equation}
Accordingly, the temporal metric satisfies
\begin{equation}
\refstepcounter{equation}
\label{homo:self_time}
\tag{E\arabic{equation}}
\tau_{ii}^{(i)}
=
\frac{Z_{ii}^{(i)}}{2P_{ii}^{(i)}}
\approx
\gamma_i.
\end{equation}

Thus, when recurrent return effects are negligible,
\begin{align*}
& R_i^{(i)}\approx 2P_{ii}^{(i)} \approx \gamma_i, \\
&Z_{ii}^{(i)}\approx Z_{i}^{(i)} \approx \gamma_i^2, \\
&\tau_{ii}^{(i)}\approx \tau_{i}^{(i)} \approx\gamma_i,
\end{align*}
showing that deterministic and stochastic self-responses are governed by
the same leading local dynamics. Their differences arise primarily from
recurrent contributions and the associated temporal overlap.

Fig.~\ref{fig:mm} summarizes the self-response metrics for a source node
under four input types. The local
contribution alone ($p=0$) already provides accurate estimates under the weakly recurrent heterogeneous networks. In both cases, errors remain below $10\%$ and
estimated-to-exact ratios remain close to unity, with consistent rank
correlations across metrics. Additional validations are provided in
SM Sec.~IV.

\subsection*{Chain-wise propagation law}
\label{sec:E:iterative}
We analyze the constant-input propagation behavior along a single path under heterogeneous in-degree conditions and regard it as a baseline model due to its simplicity, in the NSDDM setting. The transient response to constant input along the path can be characterized by the following metrics:

\DIFaddend 

\begin{equation}
\begin{aligned}
R(d)
&=
A_{m\to\cdots\to i}\,
\gamma_m\gamma_{j_1}\cdots\gamma_i,
\\
Z(d)
&=
R(d)
\left(
\gamma_m+\gamma_{j_1}+\cdots+\gamma_i
\right),
\\
\tau(d)
&=
\gamma_m+\gamma_{j_1}+\cdots+\gamma_i.
\end{aligned}
\end{equation}

By comparing the response metrics at path length $d$ and $d+1$, we obtain the following iterative relations:
\begin{equation}
\refstepcounter{equation}\label{iterative_whole}
\begin{aligned}
&\frac{Z(d+1)}{Z(d)}
=
A_{d\to d+1}
\left(
\gamma_{d+1}
+
\gamma_{d+1}^{2}\frac{R(d)}{Z(d)}
\right)
\\
&=
A_{d\to d+1}\gamma_{d+1}
\left(
1+\frac{\gamma_{d+1}}{\tau(d)}
\right)
\;\longrightarrow\;
A_{d\to d+1}\gamma_{d+1},
\ d\to\infty,
\\[4pt]
&\frac{R(d+1)}{R(d)}
=
A_{d\to d+1}\gamma_{d+1},
\\[4pt]
&\tau(d+1)-\tau(d)
=
\gamma_{d+1}.
\end{aligned}
\end{equation}
\begin{figure*}
\centering
\DIFdelbeginFL 
\DIFdelendFL \includegraphics[scale=0.4]{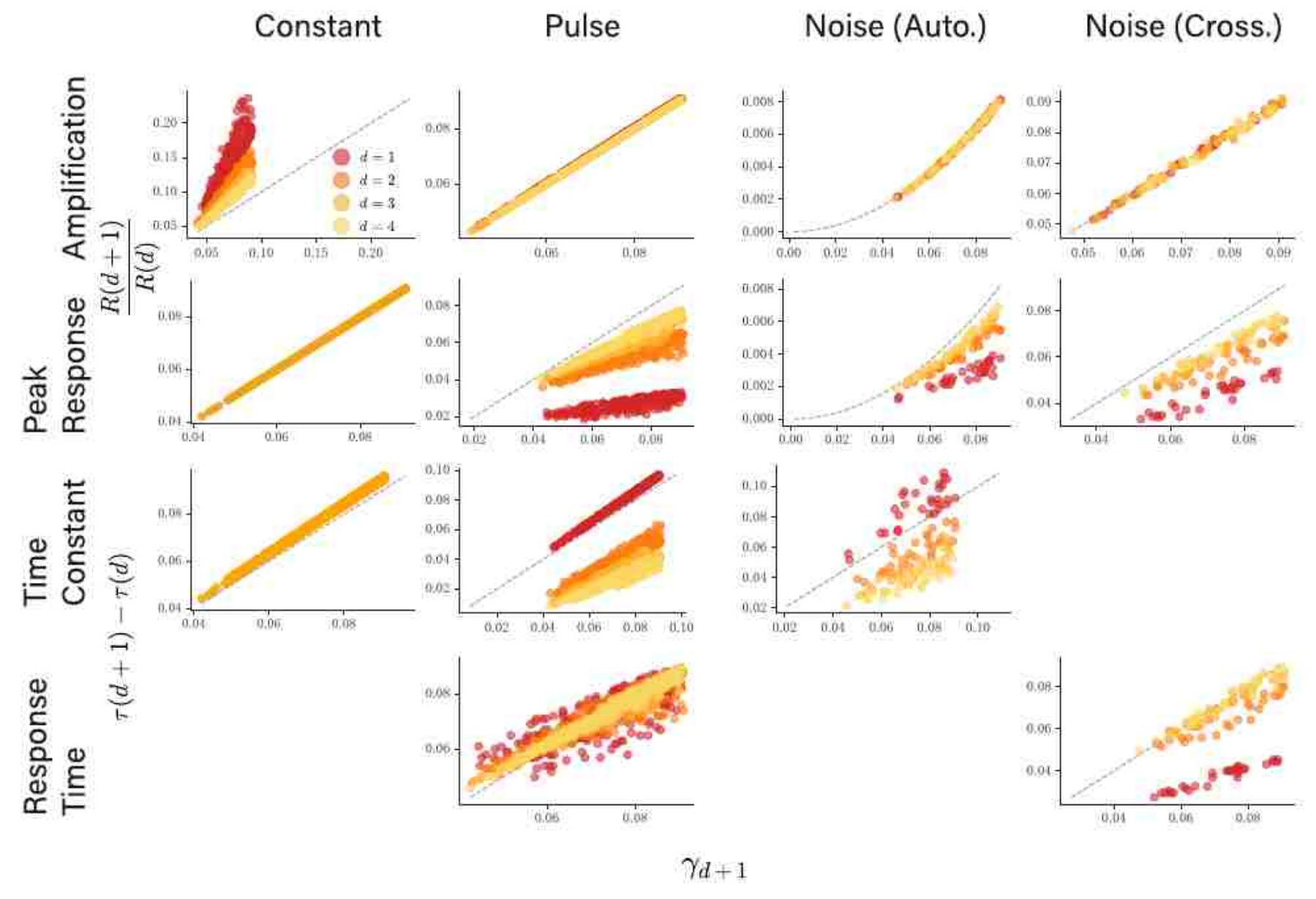}
\DIFaddendFL \caption{\label{fig:diff}
Propagation laws for heterogeneous degree configurations along a single chain. Strength metrics ($Z,R,P$) follow the scaling relation $R(d+1)/R(d) \approx 1/(\beta + D_{d+1})\equiv \gamma_{d+1}$, particularly when $d$ is large, apart from autocovariance, which obeys a squared form, $\gamma_{d+1}^2$. Temporal metrics satisfy $\tau(d+1) - \tau(d) \approx \gamma_{d+1}$. Parameters: $\beta = 10$, chain length $=5$, uniform edge weights $A_{d \to {d+1}}=1$, with input applied to the first node \DIFaddbeginFL {(see Fig}\DIFaddendFL .\DIFaddbeginFL {~S50 for any negative couplings included).}\DIFaddendFL }
\end{figure*}

\begin{figure*}
\centering
\includegraphics[scale=0.45]{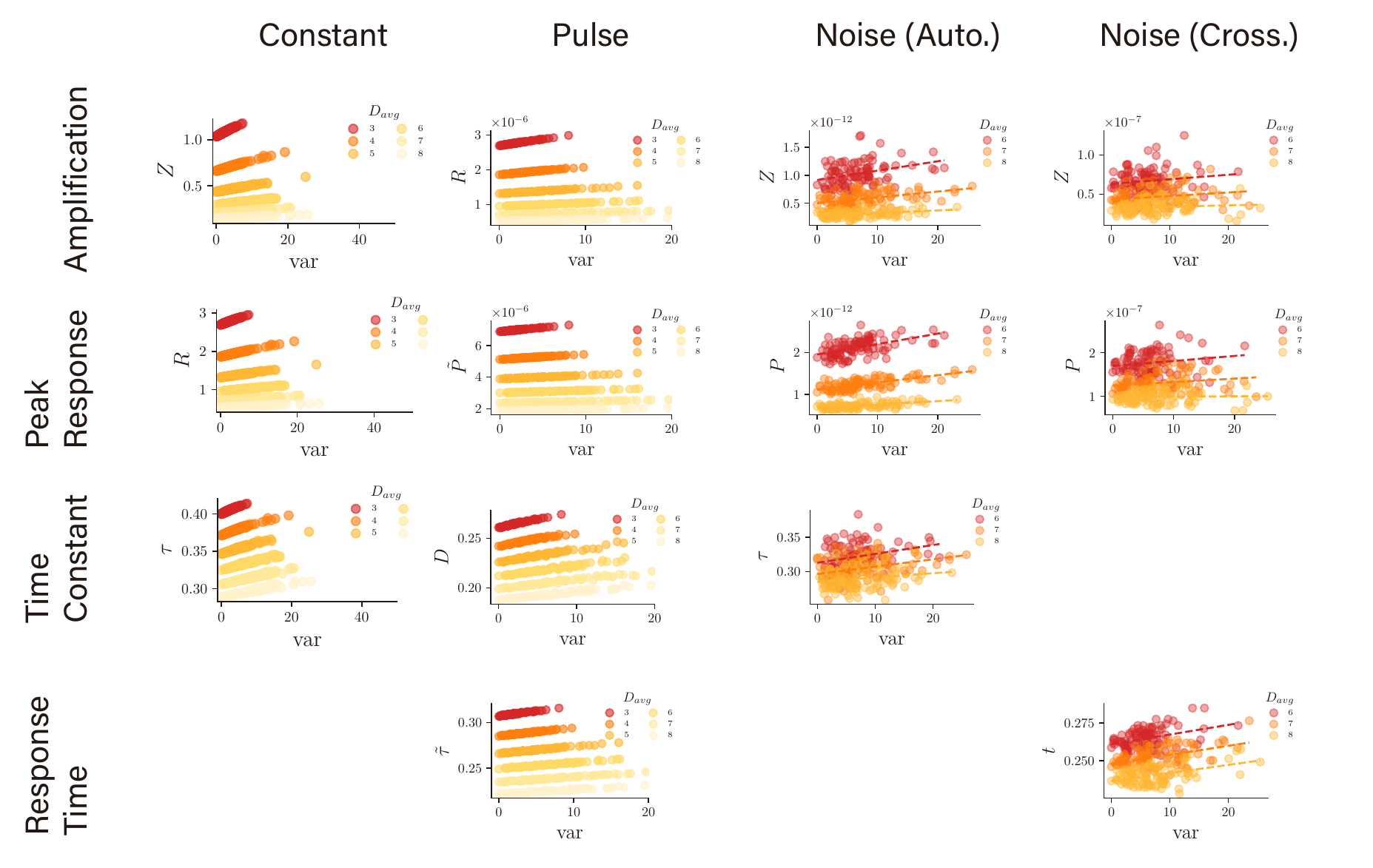}
\caption{\label{fig:heter}
Impact of mean and variance of degree distributions on a single chain. All metrics indicate that larger mean degrees suppress the metrics, whereas larger variance enhances them. Parameters: $\beta = 10$, chain length $=5$, uniform edge weights $A_{d \to {d+1}}=1$ along the chain, with input applied to the first node.}
\end{figure*}

\par For white-noise input, the corresponding autocovariance and
cross-covariance amplification, and peak response time for cross-covariance obey
\begin{equation}
\begin{aligned}
\frac{Z_{d+1,d+1}^{(m)}}{Z_{dd}^{(m)}}
&=
\left(A_{d\to d+1}\gamma_{d+1}\right)^2,
\\
\frac{Z_{d+1,m}^{(m)}}{Z_{dm}^{(m)}}
&=
A_{d\to d+1}\gamma_{d+1},
\\
t_{d+1,m}^{(m)}-t_{dm}^{(m)}
&=
\gamma_{d+1}.
\end{aligned}
\label{eq:noise_chain_simple}
\end{equation}
Thus, $Z_{dm}^{(m)}$ and $t_{dm}^{(m)}$ inherit exactly the same
multiplicative-amplitude and additive-time laws as the deterministic
response, while $Z_{dd}^{(m)}$ follows the squared amplitude law.

\par The covariance peak itself retains an additional dependence on the
temporal history of the path. Defining
$
S_d
\equiv
2\gamma_m^2+\sum_{k=1}^{d}\gamma_k^2,
$
the source--target cross-covariance peak satisfies
\begin{equation}
\frac{P_{d+1,m}^{(m)}}{P_{dm}^{(m)}}
=
\left|A_{d\to d+1}\right|\gamma_{d+1}
\sqrt{
\frac{S_d}
{S_d+\gamma_{d+1}^2}
}.
\label{eq:noise_cross_peak_chain}
\end{equation}
As $S_d$ grows relative to
$\gamma_{d+1}^2$, this correction approaches unity, giving
\[
\frac{P_{d+1,m}^{(m)}}{P_{dm}^{(m)}}
\longrightarrow
\left|A_{d\to d+1}\right|\gamma_{d+1}.
\]

\par For the autocovariance, let
$\Phi_d\equiv\Phi(w_d,w_d)$ denote the temporal-overlap factor of the
unique path $w_d=(m,\ldots,d)$. Then
\begin{equation}
\begin{aligned}
\frac{P_{d+1,d+1}^{(m)}}{P_{dd}^{(m)}}
&=
\left(A_{d\to d+1}\gamma_{d+1}\right)^2
\frac{\Phi_{d+1}}{\Phi_d},
\\
\tau_{dd}^{(m)}
&=
\frac{1}{2\Phi_d},
\\
\tau_{d+1,d+1}^{(m)}-\tau_{dd}^{(m)}
&=
\frac{1}{2\Phi_{d+1}}
-
\frac{1}{2\Phi_d}.
\end{aligned}
\label{eq:noise_auto_chain}
\end{equation}

Unlike the cross-covariance metrics, these relations generally retain
information about the temporal width accumulated along the preceding
path. Nevertheless, the local multiplicative and additive propagation
laws are approximately recovered when the response at node $d$ is
dominated by a single upstream timescale $\tau_{dd}^{(m)}$ and the newly
added local timescale is relatively fast,
$
\gamma_{d+1}\ll \tau_{dd}^{(m)}.
$
In this regime,
\[
\tau_{d+1,d+1}^{(m)}
\approx
\tau_{dd}^{(m)}+\gamma_{d+1},
\]
and consequently
\[
\frac{\Phi_{d+1}}{\Phi_d}
=
\frac{\tau_{dd}^{(m)}}{\tau_{d+1,d+1}^{(m)}}
\approx
\frac{\tau_{dd}^{(m)}}
{\tau_{dd}^{(m)}+\gamma_{d+1}}
\approx 1.
\]
Therefore,
\begin{equation}
\begin{aligned}
\tau_{d+1,d+1}^{(m)}-\tau_{dd}^{(m)}
&\approx \gamma_{d+1},
\\
\frac{P_{d+1,d+1}^{(m)}}{P_{dd}^{(m)}}
&\approx
\left(A_{d\to d+1}\gamma_{d+1}\right)^2.
\end{aligned}
\label{eq:noise_auto_local}
\end{equation}
Thus, when the newly introduced timescale is small relative to the
dominant timescale already accumulated along the path, the
autocovariance approximately recovers the same local propagation
picture: time constants add, while peak amplitudes multiply according
to the squared local gain. Otherwise, the temporal-overlap factor cannot be neglected.

\subsection*{Statistical description of the chain-wise propagation}
\label{sec:E:variance}

Above ``memoryless'' or basic ``memoryless'' behavior depends mainly on the current layer $d$ and the next layer $d+1$, with no regard for the preceding layers. The main modulation is governed by local variables $D_{d+1}$. In this view, the response metrics for a path of length $d$ take the form without loss of generality assuming unit linkage between chain $A_{i \to {i+1}}=1$ for all $i$:
\begin{equation}
\refstepcounter{equation}\label{iterative}
\tag{E\arabic{equation}}
\begin{aligned}
R\left(D_{1}, \cdots, D_{d}\right) &= f\left(D_{1}\right) \cdots f\left(D_{d}\right), \\
\tau\left(D_{1}, \cdots, D_{d}\right) &= f\left(D_{1}\right) + \cdots + f\left(D_{d}\right),
\end{aligned}
\end{equation}
where intrinsic time constant $f(D) = 1 / (\beta + D)$. Eq.~(\ref{iterative}) suggests that a higher average degree $\langle D \rangle$ leads to lower response values $R$ and $\tau$.

To isolate the effect of degree variance, we consider the case where the average degree is fixed:
\[
D_{1}+\cdots+D_{d}=d\bar{D}_d.
\]
Using the method of Lagrange multipliers with
\[
L = \tau - \lambda\left(\sum_i D_{i} - \bar{D}_d\right),
\]
and solving $\frac{\partial L}{\partial D_{i}} = f'\left(D_{i}\right) - \lambda = 0$, we obtain the condition
\[
\lambda = f^{\prime}(D_{i}) = -\left(\beta+D_{i}\right)^{-2}, \quad \forall i.
\]
Thus, the minimum value of $\tau$ is
\[
\tau = \frac{d}{\beta + \bar{D}_d}
\]
achieved when all degrees are equal: $D_{1} = \cdots = D_{d} = \bar{D}_d$, consistent with Jensen's inequality. The fixed estimation bias for the time constant does not alter the conclusion.

A similar analysis applies to the peak response $R$. Defining the Lagrangian
\[
L = R - \eta\left(\sum_i D_{i} - \bar{D}_d\right),
\]
and solving $\frac{\partial L}{\partial D_{i}} = f\left(D_{1}\right) \cdots f^{\prime}\left(D_{i}\right) \cdots f\left(D_{d}\right) - \eta = 0$, we find
\[
\eta = -\left(\beta + D_{i}\right)^{-1} R, \quad \forall i.
\]
This yields the minimum peak response
\[
R = \left(\frac{1}{\beta + \bar{D}_d}\right)^d
\]
for the homogeneous in-degree configuration. Since $\ln R$ is a convex function of $D$, this homogeneous configuration indeed minimizes the peak response.

Eq. \eqref{iterative} can alternatively be expressed as
\begin{equation}
\begin{aligned}
R(d) &= \left[\beta^{d}\left(d \sum_{n=1}^{d} \sum_{j=1}^n \frac{(-1)^{j-1}}{n} s^{(n-j)} \mu^{(j)}+1\right)\right]^{-1},\\
\tau(d)&=\frac{d\left[\sum_{n=1}^{d} \sum_{j=1}^n \frac{d-n}{n}(-1)^{j-1} s^{(n-j)} \mu^{(j)}+1\right]}{\beta\left[d \sum_{n=1}^{d} \sum_{j=1}^n \frac{(-1)^{j-1}}{n} s^{(n-j)} \mu^{(j)}+1\right]},
\end{aligned}
\end{equation}
where $s^{(k)} \equiv \sum_{1 \leqslant i_1 < \cdots < i_k \leqslant n} (D_{i_1}/\beta) \cdots (D_{i_k}/\beta)$ denote elementary symmetric polynomials and $\mu^{(k)} \equiv n^{-1} \sum_{i=1}^n (D_i/\beta)^k$ represent raw moments. For propagation only restricted to source $m$ and target $i$, the metrics simplify to
\begin{equation}
\begin{aligned}
R(D_m,D_i) &= \frac{1}{\beta^2} \frac{1}{(\mu+1)^2-\sigma^2}, \\
\tau(D_m,D_i) &= \frac{1}{\beta} \frac{2(\mu+1)}{(\mu+1)^2-\sigma^2},
\end{aligned}
\end{equation}
with $\mu \equiv \tfrac{1}{2}(D_m/\beta + D_i/\beta)$ and $\sigma^2 \equiv \tfrac{1}{2}[(D_m/\beta - \mu)^2 + (D_i/\beta - \mu)^2]$ being the mean and variance of $\{D_m, D_i\}$. These dependencies align with the general cases: increasing $\mu$ reduces both $P$ and $\tau$, while increasing $\sigma^2$ elevates both.
\begingroup
 
\begin{figure*}
\includegraphics[width=1.0\linewidth]{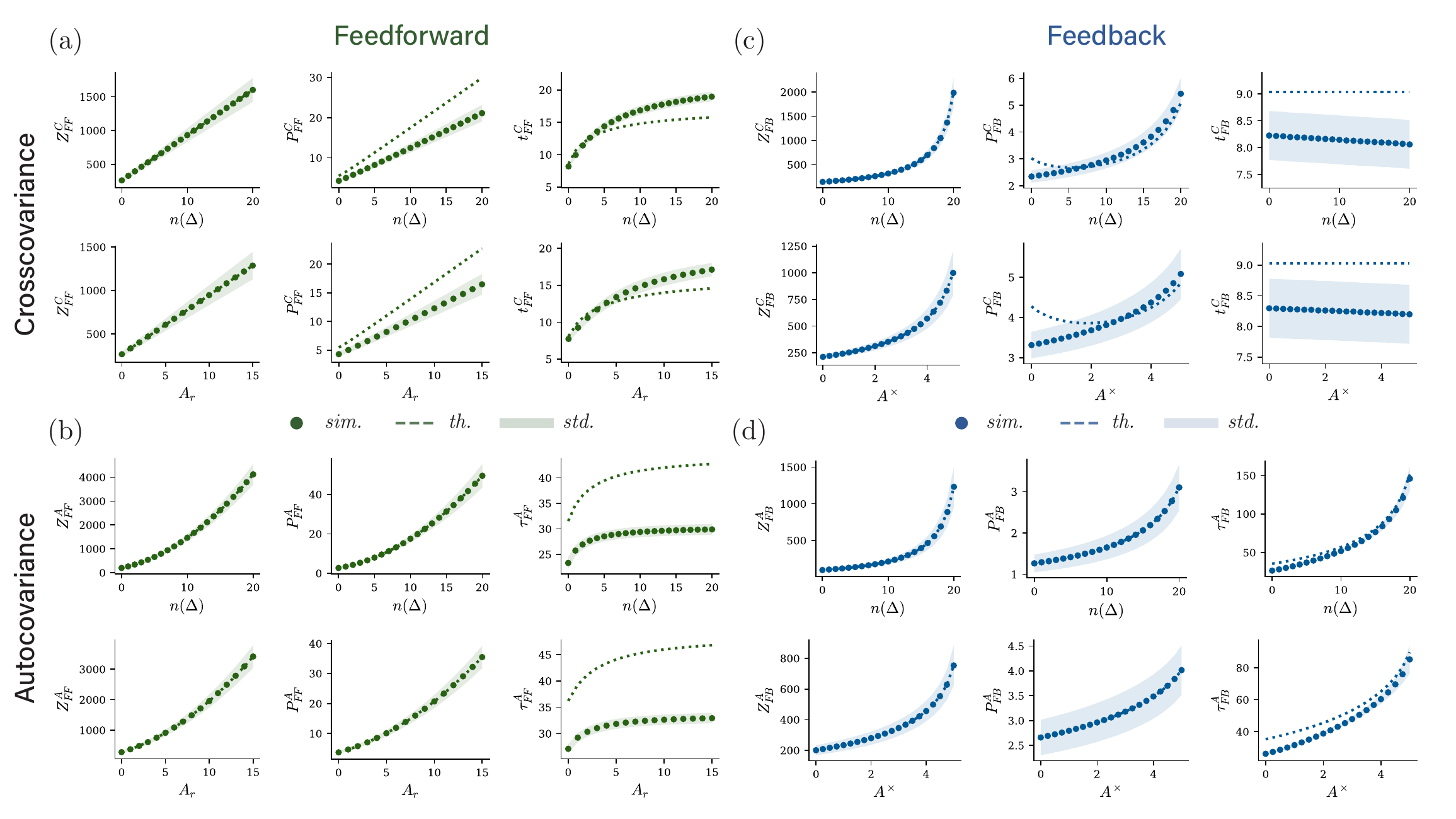}
\caption{\label{fig:six}
\textbf{White-noise input propagation under feedforward and feedback triangle motifs.}
Noise is applied to source node $m$, and responses are measured at node $i$.
For each motif, three covariance metrics are shown: amplification $Z$, peak response $P$, and timing, denoted by $t$ for crosscovariance peak time and $\tau$ for autocovariance time constant.
Metrics are plotted against the number of triangles $n(\Delta)$, the feedforward triangle weight ratio $A_r$, and the feedback triangle weight product $A^{\times}$.
Panels show crosscovariance for feedforward triangles (a), autocovariance for feedforward triangles (b), crosscovariance for feedback triangles (c), and autocovariance for feedback triangles (d).
Simulation results are shown as scatter points, shaded regions denote the standard error across $100$ random realizations, and dotted curves show the corresponding theoretical estimates.
For both feedforward and feedback motifs, strength metrics increase with triangle number and triangle weight, but with distinct scaling behavior.
Feedforward temporal metrics tend toward saturation at large triangle number or weight, whereas feedback triangles have little effect on the crosscovariance peak-response time but continuously increase the autocovariance decay time constant. Parameters: the intrinsic time constants are all fixed at $5$. The direct coupling from source node $m$ to response node $i$ is chosen from the uniform distribution between $0.25$ and $1.25$ and then kept fixed within the sweep. The two triangle-edge weights are chosen from the uniform distribution between $0.15$ and $0.75$ and then kept fixed within the sweep.
}
\end{figure*}
\subsection*{\textcolor{black}{Effects of motifs}}

The single chain is the baseline model, to which we can add motifs to observe their effects. In this subsection, we mainly focus on the effects of triangular motifs (feedforward and feedback triangles)\DIFaddbegin {, and extend that to the higher-order motifs}\DIFaddend .

\textbf{Effects of Feedforward (FF) Motifs.}
We first consider a three-node feedforward motif as a minimal example
for understanding how parallel paths of different lengths and intrinsic
timescales modify signal propagation. The motif contains a source node
$m$, a target node $i$, and $n(\Delta)$ equivalent intermediate nodes
$j$ forming the FF triangles. The local intrinsic timescale of node
$u$ is denoted by
$
\gamma_u\equiv(\beta_u+D_u)^{-1}.
$
Propagation from $m$ to $i$ occurs through a direct path
$m\to i$ and $n(\Delta)$ equivalent indirect paths
$m\to j\to i$. We define the relative structural weight of an indirect
path by
$
A_r
\equiv
({A_{m\to j}A_{j\to i}})/
     {A_{m\to i}},
$
and introduce the total feedforward gain
$
g_{\mathrm{FF}}
\equiv
n(\Delta)\gamma_j A_r$. Thus, $g_{\mathrm{FF}}$ jointly captures the number of FF triangles,
their relative coupling strength, and the intrinsic timescale of the
intermediate nodes.

\par For deterministic input, the direct and indirect path contributions
are
\[
R_1
=
\gamma_m\gamma_iA_{m\to i},
\qquad
R_2
=
\gamma_m\gamma_i\gamma_j
A_{m\to j}A_{j\to i}.
\]
The total response is therefore
\begin{equation}
\begin{aligned}
R_{\mathrm{FF}}^{(m)}
&=
R_1+n(\Delta)R_2
\\
&=
R_1(1+g_{\mathrm{FF}}).
\end{aligned}
\label{eq:FF_R}
\end{equation}
Hence, $g_{\mathrm{FF}}>1$ indicates that the combined contribution of
the longer FF paths exceeds that of the direct path.

The corresponding response time is the response-weighted average of
the two path times,
\begin{equation}
\begin{aligned}
\tau_{\mathrm{FF}}^{(m)}
&=
\gamma_m+\gamma_i
+
\frac{g_{\mathrm{FF}}}{1+g_{\mathrm{FF}}}\gamma_j
\\
&=
\gamma_m+\gamma_i
+
\gamma_j
\left(
1-\frac{1}{1+g_{\mathrm{FF}}}
\right).
\end{aligned}
\label{eq:FF_tau}
\end{equation}
It therefore satisfies
\[
\gamma_m+\gamma_i
<
\tau_{\mathrm{FF}}^{(m)}
<
\gamma_m+\gamma_i+\gamma_j.
\]
For weak FF gain, $g_{\mathrm{FF}}\ll1$, the direct path dominates,
whereas for $g_{\mathrm{FF}}\gg1$,
\[
R_{\mathrm{FF}}^{(m)}
\approx
R_1g_{\mathrm{FF}},
\qquad
\tau_{\mathrm{FF}}^{(m)}
\rightarrow
\gamma_m+\gamma_i+\gamma_j,
\]
so that the response approaches that of the longer two-step path.

\par For white-noise input at node $m$, the same feedforward gain
$g_{\mathrm{FF}}$ controls the stochastic amplification. Using the
Lyapunov identities, the autocovariance and source--target
cross-covariance amplification metrics reduce exactly to
\begin{equation}
\begin{aligned}
Z_{ii}^{(m)}
&=
\left(R_{\mathrm{FF}}^{(m)}\right)^2
=
R_1^2(1+g_{\mathrm{FF}})^2,
\\[4pt]
Z_{im}^{(m)}
&=
\gamma_mR_{\mathrm{FF}}^{(m)}
=
\gamma_mR_1(1+g_{\mathrm{FF}}),
\end{aligned}
\label{eq:FF_noise_Z}
\end{equation}
where unit noise intensity is assumed. Thus, FF triangles produce a
quadratic amplification of $Z_{ii}^{(m)}$ and a linear amplification
of $Z_{im}^{(m)}$ through the same factor $g_{\mathrm{FF}}$ (Fig.~\ref{fig:six}(a,b)).

\par The source--target cross-covariance time is determined directly by
the deterministic response times:
\begin{equation}
\begin{aligned}
t_{im}^{(m)}
&=
\tau_{\mathrm{FF}}^{(m)}-\gamma_m
\\
&=
\gamma_i
+
\frac{g_{\mathrm{FF}}}{1+g_{\mathrm{FF}}}\gamma_j.
\end{aligned}
\label{eq:FF_tim}
\end{equation}
where large $g_{FF}$ leads to saturation (Fig.~\ref{fig:six}(a)).

The temporal variance of the target response is
\begin{equation}
\left(\sigma_i^{(m)}\right)^2
=
\gamma_m^2+\gamma_i^2
+
\gamma_j^2
\frac{
g_{\mathrm{FF}}(2+g_{\mathrm{FF}})
}{
(1+g_{\mathrm{FF}})^2
}.
\label{eq:FF_sigma}
\end{equation}
The temporal variance increases monotonically with $g_{\mathrm{FF}}$
and saturates at
\[
\left(\sigma_i^{(m)}\right)^2
\rightarrow
\gamma_m^2+\gamma_i^2+\gamma_j^2,
\qquad
g_{\mathrm{FF}}\rightarrow\infty.
\]
Thus, for fixed $\gamma_j$, increasing either the number of triangles
$n(\Delta)$ or their relative weight $A_r$ progressively broadens the
response, but this broadening saturates once the indirect paths dominate.

Therefore, the source--target covariance peak is
\begin{equation}
P_{im}^{(m)}
=
\frac{
\gamma_mR_1(1+g_{\mathrm{FF}})
}{
\sqrt{
2\left[
2\gamma_m^2+\gamma_i^2
+
\gamma_j^2
\frac{
g_{\mathrm{FF}}(2+g_{\mathrm{FF}})
}{
(1+g_{\mathrm{FF}})^2
}
\right]
}
}.
\label{eq:FF_Pim}
\end{equation}
For fixed $\gamma_j$, the numerator increases linearly with
$g_{\mathrm{FF}}$, whereas the temporal-broadening
term in the denominator saturates. Consequently, for sufficiently
strong FF gain,
\[
P_{im}^{(m)}
\sim
\frac{
\gamma_mR_1
}{
\sqrt{2(2\gamma_m^2+\gamma_i^2+\gamma_j^2)}
}
g_{\mathrm{FF}},
\]
so that the source--target covariance peak approaches an approximately
linear dependence on $n(\Delta)$ and $A_r$ (Fig.~\ref{fig:six}(a)). 

For autocovariance, let
$
w_0=(m,i),
\
w_\Delta=(m,j,i)
$
denote the direct and indirect paths, respectively, and define
$
\Phi_{00}\equiv\Phi(w_0,w_0),
\
\Phi_{0\Delta}\equiv\Phi(w_0,w_\Delta), 
\Phi_{\Delta\Delta}\equiv\Phi(w_\Delta,w_\Delta).
$ For a fixed pair of path types and fixed intrinsic time constants, the temporal-overlap factors are fixed.
Since all $n(\Delta)$ indirect paths are assumed dynamically
equivalent, the target autocovariance becomes
\begin{equation}
P_{ii}^{(m)}
=
R_1^2
\left[
\Phi_{00}
+
2g_{\mathrm{FF}}\Phi_{0\Delta}
+
g_{\mathrm{FF}}^2\Phi_{\Delta\Delta}
\right].
\label{eq:FF_Pii}
\end{equation}
For fixed intrinsic timescales, $P_{ii}^{(m)}$ increases monotonically
with $g_{\mathrm{FF}}$. In particular,
\[
P_{ii}^{(m)}
\sim
R_1^2\Phi_{\Delta\Delta}g_{\mathrm{FF}}^2,
\qquad
g_{\mathrm{FF}}\gg1,
\]
showing a quadratic, and hence superlinear, dependence on the number
$n(\Delta)$ and relative weight $A_r$ of the FF paths when the indirect
paths dominate (Fig.~\ref{fig:six}(b)).

Consequently,
\begin{equation}
\tau_{ii}^{(m)}
=
\frac{
(1+g_{\mathrm{FF}})^2
}{
2\left[
\Phi_{00}
+
2g_{\mathrm{FF}}\Phi_{0\Delta}
+
g_{\mathrm{FF}}^2\Phi_{\Delta\Delta}
\right]
}.
\label{eq:FF_tauii}
\end{equation}
As $g_{\mathrm{FF}}$ increases, the autocovariance time continuously
shifts from that of the direct path toward that of the indirect path,
with
\[
\tau_{ii}^{(m)}
\rightarrow
\frac{1}{2\Phi_{\Delta\Delta}},
\qquad
g_{\mathrm{FF}}\rightarrow\infty.
\]
Hence, increasing $n(\Delta)$ or $A_r$ mainly changes the relative
contribution of the two path classes and leads to saturation once the
indirect paths dominate (Fig.~\ref{fig:six}(b)). 

Thus, unlike $Z_{ii}^{(m)}$, the autocovariance peak and time constant
depend not only on the gain
$g_{\mathrm{FF}}=n(\Delta)\gamma_jA_r$, but also on the temporal overlap
between the direct and indirect paths. Increasing $n(\Delta)$ or $A_r$
mainly strengthens the indirect contribution, producing asymptotically
quadratic growth of $Z_{ii}^{(m)}$ and $P_{ii}^{(m)}$, approximately linear growth of
$P_{im}^{(m)}$ and $Z_{im}^{(m)}$, and saturation of the temporal metrics once the
indirect paths dominate. 

\par We next generalize the three-node motif by embedding the same FF
structure into an otherwise arbitrary directed chain. Let $R(d)$ and
$\tau(d)$ denote the response amplitude and time of the original chain. The deterministic metrics become
\begin{equation}
\begin{aligned}
R_{\mathrm{FF}}(d)
&=
R(d)(1+g_{\mathrm{FF}}),
\\[4pt]
\tau_{\mathrm{FF}}(d)
&=
\tau(d)
+
\frac{g_{\mathrm{FF}}}{1+g_{\mathrm{FF}}}
\gamma_\Delta.
\end{aligned}
\label{eq:FF_general}
\end{equation}
Therefore, the effect of an FF motif can be summarized by the single effective gain. The number of triangles and their coupling strength increase the
effective feedforward gain, whereas a short intermediate timescale
suppresses it. Since
$\gamma_\Delta=(\beta_\Delta+D_\Delta)^{-1}$, increasing the local
in-degree $D_\Delta$ decreases $g_{\mathrm{FF}}$ and drives
\[
R_{\mathrm{FF}}(d)\rightarrow R(d),
\qquad
\tau_{\mathrm{FF}}(d)\rightarrow\tau(d),
\]
recovering the previously identified jamming effect of nodal in-degree. While increasing the gain factor continuously increases the peak response linearly, it saturates the time constant for longer chains, yielding
$
\tau_{\mathrm{FF}}(d) + \gamma_\Delta.
$
For stochastic inputs, the same gain completely determines the
amplification metrics $Z$, while covariance peaks and autocovariance
times may additionally retain the temporal-overlap structure of the
competing paths.

\textbf{Effects of Feedback (FB) Motifs.}
We next consider feedback (FB) triangles. Unlike FF motifs, a feedback
loop generates an infinite family of propagation walks whose lengths
increase by repeated traversals of the loop. We first consider the
three-node motif consisting of the direct edge $m\to i$ and the
feedback loop
$
m\to i\to j\to m.
$
For $n(\Delta)$ equivalent FB triangles, we define the weight and
intrinsic-timescale products of one feedback loop as
\[
A_\Delta^{\times}
\equiv
A_{m\to i}A_{i\to j}A_{j\to m},
\qquad
\Gamma_\Delta^{\times}
\equiv
\gamma_m\gamma_i\gamma_j,
\]
and its accumulated local timescale as
$
\tau_\Delta^{+}
\equiv
\gamma_m+\gamma_i+\gamma_j.
$
The effective feedback gain is then
\begin{equation}
g_{\mathrm{FB}}
\equiv
n(\Delta)\Gamma_\Delta^{\times}A_\Delta^{\times}
=
n(\Delta)\gamma_m\gamma_i\gamma_j
A_{m\to i}A_{i\to j}A_{j\to m}.
\label{eq:gFB}
\end{equation}
Thus, $g_{\mathrm{FB}}$ jointly captures the number of feedback
triangles, their loop weights, and the intrinsic timescales of the
nodes participating in the recurrent loop.

\par For deterministic input at node $m$, the direct response is
$
R_1
=
\gamma_m\gamma_iA_{m\to i}.
$
Each additional traversal of the feedback loop multiplies the walk
contribution by $g_{\mathrm{FB}}$. The total response therefore forms
a geometric series,
\begin{equation}
\begin{aligned}
R_{\mathrm{FB}}^{(m)}
&=
R_1
\left(
1+g_{\mathrm{FB}}+g_{\mathrm{FB}}^2+\cdots
\right)
\\
&=
\frac{R_1}{1-g_{\mathrm{FB}}},
\qquad
|g_{\mathrm{FB}}|<1.
\end{aligned}
\label{eq:FB_R}
\end{equation}
The condition $|g_{\mathrm{FB}}|<1$ ensures convergence of the
walk expansion; for positive feedback, $g_{\mathrm{FB}}\to1^{-}$
corresponds to approaching the stability boundary.

Each additional loop traversal also contributes the time
$\tau_\Delta^{+}$. Consequently, the response time is
\begin{equation}
\begin{aligned}
\tau_{\mathrm{FB}}^{(m)}
&=
\gamma_m+\gamma_i
+
\tau_\Delta^{+}
\frac{g_{\mathrm{FB}}}{1-g_{\mathrm{FB}}}
\\
&=
-\gamma_j
+
\frac{
\gamma_m+\gamma_i+\gamma_j
}{
1-g_{\mathrm{FB}}
}.
\end{aligned}
\label{eq:FB_tau}
\end{equation}
Thus, both response amplitude and response time increase with positive
feedback. For weak feedback, $g_{\mathrm{FB}}\ll1$,
$
R_{\mathrm{FB}}^{(m)}
\approx
R_1(1+g_{\mathrm{FB}}),
\ 
\tau_{\mathrm{FB}}^{(m)}
\approx
\gamma_m+\gamma_i
+
g_{\mathrm{FB}}\tau_\Delta^{+},
$
whereas both quantities grow superlinearly as $g_{\mathrm{FB}}\to1^{-}$.

\par For white-noise input at node $m$, the Lyapunov identities directly connect the stochastic amplification metrics to the corresponding deterministic responses. Since
\[
R_m^{(m)}
=
\frac{\gamma_m}{1-g_{\mathrm{FB}}},
\qquad
R_i^{(m)}
=
\frac{R_1}{1-g_{\mathrm{FB}}},
\]
the autocovariance and source--target amplification metrics are
\begin{equation}
\begin{aligned}
Z_{ii}^{(m)}
&=
\left(R_i^{(m)}\right)^2
=
\frac{R_1^2}{(1-g_{\mathrm{FB}})^2},
\\[4pt]
Z_{im}^{(m)}
&=
R_i^{(m)}R_m^{(m)}
=
\frac{\gamma_m R_1}{(1-g_{\mathrm{FB}})^2}.
\end{aligned}
\label{eq:FB_noise_Z}
\end{equation}
Thus, both stochastic amplification metrics inherit the feedback gain quadratically, scaling as $(1-g_{\mathrm{FB}})^{-2}$ and diverging as the positive feedback gain approaches the stability boundary, $g_{\mathrm{FB}}\to 1^{-}$ (Fig.~\ref{fig:six}(c,d)).

The deterministic response times at the source and target are
\begin{align}
\tau_m^{(m)}
&=
\gamma_m
+
\tau_\Delta^{+}
\frac{g_{\mathrm{FB}}}{1-g_{\mathrm{FB}}},
\\
\tau_i^{(m)}
&=
\gamma_m+\gamma_i
+
\tau_\Delta^{+}
\frac{g_{\mathrm{FB}}}{1-g_{\mathrm{FB}}}.
\end{align}
The recurrent loop therefore contributes the same additional delay to both nodes. Their difference is consequently independent of the feedback gain,
\begin{equation}
t_{im}^{(m)}
=
\tau_i^{(m)}-\tau_m^{(m)}
=
\gamma_i.
\label{eq:FB_tim}
\end{equation}
Hence, although feedback increases the absolute response times of both source and target, their relative cross-covariance timing remains controlled solely by the local intrinsic timescale $\gamma_i$ (Fig.~\ref{fig:six}(c); Fig. S50).

The temporal variances are
\begin{align}
\left(\sigma_i^{(m)}\right)^2
&=
\gamma_m^2+\gamma_i^2
+
\frac{g_{\mathrm{FB}}}{1-g_{\mathrm{FB}}}V_\Delta
+
\frac{g_{\mathrm{FB}}}{(1-g_{\mathrm{FB}})^2}
\left(\tau_\Delta^{+}\right)^2,
\\
\left(\sigma_m^{(m)}\right)^2
&=
\gamma_m^2
+
\frac{g_{\mathrm{FB}}}{1-g_{\mathrm{FB}}}V_\Delta
+
\frac{g_{\mathrm{FB}}}{(1-g_{\mathrm{FB}})^2}
\left(\tau_\Delta^{+}\right)^2,
\end{align}
where
$
V_\Delta
=
\gamma_m^2+\gamma_i^2+\gamma_j^2.
$
Accordingly, the source--target cross-covariance peak is
\begin{equation}
P_{im}^{(m)}
=
\frac{
\gamma_m R_1/(1-g_{\mathrm{FB}})^2
}{
\sqrt{
2\left[
\left(\sigma_i^{(m)}\right)^2
+
\left(\sigma_m^{(m)}\right)^2
\right]
}
}.
\label{eq:FB_Pim}
\end{equation}
When the recurrent contribution dominates, $g_{\mathrm{FB}}\to 1^{-}$, the temporal variances scale as $(1-g_{\mathrm{FB}})^{-2}$. The denominator of Eq.~(\ref{eq:FB_Pim}) therefore scales as $(1-g_{\mathrm{FB}})^{-1}$ (Fig.~\ref{fig:six}(c)), yielding
\[
P_{im}^{(m)}
\sim
\frac{1}{1-g_{\mathrm{FB}}}.
\]
Thus, feedback amplifies the integrated covariance more strongly than its peak:
\[
Z_{im}^{(m)}
\sim
\frac{1}{(1-g_{\mathrm{FB}})^2},
\qquad
P_{im}^{(m)}
\sim
\frac{1}{1-g_{\mathrm{FB}}},
\]
because the growing covariance is simultaneously broadened in time (Fig.~\ref{fig:six}(c)).

\par The target autocovariance contains contributions from all pairs of recurrent walks. Let $w_k$ denote the walk from $m$ to $i$ containing $k$ traversals of the feedback loop, whose deterministic contribution is $R_1 g_{\mathrm{FB}}^k$. The zero-lag autocovariance is then
\begin{equation}
P_{ii}^{(m)}
=
R_1^2
\sum_{k,l=0}^{\infty}
g_{\mathrm{FB}}^{k+l}
\Phi(w_k,w_l),
\label{eq:FB_Pii}
\end{equation}
where $\Phi(w_k,w_l)$ measures the temporal overlap between the two recurrent walks. The associated autocovariance time is
\begin{equation}
\tau_{ii}^{(m)}
=
\frac{
1
}{
2(1-g_{\mathrm{FB}})^2
\displaystyle
\sum_{k,l=0}^{\infty}
g_{\mathrm{FB}}^{k+l}
\Phi(w_k,w_l)
}.
\label{eq:FB_tauii}
\end{equation}
Unlike $Z_{ii}^{(m)}$, these quantities therefore retain information about the temporal overlap among repeated loop traversals and cannot, in general, be expressed solely through the geometric feedback factor.

A simpler scaling emerges when positive feedback is sufficiently strong that the recurrent mode dominates. As $g_{\mathrm{FB}}\to 1^{-}$,
\begin{equation}
P_{ii}^{(m)}
\sim
\frac{1}{1-g_{\mathrm{FB}}},
\qquad
\tau_{ii}^{(m)}
\sim
\frac{1}{1-g_{\mathrm{FB}}},
\label{eq:FB_noise_asymptotic}
\end{equation}
whereas
\[
Z_{ii}^{(m)}
\sim
\frac{1}{(1-g_{\mathrm{FB}})^2}.
\]
Thus, approaching the feedback-dominated regime produces two coupled effects: the autocovariance peak increases and its temporal profile broadens, each with leading scaling $(1-g_{\mathrm{FB}})^{-1}$, while the integrated amplification grows faster, as $(1-g_{\mathrm{FB}})^{-2}$ ((Fig.~\ref{fig:six}(d)). Away from this regime, the detailed temporal-overlap factors $\Phi(w_k,w_l)$ remain important, and the feedback gain alone is insufficient to determine the covariance peak or timescale.

\par We next generalize the same FB motif to an otherwise directed feedforward chain. Let $R(d)$ and $\tau(d)$ denote the response amplitude and response time at distance $d$ in the corresponding chain without feedback. Repeated traversals of the feedback loop give
\begin{equation}
\begin{aligned}
R_{\mathrm{FB}}(d)
&=
\frac{R(d)}{1-g_{\mathrm{FB}}},
\\[4pt]
\tau_{\mathrm{FB}}(d)
&=
\tau(d)
+
\frac{g_{\mathrm{FB}}}{1-g_{\mathrm{FB}}}
\tau_\Delta^{+},
\qquad
|g_{\mathrm{FB}}|<1.
\end{aligned}
\label{eq:FB_general}
\end{equation}
For white-noise input, the corresponding target autocovariance amplification follows directly from the deterministic response,
\begin{equation}
Z_{ii,\mathrm{FB}}^{(m)}
=
\frac{
Z_{ii}^{(m)}
}{
(1-g_{\mathrm{FB}})^2
}.
\label{eq:FB_general_Zii}
\end{equation}
More generally,
\[
Z_{ij,\mathrm{FB}}^{(m)}
=
R_{i,\mathrm{FB}}^{(m)}
R_{j,\mathrm{FB}}^{(m)},
\]
so the power of the feedback factor depends on whether the recurrent loop contributes to the propagation paths of one or both observed nodes. By contrast, covariance peaks and autocovariance times additionally retain the temporal-overlap structure generated by repeated loop traversals. Therefore, $g_{\mathrm{FB}}$ provides the dominant control parameter for stochastic amplification, while temporal-overlap effects determine how this amplification is distributed over time.

\textbf{FF vs FB.} Across different input types, FF and FB motifs produce qualitatively
distinct propagation effects. For FF triangles, increasing the effective
gain $g_{\mathrm{FF}}$ generally enhances both response strength and time. The peak-response metrics increase continuously with
the number and relative weight of the FF triangles, whereas the associated
time metrics increase toward the characteristic timescale of the longer
indirect path and therefore saturate as the FF contribution becomes
dominant. In contrast, FB triangles generate recurrent geometric
amplification. Increasing $g_{\mathrm{FB}}$ produces a superlinear
increase in the deterministic response strength and time constant, as
well as in the stochastic amplification metrics, with increasingly
strong enhancement as $g_{\mathrm{FB}}$ approaches the stability
boundary. Crosscovariance peaks and autocovariance times exhibit the same
enhancement, although their detailed dependence also reflects
temporal broadening and overlap. The source--target cross-covariance
time is an exception: for the three-node FB motif, the feedback-induced
delay is shared by the source and target and therefore cancels from
their relative timing. 

\medskip

\textbf{Effects of Higher-order Motifs.}
The results for triangular motifs extend naturally to motifs of arbitrary
length. Consider a higher-order motif $\mathcal{M}_k$ containing $k$
edges and attached to an otherwise directed propagation chain. We define
the multiplicative intrinsic-timescale factor and the accumulated
timescale of the motif as
\[
\Gamma_{\mathcal{M}_k}^{\times}
\equiv
\prod_{u\in\mathcal{M}_k}\gamma_u,
\qquad
\tau_{\mathcal{M}_k}^{+}
\equiv
\sum_{u\in\mathcal{M}_k}\gamma_u.
\]

For a higher-order FF motif, let $A_r^{(\mathcal{M}_k)}$ denote the
weight of the alternative motif path relative to the corresponding
direct path, and define its effective gain as
\[
g_{\mathrm{FF}}^{(\mathcal{M}_k)}
\equiv
n(\mathcal{M}_k)
\Gamma_{\mathcal{M}_k}^{\times}
A_r^{(\mathcal{M}_k)}.
\]
The response metrics are then
\begin{equation}
\begin{aligned}
R_{\mathrm{FF},\mathcal{M}_k}(d)
&=
R(d)
\left[
1+g_{\mathrm{FF}}^{(\mathcal{M}_k)}
\right],
\\[4pt]
\tau_{\mathrm{FF},\mathcal{M}_k}(d)
&=
\tau(d)
+
\frac{
g_{\mathrm{FF}}^{(\mathcal{M}_k)}
}{
1+g_{\mathrm{FF}}^{(\mathcal{M}_k)}
}
\tau_{\mathcal{M}_k}^{+}.
\end{aligned}
\label{eq:higher_FF}
\end{equation}
Thus, increasing the number or relative strength of higher-order FF
paths continuously increases the response strength, whereas the
additional response time saturates at
$\tau_{\mathcal{M}_k}^{+}$ once the alternative paths dominate.

For a higher-order FB motif, let
$A_{\mathcal{M}_k}^{\times}$ denote the product of the edge weights
along one traversal of the feedback loop, and define
\[
g_{\mathrm{FB}}^{(\mathcal{M}_k)}
\equiv
n(\mathcal{M}_k)
\Gamma_{\mathcal{M}_k}^{\times}
A_{\mathcal{M}_k}^{\times}.
\]
Repeated traversals of the loop generate a geometric series, yielding
\begin{equation}
\begin{aligned}
R_{\mathrm{FB},\mathcal{M}_k}(d)
&=
\frac{
R(d)
}{
1-g_{\mathrm{FB}}^{(\mathcal{M}_k)}
},
\\[4pt]
\tau_{\mathrm{FB},\mathcal{M}_k}(d)
&=
\tau(d)
+
\frac{
g_{\mathrm{FB}}^{(\mathcal{M}_k)}
}{
1-g_{\mathrm{FB}}^{(\mathcal{M}_k)}
}
\tau_{\mathcal{M}_k}^{+},
\qquad
\left|g_{\mathrm{FB}}^{(\mathcal{M}_k)}\right|<1.
\end{aligned}
\label{eq:higher_FB}
\end{equation}
Hence, unlike the saturating temporal effect of FF motifs, positive FB
motifs recurrently amplify both response amplitude and response time in superlinear way.
Both quantities increase superlinearly with the feedback gain and grow
rapidly as
$g_{\mathrm{FB}}^{(\mathcal{M}_k)}\to1^{-}$, corresponding to approach
toward the stability boundary.

These expressions show that the effects of motifs of arbitrary length
can be summarized by two structural quantities: the multiplicative gain
$n(\mathcal{M}_k)\Gamma_{\mathcal{M}_k}^{\times}A^{\times}$ and the accumulated
timescale $\tau_{\mathcal{M}_k}^{+}$. Increasing the in-degree of nodes
within the motif decreases their local timescales
$\gamma_u=(\beta_u+D_u)^{-1}$ and therefore suppresses both quantities.
Consequently,
\[
R_{\mathcal{M}_k}(d)\rightarrow R(d),
\qquad
\tau_{\mathcal{M}_k}(d)\rightarrow\tau(d),
\]
as the motif-node in-degrees become large, also the intrinsic time constant becomes smaller, extending the nodal
in-degree jamming effect from triangular motifs to motifs of arbitrary
order.
\endgroup

\DIFaddbegin \section{{APPLICATION FOR PRIMATE CORTEX MODEL}}
\label{sec:F}
We apply the framework to a Wilson--Cowan-type large-scale primate cortical model, where each cortical area contains interacting excitatory and inhibitory populations and interareal coupling is given by empirically measured FLNs~\cite{chaudhuri2015large, li2022hierarchical, joglekar2018inter, li2025hierarchy}. The corresponding equation is 
\begin{align}
\tau_E \frac{d}{dt}\nu_E^i
&=
-\nu_E^i
+
\beta_E
\left[
I_{\mathrm{syn},E}^i
+
I_{\mathrm{ext},E}^i
\right]_+,
\\
\tau_I \frac{d}{dt}\nu_I^i
&=
-\nu_I^i
+
\beta_I
\left[
I_{\mathrm{syn},I}^i
+
I_{\mathrm{ext},I}^i
\right]_+,
\end{align} where
\begin{align}
I_{\mathrm{syn},E}^{i}
&=
(1+\eta h_i)
\left(
w_{EE}\nu_E^{i}
+
\mu_{EE}
\sum_{j=1}^{N}
FLN_{ij}\nu_E^{j}
\right)
-
w_{EI}\nu_I^{i},
\\
I_{\mathrm{syn},I}^{i}
&=
(1+\eta h_i)
\left(
w_{IE}\nu_E^{i}
+
\mu_{IE}
\sum_{j=1}^{N}
FLN_{ij}\nu_E^{j}
\right)
-
w_{II}\nu_I^{i}.
\end{align}
In the weak-GBA regime, the parameters are $\tau_E=20\,\mathrm{ms}$, $\tau_I=10\,\mathrm{ms}$, $\beta_E=0.066\,\mathrm{Hz/pA}$, $\beta_I=0.351\,\mathrm{Hz/pA}$, $w_{EE}=24.4\,\mathrm{pA/Hz}$, $w_{IE}=12.2\,\mathrm{pA/Hz}$, $w_{EI}=19.7\,\mathrm{pA/Hz}$, $w_{II}=12.5\,\mathrm{pA/Hz}$, $\mu_{EE}=33.7\,\mathrm{pA/Hz}$, $\mu_{IE}=25.5\,\mathrm{pA/Hz}$, and $\eta=0.68$; strong GBA differs only in $w_{EI}=25.2\,\mathrm{pA/Hz}$ and $\mu_{EE}=51.5\,\mathrm{pA/Hz}$~\cite{joglekar2018inter}. Previous work showed that weak GBA preserves temporal hierarchy but yields weak responses, whereas strong GBA amplifies responses while collapsing that hierarchy~\cite{joglekar2018inter}. Thus, even within the linear regime, small parameter changes can alter the obvious quantification of responses and even the response ordering (Figs.~S55 and S56), while in the nonlinear regime, large inputs may drive the system beyond the original attractor basin (Fig.~S57), whose transient responses are hard to capture. This motivates an explicit quantification of how structured heterogeneity jointly shapes response strength and time, which our framework achieves in the linear regime, while the nonlinear regime remains an open problem.

\par {The corresponding Jacobian \mbox{
$\mathbf{H}$
}
has an E--E/E--I/I--E/I--I block structure and lies outside the NSDDM regime, due to negative couplings, positive effective self-dynamics, and violations of strict diagonal dominance, although it remains Hurwitz stable. This dense, strongly heterogeneous system exhibits approximately log-normal weights and a hierarchy-related spectral organization~\mbox{
\cite{chaudhuri2015large, science.1238406, li2022hierarchical}}\hskip0pt
. Nevertheless, our theory agrees well with simulations (Fig.~\ref{fig:newadd}(c); Fig.~S59), yielding an explicit mapping from \mbox{
$\mathbf{H}$
}
to transient response strength and time. Combining the basic Wilson-Cowan motif analysis, we provide a first step toward disentangling the roles of local dynamics and network structure. 
}
\begingroup

Building on the motif analysis introduced in the main text, we validate the perturbative theory developed around the NSDDM regime in Appendix~\ref{sec:A} using a divergent feedforward motif. The state variables are ordered as
$
\mathbf{x}
=
(E_m,I_m,E_i,I_i,E_j,I_j)^{\top},
$
where the source region $m$ projects to the two downstream regions $i$ and $j$. For a unit constant input applied to $E_m$, the reference NSDDM matrix is
\begin{equation}
\mathbf{H}_0
=
\begin{pmatrix}
-\dfrac{1}{\tau_E^m}
&0&0&0&0&0
\\[5pt]
\dfrac{w_{IE}^{m}}{\tau_I^m}
&-\dfrac{1}{\tau_I^m}
&0&0&0&0
\\[5pt]
\dfrac{w_{EE}^{im}}{\tau_E^i}
&0
&-\dfrac{1}{\tau_E^i}
&0&0&0
\\[5pt]
\dfrac{w_{IE}^{im}}{\tau_I^i}
&0
&\dfrac{w_{IE}^{i}}{\tau_I^i}
&-\dfrac{1}{\tau_I^i}
&0&0
\\[5pt]
\dfrac{w_{EE}^{jm}}{\tau_E^j}
&0&0&0
&-\dfrac{1}{\tau_E^j}
&0
\\[5pt]
\dfrac{w_{IE}^{jm}}{\tau_I^j}
&0&0&0
&\dfrac{w_{IE}^{j}}{\tau_I^j}
&-\dfrac{1}{\tau_I^j}
\end{pmatrix}.
\label{eq:WC_H0_divergent}
\end{equation}
This matrix is NSDDM when
\begin{equation}
\begin{aligned}
w_{IE}^{m}<1,\qquad
w_{EE}^{im}<1,\qquad
w_{IE}^{im}+w_{IE}^{i}<1, \\
w_{EE}^{jm}<1,\qquad
w_{IE}^{jm}+w_{IE}^{j}<1.
\label{eq:WC_H0_NSDDM_conditions}
\end{aligned}
\end{equation}

We introduce local inhibitory couplings at the source and both downstream regions through
\begin{equation}
\mathbf{E}
=
\begin{pmatrix}
0&-\dfrac{w_{EI}^{m}}{\tau_E^m}&0&0&0&0
\\[5pt]
0&0&0&0&0&0
\\[5pt]
0&0&0&-\dfrac{w_{EI}^{i}}{\tau_E^i}&0&0
\\[5pt]
0&0&0&0&0&0
\\[5pt]
0&0&0&0&0&-\dfrac{w_{EI}^{j}}{\tau_E^j}
\\[5pt]
0&0&0&0&0&0
\end{pmatrix},
\qquad
\mathbf{H}=\mathbf{H}_0+\mathbf{E}.
\label{eq:WC_inhibitory_perturbation}
\end{equation}

For the reference system, the steady-state responses at the two downstream regions are
\begin{equation}
R_i^{(m)}(\mathbf{H}_0)
=
\tau_E^m w_{EE}^{im},
\qquad
R_j^{(m)}(\mathbf{H}_0)
=
\tau_E^m w_{EE}^{jm}.
\label{eq:WC_reference_target_responses}
\end{equation}
We assume $w_{EE}^{im}>w_{EE}^{jm}$, so that
$
R_i^{(m)}(\mathbf{H}_0)>R_j^{(m)}(\mathbf{H}_0).
$

After introducing the inhibitory couplings, the two responses become
\begin{align}
R_i^{(m)}(\mathbf{H})
&=
\frac{\tau_E^m}
{1+w_{EI}^{m}w_{IE}^{m}}
\,
\frac{
w_{EE}^{im}-w_{EI}^{i}w_{IE}^{im}
}{
1+w_{EI}^{i}w_{IE}^{i}
},
\label{eq:WC_Ri_signed}
\\[5pt]
R_j^{(m)}(\mathbf{H})
&=
\frac{\tau_E^m}
{1+w_{EI}^{m}w_{IE}^{m}}
\,
\frac{
w_{EE}^{jm}-w_{EI}^{j}w_{IE}^{jm}
}{
1+w_{EI}^{j}w_{IE}^{j}
}.
\label{eq:WC_Rj_signed}
\end{align}
Thus, source inhibition $w_{EI}^{m}$ rescales both downstream responses by the same factor and does not affect their ordering. Defining
\begin{equation}
q_i
=
\frac{
w_{EE}^{im}-w_{EI}^{i}w_{IE}^{im}
}{
1+w_{EI}^{i}w_{IE}^{i}
},
\qquad
q_j
=
\frac{
w_{EE}^{jm}-w_{EI}^{j}w_{IE}^{jm}
}{
1+w_{EI}^{j}w_{IE}^{j}
},
\end{equation}
the ordering $R_i^{(m)}>R_j^{(m)}$ is equivalent to
\begin{equation}
q_i>q_j.
\label{eq:WC_exact_target_order}
\end{equation}
For fixed $w_{EI}^{j}$, the exact condition on the inhibitory coupling at region $i$ is
\begin{equation}
w_{EI}^{i}
<
\frac{
w_{EE}^{im}-q_j
}{
w_{IE}^{im}+w_{IE}^{i}q_j
}
\equiv
w_{EI,\mathrm{exact}}^{i},
\label{eq:WC_exact_order_bound}
\end{equation}
provided that the denominator is positive. Equivalently,
\begin{equation}
w_{EI,\mathrm{exact}}^{i}
=
\frac{
w_{EE}^{im}-w_{EE}^{jm}
+
w_{EI}^{j}
\left(
w_{EE}^{im}w_{IE}^{j}
+
w_{IE}^{jm}
\right)
}{
w_{IE}^{im}
+
w_{EE}^{jm}w_{IE}^{i}
+
w_{EI}^{j}
\left(
w_{IE}^{im}w_{IE}^{j}
-
w_{IE}^{i}w_{IE}^{jm}
\right)
}.
\label{eq:WC_exact_order_bound_expanded}
\end{equation}

We next compare the condition with the general perturbative criterion from lemma~\ref{proof:perturb}. For the induced infinity norm $\eta
=
\left\|
\mathbf{H}_0^{-1}\mathbf{E}
\right\|_{\infty}$, 
\begin{equation}
\eta
=
\max
\left\{
\begin{aligned}
&w_{EI}^{m},
\\
&w_{IE}^{m}w_{EI}^{m},
\\
&w_{EE}^{im}w_{EI}^{m}+w_{EI}^{i},
\\
&\left(
w_{IE}^{im}+w_{IE}^{i}w_{EE}^{im}
\right)w_{EI}^{m}
+w_{IE}^{i}w_{EI}^{i},
\\
&w_{EE}^{jm}w_{EI}^{m}+w_{EI}^{j},
\\
&\left(
w_{IE}^{jm}+w_{IE}^{j}w_{EE}^{jm}
\right)w_{EI}^{m}
+w_{IE}^{j}w_{EI}^{j}
\end{aligned}
\right\}.
\label{eq:WC_eta_divergent}
\end{equation}
For a unit input, the response bound is determined by
\begin{equation}
R_{\max}
=
\left\|
-\mathbf{H}_0^{-1}
\right\|_{\infty},
\end{equation}
or explicitly,
\begin{equation}
R_{\max}
=
\max
\left\{
\begin{aligned}
&\tau_E^m,
\\
&w_{IE}^{m}\tau_E^m+\tau_I^m,
\\
&w_{EE}^{im}\tau_E^m+\tau_E^i,
\\
&\left(
w_{IE}^{im}+w_{IE}^{i}w_{EE}^{im}
\right)\tau_E^m
+w_{IE}^{i}\tau_E^i+\tau_I^i,
\\
&w_{EE}^{jm}\tau_E^m+\tau_E^j,
\\
&\left(
w_{IE}^{jm}+w_{IE}^{j}w_{EE}^{jm}
\right)\tau_E^m
+w_{IE}^{j}\tau_E^j+\tau_I^j
\end{aligned}
\right\}.
\label{eq:WC_Rmax_divergent}
\end{equation}
The baseline response gap between the two downstream regions is
\begin{equation}
\Delta_{ij}^{(m)}
=
\tau_E^m
\left(
w_{EE}^{im}-w_{EE}^{jm}
\right).
\label{eq:WC_baseline_gap}
\end{equation}
Therefore, the perturbative condition
$
2\epsilon_R<\Delta_{ij}^{(m)}
$,
with
$
\epsilon_R=R_{\max}[\eta/(1-\eta)],
$
is equivalent to
\begin{equation}
\eta<\Theta,
\qquad
\Theta
\equiv
\frac{
\Delta_{ij}^{(m)}
}{
2R_{\max}+\Delta_{ij}^{(m)}
}.
\label{eq:WC_theta}
\end{equation}

For nonnegative inhibitory couplings, a sufficient set of explicit conditions is
\begin{align}
w_{EI}^{m}
&<\Theta,
\\
w_{EI}^{i}
&<
\min
\left\{
\Theta-w_{EE}^{im}w_{EI}^{m},
\,
\frac{
\Theta-
\left(
w_{IE}^{im}+w_{IE}^{i}w_{EE}^{im}
\right)w_{EI}^{m}
}{
w_{IE}^{i}
}
\right\},
\label{eq:WC_pert_bound_i}
\\
w_{EI}^{j}
&<
\min
\left\{
\Theta-w_{EE}^{jm}w_{EI}^{m},
\,
\frac{
\Theta-
\left(
w_{IE}^{jm}+w_{IE}^{j}w_{EE}^{jm}
\right)w_{EI}^{m}
}{
w_{IE}^{j}
}
\right\}.
\label{eq:WC_pert_bound_j}
\end{align}
These conditions guarantee preservation of the ordering
$
R_i^{(m)}>R_j^{(m)}
$,
but they are sufficient rather than necessary, as Fig.~\ref{fig:bound} shows.

\begin{figure}
\includegraphics[width=0.8\linewidth]{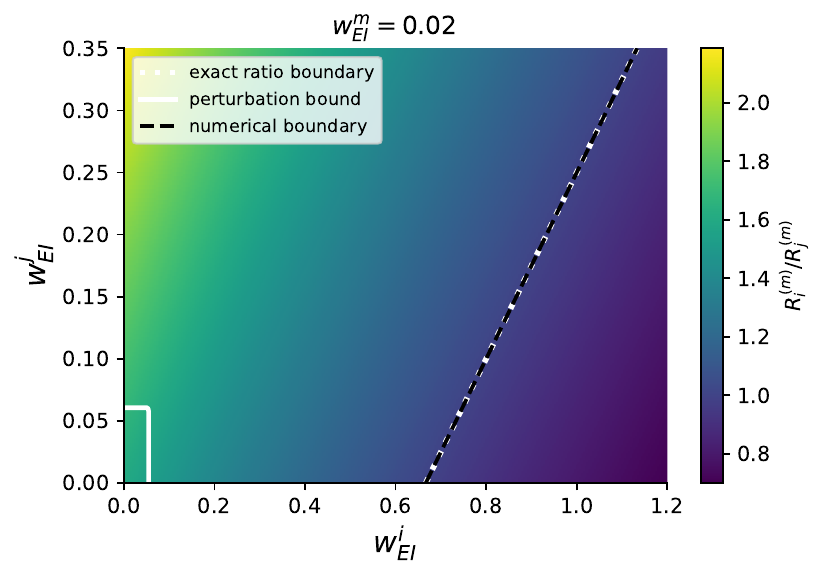}
\caption{\label{fig:bound}
\textbf{Ranking-order preservation boundary in the Wilson--Cowan model.}
The heat map shows the simulated steady-state response ratio
$R_i^{(m)}/R_j^{(m)}$ as the local inhibitory couplings
$w_{EI}^{i}$ and $w_{EI}^{j}$ are varied, with
$w_{EI}^{m}=0.02$ fixed. Values greater than $1$ indicate that the reference ordering
$R_i^{(m)}>R_j^{(m)}$ is preserved.
The dotted white curve is the exact boundary derived from
Eq.~\eqref{eq:WC_exact_order_bound}; for the parameters used here, it is
$w_{EI}^{i}=2/3+(4/3)w_{EI}^{j}$.
The dashed black curve is obtained numerically from
$R_i^{(m)}/R_j^{(m)}=1$ and coincides with the exact boundary.
The solid white curve denotes the sufficient perturbative boundary
$\|\mathbf{H}_0^{-1}\mathbf{E}\|_{\infty}=\Theta$ from 
Eq.~\eqref{eq:WC_theta}, with the guaranteed ordering-preservation region lying below and to the left of this curve.
Its smaller region illustrates the conservative perturbative criterion.
The reference NSDDM parameters are
$\tau_E^m=0.020$,
$\tau_I^m=0.010$,
$\tau_E^i=\tau_E^j=0.015$,
$\tau_I^i=\tau_I^j=0.020$,
$w_{IE}^{m}=0.50$,
$w_{EE}^{im}=0.80$,
$w_{EE}^{jm}=0.50$,
$w_{IE}^{im}=w_{IE}^{jm}=0.20$, and
$w_{IE}^{i}=w_{IE}^{j}=0.50$.
A unit constant input is applied at $E_m$.
}
\end{figure}

\clearpage
\DIFaddend 


\bibliography{apstemplate}
\end{document}